\documentclass[12pt]{article}
\usepackage{etex}
\usepackage[utf8]{inputenc}
\usepackage[TS1,T1]{fontenc}
\usepackage{textcomp}
\usepackage[a4paper,tmargin=3cm,bmargin=3cm,rmargin=2.2cm,lmargin=2.2cm]{geometry}
\usepackage{lmodern}

\usepackage{marvosym}

\usepackage{calc}
\usepackage[absolute]{textpos}
\textblockorigin{0mm}{0mm}

\usepackage{amsfonts,amssymb,amsmath,mathtools,cite,slashed}
\usepackage[amsthm,thmmarks,amsmath]{ntheorem}

\usepackage{titlesec}

\DeclareFontFamily{U}{mathx}{\hyphenchar\font45}
\DeclareFontShape{U}{mathx}{m}{n}{
     <5> <6> <7> <8> <9> <10>
     <10.95> <12> <14.4> <17.28> <20.74> <24.88>
     mathx10
     }{}
\DeclareSymbolFont{mathx}{U}{mathx}{m}{n}
\DeclareFontSubstitution{U}{mathx}{m}{n}
\DeclareMathAccent{\widecheck}    {0}{mathx}{"71}

\usepackage{graphicx}
\usepackage[all]{xy}
\usepackage{lmodern}
\usepackage[french,english]{babel}
\usepackage[pdftex]{hyperref}

\usepackage{tikz}
\usetikzlibrary{shapes.misc,arrows,decorations.markings}

\DeclareMathOperator{\Ima}{Im}            
\DeclareMathOperator{\Dom}{Dom}      
\DeclareMathOperator{\Ker}{Ker}           
\DeclareMathOperator{\Lin}{Lin}             
\DeclareMathOperator{\Res}{Res}         
\DeclareMathOperator{\spec}{sp}           
\DeclareMathOperator{\supp}{supp}      
\DeclareMathOperator{\Tad}{Tad}          
\DeclareMathOperator{\Tr}{Tr}                 
\DeclareMathOperator{\Vol}{\text{Vol}}   

\newtheorem{assumption}{Assumption}[section]
\newtheorem{theorem}[assumption]{Theorem}
\newtheorem{thm}[assumption]{Theorem}
\newtheorem{cly}[assumption]{Corollary}
\newtheorem{corollary}[assumption]{Corollary}
\newtheorem{conjecture}[assumption]{Conjecture}
\newtheorem{lem}[assumption]{Lemma}
\newtheorem{lemma}[assumption]{Lemma}
\newtheorem{defn}[assumption]{Definition}
\newtheorem{definition}[assumption]{Definition}
\newtheorem{prop}[assumption]{Proposition}
\newtheorem{rem}[assumption]{Remark}
\newtheorem{remark}[assumption]{Remark}
\newtheorem{example}[assumption]{Example}

\newcommand{\Sp}{{\rm Sp}}

\newcommand{\A}{\mathcal{A}}               
\newcommand{\Aun}{\widetilde{\mathcal{A}}}   
\renewcommand{\a}{\alpha}                     
\newcommand{\as}{\quad\mbox{as}\enspace}     
\renewcommand{\b}{\beta}                        
\newcommand{\B}{\mathcal{B}}               
\newcommand{\C}{\mathbb{C}}               
\newcommand{\CC}{\mathcal{C}}            
\newcommand{\Coo}{C^\infty}                  
\newcommand{\del}{\partial}                     
\newcommand{\delslash}{{\del\mkern-9.5mu/\,}} 
\newcommand{\DD}{\mathcal{D}}            
\newcommand{\DDD}{\vert \mathcal{D} \vert}  
\newcommand{\Ds}{{D\mkern-11.5mu/\,}} 
\newcommand{\Dslash}{{D\mkern-11.5mu/\,}} 
\newcommand{\eps}{\varepsilon}             
\newcommand{\E}{\mathcal{E}}                
\newcommand{\F}{\mathcal{F}}                 
\newcommand{\G}{\mathcal{G}}                 
\newcommand{\Ga}{\Gamma}                   
\newcommand{\ga}{\gamma}                    
\newcommand{\grad}{\text{grad }}            
\renewcommand{\H}{\mathcal{H}}            
\newcommand{\half}{{\mathchoice{\thalf}{\thalf}{\shalf}{\shalf}}}
\newcommand{\hookto}{\hookrightarrow}        
\newcommand{\I}{\makebox[0.5cm][c]{$\int\!\!\!\!\!-$}} 
\newcommand{\K}{\mathcal{K}}                 
\renewcommand{\L}{\mathcal{L}}              
\newcommand{\La}{\Lambda}                    
\newcommand{\la}{\lambda}                      
\newcommand{\<}{\langle}

\newcommand{\mop}{{\mathchoice{\mathbin{\star_{_\theta}}}F
       {\mathbin{\star_{_\theta}}}                   
       {{\star_\theta}}{{\star_\theta}}}}            
\newcommand{\M}{\mathcal{M}}              
\renewcommand{\mop}{{\mathchoice{\mathbin{\star_{_\theta}}}
       {\mathbin{\star_{_\theta}}}           
       {{\star_\theta}}{{\star_\theta}}}}    
\newcommand{\N}{\mathbb{N}}               
\newcommand{\nb}{\nabla}                     
\newcommand{\Oh}{\mathcal{O}}                
\newcommand{\om}{\omega}                   
\newcommand{\ox}{\otimes}                      
\newcommand{\pd}[2]{\frac{\partial#1}{\partial#2}}
\newcommand{\piso}[1]{\lfloor#1\rfloor}  
\newcommand{\PsiDO}{\Psi\mathrm{DO}}         
\newcommand{\Q}{\mathbb{Q}}                
\newcommand{\R}{\mathbb{R}}                

\newcommand{\sepword}[1]{\quad\mbox{#1}\quad} 
\newcommand{\set}[1]{\{\,#1\,\}}               
\newcommand{\shalf}{{\scriptstyle\frac{1}{2}}} 
\renewcommand{\SS}{\mathcal{S}}         
\newcommand{\Spin}{\text{Spin}}            
\newcommand{\SO}{\text{SO}}                  
\newcommand{\Th}{\Theta}                       
\renewcommand{\th}{\theta}                      
\newcommand{\T}{\mathbb{T}}                 
\newcommand{\thalf}{\tfrac{1}{2}}             
\newcommand{\tr}{\text{tr}}                         
\newcommand{\U}{\mathcal{U}}               
\newcommand{\x}{\times}                           
\newcommand\exter{{\textstyle\bigwedge}} 
\newcommand{\wh}{\widehat}                   
\newcommand{\wt}{\widetilde}                  
\newcommand{\Z}{\mathbb{Z}}                  
\renewcommand{\.}{\cdot}                           
\renewcommand{\:}{\colon}                         
\newcommand{\vc}{\vcentcolon =}             

\def\ee_#1{e_{{\scriptscriptstyle#1}}}       
\def\<#1,#2>{\langle#1\,,\,#2\rangle}          
\newcommand{\norm}[1]{\left\lVert#1\right\rVert}    

\newcommand{\be}{\begin{enumerate}}

\def\Xint#1{\mathchoice
	{\XXint\displaystyle\textstyle{#1}}%
	{\XXint\textstyle\scriptstyle{#1}}%
	{\XXint\scriptstyle\scriptscriptstyle{#1}}%
	{\XXint\scriptscriptstyle\scriptscriptstyle{#1}}%
	\!\int}
\def\XXint#1#2#3{{\setbox0=\hbox{$#1{#2#3}{\int}$}
	\vcenter{\hbox{$#2#3$}}\kern-0.5\wd0}}
\def\ncint{\Xint-}

\newcommand{\sg}{\sigma}                              

\newcommand{\makeletterhead}{
\begin{textblock*}{4.5cm}(\paperwidth-6cm,1cm)%
\setlength{\parindent}{0mm}%
\href{http://\urlLabo}{\includegraphics[width=4.3cm]{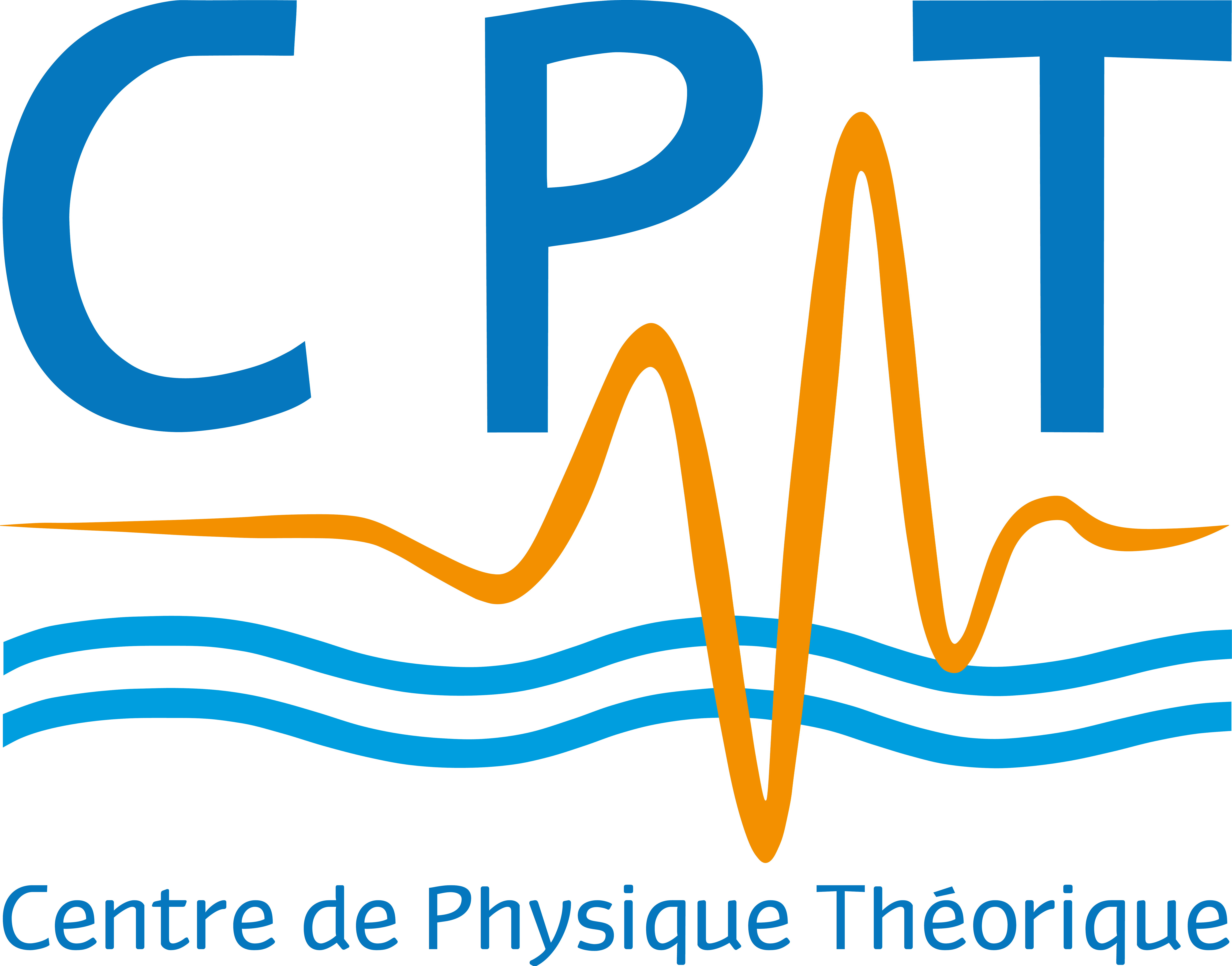}}%
\end{textblock*}
}

\newcommand{\makeletterfoot}{
\begin{textblock*}{21cm}(0cm,\paperheight-2cm)
\setlength{\parindent}{0mm}%
\setlength{\tabcolsep}{8pt}%
\centering\begin{tabular}{cccc}
\includegraphics[height=1.2cm]{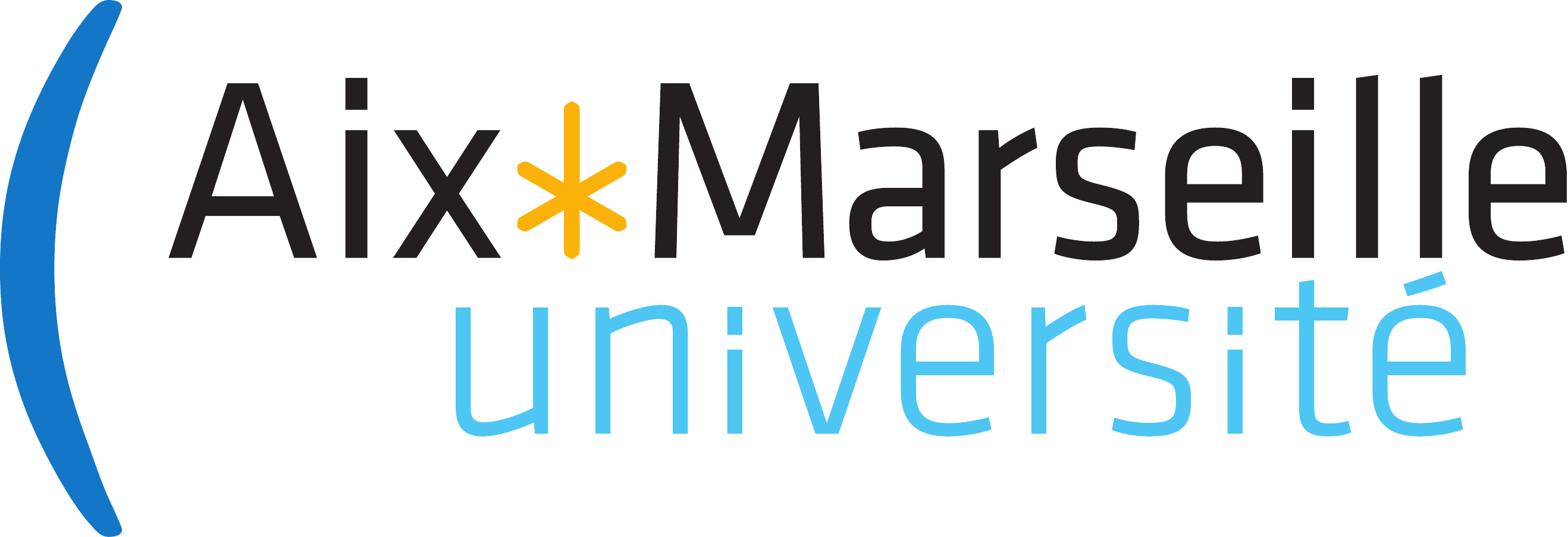}
&
\includegraphics[height=1.2cm]{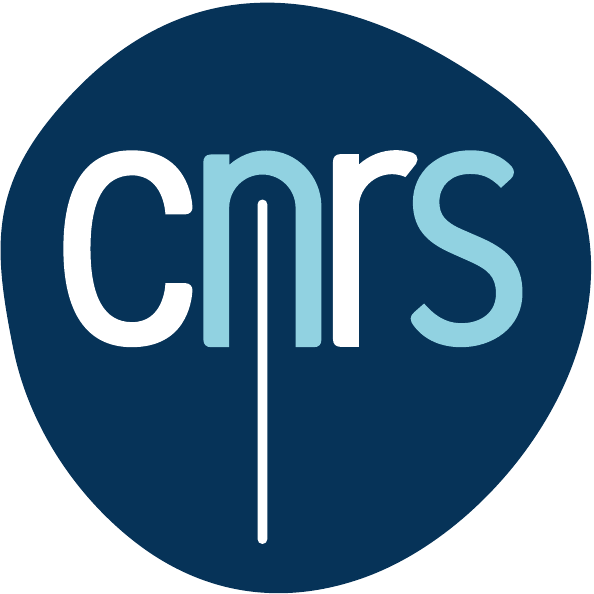}
&
\bfseries\sffamily\scriptsize%
\begin{tabular}[b]{c} 
	\adresseauteur\\
	\href{http://\urlLabo}{\urlLabo}
\end{tabular}%
&
\includegraphics[height=1.2cm]{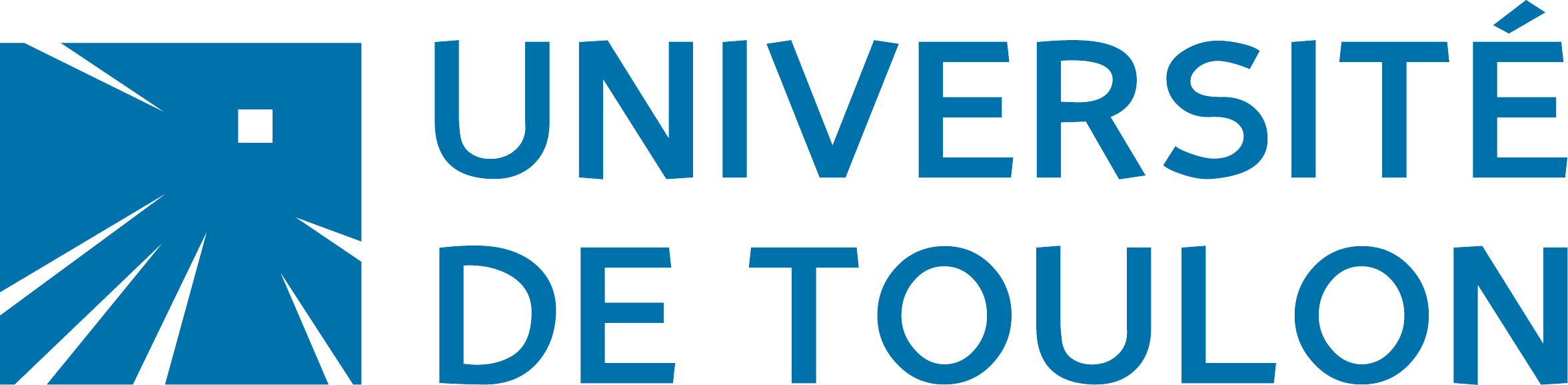}
\end{tabular}
\end{textblock*}
}

\newcommand{\adresseauteur}{
Centre de Physique Théorique (UMR 7332)\\
Case 907 - Luminy\\
F-13288 Marseille cedex 9
}
\newcommand{\logoLabo}{LogoCPT-texte.pdf}
\newcommand{\urlLabo}{www.cpt.univ-mrs.fr}

\begin{document}

\makeletterhead
\makeletterfoot

\vspace*{3cm}

\begin{center}
{\sffamily\bfseries\LARGE  Spectral Geometry}
\end{center}

\vspace{1cm}

\begin{center} {\Large
Bruno  {\scshape Iochum}} \\

\hspace{1cm}

Centre de Physique Théorique\\
Aix Marseille Univ, Université de Toulon, CNRS, CPT, Marseille, France
\end{center}

\vspace{1cm}

\begin{center}
{\sffamily  Notes of lectures delivered at the Summer School \\[2mm]
``Geometric, Algebraic and Topological Methods for Quantum Field Theory 2011''\\[2mm]
Villa de Leyva (Columbia), July 4--22 }\\[2mm]
\end{center}

\vspace{0.5cm}
\begin{center}
{\sffamily\bfseries Abstract}\\[3mm]
\begin{minipage}{0.75\textwidth}\small

The goal of these lectures is to present the few fundamentals of noncommutative geometry looking around its 
spectral approach. Strongly motivated by physics, in particular by relativity and quantum mechanics, 
Chamseddine and Connes have defined an action based on spectral considerations, the so-called spectral 
action.

The idea is to review the necessary tools which are behind this spectral action to be able to compute it first in 
the case of Riemannian manifolds (Einstein--Hilbert action). Then, all primary objects defined for manifolds will 
be generalized to reach the level of noncommutative geometry via spectral triples, with the concrete analysis 
of the noncommutative torus which is a deformation of the ordinary one.

The basics of different ingredients will be presented and studied like, Dirac operators, heat 
equation asymptotics, zeta functions  and then, how to get within the framework of operators on Hilbert 
spaces, the notion of noncommutative residue, Dixmier trace, pseudodifferential operators etc. These notions 
are appropriate in noncommutative geometry to tackle the case where the space is swapped with an algebra 
like for instance the noncommutative torus. Its non-compact generalization, namely the Moyal plane, is also 
investigated.
\end{minipage}
\end{center}
\vspace{1cm} \hspace{2cm}{\small Update: 2017, December 15}

\newpage

{\bf Motivations:}

Let us first expose few motivations from physics to study noncommutative geometry which is by essence a spectral 
geometry. Of course, precise mathematical definitions and results will be given in the other sections.

The notion of spectrum is quite important in physics, for instance in classical mechanics, the Fourier spectrum is 
essential to understand vibrations or the light spectrum in electromagnetism. 
The notion of spectral theory is also important in 
functional analysis, where the spectral theorem tells us that any selfadjoint operator $A$ can be seen as an integral 
over its spectral measure $A=\int_{a\in \Sp(a)} a \, dP_a$ if $\Sp(A)$ is the spectrum of $A$. This is of course  
essential in the axiomatic formulation of quantum mechanics, especially in the Heisenberg picture where the tools 
are the observables namely are selfadjoint operators.
\\
But this notion is also useful in geometry. In special relativity, we consider fields 
$\psi(\vec{x})$ for $\vec{x}\in \R^4$ and the electric and magnetic fields $E,\, B \in$ Function($M=\R^4,\R^3)$.
Einstein introduced in 1915 the gravitational field and the equation of motion of matter. But a problem appeared: 
what are the physical meaning of coordinates $x^\mu$ and equations of fields? 
Assume the general covariance of 
field equation. If $g_{\mu \nu}(x)$ or the tetradfield $e^I_\mu (x)$ is a solution (where $I$ is a local inertial reference 
frame), then, for any diffeomorphism $\phi$ of $M$ which is active or passive (i.e. change of coordinates), 
$e'^I_\nu (x) =\frac{\partial x^\mu}{\partial\phi(x)^\nu} \, e^I_\mu(x)$ is also a solution. As a consequence, when 
relativity became general, the points disappeared and it remained only fields on fields in the sense that there is no 
fields on a given space-time. But how to practice geometry without space, given usually by a manifold $M$? 
In this later case, the spectral approach, namely the control of eigenvalues of the scalar (or spinorial) Laplacian 
return important informations on $M$ and one can even address the question if they are sufficient: can one hear 
the shape of $M$?
\\
There are two natural points of view on the notion of space: one is based on points (of a manifold), this is the 
traditional geometrical one. The other is based on algebra and this is the spectral one. So the idea is to use algebra 
of the dual spectral quantities.
\\
This is of course more in the spirit of quantum mechanics but it remains to know what is a quantum geometry 
with bosons satisfying the Klein-Gordon equation $(\Box +m^2) \psi(\vec{x})=s_b(\vec{x})$ and fermions 
satisfying $(i\delslash-m)\psi(\vec{x})=s_f(\vec{x})$ for sources $s_b,s_f$. Here $\delslash$ can be seen as a 
square root of $\Box$ and the Dirac operator will play a key role in noncommutative geometry.

In some sense, quantum forces and general relativity drive us to a spectral approach of physics, especially of 
space-time. 

Noncommutative geometry, mainly pioneered by A. Connes (see \cite{Book,ConnesMarcolli}), is based on a 
spectral triple $(\A,\H,\DD)$ where the $*$-algebra $\A$ generalizes smooth functions on space-time $M$
(or the coordinates) with pointwise product, $\H$ generalizes the Hilbert space of above quoted spinors $\psi$ 
and $\DD$ is a selfadjoint operator on $\H$ which generalizes $\delslash$ via a connection on a vector bundle 
over $M$. The algebra $\A$ also acts, via a representation of $*$-algebra, on $\H$.

Noncommutative geometry treats space-time as quantum physics does for the phase-space since it gives a 
uncertainty principle: 
under a certain scale, phase-space points are indistinguishable. Below the scale $\Lambda^{-1}$, a certain 
renormalization is necessary. Given a geometry, the notion of action plays an essential role in physics, for instance, 
the Einstein--Hilbert action in gravity or the Yang--Mills--Higgs action in particle physics. So here, given the data 
$(\A,\H,\DD)$, the appropriate notion of action was introduced by Chamseddine and Connes \cite{CC} and defined 
as
$$
S(\DD,\Lambda,f) \vc \Tr \big(f(\DD/\Lambda)\big)
$$
where $\Lambda\in \R^+$ plays the role of a cut-off and $f$ is a positive even function. The asymptotic series in 
$\Lambda \to \infty$ yields to an effective theory. For instance, this action applied to a noncommutative model of 
space-time $M\times F$ with a fine structure for fermions encoded in a finite geometry $F$ gives rise from pure 
gravity to the standard model coupled with gravity \cite{CC2,CCM,ConnesMarcolli}. 

\vspace{0.5cm}
The purpose of these notes is mainly to compute this spectral action on few examples like manifolds and 
the noncommutative torus.
\\
In section \ref{Wod}, we present standard material on pseudodifferential operators over a compact Riemannian 
manifold. A description of the behavior of the kernel of a $\Psi$DO near the diagonal is given with the important 
example of elliptic operators. Then follows the notion of Wodzicki residue and its computation. The main point 
being to understand why it is a residue.

In section \ref{Dixmiertrace}, the link with the Dixmier trace is shown. Different subspaces of compact operators 
are described in particular, the ideal $\L^{1,\infty}(\H)$. Its definition is on purpose because in 
renormalization theory, one has to control the logarithmic divergency of the series $\sum_{n=1}^\infty n^{-1}$. We 
will see that this ``defect'' of convergence of the Riemann zeta function (in the sense that this generates a lot of 
complications of convergence in physics) is in fact an ``advantage'' because it is 
precisely the Dixmier trace and more generally the Wodzicki residue which are the right tools which mimics this 
zeta function: firstly, this controls the spectral aspects of a manifold and secondly they can be generalized to any 
spectral triple. 

In section \ref{Diracoperator}, we recall the basic definition of a Dirac (or Dirac-like) operator on a compact 
Riemannian manifold $(M,\,g)$ endowed with a vector bundle $E$. An example is the (Clifford) bundle 
$E=\CC \ell\, M$ where $\CC \ell\,T_x^*M$ is the Clifford algebra for $x \in M$. This leads to the notion of spin 
structure, spin connection $\nabla^S$ and Dirac operator $\Ds=-i c\circ\nabla^S$ where $c$ is the Clifford 
multiplication. A special focus is put on the change of metrics $g$ under conformal transformations.

In section \ref{Heatkernel} is presented the fundamentals of heat kernel theory, namely the Green function of the 
heat operator $e^{t\Delta},\,t\in \R^+$. In particular, its expansion as $t\to 0^+$ in terms of coefficients of the elliptic 
operator $\Delta$, with a method to compute the coefficients of this expansion is explained. The idea being to 
replace the Laplacian $\Delta$ by $\DD^2$ later on.

In section \ref{Noncommutative integration}, a noncommutative integration theory is developed around 
the notion of spectral triple. This means to understand the notion of differential (or pseudodifferential) 
operators in this context. Within differential calculus, the link between the one-form and the fluctuations of the 
given $\DD$ is outlined.

Section \ref{Spectral action} concerns few actions in physics, such that the Einstein--Hilbert and Yang--Mills actions. 
The spectral action $\Tr \big(f(\DD/\Lambda)\big)$ is justified and the link between its 
asymptotic expansion in $\Lambda$ and the heat kernel coefficients is given via the noncommutative integrals 
of powers of $\vert \DD\vert$.

Section \ref{Residues of series} gathers several results on the computation of a residue of a series of 
holomorphic functions, a real difficulty since one cannot commute residue and infinite sums. The notion of 
Diophantine condition appears and allows nevertheless this commutation for meromorphic extension of a 
class of zeta functions.

Section \ref{The noncommutative torus} is devoted to the computation of the spectral action on the noncommutative 
torus. After the very definitions, it is shows how to calculate with the noncommutative integral. The main 
technical difficulty stems from a Diophantine condition which seems necessary (but is sufficient) since any element 
of the smooth algebra of the torus is a series of its generators, so the previous section is fully used. All proofs are 
not given, but the reader should be aware of all the main steps.

Section \ref{The non-compact case} is an approach of non-compact spectral triples. This is mandatory for physics 
since, a priori, the space-time is not compact. After a quick review on the difficulties which occur when $M=\R^d$ 
due to the fact that the Dirac operator has a continuous spectrum, the example of the Moyal plane is analyzed. This 
plane is a non-compact version of the noncom-mutative torus. Thus, no Diophantine condition appears, but the 
price to pay is that functional analysis is deeply used. 
\bigskip

For each section, we suggest references since this review is by no means original.

\newpage
\tableofcontents
\bigskip

{\bf Notations:} \\
$\N =\{1, 2,\dots \}$ is the set of positive integers and $\N_0 =\N \cup \{0\}$ the set of non negative integers.\\
On $\R^d$, the volume form is $dx=dx^1\wedge \cdots \wedge dx^d$.\\
$\mathbb{S}^{d}$ is the sphere of radius one in dimension $d$. The induced metric: 
$$
d\xi=\vert \sum_{j=1}^d (-1)^{j-1}\xi_j\,d\xi_1\wedge \cdots \wedge \widehat{d\xi_j}\wedge\cdots \wedge 
d\xi_d \vert$$
restricts to the volume form on $\mathbb{S}^{d-1}$.\\
$M$ is a $d$-dimensional manifold with metric $g$. \\
$U,V$ are open set either in $M$ or in $\R^d$.\\
We denote by $dvol_g$ the unique volume element such that $dvol_g(\xi_1,\cdots,\xi_d)=1$ for all positively 
oriented $g$-orthonormal basis $\set{\xi_1,\cdots,\xi_d}$ of $T_xM$ for $x\in M$. Thus in a local chart 
$\sqrt{\det\,g_x} \,\vert dx\vert= \vert dvol_g \vert$.\\
When $\a \in \N^d$ is a multi-index, we define 
$$
\partial^\a_x \vcentcolon = \partial^{\a_1}_{x_1}  \partial^{\a_2}_{x_2} \cdots \partial^{\a_d}_{x_d} \,,  \qquad
\vert \a \vert \vcentcolon =\sum_{i=1}^d \a_i \,, \qquad \a ! \vcentcolon = \a_1 a_2\cdots \a_d.
$$ 
For $\xi \in \R^d$, $\vert \xi \vert \vcentcolon=\big(\sum_{k=1}^d \vert \xi_k \vert ^2\big)^{1/2}$ is the 
Euclidean metric.\\
$\H$ is a separable Hilbert space and $\B(\H), \K(\H), \L^p(\H)$ denote respectively the set of bounded, compact 
and $p$-Schatten-class operators, so $\L^1(\H)$ are trace-class operators.

\section{Wodzicki residue and kernel near the diagonal}
\label{Wod}

The aim of this section is to show that the Wodzicki's residue {\it WRes} is a trace on the set $\Psi DO(M)$ of classical 
pseudodifferential operators on a compact manifold $M$ of dimension $d$. 
\medskip

Let us first describe the steps:

- Define {\it WRes}$(P) =  2\, \underset{s=0}{\Res}\, \zeta(s)$ for $P\in \Psi DO^m$ of order $m$ and 
$\zeta: s\in\C \rightarrow \Tr(P \Delta^{-s})$, which is holomorphic when $\Re (s)\geq \frac{1}{2} (d+m)$.

- If $k^P(x,y)$ is the kernel of $P$, then its trace can be developed homogeneously as the following :
$\tr\big(k^P(x,y)\big)=\sum_{j=-(m+d)}^0 a_j(x,x-y)-c_P(x)\log\vert x-y\vert + \cdots$ 
where $a_j$ is homogeneous of degree $j$ in $y$ and $c_P$ is a density on $M$ defined by 
$c_P(x):= \frac{1}{(2\pi)^d}\int_{\mathbb{S}^{d-1}}\tr\big(\sigma^P_{-d}(x,\xi)\big) \, d\xi$; here,  
$\sigma_{-d}^P$ is the symbol of $P$ of order $-d$.
\\
The Wodzicki's residue has a simple computational form, namely\,{\it WRes}$\, P=\int_M c_P(x)\, \vert dx \vert$. 
Then, the trace property follows.
\medskip

References for this section: Classical books are \cite{Shubin,Taylor}. For an orientation more in the spirit of 
noncommutative geometry since here we follow \cite{Ponge,Ponge1} based on \cite{CM2,BG}, see also the 
excellent books \cite{NR,Paycha, Polaris, Var, VarillyS}.

\subsection{A quick overview on pseudodifferential operators}

In the following, $m\in\C$.
\begin{definition}
\label{defsym}
A symbol $\sigma(x,\xi)$ of order $m$ is a $C^\infty$ function: $(x,\xi)\in U\times \R^d \rightarrow \C$ satisfying for any compact $K\subset U$ and any $x \in K$

(i) $\vert \partial_x^\a \partial_\xi^\beta \sigma(x,\xi) \vert \leq C_{K\a \beta}\,
(1+\vert \xi\vert)^{\Re(m)-\vert \beta \vert}$, for some constant $C_{K\alpha\beta}$.

(ii) We suppose that $\sigma(x,\xi) \simeq \sum_{j\geq 0} \sigma_{m-j}(x,\xi)$ where $\sigma_k$ is 
homogeneous of degree $k$ in $\xi$ where $\simeq$ means a controlled asymptotic behavior

\centerline{$\vert \partial_x^\a \partial_\xi^\beta \big( \sigma- \sum_{j<N}\sigma_{m-j}\big)(x,\xi)
 \vert \leq C_{KN\a\beta} \, \vert \xi \vert^{\Re(m)-N-\vert \beta\vert} \text{ for }
 \vert \xi \vert \geq 1.$}
\noindent The set of symbols of order $m$ is denoted by $S^m(U\times \R^d)$.

A function $a\in C^\infty(U\times U \times \R^d)$ is an amplitude of order $m$, if for any compact $K\subset U$ 
and any $\a,\beta, \ga \in \N^d$ there exists a constant $C_{K\a\beta \ga}$ such that
\begin{align*}
\vert \partial_x^\a \partial_y^\ga \partial_\xi^\beta \,a(x,y,\xi) \vert \leq C_{K\a \beta\ga}\,
(1+\vert \xi\vert)^{\Re(m)-\vert \beta \vert},\quad \forall\, x,y\in K, \xi\in \R^d.
\end{align*} 
The set of amplitudes is written $A^m(U)$.
\end{definition}

For $\sigma\in S^m(U\times\R^d)$, we get a continuous operator 
$\sigma(\cdot,D): u\in C^\infty_c(U) \rightarrow C^\infty(U)$ given by
\begin{align}
\label{sigma(cdot,D)}
\sigma(\cdot,D)(u)(x)\vcentcolon =
\sigma(x,D)(u)\vc  \tfrac{1}{(2\pi)^d} \int_{\R^d}\sigma(x,\xi)\, \widehat{u}(\xi)\,e^{ix\cdot \xi}\, d\xi
\end{align}
where $\,\widehat{ }\,$ means the Fourier transform. This operator $\sigma(\cdot,D)$ will be also denoted by 
$Op(\sigma)$. For instance, 
\begin{align*}
\text{if }\,\sigma(x,\xi)=\sum_\a a_\a(x)\,\xi^\a \text{, then } \sigma(x,D)=\sum_{\a} a_\a(x)\,D_x^\a \text{ with } 
D_x \vcentcolon = -i\partial_x.
\end{align*}
Remark that, by transposition, there is a natural extension of $\sigma(\cdot,D)$ from the set $\DD'_c(U)$ of 
distributions with compact support in $U$ to the set of distributions $\DD'(U)$.

By definition, the leading term for $\vert \a \vert =m$ is the {\it  principal symbol} and the {\it Schwartz kernel of} 
$\sigma(x,D)$ is defined by
\begin{align}
\label{kernel}
k^{\sigma(x,D)}(x,y)\vc \tfrac{1}{(2\pi)^d} \int_{\R^d} \sigma(x,\xi) \, e^{i(x-y)\cdot \xi}\,d\xi
=\widecheck{\sigma}_{\xi\rightarrow y}(x,x-y)
\end{align}
where $\,\widecheck{}\,$ is the Fourier inverse in variable $\xi$. Similarly, the kernel of the operator $Op(a)$ 
associated to the amplitude $a$ is
\begin{align}
\label{amplitude}
k^a(x,y)\vc \tfrac{1}{(2\pi)^d} \int_{\R^d} a(x,y,\xi) \, e^{i(x-y)\cdot \xi}\,d\xi.
\end{align}

\begin{definition}
$P: C_c^\infty (U) \rightarrow C^\infty (U)$ (or $\DD'(U)$) is said to be smoothing if its kernel is in 
$C^\infty(U\times U)$ and $\Psi DO^{-\infty}(U)$ will denote the set of smoothing operators.\\
For $m\in \C$, the set $\Psi DO^m (U)$ of pseudodifferential operators of order $m$ will be the set of $P$ such that 
\begin{align*}
P: C_c^\infty (U) \rightarrow C^\infty (U), \, Pu(x)=\big(\sigma(x,D)+R\big)(u) \text{ where }
\sigma \in S^m(U\times \R^d) ,\, R\in \Psi DO^{-\infty}.
\end{align*}
$\sigma$ is called the symbol of $P$.
\end{definition}
\begin{remark}
\label{defsmoothing}
It is important to quote that a smoothing operator is a pseudodifferential operator whose amplitude is in $A^m(U)$ 
for all $m\in \R$: by \eqref{amplitude}, $a(x,y,\xi)\vc e^{-i(x-y)\cdot \xi}k(x,y)\,\phi(\xi)$ where the function 
$\phi \in C^\infty_c(\R^d)$ satisfies $\int_{\R^d}\phi(\xi)\,d\xi=(2\pi)^d$.
\end{remark}
Clearly, the main obstruction to smoothness is on the diagonal since

\begin{lemma}
$k^{\sigma(x,D)}$ is $C^\infty$ outside the diagonal.
\end{lemma}

\begin{proof}
$\widecheck{\sigma}$ is smooth since it is given for $y\neq 0$ by the oscillatory integral 
\begin{align*}
\int_{\R^d} \sigma(x,\xi)\, e^{iy\cdot \xi}\, d\xi=\int_{\R^d} (P_y^k \sigma)\, e^{iy\cdot\xi}\,d\xi
\end{align*}
where $k$ is an integer such that $k>\Re(m)+n$ and moreover $P_y=P(y,D_\xi)$ is chosen with 
$P_y(e^{iy\cdot \xi})=e^{iy \cdot \xi}$; for instance 
$P_y=\tfrac{1}{\vert y \vert^2} \sum_j y_j\tfrac{\partial\,}{\partial \xi_j}$. The last integral is absolutely 
converging.
\end{proof}
Few remarks on the duality between symbols and a subset of pseudodifferential operators:
\begin{align*}
\sigma(x,\xi)\in S^m(U\times \R^d) \longleftrightarrow k_\sigma(x,y)\in C^\infty_c(U\times U\times \R^d)
\longleftrightarrow A=Op(\sigma) \in \Psi DO^m
\end{align*}
with the definition
$$
\sigma^A(x,\xi) \vcentcolon = e^{-ix\cdot \xi} \,A(x\rightarrow e^{i x \cdot \xi})$$
where $A$ is properly supported (since $x \to e^{ix\xi}$ has not a compact support) namely, 
$A$ and its adjoint map the dual of $C^\infty(U)$ (distributions with compact support) into itself. 
Moreover,
\begin{align*}
 \sigma^A \simeq \sum_\a \tfrac{(-i)^\a}{\a !} \,\partial^\a_\xi \partial^\a_y \, k^A_\sigma (x,y,\xi) _{\vert y=x}\,,\qquad
 k^A_\sigma(x,y)=\vcentcolon \tfrac{1}{(2\pi)^d} \int_{\R^d} e^{i(x-y)\cdot \xi} \, k^A(x,y,\xi) \,d\xi \,,
\end{align*}
where $k^A(x,y,\xi)$ is the amplitude of $k_\sigma^A(x,y)$. Actually, 
$\sigma^A(x,\xi)=e^{iD_\xi\,D_y}\,k^A(x,y,\xi)_{\vert y=x}$ and 
$e^{iD_\xi\,D_y}=1+iD_\xi D_y -\tfrac{1}{2}(D_\xi D_y)^2+\cdots$. Thus $A=Op(\sigma^A) +R$ where 
$R$ is a regularizing operator on $U$.

A technical point is the following: a pseudodifferential operator $P$ is properly supported when both $P$ and $P^*$ maps $C_c^\infty(U)$ not only in $C^\infty(U)$ but in $C_c^\infty(U)$. Any $\Psi DO$ is the sum of a properly supported 
$\Psi DO$ and a smothing one.

A point of interest is that differential operators are local: if $f=0$ on $U^c$ (complementary set of $U$) then $Pf=0$ 
on $U^c$. While pseudodifferential operators are pseudo-local: $Pf$ is smooth on $U$ when $f$ is smooth.

There are two fundamental points about $\Psi DO$'s: they form an algebra and this notion is stable by 
diffeomorphism justifying its extension to manifolds and then to bundles:

\begin{theorem}
\label{pseudo}
(i) If $P_1 \in \Psi DO^{m_1}$ and $P_2 \in \Psi DO^{m_2}$ are properly supported $\Psi DO$'s, then $P_1 P_2 \in \Psi DO^{m_1+m_2}$ is properly supported with symbol
\begin{align*}
\sigma^{P_1P_2}(x,\xi)\simeq \sum_{\a \in \N^d} \tfrac{(-i)^\a}{\a !}\, \partial^\a_\xi \sigma^{P_1}(x,\xi) \,
\partial^\a_\x \sigma^{P_2} (x,\xi).
\end{align*}
The principal symbol of $P_1P_2$ is
\begin{align*}
\sigma_{m_1+m_2}^{P_1P_2}(x,\xi)=\sigma_{m_1}^{P_1}(x,\xi)\,\sigma_{m_2}^{P_2}(x,\xi).
\end{align*}
(ii) Let $P\in\Psi DO^m(U)$ and $\phi \in \text{Diff\,}(U,V)$ where $V$ is another open set of $\R^d$. The operator 
$\phi_*P: f\in C^\infty(V) \rightarrow P(f\circ \phi)\circ \phi^{-1}$ satisfies $\phi_* P \in \Psi DO ^m(V)$ and its symbol is
\begin{align*}
\sigma^{\phi_* P}(x,\xi)=\sigma^P_m\big(\phi^{-1}(x),(d\phi)^t \xi \big) + \sum_{\vert \a \vert >0} \tfrac{(-i)^\a}{\a !}\,
\phi_\a(x,\xi)\, \partial^\a_\xi \sigma^P \big( \phi^{-1}(x), (d\phi)^t \xi \big)
\end{align*}
where $\phi_\a$ is a polynomial of degree $\a$ in $\xi$. Moreover, its principal symbol is 
\begin{align*}
\sigma_m^{\phi_*P}(x,\xi)=\sigma^P_m \big(\phi^{-1}(x),(d\phi)^t \xi \big).
\end{align*}
In other terms, the principal symbol is covariant by diffeomorphism: ${\sigma^{\phi_* P}}_m=\phi_* \sigma_m^P$.\\
(iii) If $P\in \Psi DO^{m}$ is properly supported, then $P^*\in \Psi DO^{m}$ and 
\begin{align*}
\sigma^{P^*}(x,\xi) \simeq  \sum_\a \tfrac{(-i)^\a}{\a !} \,\partial^\a_\xi \partial^\a_x \,\overline{\sigma(x,\xi)}\,.
\end{align*}
\end{theorem}
While the proof of formal expressions is a direct computation, the asymptotic behavior requires some care, see
\cite{Taylor, Shubin}.

An interesting remark is in order: $\sigma^P(x,\xi)=e^{-ix\cdot \xi}\, P(x\rightarrow e^{ix\cdot \xi})$, thus the dilation 
$\xi \rightarrow t \xi$ with $t>0$ gives
\begin{align*}
t^{-m}\,e^{-itx\cdot \xi} P \,e^{itx\cdot\xi}=t^{-m}\sigma^P(x,t\xi) \simeq t^{-m}\sum_{j\geq0} \sigma^P_{m-j}(x,t\xi)
=\sigma^P_m(x,\xi)+o(t^{-1}).
\end{align*}
Thus, if $P\in \Psi DO^m(U)$ with $m\geq 0$, 
\begin{align*}
\sigma^P_m(x,\xi)= \lim_{t\to \infty} t^{-m} \, e^{-ith(x)} \,P \, e^{ith(x)}, \text{ where } h\in C^\infty(U) \text{ is (almost) 
defined by } dh(x)=\xi. 
\end{align*}

\subsection{Case of manifolds}

Let $M$ be a (compact) Riemannian manifold of dimension $d$. Thanks to Theorem \ref{pseudo}, 
the following makes sense:
\begin{definition}
$\Psi DO^m(M)$ is defined as the set of operators $P: C^\infty_c(M) \rightarrow C^\infty(M)$ such that

\qquad (i) the kernel $k^P \in C^\infty (M\times M)$ off the diagonal,

\qquad (ii) the map 
$:\,f \in C_c^\infty \big( \phi(U) \big) \rightarrow P(f \circ \phi) \circ \phi^{-1} \in C^\infty \big( \phi(U)\big)$ 
is in $\Psi DO^m\big(\phi(U)\big)$ for every coordinate chart $(U,\phi:U \to \R^d)$.
\end{definition}
\noindent Of course, this can be generalized: 
\begin{definition}
Given a vector bundle $E$ over $M$, a linear map $P: \Gamma_c^\infty(M,E) \rightarrow \Gamma^\infty (M,E)$ is in 
$\Psi DO^m(M,E)$ when $k^P$ is smooth off the diagonal, and local expressions are $\Psi DO$'s with matrix-valued 
symbols.
\end{definition}
The covariance formula implies that $\sigma_m^P$ is independent of the chosen local chart so is globally defined 
on the bundle $T^*M \rightarrow M$ and $\sigma_m^P$ is defined for every $P \in \Psi DO^m$ using overlapping 
charts and patching with partition of unity.

An important class of pseudodifferential operators are those which are invertible modulo regularizing ones:
\begin{definition}
$P \in \Psi DO^m(M,E)$ is elliptic if $\sigma_m^P(x,\xi)$ is invertible for all $\xi \in TM^*_x$, $\xi \neq 0$.
\end{definition}
This means that $\vert \sigma^P(x,\xi) \vert \geq c_1(x) \vert \xi \vert ^m$ for $\vert \xi \vert \geq c_2(x), \, x \in U$ 
where $c_1,c_2$ are strictly positive continuous functions on $U$.
\\
This also means that there exists {\it a parametrix:}
\begin{lemma}
The following are equivalent:

(i) $Op(\sigma) \in \Psi DO^m(U)$ is elliptic.

(ii) There exist $\sigma' \in S^{-m}(U\times \R^d)$ such that $\sigma \circ \sigma'=1$ or $\sigma' \circ \sigma=1$.

(iii) $Op(\sigma) \, Op(\sigma')=Op(\sigma')\,Op(\sigma) = 1$ modulo $\Psi DO^{-\infty}(U)$. \\
Thus $Op(\sigma') \in \Psi DO^{-m}(U)$ is also elliptic.
\end{lemma}
At this point, it is useful to remark that any $P \in \Psi DO^m(M,E)$ can be extended to a bounded operator on 
$L^2(M,E)$ when $\Re(m) \leq 0$. Of course, this needs an existing scalar product for given metrics on $M$ and 
$E$.

\begin{theorem}
\label{compact}
Assume $M$ is compact. If $P\in \Psi DO^{-m}(M,E)$ is elliptic with $\Re(m)>0$, then $P$ is compact, so its spectrum is discrete.
\end{theorem}

\begin{proof}
We need to get the result first for an open set $U$, for a manifold $M$ and then for a bundle $E$ over $M$.

For any $s \in \R$, the usual Sobolev spaces $H^s(\R^d)$ (with $H^0(\R^d)=L^2(\R^d)$) and $H^s_c(U)$ (defined 
as the union of all $H^s(K)$ over compact subsets $K \subset U$) or $H^s_{loc}(U)$ (defined as the set of 
distributions $u \in \DD'(U)$ such that $\phi u \in H^s(\R^d)$ for all $\phi \in C_c^\infty(U)$) can be extended for 
any manifold $M$ to the Sobolev spaces $H^s_c(M)$ (obvious definition) and $H^s_{loc}(M)$: if 
$(U,\phi:U \to \R^d)$ is a local chart and $\chi \in C^\infty_c\big(\phi(U)\big)$, we say that a 
distribution $u \in \DD'(M)$ is in $H^s_{loc}(M)$ when $\big(\phi^{-1}\big)^*(\xi \,u) \in H^s(\R^d)$. When $M$ is 
compact, $H^s_{loc}(M)=H_c^s(M)$ (thus denoted $H^s(M$). Using Rellich's theorem, the inclusion 
$H_c^s(U) \hookrightarrow H_c^t(U)$ for $s<t$ is compact. Since $P:H_c^s(M) \rightarrow H_{loc}^{s-\Re(m)}(M)$ 
is a continuous linear map for a (non-necessarily compact) manifold $M$, both results yield that $P$ is compact. 
Finally, the extended operator on a bundle is $P: L^2(M,E) \rightarrow H^{-\Re(m)}(M,E) \hookrightarrow L^2(M,E)$ 
where the second map is the continuous inclusion, so $P$ being compact as an $L^2$ operator has a 
discrete spectrum.
\end{proof}
We rephrase a previous remark (see \cite[Proposition 2.1]{Berline}):

Let $E$ be a vector bundle of rank $r$ over $M$. If $P\in \Psi DO^{-m}(M,E)$, then for any couple of sections 
$s\in \Gamma^\infty(M,E)$, $t^* \in \Gamma^\infty(M, E^*)$, the operator 
$f \in C^\infty(M) \rightarrow \langle t^*,P(fs) \rangle \in C^\infty(M)$ is in $\Psi DO^m(M)$. This means that in a local 
chart $(U,\phi)$, these operators are $r\times r$ matrices of pseudodifferential operators of order $-m$. 
The total symbol is in $C^\infty(T^*U) \otimes End(E)$ with $End(E) \simeq M_r(\C)$. 
The principal symbol can be globally defined: $\sigma_{-m}^P(x,\xi): E_x \rightarrow E_x$ for $x\in M$ and 
$\xi \in T^*_xM$, can be seen as a smooth homomorphism homogeneous of degree 
$-m$ on all fibers of $T^*M$. Moreover, we get the simple formula which could be seen as a definition of the 
principal symbol (as already noticed at the end of previous section)
\begin{align}
\label{symprinc}
\sigma_{-m}^P(x,\xi)=\lim_{t \to \infty} t^{-m} \, \big(e^{-ith} \cdot P \cdot e^{ith}\big)(x) \, \text{ for } x \in M, \, 
\xi \in T_x^*M
\end{align}
where $h \in C^\infty(M)$ is such that $d_xh =\xi$.

\subsection{Singularities of the kernel near the diagonal}

The question to be solved is to define a homogeneous distribution which is an extension on $\R^d$ of a given 
homogeneous symbol on $\R^d \backslash \set{0}$. Such extension is a regularization used for instance by 
Epstein--Glaser in quantum field theory.
\\
The Schwartz space on $\R^d$ is denoted by $\SS$ and the space of tempered distributions by $\SS'$.

\begin{definition}
For $f_\lambda(\xi)\vcentcolon = f(\lambda \xi)$, $\lambda \in \R^*_+$, define $\tau\in \SS' \rightarrow \tau_\lambda$ 
by $\langle \tau_\lambda,f \rangle \vcentcolon = \lambda^{-d} \langle \tau,f_{\lambda^{-1}} \rangle$ for all $f\in \SS$.

A distribution $\tau\in \SS'$ is homogeneous of order $m\in \C$ when $\tau_\lambda=\lambda^m \,\tau$.
\end{definition}

\begin{prop}
\label{extprop}
Let $\sigma \in C^\infty(\R^d \backslash \set{0})$ be a homogeneous symbol of order $k\in \Z$.

(i) If $k > -d$, then $\sigma$ defines a homogeneous distribution.

(ii) If $k = -d$, there exists a unique obstruction to the extension of $\sigma$ given by
\begin{align*}
c_\sigma=\int_{\mathbb{S}^{d-1}} \sigma(\xi) \,d\xi,
\end{align*}
namely, one can at best extend $\sigma$ in $\tau \in \SS'$ such that 
\begin{align}
\label{ext}
\tau_\lambda=\lambda^{-d} \big(\tau + c_\sigma \log(\lambda)\, \delta_0\big).
\end{align}
\end{prop}

\begin{proof}
$(i)$ For $k>-d$, $\sigma$ is integrable near zero, increases slowly at $\infty$, so defines by extension 
a unique distribution $\tau\in \SS'$ which will be homogeneous of order $k$.

$(ii)$Assume $k=-d$. Then $\sigma$ extends to a continuous linear form 
$L_\sigma(f) \vcentcolon = \int_{\R^d} f(\xi)\, \sigma(\xi) \,d\xi$ on 
$\SS_0 \vcentcolon = \set{f \in \SS \, \vert \, f(0)=0}$. By 
Hahn--Banach theorem, $L_\sigma$ extends to $\SS'$ and $L_\sigma \in E$ where 
$E \vcentcolon = \set{\tau \in \SS' \, \vert \tau_{\vert \SS_0}=L_\sigma}$ is given by the direction $\delta_0$. 

This affine space $E$ is stable by the endomorphism $\tau \rightarrow \lambda^d \tau_\lambda$: actually if 
$f\in \SS_0$, $f_{\lambda^{-1}} \in \SS_0$ and  
\begin{align*}
\lambda^d \langle \tau_\lambda,f \rangle=\langle \tau, f_{\lambda^{-1}}\rangle=L_\sigma(f_{\lambda^{-1}})=
\int_{\R^d} f(\lambda^{-1} \xi)\,\sigma(\xi)\, d\xi=\int_{\R^d} f(\xi) \,\sigma(\xi)\,d\xi=L_\sigma(f),
\end{align*}
thus $\lambda^d \tau_\lambda=L_\sigma$ on $\SS_0$. 

Moreover, $\lambda^d \, (\delta_0)_\lambda=\delta_0$; thus 
there exists $c(\lambda) \in \C$ such that 
\begin{align}
\label{tau}
\tau_\lambda=\lambda^{-d} \tau + c(\lambda) \, \lambda^{-d} \, \delta_0
\end{align}
for all $\tau \in E$. The computation of $c(\lambda$ for a specific example in $E$ gives 
$c(\lambda)=c_\sigma \log(\lambda)$: for instance, choose $g\in C_c^\infty([0,\infty])$ which is $1$ near 0 and define 
$\tau \in \SS'$ by 
\begin{align*}
\langle \tau, f \rangle \vcentcolon = L_\sigma \big(f-f(0)g(\vert\cdot \vert)\big)=
\int_{\R^d} \big(f(\xi)-f(0)g(\vert \xi \vert) \big) \, \sigma(\xi) \, d\xi, \quad \forall f\in \SS.
\end{align*}
Thus if $f(0)=1$, we get $c(\lambda) \lambda^{-d} \langle \delta_0,f \rangle =
c(\lambda) \lambda^{-d}$, so by \eqref{tau}
\begin{align*}
c(\lambda)\lambda^{-d}&= \langle \tau, f_{\lambda^{-1}} \rangle -\lambda^{-d}\langle \tau, f \rangle \\
&=
\int_{\R^d} \big( f(\lambda^{-1} \xi)- g(\vert \xi\vert) \big) \,\sigma (\xi) \, d\xi 
- \lambda^{-d} \int_{\R^d} \big(f(\xi)-g(\vert \xi \vert) \, \sigma (\xi) \, d \xi \\
&=- \lambda^{-d}\int_{\R^d} \big( g(\lambda \vert \xi \vert) -g(\vert \xi \vert) \, \sigma(\xi) \, d\xi= 
-\lambda^{-d}c_\sigma \int_0^\infty \big( g(\lambda \vert \xi \vert)-g(\vert \xi \vert) \, \tfrac{d\vert \xi \vert}{\vert \xi \vert} 
\end{align*}
with $c_\sigma \vcentcolon=\int_{\mathbb{S}^{d-1}} \sigma(\xi)\, d^{d-1}\xi$. Since 
\begin{align*}
\lambda \frac{d\,}{d\lambda} \int_0^\infty \big( g(\lambda \vert \xi \vert)-g(\vert \xi \vert)\big) \, 
\tfrac{d \vert \xi \vert}{\vert \xi \vert} = 
\lambda \int_0^\infty g'(\lambda \vert \xi \vert)\, d\vert \xi \vert=-g(0)=-1,
\end{align*}
we get $c(\lambda)=c_\sigma \log(\lambda)$. Thus, when 
$c_\sigma =0$, every element of $E$ is a homogeneous distribution on $\R^d$ which extends the symbol $\sigma$.

Conversely, let $\tau\in \SS'$ be a homogeneous distribution extending $\sigma$ and let $\tilde \tau \in E$. Since 
$\tau-\tilde \tau$ is supported at the origin, we can write 
$\tau=\tilde \tau +\sum_{\vert \a \vert \leq N} a_\a \,\partial^\a \delta_0$ where $a_\a \in \C$ and 
\begin{align*}
0=\tau_\lambda-\lambda^{-d} \tau=c_\sigma \lambda^{-d} \log (\lambda) \delta_0 + 
\sum_{1\leq \vert \a \vert \leq N} a_\a \lambda^{-d}\big(\lambda^{\vert \a \vert}-1 \big) \partial ^\a \delta_0.
\end{align*}
The linear independence of $(\partial^\a \delta_0)$ gives $a_\a=0, \, \forall a_\a$. So $c_\sigma=0$ and $\tau \in E$. 
The condition $c_\sigma=0$ is so necessary and sufficient to extend $\sigma$ in 
a homogeneous distribution. And in the general case, one can at best extend it in a distribution satisfying 
\eqref{ext}, but it is only possible with elements of $E$.
\end{proof}

In the following result, we are interested by the behavior near the diagonal of the kernel $k^P$ for $P\in\Psi DO$. 
For any $\tau \in \SS'$, we  choose the decomposition as $\tau=\phi \circ \tau + (1-\phi)\circ \tau$ where 
$\phi \in C_c^\infty(\R^d)$ and $\phi=1$ near $0$. We can look at the infrared behavior of $\tau$ near the origin 
and its ultraviolet behavior near infinity. Remark first that, since $\phi \circ \tau$ has a compact support, 
$(\phi \circ \tau)\,{\widecheck{ }} \in \SS'$, so the regularity of $\tau{\,\widecheck{ }}\,$ depends only of its 
ultraviolet part $\big( (1-\phi)\circ \tau \big)\,{\widecheck{ }}\, $.

\begin{prop}
\label{propasymp}
Let $P \in\Psi DO^m(U)$, $m\in \Z$. Then, in local form near the diagonal,
\begin{align*}
k^P(x,y)=\sum_{-(m+d) \leq j \leq 0}\,a_j(x,x-y)-c_P(x)\log \vert x-y\vert +\mathcal{O}(1)
\end{align*}
where $a_j(x,y) \in C^\infty \big(U\times U \backslash \set{x}\big)$ is homogeneous of order $j$ in $y$ and 
$c_P(x) \in C^\infty(U)$ is given by
\begin{align}
\label{cp}
c_P(x)=\tfrac{1}{(2\pi)^d} \int_{\mathbb{S}^{d-1}} \, \sigma_{-d}^P(x,\xi)\,d\xi.
\end{align}
\end{prop}

\begin{proof}
We know that $\sigma^P(x,\xi) \simeq \sum_{j \leq m} \, \sigma_j^P(x,\xi)$ and by \eqref{kernel}, 
$k^P(x,y)=\widecheck{{\sigma}_{\xi\rightarrow y}}(x,x-y)$ so we need to control 
$\widecheck{{\sigma}_{\xi\rightarrow y}}(x,x-y)$ 
when $y \to 0$.

{\it Assume first that $\sigma^P(x,\xi)$ is independent of $x$:}
\\
For $-d<j\leq m$, $\sigma_j(\xi)$ extends to $\tau_j \in \SS'$. For $j>-d$, this extension is homogeneous (of degree 
$j$) and unique.
\\
For $j=-d$, we may assume that $\tau_{-d}$ satisfies \eqref{ext}. Thus 
$\tau\vcentcolon=\sigma^P-\sum_{j=-d}^m\,\tau_j \in \SS'$ behaves in the ultraviolet as a integrable symbol. In 
particular $\tau\,{\widecheck{ }}\,$ is continuous near 0 and we get 
\begin{align}
\label{sig}
\widecheck{\sigma^P}\,(y)=\sum_{j=-d}^m \widecheck{ \tau_j}\,(y) + \mathcal{O}(1).
\end{align}
Note that the inverse Fourier transform of the infrared part of $\tau_j$ is in $C^\infty(\R^d)$ while those of its 
ultraviolet part is in $C^\infty \big( \R^d \backslash \set{0} \big)$, so $\widecheck{ \tau_j}$ is smooth near 0.

Moreover, for $j>-d$, $\widecheck{ \tau_j}$ is homogeneous of degree $-(d+j)$ while for $j=-d$,
\begin{align*}
\widecheck{ \tau_{-d}}(\lambda y)=\lambda^{-d} \, [ \big({\tau_{-d})}_{\lambda^{-1}}] \, \widecheck{\,}\,(y) =
[\tau_{-d}-c_{\sigma_{-d}} \log(\lambda) \delta_0]\, \widecheck{ \,}\,(y) = \widecheck{ \tau_{-d}}(y) 
- \tfrac{1}{(2\pi)^d}\,c_{\sigma_{-d}} \log \lambda.
\end{align*}
For $\lambda= \vert y \vert ^{-1}$, we get
\begin{align*}
 \widecheck{ \tau_{-d}} \big( \tfrac{y}{\vert y \vert}\big) = \widecheck{ \tau_{-d}}\big( \tfrac{y}{\vert y \vert}\big) 
 - \tfrac{1}{(2\pi)^d}\,c_{\sigma_{-d}} \log \vert y \vert.
\end{align*}
Summation over $j$ in \eqref{sig} yields the result.

{\it Assume now that $\sigma^P(x,\xi)$ is dependent of $x$:}
\\
We do the same with families $\set{\tau_x}_{x\in U}$ and $\set{\tau_{j,x}}_{x\in U}$. Their ultraviolet behaviors are 
those of smooth symbols on $U\times \R^d$, so given by smooth functions on $U\times \R^d \backslash \set{0}$ and 
for $\tau_x$ by a continuous function on $U\times \R^d$.
For the infrared part, we get smooth maps from $U$ to $\mathcal{E}(\R^d)'$ (distributions with compact 
support), thus applying inverse Fourier 
transform, we end up with smooth functions on $U \times \R^d$. Actually,  for $\tau_{j,x}$ with $j>-d$, this 
follows from the fact that it is the extension of $\sigma_j(x,\cdot)$ which is integrable near the origin: let $f \in \SS$,
\begin{align*}
\langle \phi \circ \tau_{j,x}, f\rangle=\langle  \tau_{j,x}, \phi \circ f\rangle = 
\int_{\R^d} \phi(\xi)\,f(\xi) \,\sigma_{j}(x,\xi)\, d\xi.
\end{align*}
While for $j=-d$,
\begin{align*}
\langle \phi \circ \tau_{-d,x}, f\rangle=\langle  \tau_{-d,x}, \phi \circ f\rangle = \int_{\R^d} \phi(\xi)\,\big(f(\xi)-f(0)\big) \,
\sigma_{-d}(x,\xi)\, d\xi,
\end{align*}
and the map $x \to \phi \circ \tau_{-d,x}$ is smooth from $U$ to $\mathcal{E}(\R^d)'$.
In conclusion, 
\begin{align*}
\widecheck{\sigma}_{\xi\rightarrow y}(x,y)=\sum_{-(m+d)\leq j\leq0}\, a_j(x,y)-c_P(x) \log \vert y \vert +R(x,y)
\end{align*}
where $a_j(x,y)$ is a smooth function on $U\times \R^d\backslash \set{0}$, is homogeneous of degree $j$ in $y$, 
$c_P$ is given by \eqref{cp} and $R(x,y)$ is a function, continuous on $U\times \R^d$. So we get the claimed 
asymptotic behavior.
\end{proof}

\begin{theorem}
\label{kernelbehavior}
Let $P\in \Psi DO^m(M,E)$ with $m \in \Z$. Then, for any trivializing local coordinates
\begin{align*}
\tr \big(k^P(x,y)\big)=\sum_{j=-(m+d)}^0 a_j(x,x-y) -c_P(x) \log \vert x-y \vert + \mathcal{O}(1),
\end{align*}
where $a_j$ is homogeneous of degree $j$ in $y$, $c_P$ is intrinsically locally defined by
\begin{align}
\label{cP}
c_P(x)\vcentcolon=\tfrac{1}{(2\pi)^d} \int_{\mathbb{S}^{d-1}} \tr\big(\sigma_{-d}^P(x,\xi)\big)\,d\xi.
\end{align}
Moreover, $c_P(x) \vert dx \vert$ is a 1-density over $M$ which is functorial with respect to diffeomorphisms 
$\phi$: 
\begin{align}
\label{phi cP}
c_{\phi_*P}(x)=\phi_*\big(c_p(x)\big).
\end{align}
\end{theorem}

\begin{proof}
The asymptotic behavior follows from Proposition \ref{propasymp} but we first have to understand why $c_P$ is 
well defined:

{\it Assume first that $E$ is a trivial line bundle and $P$ is a scalar $\Psi DO$.}
\\
Define a change of coordinates by $y\vcentcolon = \phi^{-1}(x)$. Thus 
$k^P(x,x') \xrightarrow{\phi_*}k^{\phi_*P}(y,y')$ with 
\begin{align*}
k^{\phi_*P}(y,y')&=\vert J_\phi(y')\vert \, k^P\big( \phi(y),\phi(y')\big)\\
&=\sum_{j=-(m+d)}^0 \vert J_\phi(y') \vert \, \big[a_j\big( \phi(y),\phi(y)-\phi(y')\big)
- c_P\big(\phi(y)\big) \,\log \vert \phi(y) -\phi(y') \vert \,\big] +\mathcal{O}(1).
\end{align*}
A Taylor expansion around $\big(\phi(y), \phi'(y)\cdot (y-y')\big)$ of $a_j$ gives 
\begin{align*}
a_j\big( \phi(y),\phi(y)-\phi(y')\big) \simeq \vert y-y' \vert^j \,a_j\big( \phi(y),\phi'(y)\cdot \tfrac{y-y'}{\vert y-y' \vert} \big) 
+ \cdots,
\end{align*}
since $a_j\big(\phi(y),\cdot\big)$ is smooth outside 0, so we get only homogeneous and continuous terms. Moreover 
the only contribution to the $\log$-term is
\begin{align*}
\vert J_\phi(y') \vert \, c_P\big(\phi(y)\big) \,\log \vert \phi(y)-\phi(y') \vert \simeq 
\vert J_\phi(y) \vert \, c_P\big(\phi(y)\big) \,\log \vert \phi(y)-\phi(y') \vert +\mathcal{O}(1)
\end{align*}
and we get
\begin{align*}
c_{\phi_*P}(y)=\vert J_\phi(y)\vert\, c_P\big( \phi(y)\big).
\end{align*}
In particular $c_P(x)\, \vert dx^1 \wedge \cdots \wedge dx^d \vert$ can be globally defined on $M$ as a 1-density. 
(Recall that a $\a$-density on a vector space $E$ of dimension $n$ is any application $f:\, \exter^nE \to \R$ such 
that for any $\lambda  \in \R$, $f(\lambda x)=\vert \lambda \vert^\a f(x)$ and the set of these densities is denoted 
$\vert \exter \vert ^\a E^*$; this is generalized to a vector bundle $E$ over $M$ where 
each fiber is $\vert \exter \vert^\a E^*_x$. The interest of the bundle of 1-densities is to give a class of objects 
directly integrable on $M$. In particular, we get here something intrinsically defined, even when the manifold is not 
oriented).

{\it General case:} 
\\
$P$ acts on section of a bundle. By a change of trivialization, the action of $P$ is conjugate on each fiber by a 
smooth matrix-valued map $A(x)$, so $k^P(x,x') \rightarrow A(x)^{-1} k^P(x,x') A(x')$. We are looking for the 
logarithmic term: only the $0$-order term in $A(x')$ will contribute and $\tr \big(  A(x)^{-1} k^P(x,x') A(x') \big)$ has the 
same logarithmic singularity than the similar term 
$\tr \big(  A(x)^{-1} k^P(x,x') A(x) \big) =\tr \big(k^P(x,x')\big)$ near the diagonal. 
Thus $c_P(x)$ is independent of a chosen trivialization.
\\
Similarly, if $P$ is not a scalar but $End(E)$-valued, the above proof can be generalized (the space of 
$C^\infty \big((M, \vert \exter \vert(M) \otimes End(E) \big)$ of $End(E)$-valued densities is a sheaf).
\end{proof}

Remark that, when $M$ is Riemannian with metric $g$ and $d_g(x,y)$ is the geodesic distance, then 
\begin{align*}
\tr \big(k^P(x,y)\big)=\sum_{j=-(m+d)}^0 a_j(x,x-y) -c_P(x) \log\big(d_g(x,y)\big)+ \mathcal{O}(1),
\end{align*}
since there exists $c>0$ such that $c^{-1} \vert x-y\vert \leq d_g(x,y) \leq c \vert x-y \vert$.

\subsection{Wodzicki residue}
\label{Wodzicki residue}

The claim is that $\int_M c_P(x) \vert dx \vert$ is a residue.
\\
For this, we embed everything in $\C$. In the same spirit as in Proposition \ref{extprop}, one obtains the following
\begin{lemma}
\label{extc}
Every $\sigma \in C^\infty \big(\R^d \backslash \set{0} \big)$ which is homogeneous of degree 
$m\in \C \backslash \Z$ can be uniquely extended to a homogeneous distribution.
\end{lemma}

\begin{definition}
Let $U$ be an open set in $\R^d$ and $\Omega$ be a domain in $\C$. 

A map $\sigma :  \Omega \to S^m(U\times \R^d)$ is said to be holomorphic when

\qquad the map: $z\in \Omega \to \sigma(z)(x,\xi)$ is analytic for all $x\in U$, $\xi \in \R^d$,

\qquad the order $m(z)$ of $\sigma(z)$ is analytic on $\Omega$,

\qquad the two bounds of Definition \ref{defsym} $(i)$ and $(ii)$ of the asymptotics 
$\sigma(z) \simeq \sum_j \sigma_{m(z)-j}(z)$ are locally uniform in $z$.
\end{definition}
\noindent This hypothesis is sufficient to get:

\quad The map: $z \to \sigma_{m(z)-j}(z)$ is holomorphic from $\Omega$ to 
$C^\infty \big(U\times \R^d \backslash \set{0}\big)$.

\quad The map $\partial_z \sigma(z)(x,\xi)$ is a classical symbol on $U\times \R^d$ and one obtains:
\begin{align*}
\partial_z \sigma(z)(x,\xi) \simeq \sum_{j\geq 0} \,\partial_z \sigma_{m(z)-j}(z)(x,\xi).
\end{align*}

\begin{definition}
The map $P:\Omega \subset \C \to \Psi DO(U)$ is said to be holomorphic if it has the decomposition
\begin{align*}
P(z)=\sigma(z)(\cdot,D)+R(z)
\end{align*}
(see definition \eqref{sigma(cdot,D)}) where $\sigma: \Omega \to S(U\times \R^d)$ and 
$R: \Omega \to C^\infty(U \times U)$ are holomorphic.
\end{definition}

As a consequence, there exists a holomorphic map from $\Omega$ into $\Psi DO(M,E)$ with a holomorphic product 
(when $M$ is compact).

\begin{example} Elliptic operators:
\end{example}

Recall that $P \in \Psi DO^m(U)$, $m\in \C$, is elliptic if there exist strictly positive continuous functions $c$ and 
$C$ on $U$ such that $\vert \sigma^P(x,\xi) \vert \geq c(x)\,\vert \xi \vert ^m$ for $\xi \vert \geq C(x), \, x\in U$.
This essentially means that $P$ is invertible modulo smoothing operators. More generally, $P\in \Psi DO^m(M,E)$ is 
elliptic if its local expression in each coordinate chart is elliptic.

Let $Q\in \Psi DO^m(M,E)$ with $\Re(m)>0$. We assume that $M$ is compact and $Q$ is elliptic. Thus $Q$ has a 
discrete spectrum and we suppose Spectrum$(Q)\cap \R^-=\emptyset$. Since we want to integrate in $\C$, we 
assume that there exists a curve $\Gamma$ coming from $+\infty$ along the real axis in the upper half plane, 
turns around the origin and goes back to infinity in the lower half plane whose interior contains the spectrum 
of $Q$. The curve $\Gamma$ must avoid branch points of $\lambda^s$ at $s=0$ (so the branch of $\lambda^s$ defined in the right half-plane is such that $1^s= 1$).

\bigskip
\centerline{
\begin{tikzpicture}[>=angle 45,line width=0.8pt,
		point spectre/.style={draw=black,cross out,line width=0.8pt,inner sep=2.2pt,outer sep=0pt},
		decoration={markings,mark=at position 0.9 with {\arrowreversed[black, line width=0.5pt,scale=2]{angle 
		60}}}
		]
\draw[->] (-1,0) -- (9,0);
\draw[->] (0,-3) -- (0,3);
\node[point spectre] at (3,1) {};
\node[point spectre] at (4.3,1.4) {};
\node[point spectre] at (2.1,-0.3) {};
\node[point spectre] at (5.3,-1.5) {};
\node[point spectre] at (7.3,-1.9) {};
\node[point spectre] at (6.4,0.2) {};
\node[point spectre] at (7.5,1.7) {};
\node[point spectre] at (1,0) {};
\node[point spectre] at (5,0) {};
\node[point spectre] at (7,0) {};
\draw[postaction={decorate}] plot[domain=-2.9:2.9,smooth] ({(\x)^2 - 0.5},\x);
\node[below=3pt] at (8.2,3.5) {$\Gamma$};
\node[anchor=south west] (0,0) {$O$};
\end{tikzpicture}
}
\bigskip

When $\Re(s)<0$, $Q^s \vcentcolon=\tfrac{1}{i2\pi}\int_\Gamma \lambda^s\,(\lambda-Q)^{-1}\,d\lambda$ makes 
sense as operator on $L^2(M,E)$ (when $\Re(s)\geq 0$, define $Q^s \vcentcolon= Q^{s-k}\,Q^k$ for $k\in \N$ large enough so $\Re(s)-k<0$; there is no unambiguity because $Q^s\,Q = Q^{s+1}$ for $\Re(s)<-1$. Note that if $\Re(m)<0$, we can define $Q^s={(Q^{-1})}^{-s}$ and apply previous definition). 

Actually, $Q^s \in \Psi DO^{ms}(M,E)$ and $(\lambda-Q)^{-1}=\sigma(\lambda)(\cdot,D)+R(\lambda)$ where 
$R(\lambda)$ is a regularizing operator and $\sigma(\lambda)(\cdot,D)$ has a symbol smooth in $\lambda$ 
such that  $\sigma(\lambda)(x,\xi) \simeq \sum_{j\geq 0} a_{-m-j}(\lambda,x,\xi)$ 
with $a_{n}(\lambda,x,\xi)$ homogeneous of degree $n$ in $(\lambda^{1/m},\xi)$.

The map $s \to Q^s$ is a one-parameter group containing $Q^0=1$ and $Q^1=Q$ which is holomorphic on 
$\Re(s)\leq 0$.

We want to integrate symbols, so we will need the set $S_{int}$ of integrable symbols. 
Using same type of arguments as in Proposition \ref{extprop} and Lemma \ref{extc}, one proves

\begin{prop}
\label{Lext}
Let 
\begin{align*}
L: \sigma \in S_{int}^\Z(\R^d) \to L(\sigma) \vcentcolon=\widecheck{\sigma}\,(0)=
\tfrac{1}{(2\pi)^d}\, \int_{\R^d} \sigma(\xi)\,d\xi.
\end{align*}
Then $L$ has a unique holomorphic extension $\widetilde L$ on $S^{\C \backslash \Z}(\R^d)$.
\\
Moreover, when $\sigma(\xi)\simeq \sum_j \sigma_{m-j}(\xi)$, $m \in \C \backslash \Z$, 
 $$
 \widetilde L(\sigma)=\big( \sigma-\sum_{j\leq N} \tau_{m-j}\big)\,\widecheck{ } \,\,(0)=
 \tfrac{1}{(2\pi)^d}\, \int_{\R^d} \big(\sigma-\sum_{j\leq N} \tau_{m-j} \big) (\xi) \, d\xi
 $$
 where $m$ is the order of 
 $\sigma$, $N$ is an integer with $N> \Re(m)+d$ and $\tau_{m-j}$ is the extension of $\sigma_{m-j}$ of 
 Lemma \ref{extc}.
\end{prop}
$\widetilde L$ is holomorphic extension of $L$ on $S^{\C \backslash \Z}(\R^d)$ which is unique since 
every element of $S^{\C \backslash \Z}(\R^d)$ is arcwise connected to $S_{int}(\R^d)$ via a holomorphic 
path within $S^{\C \backslash \Z}(\R^d)$.

This result has an important consequence here:

\begin{corollary}
\label{Lres}
If $\sigma: \,\C \to S(\R^d)$ is holomorphic and order$\big( \sigma(s) \big)=s$, then $\widetilde L\big(\sigma(s)\big)$ 
is meromorphic with at most simple poles on $\Z$ and for $p\in \Z$,
\begin{align*}
\,\underset{s=p}{\Res}\,\widetilde L\big(\sigma(s)\big)=-\tfrac{1}{(2\pi)^d} \,
\int_{\mathbb{S}^{d-1}} \,\sigma_{-d}(p)(\xi) \, d\xi.
\end{align*}
\end{corollary}

\begin{proof}
Using Lemma \ref{extc}, one proves that if $m(s)$ is holomorphic near $m(s)=p$, then 
$\widetilde L \big(\sigma(s)\big)$ is meromorphic near $p$. 

Now we look at the singularity near $p\in \Z$. In the half plane $\set{\Re(s)<p}$, only the infrared part of 
$\tau_{m-j}(s)$ is a problem since its ultraviolet part is holomorphic. For $0\leq j \leq p+m$ and $\Re(s)<p$, 
$\sigma_{s-j}(s)(\xi)$ is integrable near 0 thus defines its unique extension $\tau_{s-j}(s)$. So, the only 
possible singularity near $s=p$ could come from 
\begin{align*}
-\tfrac{1}{(2\pi)^d} \,\int_{\vert \xi \vert \leq 1}\sigma_{s-j}(s)(\xi)\, d\xi&=-\tfrac{1}{(2\pi)^d} \,\int_0^1 t^{s-j+d-1}\,dt 
\int_{\vert \xi \vert \leq 1} \sigma_{s-j}(s)(\tfrac{\xi}{\vert \xi \vert})\,d(\tfrac{\xi}{\vert \xi \vert})\\
&=-\tfrac{1}{(2\pi)^d} \, \frac{1}{s-j+d} \int_{\mathbb{S}^{d-1}}\sigma_{s-j}(s)(\xi)\, d\xi.
\end{align*}
where we used for the first equality $\sigma_{s-j}(s)(\xi)=\vert \xi \vert^{s-j} \sigma_{s-j}(s)(\tfrac{\xi}{\vert \xi \vert})$. 
Thus, $\widetilde L\big(\sigma(s)\big)$ has at most only simple pole at $s=-d+j$. 
\end{proof}
We are now ready to get the main result of this section which is due to Wodzicki \cite{Wodzicki,Wodzicki1}.
  
\begin{definition}
Let $\DD\in \Psi DO(M,E)$ be an elliptic pseudodifferential operator  of order 1 on a boundary-less compact 
manifold $M$ endowed with a vector bundle $E$. 
\\
Let $\Psi DO_{int}(M,E) \vcentcolon=\set{Q \in \Psi DO^\C (M,E)\, \,\vert\, \, \Re\big(order(Q)\big) <-d}$ be the class of 
pseudodifferential operators whose symbols are in $S_{int}$, i.e. integrable in the $\xi$-variable. 
\\
In particular, if 
$P \in \Psi DO_{int}(M,E) $, then its kernel $k^P(x,x)$ is a smooth density on the diagonal of $M \times M$ with 
values in $End(E)$.

For $P \in \Psi DO^{\Z}(M,E)$, define
\begin{align}
\label{Wresdef}
\text{WRes }  P\vc \,\underset{s=0}{\Res}\, \Tr\big(P\vert \DD\vert^{-s}\big).
\end{align}
\end{definition}
\noindent This makes sense because:

\begin{theorem}
\label{Wresresult}
(i) Let $P:\Omega\subset \C \to \Psi DO_{int}(M,E)$ be a holomorphic family. Then the functional map $\Tr:\,s\in\Omega  \to \Tr(P(s))\in \C$ has a unique analytic extension on the family $\Omega \to \Psi DO^{\C \backslash \Z}(M,E)$ still denoted by $\Tr$.

(ii) If $P \in \Psi DO^{\Z}(M,E)$, the map: $s\in \C \to \Tr \big(P\vert \DD\vert^{-s}\big)$ has at most 
simple poles on $\Z$ and
\begin{align}
\label{Wresint}
\text{WRes }  P= \int_M c_P(x) \, \vert dx \vert
\end{align}
is independent of $\DD$. Recall (see Theorem \ref{kernelbehavior}) that
$c_P(x)=\tfrac{1}{(2\pi)^d} \int_{\mathbb{S}^{d-1}} \tr\big(\sigma_{-d}^P(x,\xi)\big)\,d\xi.$

(iii) $\text{WRes}$ is a trace on the algebra $\Psi DO^\Z (M,E)$.
\end{theorem}

\begin{proof}
$(i)$ The map $s \to \Tr \big( P \vert D\vert^{-s}\big)$ is holomorphic on $\C$ and connect  
$P \in \Psi DO^{\C \backslash \Z}(M,E)$ to the set $\Psi DO_{int}(M,E)$ within $\Psi DO^{\C \backslash \Z}(M,E)$, 
so an analytic extension of $\Tr$ from $\Psi DO_{int}$ to $\Psi DO^{\C \backslash \Z}$ is necessarily unique. 

$(ii)$ one apply the above machinery:

(1) Notice that $\Tr$ is holomorphic on smoothing operator, so, using a partition of unity, we can 
reduce to a local study of scalar $\Psi DO$'s.

(2) First, fix $s=0$. We are interested in the function $L_\phi(\sigma)\vcentcolon=\Tr \big(\phi \,\sigma(x,D)\big)$
with $\sigma \in S_{int}(U\times \R^d)$ and $\phi \in C^\infty(U)$. For instance, if $P=\sigma(\cdot,D)$, 
$$
\Tr(\phi\,P)=\int_U \phi(x)\,k^P(x,x)\,\vert dx \vert=\tfrac{1}{(2\pi)^d}\int_U \phi(x)\,\sigma(x,\xi) \,d\xi\,\vert dx\vert
=\int_U \phi(x)\,L(\sigma(x,\cdot)) \,\vert dx\vert,
$$
so one extends $L_\phi$ to $S^{\C\backslash \Z}(U\times \R^d)$ with Proposition \ref{Lext} via 
$\widetilde L_\phi(\sigma)=\int_U \phi(x)\,\widetilde L_\phi \big(\sigma(x,\cdot)\big)\,\vert dx\vert$.

(3) If now $\sigma(x,\xi)=\sigma(s)(x,\xi)$ depends holomorphically on $s$, we get uniform bounds in $x$, thus we 
get, via Lemma \ref{extc} applied to $\widetilde L_\phi \big(\sigma(s)(x,\cdot)\big)$ uniformly in $x$, yielding 
a natural extension to $\widetilde L_\phi \big(\sigma(s)\big)$ which is holomorphic on $\C\backslash \Z$.
\\
When order$(\sigma(s))=s$, the map $\widetilde L_\phi \big(\sigma(s)\big)$ has at most simple poles on $\Z$ and 
for each $p\in \Z$, 
$\underset{s=p}{\Res}\,\widetilde L_\phi \big(\sigma(s)\big)=-\tfrac{1}{(2\pi)^d}\, \int_U \int_{\mathbb{S}^{d-1}} \phi(x) \,
\sigma_{-d}(p)(x,\xi)\, d\xi\,\vert dx\vert=-\int_U \phi(x) \,c_{Pp}(x)\,\vert dx\vert$ where we used \eqref{cP} with 
$P=Op\big(\sigma_p(x,\xi)\big)$.

(4) In the general case, we get a unique meromorphic extension of the usual trace $\Tr$ on $\Psi DO^\Z(M,E)$ 
that we still denoted by $\Tr$). 

\noindent When 
$P: \, \C \to \Psi DO^\Z(M,E)$ is meromorphic with order$(\big(P(s)\big)=s$, then $\Tr\big(P(s)\big)$ has at most 
poles on $\Z$ and $\underset{s=p}{\Res}\,\Tr\big(P(s)\big)=-\int_M c_{P(p)}(x) \,\vert dx \vert$ for $p\in \Z$. So we get 
the claim for the family 
$$
P(s)\vcentcolon = P\vert \DD \vert^{-s}.
$$

$(iii)$ Let $P_1,P_2\in \Psi DO^\Z(M,E)$. Since $\Tr$ is a trace on $\Psi DO^{\C\backslash\Z}(M,E)$, we get by $(i)$, 
$\Tr\big(P_1P_2 \vert \DD \vert^{-s}\big)=\Tr\big(P_2 \vert \DD \vert^{-s}P_1\big)$. Moreover 
$$
{\it WRes} \big(P_1P_2\big) = \underset{s=0}{\Res}\,\Tr\big(P_2 \vert \DD \vert^{-s}P_1\big)=
\underset{s=0}{\Res}\,\Tr\big(P_2 P_1\vert \DD \vert^{-s}\big)={\it WRes} \big(P_2P_1\big)
$$
where for the second equality we used \eqref{Wresint} so the residue depends only of the value of $P(s)$ at $s=0$. 
\end{proof}
Note that {\it WRes} is invariant by diffeomorphism: 
\begin{align}
\label{Invdiff}
\text{if } \,\phi \in \text{Diff}(M), 
{\it WRes}(P)={\it WRes}(\phi_*P) 
\end{align}
which follows from \eqref{phi cP}.
The next result is due to Guillemin and Wodzicki.

\begin{corollary}
\label{uniquetrace}
The Wodzicki residue WRes is the only trace (up to multiplication by a constant) on the algebra $\Psi DO^{-\N}(M,E)$, when $M$ is connected and $d\geq 2$.
\end{corollary}

\begin{proof}
The restriction to $d\geq 2$ is used only in the part {\it 3)} below. When $d=1$, $T^*M$ is disconnected and they 
are two residues.

{\it 1) On symbols, derivatives are commutators:}
\begin{align*}
[x^j,\sigma]=i\partial_{\xi_j} \sigma, \qquad [\xi_j,\sigma]=-i\partial_{x^j}\sigma.
\end{align*}

{\it 2) If $\sigma_{-d}^P=0$, then $\sigma^P(x,\xi)$ is a finite sum of commutators of symbols:}
\\
When $\sigma^P \simeq \sum_j \,\sigma^P_{m-j}$ with $m=order(P)$, by Euler's theorem, 
$$
\sum_{k=1}^d \xi_k \,\partial_{\xi_k} \,\sigma_{m-j}^P=(m-j)\,\sigma_{m-j}^P
$$ 
(this is false for $m=j$!) and
\begin{align*}
\sum_{k=1}^d \, [x^k,\,\xi_k \, \sigma_{m-j}^P]=i\sum_{k=1}^d\partial_{\xi_k} \xi_k \,\sigma_{m-j}^P
=i(m-j+d)\,\sigma_{m-j}^P.
\end{align*}
So $\sigma^P=\sum_{k=1}^d\,[\xi_k \,\tau,x^k]$ (where $\tau \simeq i\sum_{j\geq 0} \tfrac{1}{m-j+d}\,\sigma^P_{m-j}$ 
and here we need for $m-j=-d$ that $\sigma^P_d=0$!).

Let $T$ be another trace on $\Psi DO^\Z(M,E)$. Then $T(P)$ depends only on $\sigma^P_{-d}$ because  
$T([\cdot,\cdot])=0$.

{\it 3) We have $\int_{\mathbb{S}^{d-1}} \sigma_{-d}^P(x,\xi)\, d\vert \xi \vert =0$ if and only if $\sigma_{-d}^P$ is 
sum of derivatives:}
\\
The if part is direct (less than more !).
\\
Only if part: $\sigma_{-d}^P$ is orthogonal to constant functions on the sphere $\mathbb{S}^{d-1}$ and these are 
kernels of the Laplacian: $\Delta_{\mathbb{S}}f=0 \Longleftrightarrow df=0 \Longleftrightarrow f=cst$. Thus 
$\Delta_{\mathbb{S}^{d-1}}h={\sigma_{-d}^P}{}_{\upharpoonright \mathbb{S}^{d-1}}$ has a solution $h$ on 
$\mathbb{S}^{d-1}$. If $\tilde h(\xi) \vcentcolon = \vert \xi \vert^{-d+2}\, h\big(\tfrac{\xi}{\vert \xi \vert}\big)$ is its 
extension to $\R^d\backslash \set{0}$, then we get  
$\Delta_{\R^d} \tilde h(\xi)= \vert \xi \vert \sigma_{-d}^P \big(\tfrac{\xi}{\vert \xi \vert}\big)= \sigma_{-d}^P(\xi)$ because 
$\Delta_{\R^d}=r^{1-d}\,\partial_r\big(r^{d-1}\,\partial_r)+r^{-2}\,\Delta_{\mathbb{S}^{d-1}}$. This means that $\tilde h$ 
is a symbol of order $d-2$ and $\partial_\xi \tilde h$ is a symbol of order $d-1$. As a consequence,  
$\sigma_{-d}^P=\sum_{k=1}^d \partial^2_{\xi_k} \tilde h=-i \sum_{k=1}^d [\partial_{\xi_k} \tilde h,\,x^k]$ is a sum of 
commutators.

{\it 4) End of proof:}
\\
$\sigma_{-d}^P(x,\xi)-\tfrac{\vert \xi \vert^{-d}}{\Vol (\mathbb{S}^{d-1})}\,c_P(x)$ is a symbol of order $-d$ with zero 
integral, thus is a sum of commutators by {\it 3)} and $T(P)=T\big(Op(\vert \xi \vert ^{-d}\,c_p(x)\big)$ for all 
$T\in \Psi DO^\Z(M,E)$. In other words, the map $\mu:\,f\in C^\infty_c(U) \to T\big( Op(f\vert \xi \vert^{-d})\big)$ is 
linear, continuous and satisfies $\mu(\partial_{x^k}f)=0$ because $\partial_{x^k}(f) \,\vert \xi \vert^{-d}$ is a 
commutator if $f$ has a compact support and $U$ is homeomorphic to $\R^d$. As a consequence, $\mu$ is a 
multiple of the Lebesgue integral:
\begin{align*}
T(P)=\mu\big(c_P(x)\big)=c\, \int_Mc_P(x)\,\vert dx\vert=c\,{\it WRes}(P).
\end{align*}
\end{proof}

\begin{example} 
\label{WResLaplacian}
Laplacian on a manifold $M$: 
Let $M$ be a compact Riemannian manifold of dimension $d$ and $\Delta$ be the scalar Laplacian which is a 
differential operator of order 2. Then 
\begin{align*}
\text{WRes}\big((1+\Delta)^{-d/2}\big)=\Vol \big(\mathbb{S}^{d-1}\big)=\tfrac{2 \pi^{d/2}}{\Gamma(d/2)}\, .
\end{align*}

\begin{proof}
$(1+\Delta)^{-d/2} \in \Psi DO(M)$ has order $-d$ and its principal symbol $\sigma^{(1+\Delta)^{-d/2} }_{-d}$ satisfies 
\begin{align*}
\sigma^{(1+\Delta)^{-d/2} }_{-d}(x,\xi)=\big(g_x^{ij}\, \xi_i\xi_j\big)^{-d/2}=\vert \vert \xi\vert \vert^{-d}_x.
\end{align*}
So \eqref{Wresint} gives
\begin{align*}
{\it WRes}\big((1+\Delta)^{-d/2}\big)&=\int_M \vert dx \vert\, \int_{\mathbb{S}^{d-1}} \vert \vert \xi\vert \vert^{-d}_x \, d\xi
=\int_M \vert dx \vert \sqrt{det\,g_x}\,\Vol \big(\mathbb{S}^{d-1}\big)\\
&=\Vol \big(\mathbb{S}^{d-1}\big)\,\int_M \vert dvol_g \vert=\Vol \big(\mathbb{S}^{d-1}\big).
\end{align*}
\end{proof}
\end{example}

\newpage
\section{Dixmier trace}
\label{Dixmiertrace}

References for this section:  \cite{Landi, Ponge1,CM2,Polaris, Var, VarillyS}.
\medskip

The trace on the operators on a Hilbert space $\H$ has an interesting property, it is {\it normal}.
\\
Recall first that $\Tr$ acting on $\B(\H)$ is a particular case of a weight $\omega$ acting on a von Neumann 
algebra $\M$: it is a homogeneous additive map from positive elements 
$\M^+\vc\set{aa^* \,\,\vert \,\, a\in \M}$ to $[0,\infty]$.\\
A state is a weight $\omega\in \M^*$ (so $\omega(a)<\infty, \,\forall a \in \M$) such that $\omega(1)=1$.\\
A trace is a weight such that $\omega(aa^*)=\omega(a^*a)$ for all $a\in \M$.

\begin{definition}
A weight $\omega$ is normal if $\,\omega(\underset{\a}{\sup} \, \,a_\a) = \underset{\a}{\sup}\, \,\omega(a_\a)$ 
whenever $(a_\a)\subset \M^+$ is an increasing bounded net.
\end{definition}
This is equivalent to say that $\omega$ is lower semi-continuous with respect to the $\sigma$-weak topology.

\begin{lemma} 
The usual trace $\Tr$ is normal on $\B(\H)$.
\end{lemma}
\noindent Remark that the net $( a_\a)_\a$ converges in $\B(\H)$ and this property looks innocent since a 
trace preserves positivity.  
\\
Nevertheless it is natural to address the question: are all traces (in particular on 
an arbitrary von Neumann algebra) normal? In 1966, Dixmier answered by the negative \cite{Dixmier} by exhibiting 
non-normal, say singular, traces. Actually, his motivation was to answer the following related question: is any trace 
$\omega$ on $\B(\H)$ proportional to the usual trace on the set where $\omega$ is finite?

\medskip
The aim of this section is first to define this Dixmier trace, which essentially means 
 $\Tr_{Dix}(T) \,``="\,\lim_{N\to \infty} \tfrac{1}{\log N} \sum_{n=0}^N \mu_n(T)$, where the $\mu_n(T)$ are the 
singular values of $T$ ordered in decreasing order and then to relate this to the Wodzicki trace. 
It is a non-normal trace on some set that we have to identify. 
Naturally, the reader can feel the link with the Wodzicki trace via Proposition \ref{propasymp}. 
We will see that if $P\in \Psi DO^{-d}(M)$ where $M$ is a compact Riemannian manifold of dimension $d$, then,
\begin{align*}
\Tr_{Dix}(P)=\tfrac{1}{d} \,{\it WRes} (P)=\tfrac{1}{d}\int_M \int_{S^*M} \sigma_{-d}^P(x,\xi)\, d\xi \vert dx \vert
\end{align*}
where $S^*M$ is the cosphere bundle on $M$.

The physical motivation is quite essential: we know how $\sum_{n\in \N^*} \tfrac{1}{n}$ diverges and this is related 
to the fact the electromagnetic or Newton gravitational potentials are in $\tfrac{1}{r}$ which has the same singularity 
(in one-dimension as previous series). Actually, this (logarithmic-type) divergence appears everywhere in 
physics and explains the widely use of the Riemann zeta function $\zeta: s \in \C \to \sum_{n\in \N^*} \tfrac{1}{n^s}$. 
This is also why we have already seen a logarithmic obstruction in Theorem \ref{kernelbehavior} and define a zeta 
function associated to a pseudodifferential operator $P$ by $\zeta_P(s)=\Tr\big(P\vert \DD \vert^{-s}\big)$ in 
\eqref{Wresdef}.

We now have a quick review on the main properties of singular values of an operator.

\subsection{Singular values of compact operators}

In noncommutative geometry, infinitesimals correspond to compact operators: for $T \in \K(\H)$ (compact 
operators), define for $n\in \N$ 
$$
\mu_n(T)\vcentcolon = 
\inf \set{ \lVert T_{\upharpoonright E^\perp} \rVert  \, \vert \,E \text{ subspace of }\H \text{ with dim}(E)=n}.
$$
This could looks strange but actually, by mini-max principle, $\mu_n(T)$ is nothing else than the 
$(n+1)$th of eigenvalues of $\vert T \vert$ sorted in decreasing order. Since $\lim_{n\to \infty} \mu_n(T) =0$, for any 
$\epsilon>0$, there exists a finite-dimensional subspace $E_\epsilon$ such that 
$\norm{T_{\upharpoonright E_\epsilon^\perp}}<\epsilon$ and this property being equivalent to $T$ compact, $T$ 
deserves the name of infinitesimal.

Moreover, we have following properties:

\qquad $\mu_n(T)=\mu_n(T^*)=\mu_n(\vert T \vert)$.

\qquad $T\in \L^1(\H)$ (meaning $\lVert T \rVert_1 \vcentcolon =\Tr(\vert T \vert)<\infty$) 
$\Longleftrightarrow \sum_{n\in \N} \mu_n(T) < \infty$.

\qquad $\mu_n(A\,\,T\,B) \leq \lVert A \rVert \, \mu_n(T)\, \lVert B \rVert$ when $A,B \in \B(\H)$.

\qquad $\mu_N(U\, T \,U^*)=\mu_N(T)$ when $U$ is a unitary.

\begin{definition}
For $T\in \K(\H)$, the partial trace of order $N\in \N$ is $\sigma_N(T) \vcentcolon=\sum_{n=0}^N \mu_n(T)$.
\end{definition}
\noindent Remark that $\lVert T \rVert \leq \sigma_N(T) \leq N \lVert T \rVert$ which implies 
$\sigma_n \simeq \lVert \cdot \rVert$ on $\K(\H)$. Then 
\begin{align}
&\sigma_N(T_1+T_2) \leq\sigma_N(T_1)+ \sigma_N(T_2), \nonumber\\
&\sigma_{N_1}(T_1)+\sigma_{N_2}(T_2) \leq \sigma_{N_1+N_2}(T_1+T_2) \, \text{ when } T_1,T_2 \geq 0.
\label{sumsigma}
\end{align}
The proof of the sub-additivity is based on the fact that $\sigma_N$ is a norm on $\K(\H)$. Moreover
$$
T\geq 0 \Longrightarrow \sigma_N(T) =\sup \set{\Tr(T\,E)\,\, \vert \,\, E \text{ subspace of } \H \text{ with }\dim(E)=n}.
$$
which implies $\sigma_N(T) =\sup \set{\Tr(\lVert T\,E\rVert_1 \,\, \vert  \, \dim(E)=n}$ and gives the second inequality.
\\
The norm $\sigma_N$ can be decomposed:
$$
\sigma_N(T)=\inf \set{ \lVert x \rVert_1+N\norm{y} \, \vert \, \,T=x+y \text{ with }\,x \in \L^1(\H),\, y \in \K(\H)}.
$$
In fact if $ \tilde \sigma_N$ is the right hand-side, then the sub-additivity gives 
$\tilde \sigma_N \geq \sigma_N(T)$. To get equality, let $\xi_n\in \H$ be such that 
$\vert T \vert \xi_n =\mu_n(T) \xi_n$ and define $x_N\vcentcolon=\big(\vert T \vert-\mu_N(T)\big) E_N$, 
$y_N\vcentcolon=\mu_N(T) E_N+ \vert T \vert (1-E_N)$ where 
$E_N\vcentcolon=\sum_{n<N} \vert \xi_n\rangle \langle \xi_n\vert$. If $T=U\vert T \vert$ is the polar decomposition of 
$T$, then $T=Ux_N+Uy_N$ is a claimed decomposition of $T$ and 
$$
\tilde \sigma_N(T) \leq \lVert Ux_N \lVert_1 +N \lVert U y_N \rVert \leq \lVert x_N \rVert_1+N \lVert y_N \rVert \leq
\sum_{n<N} \big(\mu_n(T)-N\mu_n(T)\big)+N\mu_n(T) \leq \sigma_N(T).
$$
This justifies a continuous approach with the

\begin{definition}
The partial trace of $T$ of order $\lambda \in \R^+$ is 
$$
\sigma_\lambda(T) \vcentcolon = 
\inf \set{ \lVert x \rVert_1+\lambda \lVert y \rVert \, \vert \, T=x+y \text{ with }\,x \in \L^1(\H),\, y \in \K(\H)}.
$$
\end{definition}

It interpolates between  two consecutive integers since the map: $\lambda \to \sigma_\lambda(T)$ is concave for 
$T \in \K(\H)$ and moreover, it is affine between $N$ and $N+1$ because
\begin{align}
\label{sigin}
\sigma_\lambda(T)=\sigma_N(T)+(\lambda-N)\sigma_N(T), \text{ where } N=[\lambda].
\end{align}
Thus, as before, 
$$
\sigma_{\lambda_1}(T_1)+\sigma_{\lambda_2}(T_2)=
\sigma_{\lambda_1+\lambda_2}(T_1+T_2),  \text{ for } \,\lambda_1,\lambda_2\in \R^+,\,0\leq T_1,T_2 \in \K(\H).
$$
We define a real interpolate space between $\L^1(\H)$ and $\K(\H)$ by
$$
\L^{1,\infty} \vcentcolon = \set{T \in \K(\H) \,\, \vert \,\, \lVert T \rVert_{1,\infty} \vcentcolon= \sup_{\lambda\geq e} 
\tfrac{\sigma_\lambda(T)}{\log \,\lambda}< \infty}.
$$
If $\L^p(\H)$ is the ideal of operators $T$ such that 
$\Tr \big(|T|^p\big) < \infty$, so $\sigma_\lambda(T ) = \mathcal{O}(\lambda^{1-1/p})$, we have naturally 
\begin{align}
&\L^{1}(\H)\subset \L^{1,\infty} \subset \L^p(\H)\, \text{ for }p>1,\label{notsep}\\
&\norm{T}\leq \norm{T}_{1,\infty} \leq \norm{T}_1.\nonumber
\end{align}

\begin{lemma} 
$\L^{1,\infty} $ is a $C^*$-ideal of $\B(\H)$ for the norm $\lVert \cdot \rVert_{1,\infty}$.\\
Moreover, it is equal to the Macaev ideal
\begin{align*}
\L^{1,+} \vcentcolon = \set{T \in \K(\H) \,\, \vert \, \, \lVert T \rVert_{1,+} \vcentcolon 
= \sup_{N\geq 2} \tfrac{\sigma_N(T)}{\log(N)} < \infty}.
\end{align*}
\end{lemma}

\begin{proof}
$\lVert \cdot \rVert_{1,\infty}$ is a norm as supremum of norms. By \eqref{sigin},
$$
\sup_{\rho\geq e} \tfrac{\sigma_\rho(T)}{\log \rho} \leq \sup_{N\geq 2} \, \sup_{0 \leq \a \leq 1}
\frac{\sum_{n=0}^{N-1} \mu_N(T) +\a\mu_N(T)}{\log (N+\a)}
$$
and $\L^{1,+\infty}$ is a left and right ideal of $\B(\H)$ since 
$\lVert A \,T \,B \rVert_{1,\infty} \leq \lVert A \rVert\, \lVert T \rVert_{1,\infty} \lVert B \rVert$ for every $A,B \in \B(\H)$,  
$T\in \L^{1,\infty}$, and moreover 
$ \lVert T \rVert_{1,\infty}= \lVert T^* \rVert_{1,\infty}= \lVert \,\vert T\vert \, \rVert_{1,\infty}$.
\\
This ideal $\L^{1,\infty}$ is closed for $ \lVert \cdot \rVert_{1,\infty}$: this follows from a 3-$\epsilon$ argument since 
Cauchy sequences for $ \lVert \cdot \rVert_{1,\infty}$ are Cauchy sequences for each norm $\sigma_\lambda$ 
which are equivalent to $\lVert \cdot \rVert$.
\end{proof}
Despite this result, the reader should notice that $\lVert \cdot \rVert_{1,\infty} \neq \lVert \cdot \rVert_{1,+}$ since 
the norms are only equivalent.

\subsection{Dixmier trace}

We begin with a Ces\`aro mean of $\tfrac{\sigma_\rho(T)}{\log \, \rho}$ with respect of the Haar measure of the 
group $\R^*_+$:
\begin{definition}
For $\lambda \geq e$ and $T\in \K(\H)$, let
$$
\tau_\lambda(T) \vcentcolon= \tfrac{1}{\log \, \lambda} \, \int_e^\lambda \tfrac{\sigma_\rho(T)}{\log \, \rho} \, 
\tfrac{d\rho}{\rho} \,.
$$
\end{definition}
Clearly, $\sigma_\rho(T) \leq \log \rho\,\lVert T \rVert_{1,\infty}$ and $\tau_\lambda(T) \leq \lVert T \rVert_{1,\infty}$, 
thus the map: $\lambda \to \tau_\lambda(T)$  is in $C_b([e,\infty])$. It  is not additive on $\L^{1,\infty}$ but this 
defect is under control: 
$$
\tau_\lambda(T_1+T_2)-\tau_\lambda(T_1)-\tau_\lambda(T_2) \underset{\lambda \to \infty}{\simeq}\, 
\mathcal{O}\big(\tfrac{\log\, (\log \, \lambda)}{\log \, \lambda} \big) ,\, \text{ when } \, 0\leq T_1,T_2 \in \L^{1,\infty}.
$$
More precisely, using previous results, one get

\begin{lemma}
\begin{align*}
\vert \,\tau_\lambda(T_1+T_2)-\tau_\lambda(T_1)-\tau_\lambda(T_2) \, \vert \leq \, 
\big(\tfrac{\log \,2(2+\log\, \log \, \lambda)}{\log \, \lambda} \big) \,\lVert T_1+T_2 \rVert_{1,\infty} ,
\, \text{ when } \,  T_1,T_2 \in \L^{1,\infty}_+.
\end{align*}
\end{lemma}

\begin{proof}
By the sub-additivity of $\sigma_\rho$, $\tau_\lambda(T_1+T2)\leq \tau_\lambda(T_1)+\tau_\lambda(T_2)$ 
and thanks to \eqref{sumsigma}, we get $\sigma_\rho(T_1)+\sigma_\rho(T_1)\leq \sigma_{2\rho}(T_1+T_2)$. Thus 
\begin{align*}
\tau_\lambda(T_1)+\tau_\lambda(T_2)\leq \tfrac{1}{\log\,\lambda}\int_e^\lambda 
\tfrac{\sigma_{2\rho}(T_1+T_2)}{\log \, \rho} \, \tfrac{d\rho}{\rho} \leq \tfrac{1}{\log\,\lambda}\int_{2e}^{2\lambda} 
\tfrac{\sigma_{\rho}(T_1+T_2)}{\log \, \rho/2} \, \tfrac{d\rho}{\rho} 
\end{align*}
Hence, 
$(\log\,\lambda)\, \vert \,\tau_\lambda(T_1+T_2)-\tau_\lambda(T_1)-\tau_\lambda(T_2) \, \vert \leq \epsilon+\epsilon'$ 
with
\begin{align*}
&\epsilon \vc \int_e^\lambda \tfrac{\sigma_{\rho}(T_1+T_2)}{\log \, \rho} \, \tfrac{d\rho}{\rho} 
-\int_{2e}^{2\lambda} \tfrac{\sigma_{\rho}(T_1+T_2)}{\log \, \rho/2} \, \tfrac{d\rho}{\rho} \,,\\
& \epsilon' \vc \int_{2e}^{2\lambda} \sigma_{\rho}(T_1+T_2) \big( \tfrac{1}{\log\,\rho/2}-\tfrac{1}{\log\,\rho} \big) 
\, \tfrac{d\rho}{\rho}\,.
\end{align*}
By triangular inequality and the fact that $\sigma_\rho(T_1+T_2) \leq \log \rho \,\norm{T_1+T_2}_{1,\infty}$ when 
$\rho\geq e$,
\begin{align*}
\epsilon &\leq  \int_e^{2e} \tfrac{\sigma_{\rho}(T_1+T_2)}{\log \, \rho} \, \tfrac{d\rho}{\rho} 
+\int_{\lambda}^{2\lambda} \tfrac{\sigma_{\rho}(T_1+T_2)}{\log \, \rho} \, \tfrac{d\rho}{\rho} \\
& \leq \norm{T_1+T_2}_{1,\infty} \big( \int_e^{2e}  \tfrac{d\rho}{\rho}+\int_\lambda^{2\lambda}  \tfrac{d\rho}{\rho} \big)
\leq 2 \log(2) \norm{T_1+T_2}_{1,\infty}.
\end{align*}
Moreover,
\begin{align*}
\epsilon' &\leq \norm{T_1+T_2}_{1,\infty} \int_{2e}^{2\lambda} \log\,\rho \,\big( \tfrac{1}{\log\,\rho/2}
-\tfrac{1}{\log\,\rho} \big) \, \tfrac{d\rho}{\rho}
\leq  \norm{T_1+T_2}_{1,\infty} \int_{2e}^{2\lambda} \tfrac{log\,2}{\log\,\rho/2}\, \tfrac{d\rho}{\rho}\\
&\leq  \norm{T_1+T_2}_{1,\infty}\log(2)\,\log \, (\log \lambda).
\end{align*}
\end{proof}

The Dixmier's idea was to force additivity: since the map $\lambda \to \tau_\lambda(T)$ is in  $C_b([e,\infty])$ and 
$\lambda \to\big(\tfrac{\log \,2(2+\log\, \log \, \lambda)}{\log \, \lambda} \big) $ is in $C_0([e,\infty[)$, let us consider 
the $C^*$-algebra 
$$
\A\vcentcolon =C_b([e,\infty]) /C_0([e,\infty[).
$$
If $[\tau(T)]$ is the class of the map $\lambda \to \tau_\lambda(T)$ in $\A$, previous lemma shows that 
$[\tau]:\,T\to[\tau(T)]$ is additive and positive homogeneous from $\L^{1,\infty}_+$ into 
$\A$ satisfying $[\tau(UTU^*)]=[\tau(T)]$ for any unitary $U$. \\
Now let $\omega$ be a state on $\A$, namely a positive linear form on $\A$ with $\omega(1)=1$. \\
Then, $\omega \circ [\tau(\cdot)]$ is a tracial weight on $\L^{1,\infty}_+$ (a map from $\L^{1,\infty}_+$ to 
$\R^+$ which is additive, homogeneous and invariant under $T\to UTU^*$). Since $\L^{1,\infty}$ is a $C^*$-ideal of 
$\B(\H)$, each of its element is generated by (at most) four positive elements, and this map can be extended to 
a map $\omega \circ [\tau(\cdot)]: T \in \L^{1,\infty} \to \omega([\tau(T)] )\in \C$ 
such that $\omega([\tau(T_1T_2)])=\omega([\tau(T_2T_1)])$ for $T_1,T_2 \in \L^{1,\infty}$. This leads to the 
following

\begin{definition}
The Dixmier trace $\Tr_\omega$ associated to a state $\omega$ on $\A\vcentcolon =C_b([e,\infty]) /C_0([e,\infty[)$ is
$$
\Tr_{\omega}(\cdot) \vcentcolon = \omega \circ [\tau(\cdot)].
$$
\end{definition}

\begin{theorem}
$\Tr_{\omega}$ is a trace on $\L^{1,\infty}$ which depends only on the locally convex topology of $\H$, not of its 
scalar product.
\end{theorem}

\begin{proof}
We already know that $\Tr_\omega$ is a trace. \\
If $\langle \cdot, \cdot \rangle'$ is another scalar product on $\H$ giving the same topology as 
$\langle \cdot, \cdot \rangle$, then there exist an invertible $U\in \B(\H)$ with 
$\langle \cdot, \cdot \rangle'=\langle U\cdot, U\cdot \rangle$. Let $\H'$ be the Hilbert space for 
$\langle \cdot, \cdot \rangle'$ and $\Tr'_{\omega}$ be the associated Dixmier trace to a given state $\omega$.
 Since the singular value of 
$U^{-1}TU \in \K_+(\H')$ are the same of $T\in \K_+(\H)$, we get $\L^{1,\infty}(\H')=\L^{1,\infty}(\H)$ and 
\begin{align*}
\Tr_{\omega}'(T)=\Tr_{\omega}(U^{-1} T U)=\Tr_{\omega}(T) \text{ for }T\in \L^{1,\infty}_+.
\end{align*}
\end{proof}
\noindent {\it Two important points:} 

1) Note that $\Tr_{\omega}(T)=0$ if $T\in \L^1(\H)$ and more generally all Dixmier traces vanish on the closure 
for the norm $\norm{.}_{1,\infty}$ of the ideal of finite rank operators. In particular, Dixmier traces are not normal.

2) The $C^*$-algebra $\A$ is not separable, so it is impossible to exhibit any state $\omega$! Despite the inclusions 
\eqref{notsep} and the fact that the $\L^p(\H)$ are separable ideals for $p\geq 1$, $\L^{1,\infty}$ is not a separable.
\\
Moreover, as for Lebesgue integral, there are sets which are not measurable. For instance, a function 
$f \in C_b([e, \infty])$ has a limit $\ell =\lim_{\lambda \to \infty} f(\lambda)$ if and only if $\ell = \omega(f)$ for all 
state $\omega$. 

\begin{definition}
The operator $T \in \L^{1,\infty}$ is said to be measurable if $\Tr_\omega(T)$ is independent of $\omega$. In this 
case, $\Tr_\omega$ is denoted $\Tr_{Dix}$.
\end{definition}

\begin{lemma}
The operator $T \in \L^{1,\infty}$ is measurable and $\Tr_\omega(T)=\ell$ if and only if the map
$\lambda\in \R^+ \to \tau_\lambda(T)\in \A$ converges at infinity to $\ell$.
\end{lemma}
\begin{proof}
If $\lim_{\tau \to \infty} \tau_\lambda(T)=\ell$, then 
$\Tr_\omega(T)=\omega \big( \tau(T)\big)=\omega(\ell) =\ell \, \omega(1)=\ell$.

Conversely, assume $T$ is measurable and $\ell=\Tr_\omega(T)$ for any state $\omega$. Then we get, 
$\omega\big( \tau(T)-\ell \big)=\Tr_\omega(T)-\ell=0$. Since the set of states  separate the points of $\A$, 
$\tau(T)=\ell$ and $\lim_{\tau \to \infty} \tau_\lambda(T)=\ell$.
\end{proof}

After Dixmier, the singular (i.e. non normal) traces have been deeply investigated, see for instance the recent 
\cite{LPS,LS,LS1}, but we do not enter into this framework and technically, we just make the following 
characterization of measurability:
\begin{remark}
If $T \in \K_+(\H)$, then $T$ is measurable if and only if 
$\lim_{N\to \infty} \tfrac{1}{\log \,N} \, \sum_{n=0}^N \mu_n(T)$ exists. \\
{\rm Actually, if $\ell=\lim_{N\to \infty} \tfrac{1}{\log \,N} \, \sum_{n=0}^N \mu_n(T)$ since 
$\Tr_\omega(T)=\ell$ for any $\omega$, so $\Tr_{Dix}=\ell$ and the converse is proved in \cite{LSS}.}
\end{remark}

\begin{example} 
\label{Dixtrace}
Computation of the Dixmier trace of the inverse Laplacian on the torus:
\end{example}
Let $\T^d=\R^d /2\pi \Z^d$ be the d-dimensional torus and $\Delta=-\sum_{i=1}^d \partial_{x^i}^2$ be the scalar 
Laplacian seen as unbounded operator on $\H=L^2(\T^d)$. We want to compute 
$\Tr_\omega\big((1+\Delta)^{-p} \big)$ for $\tfrac{d}{2}\leq p\in \N^*$. We use $1+ \Delta$ to avoid the kernel 
problem with the inverse. As the following proof shows, 1 can be replaced by any $\epsilon>0$ and the result 
does not depends on $\epsilon$.

Notice that the functions $e_k(x)\vcentcolon = \tfrac{1}{\sqrt{2\pi}}\, e^{ik\cdot x}$ with $x\in \T^d,\,k\in \Z^d=(\T^d)^*$ form a 
basis of $\H$ of  eigenvectors: $\Delta \,e_k=\vert k \vert^2\,e_k$. Moreover, for $t\in \R^*_+$,
$$
e^t \Tr\big(e^{-t(1+\Delta)}\big)=\sum_{k\in \Z^d} e^{-t\vert k \vert^2}=\big(\sum_{k\in \Z}e^{-t k^2}\big)^d.
$$
We know that $\vert \int_{-\infty}^\infty e^{-tx^2}\,dx -\sum_{k\in \Z} e^{-t k^2} \vert \leq 1$ for any $t>0$, and since the first 
integral is $\sqrt{\tfrac{\pi}{t}}$, we get 
$e^t \Tr\big(e^{-t(1+\Delta)}\big) \underset{t \downarrow 0^+}{\simeq} \big(\tfrac{\pi}{t} \big)^{d/2}
=\vcentcolon \a \,t^{-d/2}$.
\\
We will use a Tauberian theorem: 
$\mu_n\big( (1+\Delta)^{-d/2}\big) \underset{n\to \infty}{\simeq} [\a\, \tfrac{1}{\Gamma(d/2+1)}]\,\tfrac{1}{n}$, 
see \cite{Hardy} (one needs to estimates the cardinality of the set $\set{k\in \Z^d \, \vert \, \vert k \vert^2 \leq n}$, see 
\cite{Polaris}). 
Thus 
$$
\lim_{N \to \infty} \tfrac{1}{\log \, N} \sum_{n=0}^N \mu_n\big( (1+\Delta)^{-d/2} \big) = 
\tfrac{\a}{\Gamma(d/2+1)}=\tfrac{\pi^{d/2}}{\Gamma(d/2+1)}.
$$
So $(1+\Delta)^{-d/2}$ is measurable and 
$$\Tr_{Dix} \big( (1+\Delta)^{-d/2} \big)=
\Tr_{\omega} \big( (1+\Delta)^{-d/2} \big)=\tfrac{\pi^{d/2}}{\Gamma(d/2+1)}.
$$
Since $(1+\Delta)^{-p}$ is traceable for $p>\tfrac{d}{2}$, $\Tr_{Dix}\big( (1+\Delta)^{-p} \big)=0$.
\\
This result has been generalized in Connes' trace theorem \cite{Connes88}:

\begin{theorem}
\label{TrDix}
Let $M$ be a compact Riemannian manifold of dimension $d$, $E$ a vector bundle over $M$ and 
$P \in \Psi DO^{-d}(M,E)$. Then, $P\in \L^{1,\infty}$, is measurable and 
$$
\Tr_{Dix}(P)=\tfrac{1}{d}\, \text{WRes}(P).
$$
\end{theorem}

\begin{proof}
Since {\it WRes} and $\Tr_{Dix}$ are traces on $\Psi DO^{-m}(M,E)$, $m\in \N$, $\Tr_{Dix}=c\,{\it WRes}$ for 
some constant $c$ 
using Corollary \ref{uniquetrace}. Above example, when compare with Example \ref{WResLaplacian}, shows that the 
$c=\tfrac{1}{d}$.
\end{proof}

\newpage
\section{Dirac operator}
\label{Diracoperator}

There are several ways to define a Dirac-like operator. The best one is to define Clifford algebras, their 
representations, the notion of Clifford modules, spin$^c$ structures on orientable manifolds $M$ defined by Morita 
equivalence between the $C^*$-algebras $C(M)$ and $\Gamma(\CC \ell\, M)$ (this approach is more of the spirit of 
noncommutative geometry). Then the notion of spin structure and finally, with the notion of spin and Clifford 
connection, we reach the definition of a (generalized) Dirac operator.
\\
Here we try to bypass this approach to save time.

References: a classical book is \cite{Lawson}, but I recommend \cite{Ginoux}. Here, we follow \cite{Ponge}, but 
see also \cite{Polaris}.

\subsection{Definition and main properties}

Let $(M,\,g)$ be a compact Riemannian manifold with metric $g$, of dimension $d$ and $E$ be a vector bundle 
over $M$. An example is the (Clifford) bundle $E=\CC \ell\, T^*M$ where the fiber $\CC \ell\,T_x^*M$ is the Clifford 
algebra of the real vector space $T^*_xM$ for $x \in M$ endowed with the nondegenerate quadratic form $g$.

Given a connection $\nabla$ on $E$, recall that a differential operator $P$ of order $m$ on $E$ is an element of 
$$
\text{Diff}^{\,m} (M,E)\vcentcolon = \Gamma\big(M,End(E)\big) \,\cdot\,Vect\set{\nabla_{X_1}
\cdots \nabla_{X_j} \, \vert \, X_j \in \Gamma(M,TM), \,j\leq m}.
$$
In particular, $\text{Diff}^{\,m} (M,E)$ is a subalgebra of $End\big( \Gamma(M,E)\big)$ and the operator $P$ 
has a principal symbol $\sigma_m^P$ in  $\Gamma\big(T^*M,\pi^*End(E)\big)$ where $\pi:T^*M \to M$ is the 
canonical submersion and $\sigma_m^P(x,\xi)$ is given by \eqref{symprinc}.

{\it An example:} Let $E=\exter \, T^*M$. The exterior product and the contraction given on $\omega,\omega_j \in E$ 
by
\begin{align*}
& \epsilon(\omega_1) \, \omega_2 \vcentcolon = \omega_1 \wedge \omega_2,\\
& \iota(\omega)\,(\omega_1 \wedge \cdots \wedge \omega_m) \vcentcolon =
 \sum_{j=1}^m (-1)^{j-1}g(\omega,\omega_j)\, 
\omega_1 \wedge \cdots \wedge \widehat{\omega_j} \wedge \cdots \wedge \omega_m
\end{align*} 
suggest the following definition $c(\omega) \vcentcolon = \epsilon (\omega)+ \iota(\omega)$ and one checks that
\begin{align}
\label{cliffmetric}
c(\omega_1)\,c(\omega_2)+c(\omega_2)c(\omega_1)=2g(\omega_1,\omega_2) \, \text{id}_E.
\end{align}
$E$ has a natural scalar product: if $e_1,\cdots,e_d$ is an orthonormal basis of $T_x^*M$,  then the scalar product 
is chosen such that $e_{i_1}\wedge\cdots \wedge e_{i_p}$ for $i_1<\cdots<i_p$ is an orthonormal basis.
\\
If $d \in \text{Diff}^{\,1}$ is the exterior derivative and $d^*$ is its adjoint for the deduced scalar product on 
$\Gamma(M,E)$, then their principal symbols are
\begin{align}
\label{dandd1}
& \sigma_1^d(\omega)=i\epsilon (\omega),\\
& \sigma_1^{d^*}(\omega)=-i \iota(\omega). \label{dandd2}
\end{align}
This follows from 
$\sigma_1^d(x,\xi)=\lim_{t \to \infty} \tfrac{1}{t}\big(e^{-ith(x)}de^{ith(x}\big)(x)=lim_{t \to \infty} 
\tfrac{1}{t} \,it\,d_xh=i\,d_xh=i\,\xi$ where $h$ is such that $d_xh=\xi$, so $\sigma_1^d(x,\xi)=i\,\xi$ and 
similarly for $ \sigma_1^{d^*}$.

More generally, if $P\in\text{Diff}^{\,m}(M)$, $\sigma_m^P(dh)=\tfrac{1}{i^mm!} \,(ad\,h)^m(P)$ with 
$ad\, h=[\cdot,h]$ and 
$\sigma_m^{P^*}(\omega)=\sigma_m^P(\omega)^*$ where the adjoint $P^*$ is for the scalar product on 
$\Gamma(M,E)$ associated to an hermitean metric on $E$: 
$\langle \psi,\psi' \rangle\vcentcolon = \int_M \langle \psi(x),\psi'(x) \rangle_x \, \vert dx \vert$ is a scalar 
product on the space $\Gamma(M,E)$.

\begin{definition}
\label{deflaplace}
The operator $P\in \text{Diff}^{\,\,2}(M,E)$ is called a generalized Laplacian when its symbol satisfies 
$\sigma_2^P(x,\xi)=\vert \xi \vert^2_x\, \,id_{E_x}$ for $x\in M,\, \xi \in T^*_xM$ (note that $\vert \xi \vert_x$ depends 
on the metric $g$).
\end{definition}
This is equivalent to say that, in local coordinates, 
$P=-\sum_{i,j} g^{ij}(x)\partial_{x^i}\partial_{x^j}+b^j(x)\partial_{x^j}+c(x)$ where the $b^j$ are smooth and $c$ is in 
$\Gamma\big(M,End(E)\big)$.

\begin{definition}
\label{defdir}
Assume that $E=E^+ \oplus E^-$ is a $\Z_2$-graded vector bundle.
\\
When $D \in\text{Diff}^{\,\,1}(M,E)$ and $D=\left(\begin{smallmatrix} 0 & D^+ \\ D^- & 0 \end{smallmatrix}\right)$ ($D$ 
is odd) where $D^{\pm}:\, \Gamma(M,E^{\mp}) \to \Gamma(M,E^{\pm})$, $D$ is called a Dirac operator if 
$D^2=\left(\begin{smallmatrix} D^-D^+ & 0\\ 0 & D^+D^-\end{smallmatrix}\right)$ is a generalized 
Laplacian.
\end{definition}
A good example is given by $E=\exter \,T^*M=\exter^{even} \,T^*M \oplus \exter^{odd}\, T^*M$ and the de Rham 
operator $D\vcentcolon= d+d^*$. It is a Dirac operator since $D^2=dd^*+d^*d$ is a generalized Laplacian according 
to \eqref{dandd1} \eqref{dandd2}. $D^2$ is also called the Laplace--Beltrami operator.

\begin{definition}
Define $\CC \ell \,M$ as the vector bundle over $M$ whose fiber in $x\in M$ is the Clifford algebra 
$\CC \ell \,T^*_xM$ (or $\CC \ell \,T_xM$ using the musical isomorphism $X \in TM \leftrightarrow X^\flat \in T^*M$).

A bundle $E$ is called a Clifford bundle over $M$ when there exists a $\Z_2$-graduate action 
$c: \, \Gamma(M,\CC \ell \,M) \to End\big(\Gamma(M,E)\big)$.
\end{definition}

The main idea which drives this definition is that Clifford actions correspond to principal symbols of Dirac operators: 
\begin{prop}
\label{carac}
If $E$ is a Clifford module, every odd $D \in \text{Diff}^{\,\,1}$ such that $[D,\,f]=i\,c(df)$ for $f \in C^\infty(M)$ is a 
Dirac operator.
\\
Conversely, if $D$ is a Dirac operator, there exists a Clifford action $c$ with $c(df)=-i\,[D,\,f]$.
\end{prop}

\begin{proof}
Let $x\in M$, $\xi \in T_x^*M$ and $f \in C^\infty(M)$ such that $d_xf=\xi$. Then 
\begin{align*}
\sigma_1^D(df)(x)=\big(\tfrac{1}{i} \,ad \, f)D= -i[D,\,f]=c(df),
\end{align*}
so, thanks to Theorem \ref{pseudo}, 
$\sigma_2^{D^2}(x,\xi)=\big(\sigma_1^D(x,\xi)\big)^2=\vert \xi \vert_x^2 \,\,id_{E_x}$ 
and $D^2$ is a generalized Laplacian.

Conversely, if $D$ is a Dirac operator, then we can define $c(df)\vcentcolon =i [D,\,f]$. This makes sense since 
$D\in \text{Diff}^{\,1}$ and for $f\in C^\infty(M)$, $x\in M$, $[D,f](x)=i\sigma_1^D(df)(x)=i \sigma_1^D(x,d_xf)$ is an 
endomorphism of $E_x$ depending only on $d_xf$. So $c$ can be extended to 
the whole $T^*M$ with $c(x,\xi)\vcentcolon= c(dh)(x)=i\sigma_1^D(x,\xi)$ where $h\in C^\infty(M)$ is chosen such 
that $\xi=d_xh$. The map $\xi \to c(x,\xi)$ is linear from $T_x^*M$ to $End(E_x)$ and 
$c(x,\xi)^2=\sigma_1^D(x,\xi)^2=\sigma_2^{D^2}(x,\xi)=\vert \xi \vert_x^2$ for each $\xi \in T_x^*M$. Thus $c$ 
can be extended to an morphism of algebras from $\CC \ell (T^*_xM)$ in $End(E_x)$. This gives a Clifford 
action on the bundle $E$.
\end{proof}
Consider previous example: $E=\exter \,T^*M=\exter^{even} \,T^*M \oplus \exter^{odd}\, T^*M$ is a 
Clifford module for $c\vcentcolon = i(\epsilon + \iota)$ coming from the Dirac operator $D=d+d^*$: by 
\eqref{dandd1} and \eqref{dandd2} 
$$
i[D,f]=i[d+d^*,\,f]=i\big(i\sigma_1^d(df)-i\sigma_1^{d^*}(df)\big)=-i(\epsilon+\iota)(df).
$$

\begin{definition}
Let $E$ be a Clifford module over $M$. A connection $\nabla$ on $E$ is a Clifford connection if for 
$a\in \Gamma(M,\CC \ell \,M)$ and $X\in \Gamma(M,TM)$, $[\nabla_X, c(a)]=c(\nabla^{LC}_X \, a)$ where 
$\nabla^{LC}_X$ is the Levi-Civita connection after its extension to the bundle $\CC \ell \,M$ (here $\CC\ell\,M$ is 
the bundle with fiber $\CC\ell\, T_xM$). 
\\
A Dirac operator $D_\nabla$ is associated to a Clifford connection $\nabla$ in the following way: 
$$
D_\nabla \vcentcolon = -i\,c \circ \nabla, \qquad 
\Gamma(M,E) \xrightarrow{\nabla} \Gamma(M,T^*M\otimes E) \xrightarrow{c\otimes 1} \Gamma(M,E).
$$
where we use $c$ for $c\otimes 1$.
\end{definition}

Thus if in local coordinates, $\nabla=\sum_{j=1}^d dx^j \otimes \nabla_{\partial_j}$, the associated Dirac operator is 
given by $D_\nabla=-i\sum_j c(dx^j) \, \nabla_{\partial_i}$. In particular, for $f \in C^\infty(M)$, 
$$[D_\nabla,\,f \,id_E]=-i\sum_{i=1}^d  c(dx^i) \,[\nabla_{\partial_j},\,f]=\sum_{j=1}^d-i\, c(dx^j) \,\partial_j f=-ic(df).$$
By Proposition \ref{carac}, $D_\nabla$ deserves the name of Dirac operator!

{\it Examples:} 

1) For the previous example $E=\bigwedge T^*M$,  the Levi-Civita connection is indeed a Clifford 
connection whose associated Dirac operator coincides with the de Rham operator $D=d+d^*$.

2) {\it The spinor bundle}: 
Recall that the spin group $\Spin_d$ is the non-trivial two-fold covering of $SO_d$, so we have
$$
0 \longrightarrow \Z_2 \longrightarrow \Spin_d \overset{\xi}{\longrightarrow}SO_d \longrightarrow 1.
$$
Let $\SO(TM) \to M$ be the $\SO_d$-principal bundle of positively oriented orthonormal frames on $TM$ 
of an oriented Riemannian manifold $M$ of dimension $d$.

A {\it spin structure} on an oriented d-dimensional Riemannian manifold $(M,g)$ is a $\Spin_d$-principal bundle 
$\Spin(TM) \overset{\pi}{\longrightarrow} M$ 
with a two-fold covering map $\Spin(TM) \overset{\eta}{\longrightarrow} \SO(TM)$ such that the following diagram 
commutes:

\begin{equation*}
\xymatrix@R=4pt@C=40pt    
{
{\Spin(TM)\times  \Spin_d} \ar[r] \ar[dd]_-{\eta \times \xi}
& {\Spin(TM)} \ar[dd]_-{\eta} \ar[rd]^{\pi}  \\
{} &{}&{M}\\
{\SO(TM)\times  \SO_d} \ar[r]
& {\SO(TM)} \ar[ru]_{\pi}
}
\end{equation*}
where the horizontal  maps are the actions of $\Spin_d$ and $\SO_d$ on the principal fiber bundles 
$\Spin(TM)$ and $\SO(TM)$.

A {\it spin manifold} is an oriented Riemannian manifold admitting a spin structure.
\\
In above definition, one can replace $\Spin_d$ by the group $\Spin^c_d$ which is a central extension of 
$SO_d$ by $\T$:
$$
0 \longrightarrow \T \longrightarrow \Spin^c_d \overset{\xi}{\longrightarrow}SO_d \longrightarrow 1.
$$

An oriented Riemannian manifold $(M,g)$ is spin if and only if the second Stiefel--Whitney class of its tangent 
bundle vanishes. Thus a manifold is spin if and only both its first and second Stiefel--Whitney classes vanish 
(the vanishing of the first one being equivalent to the orientability of the manifold). In this case, the set of spin 
structures on $(M,g)$ stands in one-to-one correspondence with $H^1(M, \Z_2)$. 
In particular the existence of a spin structure does not depend on the metric or the orientation 
of a given manifold. Note that all manifolds of dimension $d\leq 4$ have spin$^c$ structures but $\C P^2$ is a 
4-dimensional (complex) manifold without spin structures. 

Let $\rho$ be an irreducible representation of $\CC \ell \,\C^d \to End_\C(\Sigma_d)$ with 
$\Sigma_d \simeq\C^{2^{\lfloor d/2 \rfloor}}$ as set of complex spinors. Of course, $\CC \ell\,\C^d$ is endowed 
with its canonical complex bilinear form. 

The {\it spinor bundle} $S$ of $M$ is the complex vector bundle associated to the principal bundle $\Spin(TM)$ 
with the spinor representation, namely $S\vc \Spin(TM) \,\times_{\rho_d} \, \Sigma_{d}$. Here $\rho_d$ is a 
representation of $\Spin_d$ on Aut$(\Sigma_d)$ which is the restriction of $\rho$.

More precisely, if $d=2m$ is even, $\rho_d=\rho^++\rho^-$ where $\rho^\pm$ are two nonequivalent irreducible 
complex representations of $\Spin_{2m}$ and $\Sigma_{2m}=\Sigma^+_{2m} \oplus \Sigma^-_{2m}$, while for 
$d=2m+1$ odd, the spinor representation $\rho_d$ is irreducible.
\\
In practice, $M$ is a spin manifold means that there exists a Clifford bundle $S=S^+ \oplus S^-$ such 
that $S \simeq \exter \,T^*M$. Due to the dimension of $M$, the Clifford bundle has fiber
$$
\CC \ell _x M=\left\{\begin{array}{ll} M_{2^{m}}(\C) \text{ when } d=2m \text{ is even,}\\
M_{2^m}(\C) \oplus M_{2^m}(\C) \text{ when } d=2m+1. \end{array} \right.
$$
Locally, the spinor bundle satisfies $S \simeq M \times \C^{d/2}$. 
\\
{\it A spin connection} $\nabla^S:\, \Gamma^\infty(M,S) \to \Gamma^\infty (M,S)\otimes \Gamma^\infty (M,T^*M)$ is 
any connection which is compatible with Clifford action:
$$[\nabla^S,c(\cdot)]=c(\nabla^{LC} \cdot).
$$
It is uniquely determined by the choice of a spin structure on $M$ (once an orientation of $M$ is chosen). 

\begin{definition}
{\underline {\it The}} Dirac (also called Atiyah--Singer) operator given by the spin structure is
\begin{align}
\label{theDirac}
\Ds \vcentcolon = -i \,c \circ \nabla^S.
\end{align}
\end{definition}
In coordinates,
\begin{align}
\label{Diracco}
\Ds=-ic(dx^j)\big(\partial_j-\omega_j(x)\big)
\end{align}
where $\omega_j$ is the spin connection part which can be computed in the coordinate basis
$$
\omega_j=\tfrac{1}{4} \big(\Gamma_{ji}^k\,g_{kl}-\partial_i(h_j^\a)\delta_{\a\beta}h_l^\beta \big) c(dx^i)\,c(dx^l)
$$
where the matrix $H\vc[h^\a_j]$ is such that $H^tH=[g_{ij}]$ (we use Latin letters for coordinate basis indices 
and Greek letters for orthonormal basis indices).

This gives $\sigma_1^{D}(x,\xi)=c(\xi) +i c(dx^j)\,\omega_j(x)$. Thus in {\it normal coordinates around $x_0$}, 
\begin{align}
& c(dx^j)(x_0)=\ga^j, \nonumber\\
& \sigma_1^{D}(x_0,\xi)=c(\xi)=\ga^j \xi_j \label{gammamatrices}
\end{align}
where the $\ga$'s are constant hermitean matrices. 

A fundamental result concerning a Dirac operator (definition \ref{defdir}) is its unique continuation property: 
if $\psi$ satisfies $D \psi=0$ and $\psi$ vanishes on an open subset of the smooth manifold $M$ (with or 
without boundary), then $\psi$ also vanishes on the whole connected component of $M$.

The Hilbert space of spinors is 
\begin{align}
\label{Hilbert}
\H = L^2\big((M,g),S)\big) \vc \set{\psi\in \Gamma^\infty(M,S) \, \vert \, \int_M \langle \psi,\psi \rangle_x \, 
dvol_g(x)<\infty}
\end{align}
where we have a scalar product which is $C^\infty(M)$-valued. On its domain $\Gamma^\infty(M,S)$, the Dirac 
operator is symmetric: 
$\langle \psi,\Ds \phi \rangle=\langle \Ds\psi, \phi \rangle$. Moreover, it has a selfadjoint closure (which is 
$\Ds^{**}$):

\begin{theorem}
\label{D=D*}
Let $(M,g)$ be an oriented compact Riemannian spin manifold without boundary. By extension to $\H$, $\Ds$ is 
essentially selfadjoint on its original domain $\Gamma^\infty(M,S)$. It is a differential (unbounded) operator 
of order one which is elliptic.
\end{theorem}
\noindent See \cite{Var, VarillyS, Polaris,Lawson} for a proof.

There is a nice formula which relates the Dirac operator $\Ds$ to the spinor Laplacian 
$$
\Delta^S \vcentcolon = -\Tr_g(\nabla^S \circ \nabla^S): \, \Gamma^\infty(M,S) \to \Gamma^\infty(M,S).
$$
Before to give it, we need to fix few notations: let $R\in \Gamma^\infty \big(M,\exter^2\,T^*M \otimes End(TM)\big)$ 
be the {\it Riemann curvature tensor} with components 
$R_{ijkl}\vcentcolon=g(\partial_i,R(\partial_k,\partial_l)\partial_j)$, the {\it Ricci tensor} components are 
$R_{jl} \vcentcolon =g^{ik}T_{ijkl}$ and the {\it scalar curvature} is $s \vcentcolon = g^{jl}R_{jl}$.

\begin{prop}
\label{Lich}
Schr\"odinger--Lichnerowicz formula: with same hypothesis, 
\begin{align}
\label{SLform}
\Ds^2=\Delta^S+\tfrac{1}{4} s
\end{align}
where $s$ is the scalar curvature of $M$.
\end{prop}
The proof is just a lengthy computation (see for instance \cite{Polaris}).
 \\
We already know via Theorems \ref{compact} and \ref{D=D*} that $\Ds^{-1}$ is compact so 
has a discrete spectrum. 
\\
For $T \in \K_+(\H)$, we denote by $\{\lambda_n(T)\}_{n\in \N}$ its spectrum sorted in 
decreasing order including multiplicity (and in increasing order for an unbounded positive operator $T$ such that 
$T^{-1}$ is compact) and by 
$N_{T}(\lambda) \vcentcolon = \#\set{\lambda_n(T)\,\, \vert \,\, \lambda_n  \leq \lambda}$ its counting function.

\begin{theorem}
\label{Weyl}
With same hypothesis, the asymptotics of the Dirac operator counting function is 
$N_{\lvert \Ds \rvert}(\lambda) \underset{\lambda \to \infty}{\sim} \, 
\tfrac{2^d\text{Vol}(\mathbb{S}^{d-1})}{d(2\pi)^d} \, \text{Vol}(M) \, \lambda^{d}$ where $\Vol (M)=\int_M dvol$.
\end{theorem}

\begin{proof}
By Weyl's theorem, we know the asymptotics of $N_{\Delta}(\lambda)$ for the the scalar Laplacian 
$\Delta \vcentcolon = -\Tr_g\big(\nabla^{T^*M \otimes T^*M} \circ \nabla^{T^*M} \big)$ which in coordinates is 
$\Delta= -g^{ij}(\partial_i\partial_j-\Gamma_{ij}^k \partial_k)$. It is given by:
$$
N_\Delta(\lambda)  \underset{\lambda \to \infty}{\sim} \,\tfrac{\text{Vol}(\mathbb{S}^{d-1})}{d(2\pi)^d}  \text{ Vol}(M) \, 
\lambda^{d/2}.
$$
For the spinor Laplacian, we get the same formula with an extra factor of $\Tr(1_S)=2^d$ and 
Proposition \ref{Lich} shows that $N_{\Ds^2}(\lambda)$ has the same asymptotics than $N_\Delta(\lambda)$ 
since $s$ gives rise to a bounded operator.
\end{proof}
We already encounter such computation in Example \ref{Dixtrace}.

\subsection{Dirac operators and change of metrics}

Recall that the spinor bundle $S_g$ and square integrable spinors $\H_g$ defined in \eqref{Hilbert} depends on the 
chosen metric $g$, so we note $M_g$ instead of $M$ and $\H_g:=L^2(M_g,S_g)$ and a natural question is: what 
happens to a Dirac operator when the metric changes?

Let $g'$ be another Riemannian metric on $M$. Since the space of $d$-forms is one-dimensional, there exists a 
positive function $f_{g,g'}:M \to \R^+$ such that $dvol_{g'}=f_{g',g}\,dvol_{g}$. 

Let $I_{g,g'}(x):  S_g \to  S_{g'}$ the natural injection on the spinors spaces above point $x\in M$ 
which is a pointwise linear isometry: $\vert I_{g,g'}(x) \,\psi(x) \vert_{g'}=\vert \psi(x) \vert_{g}$. Let us first see its 
construction: there always exists a $g$-symmetric automorphism $H_{g,g'}$ of the $2^{\lfloor d/2 \rfloor}$-
dimensional vector 
space $TM$ such that $g'(X,Y)=g(H_{g,g'}X,Y)$ for $X,Y \in TM$ so define $\iota_{g,g'} \,X\vc H_{g,g'}^{-1/2}\, X$. 
Note that $\iota_{g,g'}$ commutes with right action of the orthogonal group $O_d$ and can be lifted up to a 
diffeomorphism $Pin_d$-equivariant on the spin structures associated to $g$ and $g'$ and this lift is denoted by 
$I_{g,g'}$ (see \cite{Bourguignon}). This isometry is extended as operator on the Hilbert spaces 
$I_{g,g'}: \H_g  \to \H_{g'}$ with $(I_{g,g'} \,\psi)(x):=I_{g,g'}(x)\,\psi(x)$. 

Now define
\begin{align}
\label{UI}
U_{g',g}:=\sqrt{f_{g,g'}}\,I_{g',g} \,: \H_{g'} \to \H_{g}\,.
\end{align}
Then by construction, $U_{g',g}$ is a unitary operator from $\H_{g'}$ onto $\H_{g}$: for $\psi' \in \H_{g'}$, 
\begin{align*}
\langle U_{g',g} \psi',U_{g',g} \psi' \rangle_{\H_g}&=\int_M \vert U_{g,g'} \psi' \vert^2_g \,dvol_g =
\int_M \vert I_{g',g} \psi' \vert ^2_g \,f_{g',g} \,dvol_g = \int_M \vert \psi'\vert^2_{g'} \, dvol_{g'}\\
&= 
\langle \psi_{g'},\psi_{g'} \rangle_{\H_{g'}}.
\end{align*}
So we can realize $\Ds_{g'}$ as an operator $D_{g'}$ acting on $\H_g$ with
\begin{align}
\label{UDU*}
D_{g'}:\H_g\to\H_{g} , \quad D_{g'}:=U_{g,g'}^{-1}\, \Ds_{g'} \, U_{g,g'}.
\end{align}
This is an unbounded operator on $\H_g$ which has the same eigenvalues as $\Ds_{g'}$. 

In the same vein, the $k$-th Sobolev space $H^k(M_g,S_g)$ (which is the completion of the space 
$\Gamma^\infty(M_g,S_g)$ under the norm 
$\norm{\psi}^2_k=\sum_{j=0}^k \int_M \vert \nabla^j \psi (x) \vert^2 \,dx$; be careful, $\nabla$ applied to 
$\nabla \psi$ is the tensor product connection on $T^*M_g\otimes S_g$ etc, see Theorem \ref{compact}) can be 
transported: the map $U_{g,g'}:\,  H^k(M_{g},S_{g}) \to H^k(M_{g'},S_{g'})$ is an isomorphism, see\cite{Schopka}. 
In particular, (after the transport map $U$), the domain of $D_{g'}$ and $\Ds_{g'}$ are the same.

A nice example of this situation is when $g'$ is in the conformal class of $g$ where we can compute explicitly 
$\Ds_{g'}$ and $D_{g'}$ \cite{Bourguignon,Ginoux, Hijazi, Baer}.

\begin{theorem}
Let $g'=e^{2h}g$ be a conformal transformation of $g$ with $h\in C^\infty(M,\R)$. Then there exists an  
isometry $I_{g,g'}$ between the spinor bundle $S_g$ and $S_{g'}$ such that 
\begin{align*}
& \Ds_{g'} \,I_{g,g'} \, \psi= e^{-h}  \, I_{g,g'}\, \big(\Ds_g \,\psi-i  \tfrac{d-1}{2} \, c_g(\grad h)\,\psi \big),\\
& \Ds_{g'}=e^{-\tfrac{d+1}{2}h} \, I_{g,g'}\, \Ds_g \, I_{g,g'}^{-1} \, e^{\tfrac{d-1}{2}h},\\
& D_{g'}=e^{-h/2}\,\Ds_g\, e^{-h/2}.
\end{align*}
for $\psi \in \Gamma^\infty(M,S_g)$.
\end{theorem}

\begin{proof}
The isometry $X \to X' \vc e^{-h} X $ from $(TM, g)$ onto $(TM, g')$ defines a principal bundle isomorphism 
$SO_g(TM) \to SO_{g'}(TM)$ lifting to the spin level. More precisely, it induces	a vector-bundle isomorphism 
$I_{g,g'}:\,S_g \to S_{g'}$, preserving the pointwise hermitean inner product (i.e. $H_{g,g'}=e^{2h}$), such that 
$e^{-h}\, c_{g'}(X)\, I_{g,g'} \,\psi=c_{g'}(X') I_{g,g'} \,\psi=I_{g,g'} \,c_g(X) \psi$. 

For a connection 
$\widetilde \nabla$ compatible with a metric $k$ and without torsion, we have for $X,Y,Z$ in $\Gamma^\infty(TM)$
\begin{align}
\label{gnabla}
2k(\widetilde\nabla_XY,Z)&=k([X,Y],Z)+k([Z,Y],X)+k([Z,X],Y)  \nonumber\\
&\qquad+X\cdot k(Y,Z)+Y\cdot k(X,Z)-Z\cdot k(X,Y)
\end{align}
which is obtained via $k(\widetilde \nabla_X Y,Z)+ k(Y,\widetilde\nabla_X Z)=X \cdot k(Y,Z)$ minus two cyclic 
permutations.\\
The set $\set{e'_j \vc e^{-h}e_j \, \vert \,1\leq j \leq d}$ is a local $g'$-orthonormal basis of $TU$ 
for $g'$ if and only if $\set{e_j}$ is a local $g$-orthonormal basis of $TU$ where $U$ is a trivializing open subset 
of $M$. Applying \eqref{gnabla} to ${\nabla'}^{LC}$, we get
\begin{align*}
2g'({\nabla'}^{LC}_X e'_i,e'_j)& = e^{2h}g([X,e^{-h}e_i],e^{-h}e_j)+e^{2h}g([e^{-h}e_j,e^{-h}e_i],X)
+e^{2h}g([e^{-h}e_j,X],e^{-h}e_i) \\
&\qquad+X\cdot e^{2h}g(e^{-h}e_i,e^{-h}e_j)+e^{-h}e_i\cdot e^{2h}g(X,e^{-h}e_j) 
- e^{-h}e_j\cdot e^{2h}g(X,e^{-h}e_i)\\
&=2g(\nabla_X^{LC} e_i, e_j) +2(e_i\cdot h)\, g(e_j,X)- 2(e_j\cdot h)\,g(e_i,X).
\end{align*}
Since $\nabla^{S_g}_X\,\psi=-\tfrac{1}{4} g(\nabla_X^{LC} e_i,e_j)\,c_g(e_i)\,c_g(e_j) \,\psi$, for 
$\psi \in \Gamma^\infty(U,S_g)$,
\begin{align}
\label{Delta-Delta'}
\nabla_{e_k}^{S_{g'}}\, I_{g,g'} \,\psi=I_{g,g'} \,\big[ \nabla^{S_g}_{e_k} 
+\tfrac{1}{2} c_{g}(e_k)\,c_g(\grad h)-\tfrac{1}{2} e_k(h) \big] \,\psi.
\end{align}
Hence
\begin{align*}
\Ds_{g'}\, I_{g,g'}\,\psi &=-ic_{g'}(e_k'\,^)\nabla^{S_{g'}}_{e'_k} \, I_{g,g'} \,\psi
= -ie^{-h}\,c_{g'}(e_k')\nabla^{S_{g'}}_{e_k}\, I_{g,g'}\, \psi\\
&=-ie^{-h}\,c_{g'}(e_k') I_{g,g'} \, \big[\nabla^{S_g}_{e_k} 
+\tfrac{1}{2} \, c_{g}(e_k)\,c_g(\grad h) -\tfrac{1}{2} e_k(h) \big] \,\psi\\
&= -ie^{-h} \, I_{g,g'}\, c_{g}(e_k)\,\big[\nabla^{S_g}_{e_k} 
+\tfrac{1}{2} \, c_{g}(e_k)\,c_g(\grad h) -\tfrac{1}{2} e_k(h) \big]\,\psi\\
&= e^{-h}  \, I_{g,g'}\, \big[\Ds_g \,\psi-i  \tfrac{d-1}{2} \, c_g(\grad h)\big]\,\psi.
\end{align*}
So, using $[\Ds_g\,,f]=-ic_g(\grad f)$ for $f=e^{-\tfrac{d-1}{2}h}$,
\begin{align*}
\Ds_{g'} \,e^{-\tfrac{d-1}{2}h}\, I_{g,g'} \, \psi &= e^{-h} \,   I_{g,g'}\, \big[ \Ds_g \,e^{-\tfrac{d-1}{2}h} \,\psi 
- i\tfrac{d-1}{2} \, e^{-\tfrac{d-1}{2}h}\,c_g(\grad h) \,\psi \big] \\
&=e^{-h} \,   I_{g,g'} \big[ \,e^{-\tfrac{d-1}{2}h}\,\Ds_g \,\psi + [\Ds_g, e^{-\tfrac{d-1}{2}h}] \,\psi- i\tfrac{d-1}{2} \, 
e^{-\tfrac{d-1}{2}h}\,c_g( \grad h) \,\psi \big] \\
&=e^{-h} \,   I_{g,g'} \big[ \,e^{-\tfrac{d-1}{2}h}\,\Ds_g \,\psi -\tfrac{d-1}{2} \,e^{-\tfrac{d-1}{2}h}\,(-i)c_g(\grad h) \,\psi \\
&\hspace{2cm}- i\tfrac{d-1}{2} \, e^{-\tfrac{d-1}{2}h}\,c_g(\grad h)\,\psi\big]   \\
&= e^{-\tfrac{d+1}{2}h}  I_{g,g'}  \,\Ds_g \,\psi.
\end{align*}
Thus $\Ds_{g'}=e^{-\tfrac{d+1}{2}h} \, I_{g,g'}\, \Ds_g \, I_{g,g'}^{-1} \, e^{\tfrac{d-1}{2}h}$ and since 
$dvol_{g'}=e^{dh}\, dvol_g$, using \eqref{UI}
\begin{align*}
\Ds_{g'}=e^{-h/2}\, U_{g',g}^{-1} \, \Ds_g \, U_{g',g} \, e^{-h/2}.
\end{align*}
Finally, \eqref{UDU*} yields $D_{g'}=e^{-h/2}\,\Ds_g\, e^{-h/2}$.
\end{proof}
Note that $D_{g'}$ is not a Dirac operator as defined in \eqref{theDirac} since its principal symbol has an 
$x$-dependence:  $\sigma^{D_{g'}}(x,\xi)=e^{-h(x)}\, c_{g}(\xi)$. \\
The principal symbols of $\Ds_{g'}$ and $\Ds_g$ are related by
\begin{align*}
\sigma^{\Ds_{g'}}_d(x,\xi)=e^{-h(x)/2} \,U_{g',g}^{-1}(x)  \,\sigma^{\Ds_g}_d(x,\xi) \, U_{g',g}(x) \,e^{-h(x)/2}, 
\quad \xi \in T^*_xM.
\end{align*}
Thus
\begin{align}
\label{cg}
c_{g'}(\xi)=e^{-h(x)} \,U_{g',g}^{-1}(x)  \, c_g(\xi) \, U_{g',g}(x),  \quad \xi \in T^*_xM.
\end{align}
Using $c_g(\xi)c_g(\eta)+c_g(\eta)c_g(\xi)=2g(\xi,\eta)\, \text{id}_{S_g}$, formula \eqref{cg} gives a verification of 
the formula $g'(\xi,\eta)=e^{-2h} g(\xi,\eta)$.

Note that two volume forms $\mu,\mu'$ on a compact connected manifold $M$ are related by an orientation 
preserving diffeomorphism $\a$ of $M$ in the following sense \cite{Moser}: there exists a constant 
$c={(\int_M \mu')}^{-1}\,\int_M \mu$ such that $\mu=c\,\,\a^*\mu'$ where $\a^*\mu'$ is the pull-back of $\mu'$ (i.e. 
$\int_{\a(S)} \a^*\mu=\int_S \mu$ for any set $S\subset M$). The proof is based on the construction of an orientation 
preserving automorphism homotopic to the identity.

It is also natural to look at the changes on a Dirac operator when the metric $g$ is modified by a diffeomorphism 
$\a$ which preserves the spin structure. The diffeomorphism $\a$ can be lifted to a diffeomorphism 
$O_d$-equivariant on the $O_d$-principal bundle of $g$-orthonormal frames with 
$\tilde{\a} \vc  H_{\a^*g,g}^{-1/2}\,T\a$, and this lift also exists on $S_g$ when $\a$ preserves both the orientation 
and the spin structure. However, the last lift is defined up to a $\Z_2$-action which disappears if $\a$ is 
connected to the identity.

The pull-back $g'\vc \a^*g$ of the metric $g$ is defined by 
$(\a^*g)_x(\xi,\eta)=g_{\a(x)}(\a_*(\xi), \a_*\eta)$, $x\in M$, where $\a_*$ is the push-forward map 
$: \,T_xM \to T_{\a(x)}M$. Of course, the metric $g'$ and $g$ are different but the geodesic distances are the 
same. Let us check that $d_{g'}=\a^*d_g$:
\\
In local coordinates, we note $\partial_\mu \vc \partial/\partial x^\mu$ and 
$\partial'_{\mu'} \vc \partial/\partial (\a(x))^{\mu'}$. Thus $\partial' = (\Lambda^{-1T})\partial$ where 
$\Lambda^{\mu'}\,_{\mu} \vc \partial (\a(x))^{\mu'}/\partial x^\mu$. The dependence in the metric $g$ of Cristoffel 
symbols is  $\Gamma_{\mu \nu}^\rho = \half g^{\rho \beta}(\partial_\mu g_{\beta\nu}+ 
\partial_\nu g_{\mu\beta} -\partial_\beta g_{\mu\nu})$. Thus the same symbols $\Gamma'$ associated to $g'$ are 
\begin{align}
\label{Cristoffel}
{\Gamma'}^{\,\rho'}_{\mu' \nu'}=\Lambda^{\rho'}\,_\rho ( \Lambda^{-1T})_{\mu'}\,^\mu ( \Lambda^{-1T})_{\nu'}\,^\nu \,
\Gamma^{\rho}_{\mu\nu} + \Lambda^{\rho'}\,_\rho ( \Lambda^{-1T})_{\mu'}\,^\mu \, \partial _\mu \,
( \Lambda^{-1T})_{\nu'}\,^\nu \,.
\end{align}
The geodesic equation is 
$\ddot{x}^\rho+ \Gamma^\rho_{\mu\nu}\, \dot{x}^\mu\, \dot{x}^\nu=0$ for all $\rho$ (note that neither $x^\mu$ nor 
$\ddot{x}^\mu$ are 4-vectors in the sense that they are not transformed like 
$v'^{\,\mu'}=\Lambda^{\mu'}\,_\mu \,v^\mu$, while $\dot{x}^\mu$ is a 4-vector; in fact 
$\ddot{\a(x)}\,^{\rho'}=\Lambda^{\rho'}\,_\rho \, \ddot{x}^\rho + \partial_\mu \Lambda^{\rho'}\,_\rho \,\dot{x}^\mu \, 
\dot{x}^\rho$. This relation and \eqref{Cristoffel} give the invariance of the geodesic equation and the same for the 
distance since for any path $\ga$ joining points $x=\ga(0),\,y=\ga(1)$
$$
\int_0^1\sqrt{(\a^*g)_{\ga(t)}\big(\ga'(t),\ga'(t)\big)} \,dt = 
\int_0^1\sqrt{g_{\a \circ \ga(t)}\big((\a \circ \ga)'(t),(\a
\circ \ga)'(t)\big)} \,dt
$$
and $(\a\circ \ga)(0)=\a(x), (\a\circ\ga)(1)=\a(y)$. Note that $\a$ is an isometry only if $\a^*d_g=d_g$.

Recall that the principal symbol of a Dirac operator $D$ is $\sigma_d^D(x,\xi)=c_g(\xi)$ so gives the metric $g$ by 
\eqref{cliffmetric} as we checked above. This information will be used later in the definition of a spectral triple. A 
commutative spectral triple associated to a manifold generates the so-called Connes' distance which is nothing 
else but the metric distance; see the remark after \eqref{Cdist}. Again, the link between $d_{\a^*g}$ and 
$d_{g}$ is explained by \eqref{UDU*}, since the unitary induces an automorphism of the $C^*$-algebra 
$C^\infty(M)$.

\newpage
\section{Heat kernel expansion}
\label{Heatkernel}

References for this section: \cite{Gilkey,Gilkey2,Berline} and especially \cite{Vassilevich}.

Recall that the heat kernel is a Green function of the heat operator $e^{t\Delta}$ (recall that $-\Delta$ is a 
positive operator) which measures the temperature evolution in a 
domain whose boundary has a given temperature. For instance, the heat kernel of the Euclidean space 
$\R^d$ is 
\begin{align}
\label{hkt}
k_t(x,y) = \tfrac{1}{(4\pi t)^{d/2}} e^{-|x-y|^2/4t} \text{ for } x\neq y
\end{align} 
and it solves the heat equation
$$
\left\{ \begin{array}{ll}
\partial_t k_t(x,y) = \Delta_x k_t(x,y),\quad \forall t > 0,\,x,y \in \R^d\\
\text{initial condition: }\lim_{t\downarrow 0} k_t(x,y) = \delta(x-y).
\end{array}\right.
$$
Actually, $k_t(x,y)=\tfrac{1}{2\pi}\int_{-\infty}^\infty e^{-ts^2}\,e^{is(x-y)}\,ds$ when $d=1$.\\
Note that for $f\in\DD(\R^d)$, we have $ \lim_{t\downarrow 0}\int_{\mathbf{R}^d} k_t(x,y)\,f(y)\,dy = f(x)$.
\\
For a connected domain (or manifold with boundary with vector bundle $V$) $U$, let $\lambda_n$ be 
the eigenvalues for the Dirichlet problem of minus the Laplacian 
$$
\left\{ \begin{array}{ll}
-\Delta \phi = \lambda \psi  & \mathrm{in\ }\ U\\ \psi=0 & \mathrm{on\ }\ \partial U. 
\end{array}\right.
$$
If $\psi_n\in L^2(U)$ are the normalized eigenfunctions, the inverse Dirichlet Laplacian $\Delta^{-1}$ is a selfadjoint  
compact operator, $ 0 \leq \lambda_1 \leq \lambda_2\le \lambda_3\leq\cdots, \lambda_n\to\infty$.

The interest for the heat kernel is that, if $f(x)=\int_0^\infty dt\, e^{\,-tx} \,\phi(x)$ is the Laplace transform of 
$\phi$, then $\Tr\big(f(-\Delta)\big)=\int_0^\infty dt\, \phi(t)\,\Tr\big(e^{t\,\Delta}\big)$ (if everything makes sense)
is controlled by $\Tr\big(e^{t\Delta}\big)=\int_M dvol(x)\,\tr_{V_x}k_t(x,x)$ since 
$\Tr\big(e^{\,t\,\Delta}\big)=\sum_{n=1}^\infty e^{\,t\,\lambda_n}$ and
$$
k_t(x,y)=\langle x,e^{t\,\Delta}\,y \rangle=\sum_{n,m=1}^\infty 
\langle x, \psi_m\rangle\,\langle \psi_m , e^{t\,\Delta}\psi_n \rangle\,\langle \psi_n,y \rangle
= \sum_{n=1}^\infty \overline{{\psi_n(x)}} \,\psi_n(y)\, e^{t\,\lambda_n}.
$$
So it is useful to know the asymptotics of the heat kernel $k_t$ on the diagonal of $M\times M$ especially near $t=0$.

\subsection{The asymptotics of heat kernel}

Let now $M$ be a smooth compact Riemannian manifold without boundary, $V$ be a vector bundle over $M$ 
and $P\in \Psi DO^{m}(M,V)$ be a positive elliptic operator of order $m>0$. If $k_t(x,y)$ is the kernel of the heat 
operator $e^{-tP}$, then the following asymptotics exits on the diagonal: 
$$
k_t(x,x) \underset{t \downarrow 0^+}{\sim} \sum_{k=0}^\infty a_k(x) \,t^{(-d+k)/m}
$$
which means that
$$
\bigl| k_t(x,x)-\sum_{k=0}^na_k(x) \,t^{(-d+k)/m} \bigl|_{\infty,n} < c_n \,t^n \text{ for } 0<t<1
$$
where $\vert f \vert_{\infty,n}:=\sup_{x\in M} \sum_{\vert\a\vert\leq n} \vert \partial_x^\a f \vert$ (since $P$ is elliptic, 
$k_t(x,y)$ is a smooth function of $(t,x,y)$ for $t>0$, see \cite[section 1.6, 1.7]{Gilkey}).

More generally, we will use
$$
k(t,f,P)\vc \Tr\big(f\,e^{-tP}\big)
$$
where $f$ is a smooth function. We have similarly
\begin{align}
\label{heatasympt}
k(t,f,P) \underset{t \downarrow 0^+}{\sim} \sum_{k=0}^\infty a_k(f,P) \,t^{(-d+k)/m}.
\end{align}
The utility of function $f$ will appear later for the computation of coefficients $a_k$.\\
The following points are of importance:

\quad 1) The existence of this asymptotics is non-trivial \cite{Gilkey,Gilkey2}.

\quad 2) The coefficients $a_{2k}(f,P)$ can be computed locally as integral of local invariants: Recall that a locally 
computable quantity is the integral on the manifold of a local frame-independent smooth 
function of one variable, depending only on a finite number of derivatives of a finite number of terms in the 
asymptotic expansion of the total symbol of $P$. \\
In noncommutative geometry, local generally means that it is concentrated at infinity in momentum space.

\quad 3) The odd coefficients are zero: $a_{2k+1}(f,P)=0$.
\\
For instance, let us assume from now on that $P$ is a Laplace type operator of the form 
\begin{align}
\label{Lapl}
P = - (g^{\mu\nu} \del_\mu \del_\nu + \mathbb{A}^\mu\del_\mu +\mathbb{B}) 
\end{align}
where $(g^{\mu\nu})_{1\leq \mu,\nu\leq d}$ is the inverse matrix associated to the metric $g$ on $M$, and 
$\mathbb{A}^\mu$ and $\mathbb{B}$ are smooth $L(V)$-sections on $M$ (endomorphisms) (see also Definition 
\ref{deflaplace}). Then (see \cite[Lemma 1.2.1]{Gilkey2}) there is a unique connection $\nabla$ on $V$ and a 
unique endomorphism $E$ such that 
\begin{align*}
&P=  -(\Tr_g \nabla^2  + E), \quad \nabla^2(X,Y):= [\nabla_X,\nabla_Y] -\nabla_{\nabla^{LC}_X Y} \, ,
\end{align*}
$X,Y$ are vector fields on $M$ and $\nabla^{LC}$ is the Levi-Civita connection on $M$. Locally 
$$
\Tr_g \nabla^2 := g^{\mu\nu}(\nabla_\mu \nabla_\nu -\Ga^{\rho}_{\mu\nu} \nabla_\rho)
$$
where $\Ga^{\rho}_{\mu \nu}$ are the Christoffel coefficients of $\nabla^{LC}$.
Moreover (with local frames of $T^*M$ and $V$), $\nabla =
dx^\mu\ox (\del_\mu +\om_\mu)$ and $E$ are related to $g^{\mu\nu}$,
$\mathbb{A}^\mu$ and $\mathbb{B}$ through
\begin{align}
\om_\nu&=  \half g_{\nu\mu}(\mathbb{A}^\mu +g^{\sg\eps} \Ga_{\sg \eps}^{\mu}\ \text{id}_V )
\label{omeganu}\, ,\\
E&=  \mathbb{B}-g^{\nu\mu}(\del_{\nu} \om_\mu +\om_\nu\om_\mu -\om_\sg
\Ga_{\nu\mu}^\sg ) \label{EEquation}  \, .
\end{align}
In this case, the coefficients $a_k(f,P)=\int_M dvol_g \, \tr_V \big(f(x) \, a_k(P)(x)\big)$ and the 
$a_k(P)=c_i\,\a_k^i(P)$ are linear combination with constants $c_i$ of all possible independent invariants 
$\a_k^i(P)$ of dimension $k$ constructed from $E, \Omega, R$ and their 
derivatives ($\Omega$ is the curvature of the connection $\omega$, and $R$ is the Riemann curvature tensor). 
As an example, for $k=2$, $E$ and $s$ are the only independent invariants. \\
Point 3) follow since there is no odd-dimension invariant.

\subsection{Computations of heat kernel coefficients}

The computation of coefficients $a_k(f,P)$ is made by induction using first a variational method: for any 
smooth functions 
$f,h$ one has
\begin{align}
&\tfrac{d\,}{d\epsilon}_{\vert \epsilon=0}\,a_k(1,e^{-2\epsilon f}P)=(d-k)\,a_k(f,P),\label{akepsilon1}\\
& \tfrac{d\,}{d\epsilon}_{\vert \epsilon=0}\,a_k(1,P-\epsilon h)=a_{k-2}(h,P),\label{akepsilon2}\\
& \tfrac{d\,}{d\epsilon}_{\vert \epsilon=0}\,a_{d-2}(e^{-2\epsilon f}h, e^{-2\epsilon f}P)=0.\label{akepsilon3}
\end{align}
The first equation follows from
\begin{align*}
\tfrac{d\,}{d\epsilon}_{\vert \epsilon=0}\,\Tr\big(e^{-e^{-2\epsilon f} \,t P} \big)=-2t \tfrac{d\,}{dt} \Tr \big(f\,e^{-tP}\big)
\end{align*}
with an expansion in power series in $t$. Same method for \eqref{akepsilon2}.\\
For the proof of \eqref{akepsilon3}, we use $P(\epsilon,\delta) \vc e^{-2f}(P-\delta h)$; with \eqref{akepsilon1} 
for $k=d$, 
\begin{align*}
0=\tfrac{d\,}{d\epsilon}_{\vert \epsilon=0}\,a_d \big(1,P(\epsilon,\delta)\big),
\end{align*}
thus after a variation of $\delta$, 
\begin{align*}
0=\tfrac{d\,}{d\delta}_{\vert \delta=0}\, \tfrac{d\,}{d\epsilon}_{\vert \epsilon=0}\,a_d \big(1,P(\epsilon,\delta)\big)=
\tfrac{d\,}{d\epsilon}_{\vert \epsilon=0}\,\tfrac{d\,}{d\delta}_{\vert \delta=0}\, a_d \big(1,P(\epsilon,\delta)\big),
\end{align*}
we derive \eqref{akepsilon3} from \eqref{akepsilon2}.

The idea behind equations  \eqref{akepsilon1}, \eqref{akepsilon2} and \eqref{akepsilon3} is that \eqref{akepsilon2} 
shows dependence of coefficients $a_k$ on $E$, while the two others describe their behaviors under local scale 
transformations.

Then, the $a_k(P)=c_i \, \a_k^i(P)$ are computed with arbitrary constants $c_i$ (they are dependent only of the 
dimension $d$) and these constants are inductively calculated using \eqref{akepsilon1}, \eqref{akepsilon2} and 
\eqref{akepsilon3}. If $s$ is the scalar curvature and `;' denote multiple covariant derivative with respect to Levi-Civita 
connection on $M$, one finds, with rescaled $\a$'s,
\begin{align}
\label{coefficientsa}
a_{0}(f,P)&=(4\pi)^{-d/2}\,\int_M dvol_g\,\tr_V ( \a_0f) \, , \nonumber\\
a_{2}(f,P)&=\tfrac{(4\pi)^{-d/2}}{6} \int_Mdvol_g\, \tr_V\big[(f(\a_1E +\a_2 s)\big], \\
a_{4}(f,P)&=\tfrac{(4 \pi )^{-d/2}}{360} \int_M dvol_g \, \tr_V \big[ f(\a_3 E_{;kk}+\a_4Es+\a_5E^2+\a_6R_{;kk}
+\a_7 s^2 \nonumber\\
& \hspace{4.5cm} +\a_8 R_{ij}R_{ij}+\a_9 R_{ijkl}R_{ijkl}+\a_{10} \Omega_{ij}\Omega_{ij})  \big] \nonumber.
\end{align}
In $a_4$, they are no other invariants: for instance, $R_{ij;ij}$ is proportional to $R_{;ij}$.\\
Using the scalar Laplacian on the circle, one finds $\a_0=1$.\\
Using \eqref{akepsilon2} with $k=2$, under the change $P \to P-\epsilon h$, $E$ becomes $E+\epsilon h$, so
\begin{align*}
\tfrac{1}{6}\int_M dvol_g \,\tr_V(\a_1 h)=\int_M dvol_g\,\tr_V(h)
\end{align*}
yielding $\a_1=6$. For $k=4$, it gives now:
\begin{align*}
\tfrac{1}{360}\int_M dvol_g\,\tr_V(\a_4 h s+2\a_5 hE)=\tfrac{1}{6}\int_M dvol_g \,\tr_V(\a_1hE+\a_2hs),
\end{align*}
thus $\a_5=180$ and $\a_4=60 \a_2$.\\
To go further, one considers the scale transformation on $P$ given in \eqref{akepsilon1} and \eqref{akepsilon3}.
In \eqref{akepsilon1}, $P$ is transformed covariantly, the metric $g$ is changed into $e^{-2\epsilon f}g$ implying 
conformal transformation of the Riemann tensor, Ricci tensor and scalar curvature giving the modifications on 
$\omega$ and $E$ via \eqref{omeganu}, \eqref{EEquation}. This gives (we collect here all terms appearing in 
$a_2$ and only few terms appearing in $a_4$)
\begin{align*}
&\tfrac{d\,}{d\epsilon}_{\vert \epsilon=0}\,dvol_g=d \,f\,dvol_g,\\
& \tfrac{d\,}{d\epsilon}_{\vert \epsilon=0}\,E=-2fE+\tfrac{1}{2}(d-2)\,f_{;ii}\,,\\
& \tfrac{d\,}{d\epsilon}_{\vert \epsilon=0}\,s=-2fs-2(d-1)\,f_{;ii}\,,\\
& \tfrac{d\,}{d\epsilon}_{\vert \epsilon=0}\,Es=-4fEs+\tfrac{1}{2}(d-2)s\,f_{;ii}-2(d-1)f_{;ii}E\,,\\
& \tfrac{d\,}{d\epsilon}_{\vert \epsilon=0}\,E^2=-4fE^2+(d-2)\,f_{;ii}E \,,\\
& \tfrac{d\,}{d\epsilon}_{\vert \epsilon=0}\,s^2=-4fs^2-4(d-1)\,f_{;ii}s\,,\\
& \tfrac{d\,}{d\epsilon}_{\vert \epsilon=0}\,R_{ijkl}=
-2fR_{ijkl} +\delta_{jl}f_{;ik}+\delta_{;ik}f_{;jl}-\delta_{il}f_{;jk}-\delta_{jk}f_{;il}\,,\\
& \tfrac{d\,}{d\epsilon}_{\vert \epsilon=0}\,\Omega_{ij}\Omega_{ij}=-4f\,\Omega_{ij}\Omega_{ij}\,,\\
&\cdots
\end{align*}
Applying \eqref{akepsilon3} with $d=4$, we get
\begin{align*}
&\tfrac{d\,}{d\epsilon}_{\vert \epsilon=0}\,a_2(e^{-2\epsilon f}h,e^{-2\epsilon f}P)=0.
\end{align*}
Picking terms with $\int_M dvol_g\,\tr_V (hf_{;ii})$, we find $\a_1=6\a_2$, so $\a_2=1$ and $\a_4=60$. Thus 
$a_2(f,P)$ has been determined.\\
Similar method gives $a_4(f,P)$, but only after lengthy computation despite the use of Gauss--Bonnet theorem 
for the determination of $\a_{10}$! One finds:
$$
\a_3=60,\,\a_5=180,\,\a_6=12,\,\a_7=5,\,\a_8=-2,\,\a_9=2,\,\a_{10}=30.
$$

The coefficient $a_6$ was computed by Gilkey, $a_8$ by Amsterdamski, Berkin and O'Connor and $a_{10}$ in 
1998 by van de Ven \cite{Ven}. Some higher coefficients are known in flat spaces.

\subsection{Wodzicki residue and heat expansion}
\label{Wodzicki residue and heat expansion}
 
Wodzicki has proved that, in \eqref{heatasympt}, $a_k(P)(x) =\tfrac{1}{m}\,c_{P^{(k-d)/m}}(x)$ is true not only 
for $k=0$ as seen in Theorem \ref{TrDix} (where $P\leftrightarrow P^{-1}$), but for all $k\in \N$. In this section, we will 
prove this result when $P$ is is the inverse of a Dirac operator and this will be generalized in the next section.

Let $M$ be a compact Riemannian manifold of dimension $d$ even, $E$ a Clifford module over $M$ and 
$D$ be the Dirac operator (definition \ref{defdir}) given by a Clifford connection on $E$. By Theorem \ref{D=D*}, 
$D$ is a selfadjoint (unbounded) operator on $\H\vcentcolon = L^2(M,S)$. 

We are going to use the heat operator $e^{-tD^2}$ since $D^2$ is related to the Laplacian via the 
Schr\" odinger--Lichnerowicz formula \eqref{defdir} and since the asymptotics of the heat kernel of this Laplacian is 
known.

For $t>0$, we have $e^{-tD^2} \in \L^1$: the result follows from the decomposition  
$$
e^{-tD^2}= (1+D^2)^{(d+1)/2} \, e^{-tD^2} (1+D^2)^{-(d+1)/2},
$$
since $(1+D^2)^{-(d+1)/2} \in \L^1$ and the function: $\lambda \to (1+\lambda^2)^{(d+1)/2}e^{-t \lambda^2}$ 
is bounded. \\
Thus 
$\Tr \big( e^{-tD^2}\big)=\sum_n e^{-t\lambda_n^2} < \infty$. \\
Another argument is the following: 
$(1+D^2)^{-d/2}$ maps $L^2(M,S)$ into the Sobolev space $H^k(M,S)$ (see Theorem \ref{compact}) and the 
injection $H^k(M,S) \hookrightarrow L^2(M,S)$ is Hilbert--Schmidt operator for $k > \half d$. Thus $t\to e^{-tD^2}$ is
a semigroup of Hilbert--Schmidt operators for $t>0$.

Moreover, the operator $e^{-tD^2}$ has a smooth kernel since it is regularizing, see Remark \ref{defsmoothing} (or 
 \cite[Theorem 6.2]{Lawson}) and the asymptotics of its kernel is (recall \eqref{hkt}), see \cite[Theorem 2.30]{Berline}:
$$
k_t(x,y) \underset{t \downarrow 0^+}{\sim} \tfrac{1}{(4\pi t)^{d/2}} \sqrt{det\,g_x} \sum_{j\geq0} k_j(x,y) \, t^j 
\, e^{-d_g(x,y)^2/4t}
$$
where $k_j$ is a smooth section on $E^*\otimes E$. Thus 
\begin{align}
\label{traceasympt}
\Tr \big( e^{-tD^2}\big) \underset{t \downarrow 0^+}{\sim}  \sum_{j\geq 0} t^{(j-d)/2}\,a_j(D^2)
\end{align}
with for $j\in \N$,
$$
\left \{\begin{array}{ll} a_{2j}(D^2)\vcentcolon = \tfrac{1}{(4\pi)^{d/2}}\int_M \tr\big( k_j(x,x)\big) \sqrt{det\,g_x} 
\,\vert dx \vert  ,\\ 
a_{2j+1}(D^2)=0.
\end{array} \right.
$$
The aim now is to compute $\text{\it WRes}\big(D^{-p}\big)$ for an integer $p$ such that $0\leq p \leq d$:

\begin{theorem}
\label{WresD-p}
For any integer $p$, $0\leq p \leq d$, $D^{-p} \in \Psi DO^{-p}(M,E)$ and 
\begin{align*}
&\text{WRes}\big( D^{-p}\big)=0,\quad \text{for odd } p,\\
&\text{WRes}\big( D^{-p}\big)=\tfrac{2}{\Gamma(p/2)} a_{d-p}(D^2)= \tfrac{2}{(4\pi)^{d/2}\Gamma(p/2)} \int_M 
\tr \big( k_{(d-p)/2}(x,x) \big) \,dvol_g(x),\quad \text{for even }p.
\end{align*}
\end{theorem}

\begin{proof}
Assume $D$ is invertible, otherwise swap $D$ for the invertible operator $D +P$ where $P$ is the projection 
on the kernel of $D$. Since the kernel is finite dimensional, $P$ has a finite rank and generates a smoothing 
operator. 
\\
Since the trace of $\gamma$-matrices in \eqref{gammamatrices} is zero, $\text{WRes}\big( D^{-p}\big)=0$ for 
$p$ odd. 
\\
Assuming now that $p$ is an even integer, by spectral theory,
\begin{align*}
D^{-p}=\vert D \vert^{-p}=\tfrac{1}{\Ga(p/2)} \int_0^\infty t^{p/2}\,e^{-t\,D^2}\, t^{-1}dt=\tfrac{1}{\Ga(p/2)}\, (\int_0^\epsilon 
+\int_\epsilon^\infty) \,t^{p/2}\,e^{-t\,D^2}\, t^{-1}dt
\end{align*}
The second integral is a smooth operator since the map $x \to \int_\epsilon^\infty  t^{p/2}\,e^{-tx^2}\, t^{-1}dt$ is in 
the Schwartz space $\SS$.\\
Define the first integral as the operator $D_\epsilon^{-p}$ and choose $\epsilon$ small enough 
such that for $0<t\leq \epsilon$ and $x$ and $y$ close enough,
$$
\vert \, k_t(x,y)-\tfrac{1}{(4\pi t)^{d/2}} \sum_{j=0}^{(d-p)/2} t^j \,\sqrt{\det\, g_x}\, k_j(x,y)\,e^{-d_g(x,y)^2/4t} \,\vert \leq 
c \,t^{p/2}\, e^{-d_g(x,y)^2/4t}.
$$
Thus
\begin{align*}
\Ga(\tfrac{d}{2})\,\tr \big( k^{D_{\epsilon}^{-p}}(x,y) \big)&=\int_0^\infty t^{p/2} \, \tr \big(k_t(x,y) \big) \,t^{-1} dt \\
&=\tfrac{\sqrt{\det\, g_x}}{(4\pi)^{d/2}} \sum_{j=0}^{(d-p)/2} \tr \big(k_j(x,y)\big) 
\int_0^\epsilon t^{j-(p-d)/2} e^{-d_g(x,y)^2/4t} \,dt\\
& \hspace{3cm}+ \mathcal{O} \big( \int_0^\epsilon e^{-d_g(x,y)^2/4t} dt \big).
\end{align*}
For $m$ integer and $\mu>0$, we get after a change of variable $t \to t^{-1}$, 
\begin{align*}
\int_0^\epsilon t^m \,e^{-\mu/t} \,t^{-1} dt=\mu^m \int_{\mu \epsilon}^\infty t^{-m}\,e^{-t} t^{-1}dt=
\left \{ \begin{array}{lll} \text{Polynomial in }\tfrac{1}{\mu} + \mathcal{O}(1) \text{ for }m<0 ,\\
-\log \, \mu +  \mathcal{O}(1) \text{ for }m=0,\\
 \mathcal{O}(1) \text{ for }m>0.
 \end{array} \right.
\end{align*}
Thus, the logarithmic behavior of $\Ga(\tfrac{d}{2})\,\tr \big( k^{D_{\epsilon}^{-p}}(x,y) \big)$ comes from
\begin{align*}
&\tfrac{\sqrt{\det\, g_x}}{(4\pi)^{d/2}} \,\tr \big(k_{(d-p)/2}(x,y)\big) \int_0^\epsilon e^{-d_g(x,y)^2/4t} t^{-1} dt\\
&\hspace{3cm}= \tfrac{\sqrt{\det\, g_x}}{(4\pi)^{d/2}} \,\tr\big(k_{(d-p)/2}(x,y)\big) \, \big( -\log \big(d_g(x,y)^2/4 \big) 
+ \mathcal{O}(1) \big)\\
&\hspace{3cm}= \tfrac{\sqrt{\det\, g_x}}{(4\pi)^{d/2}} \,\tr\big(k_{(d-p)/2}(x,y)\big) \, \big( -2\log \big(d_g(x,y) \big) 
+ \mathcal{O}(1) \big).
\end{align*}
Thus (see Theorem \ref{kernelbehavior} for the sign)
\begin{align*}
\text{\quad \it WRes}(\big( D^{-p}\big) & =\text{\it WRes}\big( D_\epsilon^{-p}\big)\\
&= \int_M c_{D_\epsilon^{-p}}(x)\, \vert dx \vert = \tfrac{2}{(4\pi)^{d/2}} \int_M \tr \big(k_{(d-p)/2}(x,x) \big) \sqrt{\det \,g_x}\, \vert dx \vert,
\end{align*}
which is, by definition, $\tfrac{2}{\Ga(p/2)}\,a_{d-p}(D^2)$.
\end{proof}
Few remarks are in order:

1) If $p=d$ is even, {\it WRes}$\big(D^{-d}\big)=\tfrac{2}{\Ga(p/2)}\,a_{0}(D^2)= \tfrac{2}{\Ga(p/2)}\,
\tfrac{Rank(E)}{(4\pi)^{d/2}} \Vol (M)$. 
\\
Since $\Tr(e^{-tD^2}) \underset{t \downarrow 0^+}{\sim} a_0(D^2)\,t^{-d/2}$, the Tauberian theorem used in Example 
\ref{Dixtrace} implies that $D^{-d}=(D^{-2})^{d/2}$ is measurable and we obtain Connes' trace theorem \ref{TrDix}
$$
\Tr_{Dix}(D^{-d})=\Tr_\omega (D^{-d})=\tfrac{a_0(D^2)}{\Ga(d/2+1)}=\tfrac{1}{d} \,\text{\it WRes} (D^{-d}).
$$

2) When $D=\Ds$ and $E$ is the spinor bundle, the Seeley-deWit coefficient $a_2(\Ds^2)$ (see 
\eqref{coefficientsa} with $f=1$) can be easily computed (see \cite{Gilkey,Polaris}): if $s$ is the scalar curvature,
\begin{align}
\label{a2(D2)}
a_2(\Ds^2)=-\tfrac{1}{12(4\pi)^{d/2}}\,\int_M s(x) \,dvol_g(x).
\end{align}
So {\it WRes}$\big(\Ds^{-d+2}\big)=\tfrac{2}{\Ga(d/2-1)}\,a_2(\Ds^2)=c\, \int_M s(x) \,dvol_g(x)$. This is a quite 
important result since this last integral is nothing else but the Einstein--Hilbert action \eqref{EH action}. 
In dimension 4, this is an example of invariant by diffeomorphisms, see \eqref{Invdiff}.

\newpage
\section{Noncommutative integration}
\label{Noncommutative integration}

We already saw that the Wodzicki residue is a trace and, as such, can be viewed as an integral. But of course, it is 
quite natural to relate this integral to zeta functions used in \eqref{Wresdef}: with notations of Section 
\ref{Wodzicki residue}, let $P\in \Psi DO^\Z(M,E)$ and $D \in \Psi DO^1(M,E)$ which is elliptic. The definition of 
zeta function
$$
\zeta_D^P(s) \vcentcolon = \Tr \big(P\, \vert D \vert^{-s} \big)
$$
has been useful to prove that $\text{\it WRes} \,P= \underset{s=0}{\Res}\,\zeta_D^P(s)=\int_M c_P(x) \, \vert dx\vert$.

The aim now is to extend this notion to noncommutative spaces encoded in the notion of spectral triple.

References: \cite{Book, ConnesMarcolli,Polaris,CM2,EILS}.

\subsection{Notion of spectral triple}

The main properties of a compact spin Riemannian manifold $M$ can be recaptured using the following triple 
$(\A=C^\infty(M), \H=L^2(M,S), \Ds)$. The coordinates $x=(x^1,\cdots,x^d)$ are exchanged with the algebra 
$C^\infty(M)$, 
the Dirac operator $\DD$ gives the dimension $d$ as we saw in Theorem \ref{Weyl}, but also the metric of $M$ via 
Connes formula and more generally generates a quantized calculus. The idea of noncommutative geometry is to 
forget about the commutativity of the algebra and to impose axioms on a triplet $(\A,\H,\DD)$ to generalize the above 
one in order to be able to obtain appropriate definitions of important notions: pseudodifferential operators, measure 
and integration theory, $KO$-theory, orientability, Poincar\'e duality, Hochschild (co)homology etc.

An important remark, probably due to Atiyah, is that the commutator of a pseudodifferential operator of order 1 (resp. 
order 0) with the multiplication by a function is a bounded operator (resp. compact). This is at the origin of the 
notion of Fredholm module (or K-cycle) with its K-homology class and via duality to its K-theory culminating with the 
Kasparov KK-theory. Thus, it is quite natural to define (unbounded) Fredholm module since for instance $\Ds$ is 
unbounded:

\begin{definition}
\label{deftriplet}
A spectral triple $(\A,\H,\DD)$ is the data of an involutive (unital) algebra $\A$ with a faithful representation $\pi$ on 
a Hilbert space $\H$ and a selfadjoint operator $\DD$ with compact resolvent (thus with discrete spectrum) such that 
$[\DD,\pi(a)]$ extends to a bounded operator for any $a \in \A$.
\end{definition}
\noindent We could impose the existence of a $C^*$-algebra $A$ such that 
$$\A:=\set{ a\in A \, \vert \, [\DD,\pi(a)] \text{ is bounded} }$$ is norm dense in $A$ so $\A$ is a pre-$C^*$-algebra 
stable by holomorphic calculus. Such $\A$ is always a $^*$-subalgebra of $A$.

When there is no confusion, we will write $a$ instead of $\pi(a)$.
\\
We now give useful definitions:

\begin{definition}
\label{defspectraltriple}
Let $(\A,\H,\DD)$ be a spectral triple.

It is even if there is a grading operator $\chi$ such that $\chi=\chi^*$, 
$$
[\chi,\pi(a)]=0, \, \forall a \in \A\text{ and }\DD \chi=-\chi \DD.
$$

It is real of KO-dimension $d \in \Z/8$ if there is an antilinear isometry $J:\H \to \H$ such that
$$
J\DD=\epsilon \,\DD J,\qquad J^2=\epsilon',\qquad J\chi=\epsilon'' \,\chi J
$$
with the following table for the signs $\epsilon, \epsilon',\epsilon''$
\begin{align}
\label{commu}
\begin{tabular}{|c| cccccccc|}
\hline
d & 0 & 1 & 2 & 3 & 4 & 5 & 6 & 7 \\
\hline
$\epsilon$ & 1&-1&1&1&1&-1&1&1\\
$\epsilon'$ &1&1&-1&-1&-1&-1&1&1\\
$\epsilon''$ &1& &-1 & &1 & &-1&\\
\hline
\end{tabular}
\end{align}
and the following commutation rules
\begin{align}
\label{oppcommut}
[\pi(a),\pi(b)^\circ)=0,\qquad \big[ [\DD,\pi(a)],\pi(b)^\circ \big]=0, \, \forall a,b\in \A
\end{align}
where $\pi(a)^\circ \vcentcolon = J \pi(a^*)J^{-1}$ is a representation of the opposite algebra $\A^\circ$.

It is $d$-summable (or has metric dimension $d$) if the singular values of $\DD$ behave like 
$\mu_n(\DD^{-1})=\mathcal{O}(n^{-1/d})$.

It is regular if $\A$ and $[\DD,\A]$ are in the domain of $\delta^n$ for all $n\in \N$ where 
$$
\delta(T) \vc [\vert \DD \vert ,T].
$$
(Recall that the domain of the unbounded derivation $\delta$ is the set of all bounded operators $T$ on $\H$ which map 
$\Dom(\vert \DD \vert)\subset \H$ into itself and $[\vert \DD \vert,T]$ can be (uniquely) extended to a bounded operator.)

It satisfies the finiteness condition if the space of smooth vectors $\H^\infty \vc \bigcap_k \Dom \DD^k$ is a finitely 
projective left $\A$-module.

It satisfies the orientation condition if there is a Hochschild cycle $c\in Z_d(\A,\A \otimes \A^\circ)$ such that 
$\pi_\DD(c)=\chi$, where 
$\pi_\DD\big( (a\otimes b^\circ)\otimes a_1 \otimes \cdots \otimes a_d \big) \vc \pi(a)\pi(b)^\circ[\DD,\pi(a_1)]
\cdots [\DD,\pi(a_d)]$ and $d$ is its metric dimension.
\end{definition}

The above definition of $KO$-dimension comes from the fact that a Dirac operator is a square root a Laplacian. 
This generates a sign problem which corresponds to a choice of a spin structure (or orientation). Up to some 
subtleties, the choice of a manifold of a chosen homotopy needs a Poincar\'e duality between homology and 
cohomology and the necessary refinement yields to the $KO$-homology introduced by Atiyah and Singer. 

An interesting example of noncommutative space of non-zero $KO$-dimension is given by the finite part of the 
noncommutative standard model \cite{Connesmixing, CCM,ConnesMarcolli}.

Moreover, the reality (or charge conjugation in the commutative case) operator $J$ is related to the problem of the 
adjoint: If $\M$ is a von Neumann algebra acting on the Hilbert space $\H$ with a cyclic and separating vector 
$\xi \in \H$ (which means $\M\xi$ is dense in $\H$ and $a\xi=0$ implies $a=0$, for $a\in \M$), then the closure $S$ 
of the map: $a\xi \to a^*\xi$ has an unbounded extension to $\H$ with a 
polar decomposition $S=J  \Delta^{1/2}$ where $\Delta\vc S^*S$ is a positive operator and $J$ is antilinear 
operator such that $J \M J^{-1}=\M'$, see Tomita theory in \cite{Takesaki}. This explains the commutation relations 
\eqref{oppcommut}. Moreover $\Delta^{it}\M\Delta^{-it}=\M$, a point related to Definition \ref{defOP}.

A fundamental point is that a reconstruction of the manifold is possible, starting only with a spectral triple where the 
algebra is commutative (see \cite{CReconstruction} for a more precise formulation, and also \cite{RV}):

\begin{theorem}
\label{Recons}
\cite{CReconstruction} Given a commutative spectral triple $(\A,\H,\DD)$ satisfying the above axioms, then there 
exists a compact spin$^c$ manifold $M$ such that $\A \simeq C^\infty(M)$ and $\DD$ is a Dirac operator.
\end{theorem}

The manifold is known as a set, $M=\Sp(\A)=\Sp(A)$. Notice that $\DD$ is known only 
via its principal symbol, so is not unique. $J$ encodes the nuance between spin 
and spin$^c$ structures. The spectral action selects the Levi-Civita connection so {\it the} Dirac operator $\Ds$.

The way, the operator $\DD$ recaptures the original Riemannian metric $g$ of $M$ is via the Connes' distance:
\begin{definition}
Given a spectral triple $(\A,\H,\DD)$, 
\begin{align}
\label{Cdist}
d(\phi_1,\phi_2) \vc \sup \set{\vert \phi_1(a)-\phi_2(a) \vert \,  \,\vert \, \, \norm{[\DD,\pi(a)]}\leq 1, a \in \A}
\end{align}
defines a distance (eventually infinite) between two states $\phi_1,\phi_2$ on the $C^*$-algebra $A$.
\end{definition}
\noindent In a commutative geometry, any point $x\in M$ defines a state via 
$\phi_x :a\in C^\infty(M) \to a(x)\in \C$. Since the geodesic distance is also given by 
$$d_g(x,y)=\sup \set{\vert a(x)-a(y) \vert \,  \,\vert \, \, a \in C^\infty(M), \norm{\text{grad }a }_\infty \leq 1},$$
we get $d(x,y)=d_g(x,y)$ because $\norm{c(da)}=\norm{\text{grad }a}_\infty$. Recall that $g$ is uniquely determined 
by its distance function by Myers--Steenrod theorem: if $\a:\,(M,g) \to (M',g')$ is a bijection such that 
$d_{g'}\big(\a(x),\a(y)\big)=d_g(x,y)$ for $x,y \in M$, then $g=\a^*g'$.

The role of $\DD$ is non only to provide a metric by \eqref{Cdist}, but its homotopy class represents the 
$K$-homology fundamental class of the noncommutative space $\A$.

It is known that one cannot hear the shape of a drum since the knowledge of the spectrum of a Laplacian does not 
determine the metric of the manifold, even if its conformal class is given \cite{Brooks}. But Theorem 
\ref{Recons} shows that one can hear the shape of a spinorial drum (or better say, of a spectral triple) since the 
knowledge of the spectrum of the 
Dirac operator and the volume form, via its cohomological content, is sufficient to recapture the metric and spin 
structure. See however the more precise refinement made in \cite{ConnesUnitary}: for instance, if $(M,g)$ is a 
compact oriented smooth Riemannian manifold, the spectral triple 
$\big(L^\infty(M), L^2(M,\exter \,T^*M),\DD\big)$ where $\DD = d + d^*$ is the signature operator (see example 
after definition \ref{defdir}) uniquely determines the manifold $M$.

\subsection{Notion of pseudodifferential operators}

\begin{definition}
\label{defOP}
Let $(\A,\H,\DD)$ be a spectral triple.
\\
For $t\in \R$ define the map $F_t:\, T \in \B(\H) \to e^{it \DDD}T e^{-i t\DDD}$ and for $\a \in \R$
\begin{align*}
OP^0 & \vc \set{T \,\, \vert\,\, t \to F_t(T) \in C^\infty \big(\R,\B(\H)\big)} \text{ is the set of operators or order 
$\leq 0$}, \\
OP^\a &\vc \set{ T \,\, \vert \,\, T \DDD^{-\a} \in OP^0}  \text{ is the set of operators of order $\leq \a$}.
\end{align*}
Moreover, we set 
$$
\delta(T) \vc [ \DDD,T], \qquad \nabla(T) \vc [\DD^2,T].
$$
\end{definition}

For instance, $C^\infty(M)=OP^0 \bigcap L^\infty(M)$ and $L^\infty(M)$ is the von Neumann algebra generated by 
$\A=C^\infty(M)$.

\begin{prop}
Assume that $(\A,\H,\DD)$ is regular so $\A \subset OP^0=\bigcap_{k\geq 0} \Dom\,\delta^k \subset \B(\H)$. 
Then, for any 
$\a, \beta \in \R$,
$$
OP^\a \,OP^\beta \subset OP^{\a+\beta},
\quad OP^\a \subset OP^\beta\text{ if }\a\leq \beta, \quad \delta(OP^\a) \subset OP^\a,
\quad \nabla(OP^\a) \subset OP^{\a+1}.
$$
\end{prop}
As an example, let us compute the order of $X=a \,\DDD\,[\DD,b] \,\DD^{-3}$: since the order of $a$ is 0, of $\DDD$ 
is 1, of $[\DD,b]$ is 0 and of $\DD^{-3}$ is -3, we get $X \in OP^{-2}$.

\begin{definition}
\label{defDA}
Let $(\A,\H,\DD)$ be a spectral triple and $\DD(\A)$ be the polynomial algebra generated by 
$\A$, $\A^\circ$, $\DD$ and $\DDD$. 
\\
Define the set of pseudodifferential operators as
$$
\Psi(\A)\vc \set{ T \,\, \vert\, \, \forall N\in \N,\,\exists P \in \DD(\A), \,R \in OP^{-N}, \,p\in \N \text{ such that } 
T=P \DDD^{-p} +R}
$$
\end{definition}
The idea behind this definition is that we want to work modulo the set $OP^{-\infty}$ {\it of smoothing operators}. This 
explains the presence of the arbitrary $N$ and $R$. In the commutative case of a manifold $M$ with spectral triple 
$\big( C^\infty(M),L^2(M,E),\DD \big)$ where $\DD \in \text{Diff}^{\,1}(M,E)$, we get the natural inclusion 
$\Psi \big(C^\infty(M) \big) \subset \Psi DO(M,E)$. \\
The reader should be aware that Definition \ref{defDA} is not exactly the same as in 
\cite{ConnesMarcolli, Polaris, CM2} since it pays attention to the reality operator $J$ when it is present.

\subsection{Zeta-functions and dimension spectrum}

\begin{definition}
\label{ncintdef}
For $P \in \Psi^*(\A)$, we define the zeta-function associated to $P$ (and $\DD$) by 
\begin{align}
\label{zetaPD}
\zeta_\DD^P: \, s \in \C \to \Tr \big(P \DDD^{-s} \big)
\end{align}
which makes sense since for $\Re(s) \gg 1$, $P \DDD^{-s} \in \L^1(\H)$.

The dimension spectrum $Sd(\A,\H,\DD)$ of $(\A,\H,\DD)$ is the set of all poles of 
$\,\zeta_\DD^P(s)$ such that $P \in \Psi(\A) \cap OP^0$. It is said simple if it contains poles of order at most one.

The noncommutative integral of $P$ is defined by
\begin{align}
\label{ncP}
\ncint P \vc \underset{s=0}{\Res} \,\zeta_\DD^P(s).
\end{align}
\end{definition}
In \eqref{zetaPD}, we assume $\DD$ invertible since otherwise, one can replace $\DD$ by the invertible operator 
$\DD + P$, $P$ being the projection on $\Ker\DD$. This change does not modify the computation of 
the integrals $\ncint$ which follow since $\ncint X=0$ when $X$ is a trace-class operator.

The notion of dimension spectrum contains more informations than the usual dimension even 
for a manifold as we will see in Proposition \ref{spectrcomm}.

\begin{remark}
\label{remark-spectrum}
If $Sp(\A,\H,\DD)$ denotes the set of all poles of the functions $s\mapsto \Tr \big(P |D|^{-s}\big)$ where $P$
is any pseudodifferential operator, then, $Sd(\A,\H,\DD) \subseteq Sp(\A,\H,\DD)$.

When $Sp(\A,\H,\DD)=\Z$, $Sd(\A,\H,\DD) = \set{n-k \ : \ k\in \N_0}$: indeed, if  $P$ is a pseudodifferential operator
in $OP^0$, and $q\in \N$ is such that $q>n$, $P|D|^{-s}$ is in $OP^{-\Re(s)}$ so is trace-class for $s$ in a 
neighborhood of $q$; as a consequence, $q$ cannot be a pole of $s\mapsto \Tr \big(P|D|^{-s}\big)$.
\end{remark}

Due to the little difference of behavior between scalar and nonscalar pseudodifferential operators (i.e. 
when coefficients like $[\DD, a]$, $a\in \A$ appears in $P$ of Definition \ref{defDA}), it is convenient to also 
introduce

\begin{definition}
\label{defpseudo1} Let $\DD_{1}(\A)$ be the algebra generated by $\A$, $J\A J^{-1}$ and $\DD$, and 
$\Psi_{1}(\A)$ be the set of pseudodifferential operators constructed as before with $\DD_{1}(\A)$ instead of 
$\DD(\A)$. Note that $\Psi_{1}(\A)$ is subalgebra of $\Psi(\A)$.
\end{definition}

Remark that $\Psi_1(\A)$ does not
necessarily contain operators such as $|D|^k$ where $k\in \Z$ is odd.
This algebra
is similar to the one defined in \cite{CC1}.

\subsection{\texorpdfstring{One-forms and fluctuations of $\DD$}{One-forms and fluctuations of D}}
\label{oneform}

The unitary group $\U(\A)$ of $\A$ gives rise to the automorphism $\a_u: \, a \in \A  \to uau^*\in \A$. 
This defines the inner automorphisms group $Inn(\A)$ which is a normal subgroup of the 
automorphisms $Aut(\A) \vc \set{\a \in Aut(A) \, \vert \, \a(\A) \subset \A}$. For instance, in case of a gauge theory, the 
algebra $\A=C^\infty \big(M,M_n(\C) \big)\simeq C^\infty(M) \otimes M_n(\C)$ is typically used. Then, $Inn(\A)$ is 
locally isomorphic to $\mathcal{G}=C^\infty \big(M,PSU(n)\big)$. Since 
$Aut\big(C^\infty(M)\big) \simeq \text{Diff}(M)$, we get a complete parallel analogy between following two exact 
sequences:

\centerline{\begin{tabular}{ccccccccc}
$1 $ & $\longrightarrow$ & $Inn(\A)$ & $\longrightarrow$ & $Aut(\A)$ & $\longrightarrow$ & $Aut(\A)/Inn(\A)$ & 
$\longrightarrow$ &$1$, \\
$1$ & $\longrightarrow$ & $\mathcal{G}$ & $\longrightarrow$ & $\mathcal{G} \rtimes \text{Diff}(M)$ &
$\longrightarrow$ & $\text{Diff}(M)$ & $\longrightarrow$ & $1$.
\end{tabular}}
\noindent
This justifies that the internal symmetries of physics have to be replaced by the inner automorphisms. \\
We are looking for an equivalence relation between $(\A,\H,\DD)$ and $(\A',\H',\DD')$ giving rise to the same 
geometry. Of course, we could use unitary equivalence: there exists a unitary $U:\, \H \to \H'$ such that 
$\DD'=U\DD U^*$, $U\pi(a)U^* \vc \pi \big(\a(a)\big)$ for some $\a \in Aut(\A)$, and in the even real case 
$[U,\chi]=[U,J]=0$. But this is not useful since it does not change the metric \eqref{Cdist}. So we need to vary 
not only $\DD$ but the algebra and its representation. 

{\it The appropriate framework for inner fluctuations of a spectral triple $(\A,\H,\DD)$ is Morita equivalence that we 
describe now:} 

$\A$ is Morita equivalent to $\B$ if there is a finite projective right $\A$-module $\E$ such that 
$\B\simeq End_\A(\E)$. Thus $\B$ acts on $\H'=\E \otimes_\A \H$  and $\H'$ is endowed with scalar product 
$\langle r \otimes \eta , s \otimes \xi \rangle \vc \langle \eta, \pi(r\vert s)\xi \rangle$ where $(\cdot \vert \cdot)$ is a 
pairing $\E \times \E \to \A$ that is $\A$-linear in the second variable and satisfies $(r \vert s)=(s\vert r)^*$, 
$(r \vert sa)=(r\vert s)a$ and $(s\vert s)>0$ for $0\neq s \in \E$ (this can be seen as a $\A$-valued inner product).

A natural operator $\DD'$ associated to $\B$ and $\H'$ is a linear map
$\DD'(r \otimes \eta)=r \otimes \DD\eta + (\nabla r)\eta$ where $\nabla: \E \to \E\otimes_\A \Omega_\DD^1(\A)$ is 
a linear map obeying to Leibniz rule $\nabla (r a)=(\nabla r)a+r \otimes [\DD,a]$ for $r\in \E, \,a\in \A$ where we took 
the following

\begin{definition}
\label{Omega1}
Let $(\A,\H,\DD)$ be a spectral triple. The set of one-forms is defined as
\begin{align*}
\Omega^1_\DD(\A) \vc \text{span}\set{ a \, db\, \,\vert\, \, a,b \in \A}, \quad db \vc [\DD,b].
\end{align*}
It is a $\A$-bimodule.
\end{definition}
Such $\nabla$ is called a connection on $\E$ and by a result of Cuntz--Quillen, 
only projective modules admit (universal) connections (see \cite{Polaris}[Proposition 8.3]). 
Since we want $\DD'$ selfadjoint,  $\nabla$ must be hermitean with respect to $\DD$ which means: 
$\pi \big( \,(r \vert \nabla s)-(\nabla r \vert s) \,\big)=[\DD, \pi(r\vert s)]$.

In particular, when $\E=\A$ (any algebra is Morita equivalent to itself) and $\A$ is regarded as a right $\A$-module, 
 $\E$ has a natural hermitean connection with respect to $\DD$ given by 
$Ad_\DD:\, a \in A \to [\DD,a]\in \Omega^1_\DD(\A)$ and using the Leibniz rule, any another hermitean connection 
$\nabla$ must verify: $\nabla a =Ad_\DD \,a +A \,a$ where $A=A^* \in \Omega_\DD^1(\A)$. 
So this process, which does not change neither the algebra $\A$ nor the Hilbert space $\H$, gives a natural 
hermitean fluctuation of $\DD$:
$$
\DD \to \DD_A \vc \DD+A \text{ with }A=A^* \in \Omega_\DD^1(\A).
$$
In conclusion, the Morita equivalent geometries for $(\A,\H,\DD)$ keeping fixed $\A$ and $\H$ is an affine space 
modelled on the selfadjoint part of $\Omega_\DD^1(\A)$. 
\\
For instance, in commutative geometries, $\Omega_\Ds^1 \big(C^\infty(M)\big)=\set{c(da) \, \vert \, a \in C^\infty(M)}$.

When a reality operator $J$ exists, we also want $\DD_A J=\epsilon \, J \DD_A$, so we choose
\begin{align}
\label{fluct}
\DD_{\wt A} \vc \DD + \wt A, \quad \wt A \vc A + \epsilon JAJ^{-1}, \quad A=A^*.
\end{align}

The next two results show that, with the same algebra $\A$ and Hilbert space $\H$, a fluctuation of $\DD$ 
still give rise to a spectral triple $(\A,\H,\DD_{A})$ or $(\A,\H,\DD_{\wt A})$.

\begin{lemma}
\label{compres}
Let $(\A,\H,\DD)$ be a spectral triple with a reality operator $J$ and chirality $\chi$. If $A \in \Omega^1_{\DD}$ 
is a one-form, the fluctuated Dirac operator $\DD_{ A} $ or $\DD_{\wt A} $ is an operator with compact resolvent, 
and in particular its kernel is a finite dimensional space. This space is invariant by $J$ and $\chi$.
\end{lemma}

\begin{proof}
Let $T$ be a bounded operator and let $z$ be in the resolvent of $\DD+T$ and $z'$ be in the resolvent of $\DD$. 
Then
$$
(\DD+T-z)^{-1}=(\DD-z')^{-1} \, [1-(T+z'-z)(\DD+T-z)^{-1}].
$$
Since $(\DD-z')^{-1}$ is compact by hypothesis and since the term in
bracket is bounded, $\DD+T$ has a compact resolvent. Applying this to
$T=A+\epsilon JAJ^{-1}$, $\DD_A$ has a finite dimensional kernel (see for instance \cite[Theorem 6.29]{Kato}).

Since according to the dimension, $J^2=\pm 1$, $J$ commutes or anticommutes with
$\chi$,  $\chi$ commutes with the elements in the algebra $\A$ and $\DD \chi=-\chi \DD$, see \eqref{commu}, 
we get $\DD_A \chi=-\chi \DD_A$ and $\DD_A J=\pm J \DD_A$ which gives the result.
\end{proof}

Note that $\U(\A)$ acts on $\DD$ by $\DD \to \DD_u=u\DD u^*$ leaving invariant the spectrum of 
$\DD$. Since $\DD_u=\DD + u[\DD,u^*]$ and in a $C^*$-algebra, any element $a$ is a linear combination of at 
most four unitaries, Definition \ref{Omega1} is quite natural.

The inner automorphisms of a spectral triple correspond to inner fluctuation of the metric defined by \eqref{Cdist}. 

One checks directly that a fluctuation of a fluctuation is a fluctuation and that the unitary group $\U(\A)$ is
gauge compatible for the adjoint representation:

\begin{lemma}
\label{fluctuationoffluctuation}
Let $(\A,\H,\DD)$ be a spectral triple (which is eventually real) and $A \in \Omega^1_\DD(\A)$, $A=A^*$.

(i) If $B\in \Omega^1_{\DD_{A}}(\A)$ ( or $B\in \Omega^1_{\DD_{\wt A}}(\A)$), \\
\centerline{$\DD_{B}=\DD_{C}$ (or $\DD_{\wt B}=\DD_{\wt C}$) with $C\vc A+B$.}

(ii) Let $u \in \U(\A)$. Then $U_u \vc uJuJ^{-1}$ is a unitary of $\H$ such that  
$$
U_u \, \DD_{\wt A} \, {U_u}^*=\DD_{\widetilde{\ga_u(A)}}, \text{ where } \ga_u(A) \vc u[\DD,\,u^*] + uAu^*.
$$
\end{lemma}

\begin{remark}
To be an inner fluctuation is not a symmetric relation. It can append that $\DD_A=0$ with $\DD \neq 0$.
\end{remark}

\begin{lemma}
\label{adjoint}
Let $(\A,\DD,\H)$ be a spectral triple and $X \in \Psi(\A)$. Then 
\begin{align*}
\ncint X^*=\overline{  \ncint X}.
\end{align*}
If the spectral triple is real, then, for $X \in \Psi(\A)$, $JXJ^{-1} \in \Psi(\A)$ and
$$
\ncint JXJ^{-1}=\ncint X^*=\overline{  \ncint X}.
$$
\end{lemma}

\begin{proof}
The first result follows from (for $\Re s$ large enough, so the operators are traceable)
\begin{align*}
\Tr(X^*\vert \DD\vert^{-s})=\Tr \big((\vert
\DD\vert^{-\bar{s}})X)^*\big)=\overline{ \Tr(\vert \DD \vert^{-\bar{s}}
X)}=\overline{\Tr(X\vert \DD \vert^{-\bar{s}})}.
\end{align*}

The second result is due to the anti-linearity of $J$,
$\Tr(JYJ^{-1})=\overline{\Tr(Y)}$, and $J\vert \DD \vert=\vert \DD \vert J$, so 
\begin{align*}
\Tr(X \vert \DD \vert^{-s})=\overline{\Tr(JX \vert \DD
\vert^{-s}J^{-1})}=\overline{\Tr(JXJ^{-1}\vert \DD \vert^{-\bar{s}})}.
\end{align*}
\end{proof}

\begin{corollary}
\label{reel}
For any one-form $A=A^*$, and for $k,\,l \in \N$,
$$
\ncint A^l \,\DD^{-k} \in \R,\quad \ncint \big(A\DD^{-1}\big)^k \in \R,
\quad \ncint A^l \,\vert \DD\vert^{-k} \in \R,\quad  \ncint \chi A^l\,\vert \DD \vert ^{-k} \in \R, 
\,\, \ncint A^l\,\DD \, \vert \DD \vert^{-k} \in \R.
$$
\end{corollary}

We remark that the fluctuations leave invariant the first term of the spectral action (\ref{asympspectral}). This is 
a generalization of the fact that in the commutative case, the noncommutative integral
depends only on the principal symbol of the Dirac operator $\DD$ and this symbol is stable by adding a gauge
potential like in $\DD+A$. Note however that the symmetrized gauge potential $A+\epsilon JAJ^{-1}$ is always 
zero in this case for any selfadjoint one-form $A$, see \eqref{JAJ}.

\begin{theorem}
\label{difference}
Let $(\A,\H,\DD)$ be a regular spectral triple which is simple and of dimension $d$. Let $A\in \Omega_\DD^1(\A)$ 
be a selfadjoint gauge potential. Then,
\begin{align}
\zeta_{D_{\wt A}}(0)=\zeta_{D}(0)+\sum_{q=1}^{d} \tfrac{(-1)^{q}}{q} \ncint (\wt AD^{-1})^{q}. \label{termconstanttilde}
\end{align}
\end{theorem}
The proof needs few preliminaries.

\begin{definition}
For an operator $T$, define the one-parameter group and notation
\begin{align*}
& \sigma_z(T) \vc|D|^{z}T|D|^{-z}, \, z\in \C.\\
&\epsilon(T)\vc\nabla(T)D^{-2}, \text{(recall that }\nabla (T)=[\DD^2,T]).
\end{align*}
\end{definition}
\noindent The expansion of the one-parameter group $\sigma_z$ gives for $T\in OP^q$
\begin{equation}
\label{one-par}
\sigma_{z}(T)\sim \sum_{r=0}^{N} g(z,r) \,\eps^r(T)  \mod OP^{-N-1+q}
\end{equation}
where $g(z,r) \vc \tfrac{1}{r!}(\tfrac{z}{2})\cdots(\tfrac{z}{2}-(r-1))=\genfrac(){0pt}{1}{z/2}{r}$
with the convention $g(z,0)\vc1$.

We fix a regular spectral triple $(\A,\H,\DD)$ of dimension $d$ and a self-adjoint 1-form $A$.
\\
Despite previous remark before Lemma \ref{adjoint}, we pay attention here to the kernel of $\DD_A$ since this 
operator can be non-invertible even if $\DD$ is, so we define
\begin{align}
\DD_A &\vc \DD+\wt A \text{  where } \wt A\vc A +\eps JAJ^{-1} ,\nonumber\\
 D_A &\vc  \DD_A + P_A \label{DAdroit}
\end{align}
where $P_A$ is the projection on $\Ker \DD_A$. Remark that
$\wt A \in\DD(\A)\cap OP^0$ and $\DD_A\in \DD(\A)\cap OP^1$.
\\
We note
$$
V_A\vc P_A - P_0\,.
$$
As the following lemma shows, $V_A$ is a smoothing operator:

\begin{lemma}
\label{finiterank}
(i) $\bigcap_{k\geq 1} \Dom (\DD_A)^{k} \subseteq \bigcap_{k\geq 1}
\Dom |D|^k$.

(ii) $\Ker \DD_A \subseteq \bigcap_{k\geq 1} \Dom |D|^k$.

(iii) For any $\a, \beta \in \R$, $|D|^\beta P_A |D|^\a$ is bounded.

(iv) $P_A \in OP^{-\infty}$.
\end{lemma}
\begin{proof}
$(i)$ Let us define for any $p\in \N$, $R_p \vc (\DD_A)^p -\DD^p$, so
$R_p \in OP^{p-1}$ and moreover $R_p \big(\Dom |D|^p\big)\subseteq \Dom |D|$.

Let us fix $k\in \N$, $k\geq 2$. Since $\Dom \DD_A = \Dom \DD =\Dom |D|$, we have
$$
\Dom (\DD_A)^k = \set{\phi \in \Dom |D| \ : \ (\DD^j + R_j)\,\phi \in \Dom |D| \ , \ \forall j\ \  1\leq j\leq k-1  }.
$$
Let $\phi \in \Dom (\DD_A)^k$. We prove by recurrence that for any
$j\in\set{1,\cdots,k-1}$, $\phi \in \Dom |D|^{j+1}$:

We have $\phi\in \Dom |D|$ and $(\DD + R_1)\, \phi \in \Dom |D|$. Thus, since $R_1\,\phi \in \Dom |D|$, 
we have $\DD \phi \in \Dom |D|$, which proves that $\phi \in \Dom |D|^2$.
Hence, case $j=1$ is done.

Suppose now $\phi \in \Dom |D|^{j+1}$ for a $j \in \set{1,\cdots,k-2}$.
Since $(\DD^{j+1} + R_{j+1})\, \phi \in \Dom |D|$, and $R_{j+1}\, \phi \in \Dom |D|$, we get $\DD^{j+1}\,
\phi \in \Dom |D|$,
which proves that $\phi\in \Dom |D|^{j+2}$.

Finally, if we set $j=k-1$, we get $\phi \in \Dom |D|^{k}$, so $\Dom (\DD_A)^k \subseteq \Dom |D|^k$.

$(ii)$ follows from $\Ker \DD_A \subseteq \bigcap_{k\geq 1} \Dom (\DD_A)^k$ and $(i)$.

$(iii)$ Let us first check that $|D|^\a P_A$ is bounded. We define $D_0$ as the operator with domain 
$\Dom D_0 = \Ima P_A \cap \Dom |D|^\a$ and such that $D_0\, \phi = |D|^\a\, \phi.$ Since $\Dom D_0$ is
finite dimensional, $D_0$ extends as a bounded operator on $\H$ with finite rank.
We have
$$
\sup_{\phi \in \Dom |D|^\a P_A,\ \norm{\phi}\leq 1} \norm{|D|^\a P_A\, \phi} \leq
\sup_{\phi \in \Dom D_0,\ \norm{\phi}\leq 1} \norm{|D|^\a\, \phi} = \norm{D_0}<\infty
$$
so $|D|^\a P_A$ is bounded. We can remark that by $(ii)$, $\Dom D_0 = \Ima P_A$ and
$\Dom |D|^\a P_A = \H$.

Let us prove now that $P_A |D|^\a$ is bounded:
Let $\phi\in \Dom P_A |D|^\a = \Dom |D|^\a$. By $(ii)$, we have $\Ima P_A \subseteq \Dom |D|^\a$ so we get
\begin{align*}
\norm{P_A |D|^\a\,\phi} & \leq \sup_{\psi \in \Ima P_A,\ \norm{\psi}\leq 1} |<\psi,|D|^\a\, \phi>|
\leq \sup_{\psi \in \Ima P_A,\ \norm{\psi}\leq 1} |<|D|^\a\psi,\phi>| \\
&\leq \sup_{\psi \in \Ima P_A,\ \norm{\psi}\leq 1} \norm{|D|^\a\psi}\norm{\phi} = \norm{D_0} \norm{\phi}.
\end{align*}

$(iv)$ For any $k\in \N_0$ and $t\in \R$, $\delta^k(P_A)|D|^t$ is a linear combination of terms of the form 
$|D|^\beta P_A |D|^\a$, so the result follows from $(iii)$.
\end{proof}

\begin{remark}
We will see later on the noncommutative torus example how important is the difference
between $\DD_{A}$ and $\DD+A$. In particular, the inclusion $\Ker \DD \subseteq \Ker \DD +A$ is not satisfied 
since $A$ does not preserve $\Ker \DD$ contrarily to $\wt A$.
\end{remark}

Let us define
\begin{align*}
&X \vc \DD_{A}^2-\DD^2 =\wt A \DD + \DD \wt A + \wt A^2 ,\\
&X_V \vc X+V_A,
\end{align*}
thus $X\in \DD_1(\A)\cap OP^1$ and by Lemma \ref{finiterank},
\begin{equation}
\label{xvsim}
X_V \sim X \mod OP^{-\infty}.
\end{equation}
\noindent 
We will use
$$
Y\vc\log(D_A^2) -\log (D^2)
$$
which makes sense since $D_A^2 = \DD_A^2 + P_A$ is invertible for any
$A$. By definition of $X_V$, we get
$$
Y= \log (D^2 + X_V) -\log (D^2).
$$

\begin{lemma} 
\label{2dev}

(i) $Y$ is a pseudodifferential operator in $OP^{-1}$ with the
following
expansion for any $N\in\N$
$$
Y \sim \sum_{p=1}^N\sum_{k_1,\cdots,k_p
=0}^{N-p}\tfrac{(-1)^{|k|_1+p+1}}{|k|_1+p}
\nabla^{k_p}(X\nabla^{k_{p-1}}(\cdots
X\nabla^{k_1}(X)\cdots)) D^{-2(|k|_1+p)} \mod OP^{-N-1}.
$$

(ii) For any $N\in\N$ and $s\in \C$,
\begin{align}
\label{expansion} 
|D_A|^{-s} \sim |D|^{-s} + \sum_{p=1}^N K_p(Y,s) |D|^{-s} \mod OP^{-N-1-\Re(s)}
\end{align}
with $K_p(Y,s)\in OP^{-p}$.
\end{lemma}

\begin{proof}
$(i)$ We follow \cite[Lemma 2.2]{CC1}. By functional calculus, $Y=\int_0^\infty I(\la)\, d\la$, where
$$
I(\la)\sim\sum_{p=1}^N(-1)^{p+1}\big((D^2+\la)^{-1}X_V\big)^{p} (D^2+\la)^{-1} \mod OP^{-N-3}.
$$
By (\ref{xvsim}), $\big((D^2+\la)^{-1}X_V\big)^{p} \sim \big((D^2+\la)^{-1}X\big)^{p} \mod OP^{-\infty}$ and we get
$$
I(\la)\sim\sum_{p=1}^N(-1)^{p+1}\big((D^2+\la)^{-1}X\big)^{p} (D^2+\la)^{-1} \mod OP^{-N-3}.
$$
We set $A_p(X)\vc\big((D^2+\la)^{-1}X\big)^{p}(D^2+\la)^{-1}$ and $L\vc(D^2+\la)^{-1}\in OP^{-2}$ for a fixed $\la$.
Since $[D^2 + \la,X]\sim \nabla(X) \mod OP^{-\infty}$, a recurrence proves that if $T$ is an operator
in $OP^{r}$, then, for $q\in \N_0$,
$$
A_1(T)=L T L \sim \sum_{k=0}^q (-1)^k\nabla^k(T) L^{k+2} \mod OP^{r-q-5}.
$$
With $A_p(X)=LX A_{p-1}(X)$, another recurrence gives, for any $q\in \N_0$,
$$
A_p(X)\sim \sum_{k_1,\cdots,k_p =0}^q (-1)^{|k|_1}\nabla^{k_p}
(X \nabla^{k_{p-1}}(\cdots X\nabla^{k_1}(X)\cdots)) L^{|k|_1+p+1} \mod OP^{-q-p-3},
$$
which entails that
$$
I(\la)\sim\sum_{p=1}^N(-1)^{p+1}\sum_{k_1,\cdots,k_p =0}^{N-p}(-1)^{|k|_1}\nabla^{k_p}(X\nabla^{k_{p-1}}
(\cdots X\nabla^{k_1}(X)\cdots)) L^{|k|_1+p+1} \mod OP^{-N-3}.
$$

With $\int_{0}^\infty (D^2+\la)^{-(|k|_1+p+1)}d\la = \tfrac{1}{|k|_1+p} D^{-2(|k|_1+p)}$, we get the result 
provided we control the remainders. Such a control is given in \cite[(2.27)]{CC1}.

$(ii)$ Applied to $|D_A|^{-s}=e^{B-(s/2)Y}e^{-B}\, |D|^{-s}$ where $B\vc (-s/2)\log(D^2)$, the Duhamel's 
expansion formula 
$$
e^{U+V}e^{-U}=\sum_{n=0}^\infty \int_{0\leq t_1\leq t_2\leq\cdots\leq t_n \leq1} V(t_1)\cdots V(t_n) \,dt_1 \cdots dt_n
$$
with $V(t)\vc e^{tU}Ve^{-tU}$ gives
\begin{equation}
\label{egalite-DAs}
|D_A|^{-s} = |D|^{-s} + \sum_{p=1}^\infty K_p(Y,s)|D|^{-s}\, .
\end{equation}
and each $K_p(Y,s)$ is in $OP^{-p}$.
\end{proof}

\begin{corollary}
\label{eps-pdo} For any $p\in\N$ and $r_1,\cdots,r_p \in \N_0$, $\eps^{r_1}(Y)\cdots \eps^{r_p}(Y) \in \Psi_1(\A)$.
\end{corollary}

\begin{proof}
If for any $q\in \N$ and $k=(k_1,\cdots,k_q)\in \N_0^q$,
$$
\Ga_q^k(X)\vc\tfrac{(-1)^{|k|_1+q+1}}{|k|_1+q} \nabla^{k_q} (X\nabla^{k_{q-1}}(\cdots X\nabla^{k_1}(X)\cdots)),
$$
then, $\Ga_q^k(X) \in OP^{|k|_1+q}$. For any $N\in\N$,
\begin{equation}
\label{Ydev}
Y \sim \sum_{q=1}^N\sum_{k_1,\cdots,k_q =0}^{N-q} \Ga_q^k(X) D^{-2(|k|_1+q)} \mod OP^{-N-1}.
\end{equation}
Since the $\Ga_q^k(X)$ are in $\DD(\A)$, this proves with (\ref{Ydev}) that $Y$
and thus $\eps^r(Y)= \nabla^r(Y)D^{-2r}$, are also in $\Psi_1(\A)$.
\end{proof}

\medskip

\begin{proof}[Proof of Theorem \ref{difference}] Again, we follow \cite{CC1}. Since the spectral triple is simple, 
equation (\ref{egalite-DAs}) entails that
$$
\zeta_{D_A}(0)-\zeta_{D}(0) = \Tr (K_1(Y,s)|D|^{-s})_{|s=0} \, .
$$
Thus, with (\ref{one-par}), we get $\zeta_{D_A}(0)-\zeta_{D}(0) = -\half \ncint Y$. 
\\
Now the conclusion follows from 
$\ncint \log\big((1+S)(1+T)\big)=\ncint \log(1+S)+ \ncint \log (1+T)$ for $S,T\in \Psi(\A) \cap OP^{-1}$ (since 
$\log(1+S)=\sum_{n=1}^\infty \tfrac{(-1)^{n+1}}{n} S^n$) with $S=AD^{-1}$ and $T=DAD^{-2}$; so 
$\ncint \log(1+X D^{-2})=2\ncint \log(1+AD^{-1})$ and 
\begin{align*}
-\half \ncint Y = \sum_{q=1}^{d} \tfrac{(-1)^{q}}{q} \ncint (\wt AD^{-1})^{q}.
\end{align*}
\end{proof}

\begin{lemma}
\label{Res-zeta-n-k}
For any $k\in \N_0$,
$$
\underset{s=d-k}{\Res} \, \zeta_{D_A}(s)= \underset{s=d-k}{\Res} \,\zeta_{D}(s) +
\sum_{p=1}^k \sum_{r_1,\cdots, r_p =0}^{k-p} \underset{s=d-k}{\Res} \, h(s,r,p) \, \Tr\big(\eps^{r_1}(Y)
\cdots\eps^{r_p}(Y) |D|^{-s}\big),
$$
where
$$
h(s,r,p)\vc(-s/2)^p\int_{0\leq t_1\leq \cdots \leq t_p\leq 1} g(-st_1,r_1)\cdots g(-st_p,r_p) \, dt\, .
$$
\end{lemma}

\begin{proof}
By Lemma \ref{2dev} $(ii)$, $|D_A|^{-s} \sim |D|^{-s} + \sum_{p=1}^k K_p(Y,s)
|D|^{-s} \mod OP^{-(k+1)-\Re(s)}$, where the convention $\sum_{\emptyset}=0$ is used.
Thus, we get for $s$ in a neighborhood of $d-k$,
\begin{align*}
|D_A|^{-s}-|D|^{-s} - \sum_{p=1}^k K_p(Y,s) |D|^{-s} \in OP^{-(k+1)-\Re(s)}\subseteq \L^1(\H)
\end{align*}
which gives
\begin{equation}
\label{res-n-k-interm}
\underset{s=d-k}{\Res} \, \zeta_{D_A}(s)= \underset{s=d-k}{\Res} \,\zeta_{D}(s) + 
\sum_{p=1}^k \underset{s=d-k}{\Res} \,\Tr \big(K_p(Y,s) |D|^{-s}\big).
\end{equation}
Let us fix $1\leq p\leq k$ and $N\in \N$. By (\ref{one-par}) we get
\begin{align}
\label{K_p}
K_p(Y,s)\sim (-\tfrac s2)^p \int_{0\leq t_1\leq \cdots t_p \leq 1} \sum_{r_1,\cdots,r_p =0}^N g(-st_1,r_1)
&\cdots g(-st_p,r_p) \nonumber\\
&\eps^{r_1}(Y)\cdots \eps^{r_p}(Y)\, dt \, \mod OP^{-N-p-1}.
\end{align}
If we now take $N=k-p$, we get for $s$ in a neighborhood of $d-k$
$$
K_p(Y,s)|D|^{-s} - \sum_{r_1,\cdots,r_p =0}^{k-p}h(s,r,p)\,\eps^{r_1}(Y)\cdots
\eps^{r_p}(Y)|D|^{-s} \in OP^{-k-1-\Re(s)} \subseteq \L^1(\H)
$$
so (\ref{res-n-k-interm}) gives the result.
\end{proof}

Our operators $|D_A|^k$ are pseudodifferential operators:

\begin{lemma}
For any $k\in \Z$, $\vert D_{A} \vert^k \in \Psi^k(\A)$.
\end{lemma}

\begin{proof}
Using (\ref{K_p}), we see that $K_p(Y,s)$ is a pseudodifferential
operator in $OP^{-p}$, so (\ref{expansion}) proves that $|D_A|^k$
is a pseudodifferential operator in $OP^k$.
\end{proof}

The following result is quite important since it shows that one
can use $\ncint$ for $D$ or $D_{A}$:
\begin{prop}
\label{ncintfluctuated}
If the spectral triple is simple, $\underset{s=0}{\Res} \, \Tr \big(P |D_A|^{-s}\big) = \ncint P$
for any pseudodifferential operator $P$. In particular, for any $k\in \N_0$
$$
\ncint |D_A|^{-(d-k)}=\underset{s=d-k}{\Res} \,\zeta_{D_A}(s) .
$$
\end{prop}

\begin{proof}
Suppose $P\in OP^{k}$ with $k\in \Z$ and let us fix $p\geq 1$.
With (\ref{K_p}), we see that for any $N\in \N$,
$$
PK_p(Y,s)|D|^{-s}\sim \sum_{r_1,\cdots,r_p =0}^N h(s,r,p) \,
P\eps^{r_1}(Y)\cdots \eps^{r_p}(Y)|D|^{-s} \mod OP^{-N-p-1+k-\Re(s)}.
$$
Thus if we take $N=d-p+k$, we get
$$
 \underset{s=0}{\Res} \,\Tr \big(P K_p (Y,s) |D|^{-s}\big) =  \sum_{r_1,\cdots,r_p =0}^{n-p+k} \underset{s=0}{\Res} \,\,
h(s,r,p) \, \Tr \big(P\eps^{r_1}(Y)\cdots \eps^{r_p}(Y) |D|^{-s}\big).
$$
Since $s=0$ is a zero of the analytic function $s\mapsto h(s,r,p)$ and 
$s\mapsto \Tr P\eps^{r_1}(Y)\cdots \eps^{r_p}(Y)|D|^{-s}$ has only simple poles by hypothesis, we get 
 $\underset{s=0}{\Res} \, h(s,r,p) \, \Tr \big(P\eps^{r_1}(Y)\cdots \eps^{r_p}(Y) |D|^{-s}\big)=0$ and
\begin{equation}
\label{res0K_p}
\underset{s=0}{\Res} \, \Tr \big(P K_p (Y,s) |D|^{-s}\big)=0.
\end{equation}
Using (\ref{expansion}), $ P|D_A|^{-s} \sim P |D|^{-s} + \sum_{p=1}^{k+d} PK_p(Y,s) |D|^{-s} \mod OP^{-d-1-\Re(s)}$ 
and thus,
\begin{equation}\label{res0PD_A}
\underset{s=0}{\Res} \, \Tr (P|D_A|^{-s}) =\ncint P + \sum_{p=1}^{k+d} \, \underset{s=0}{\Res} \, 
\Tr \big(PK_p(Y,s) |D|^{-s}\big).
\end{equation}
The result now follows from (\ref{res0K_p}) and (\ref{res0PD_A}). To get the last equality, one uses the 
pseudodifferential operator $|D_A|^{-(d-k)}$.
\end{proof}

\begin{prop}
\label{invariance1} If the spectral triple is simple, then
\begin{align}
\label{invDAd}
\ncint {|D_{A}|}^{-d}&=\ncint |D|^{-d}.
\end{align}
\end{prop}

\begin{proof} Lemma \ref{Res-zeta-n-k} and previous proposition for $k=0$.
\end{proof}

\begin{lemma}
\label{residus-particuliers}
If the spectral triple is simple,
\begin{align*}
\hspace{-2.2cm} (i) & \quad \ncint |D_A|^{-(d-1)}= \ncint |D|^{-(d-1)} -(\tfrac{d-1}{2})\ncint X|D|^{-d-1}.\\
\hspace{-2.2cm} (ii)  & \quad \ncint |D_A|^{-(d-2)}= \ncint |D|^{-(d-2)}+\tfrac{d-2}{2}\big(-\ncint X|D|^{-d} + \tfrac{d}{4}
\ncint X^2 |D|^{-2-d} \big).
\end{align*}
\end{lemma}

\begin{proof}
$(i)$ By (\ref{expansion}),
$$
\underset{s=d-1}{\Res} \, \zeta_{D_A}(s) -\zeta_{D}(s)= \underset{s=d-1}{\Res} \,(-s/2)
\Tr \big(Y |D|^{-s}\big) = -\tfrac{d-1}{2} \, \underset{s=0}{\Res} \, \Tr \big(Y|D|^{-(d-1)}|D|^{-s} \big)
$$
where for the last equality we use the simple dimension spectrum
hypothesis. Lemma \ref{2dev} $(i)$ yields
$Y\sim XD^{-2} \mod OP^{-2}$ and $Y|D|^{-(d-1)}\sim X|D|^{-d-1} \mod OP^{-d-1}\subseteq \L^1(\H)$. Thus,
$$
\underset{s=0}{\Res} \, \Tr \big( Y|D|^{-(d-1)}|D|^{-s}\big) = \underset{s=0}{\Res} \,
\Tr \big(X|D|^{-d-1} |D|^{-s}\big) = \ncint  X |D|^{-d-1}.
$$
$(ii)$ Lemma \ref{Res-zeta-n-k} $(ii)$ gives
$$
\underset{s=d-2}{\Res} \, \zeta_{D_A}(s) = \underset{s=d-2}{\Res} \, \zeta_{D}(s) +
\underset{s=d-2}{\Res} \, \sum_{r=0}^1 h(s,r,1) \, \Tr \big(\eps^r(Y)|D|^{-s}\big) + h(s,0,2) \,\Tr \big(Y^2 |D|^{-s}\big).
$$
We have $h(s,0,1)=-\tfrac s2$, $h(s,1,1)= \half(\tfrac s2)^2$ and $h(s,0,2)= \half (\tfrac s2)^2$. Using again 
Lemma \ref{2dev} $(i)$,
$$
Y\sim XD^{-2}-\half \nabla(X)D^{-4} -\half X^2 D^{-4} \mod OP^{-3}.
$$
Thus,
$$
\underset{s=d-2}{\Res} \, \Tr \big(Y|D|^{-s}\big) =
\ncint X|D|^{-d} -\half \ncint
(\nabla(X)+X^2)|D|^{-2-d}.
$$
Moreover, using $\ncint \nabla(X)|D|^{-k}=0$ for any $k\geq 0$ since $\ncint$ is a trace,
$$
\underset{s=d-2}{\Res} \, \Tr \big(\eps(Y)|D|^{-s}\big) = \underset{s=d-2}{\Res} \, \Tr
\big(\nabla(X)D^{-4}|D|^{-s}\big) = \ncint \nabla(X)|D|^{-2-d}=0.
$$
Similarly, since $Y\sim XD^{-2}$ mod $OP^{-2}$ and $Y^2\sim X^2D^{-4} \mod OP^{-3}$, we get
$$
\underset{s=d-2}{\Res} \, \Tr \big(Y^2 |D|^{-s}\big) = \underset{s=d-2}{\Res} \, \Tr
\big(X^2D^{-4}|D|^{-s}\big) = \ncint X^2 |D|^{-2-d}.
$$
Thus,
\begin{align*}
\underset{s=d-2}{\Res} \, \zeta_{D_A}(s) = \underset{s=d-2}{\Res} \,\zeta_{D}(s) +
&(-\tfrac {d-2}{2})(\ncint X|D|^{-d}-\half \ncint (\nabla(X)+X^2)|D|^{-2-d})\\
& \quad +\half(\tfrac {d-2}{2})^2 \ncint \nabla(X)|D|^{-2-n}+\half(\tfrac {d-2}{2})^2 \ncint X^2 |D|^{-2-d}.
\end{align*}
Finally,
$$
\underset{s=d-2}{\Res} \, \zeta_{D_A}(s) = \underset{s=d-2}{\Res} \,\zeta_{D}(s) +
(-\tfrac {d-2}{2}) \big(\ncint X|D|^{-d} - \half \ncint X^2|D|^{-2-d}\big)+\half(\tfrac {d-2}{2})^2 \ncint X^2 |D|^{-2-d}
$$
and the result follows from Proposition \ref{ncintfluctuated}.
\end{proof}

\begin{corollary}
\label{res-n-2-A}
If the spectral triple satisfies $\ncint |D|^{-(d-2)}= \ncint \wt A \DD |D|^{-d} = \ncint  \DD \wt A |D|^{-d}=0$, 
then
$$
 \ncint |D_A|^{-(d-2)} = \tfrac{d(d-2)}{4} \big(\ncint \wt A \DD \wt A \DD |D|^{-d-2}+\tfrac{d-2}{d}
 \ncint \wt A ^2|D|^{-d}\big).
$$
\end{corollary}

\begin{proof}
By previous lemma,
$$
\underset{s=d-2}{\Res} \, \zeta_{D_A}(s) = \tfrac{d-2}{2}\big( -\ncint \wt A^2 |D|^{-d}
+\tfrac{d}{4}\ncint (\,\wt A\DD \wt A \DD+ \DD \wt A \DD \wt A+
\wt A \DD^2 \wt A + \DD \wt A^2 \DD \,) |D|^{-d-2} \big).
$$
Since $\nabla(\wt A) \in OP^1$, the trace property of $\ncint$ yields the result.
\end{proof}

\subsection{Tadpole}

In \cite{ConnesMarcolli}, the following definition is introduced:

\begin{definition}
\label{Deftadpole}
In $(\A,\,\H,\,\DD)$, the tadpole $Tad_{\DD+ A}(k)$ of order $k$, for  $k \in \set{d-l \, : \, l \in \N}$  is
the term linear in $A=A^{*}\in \Omega_\DD^1$, in the $\Lambda^k$ term of 
\eqref{asympspectral} (considered as an infinite series) where $\DD\to\DD+ A$.

If moreover, the triple $(\A,\,\H,\,\DD,\,J)$ is real, the tadpole $Tad_{\DD+\tilde A}(k)$ is the term linear in $A$, 
in the $\Lambda^k$ term of \eqref{asympspectral} where $\DD\to\DD+\wt A$.
\end{definition}

\begin{prop}
\label{valeurtadpole}
Let $(\A,\,\H,\,\DD)$ be a spectral triple of dimension $d$ with simple dimension spectrum. Then 
\begin{align}
    \label{tadpolen-k}
&\Tad_{\DD+A}(d-k) =-(d-k)\ncint A \DD |\DD|^{-(d-k) -2}, \quad \forall  k\neq d,\\
    \label{tadpole0}
& \Tad_{\DD+ A}(0)=-\ncint  A \DD^{-1}.
\end{align}
Moreover, if the triple is real,  $\Tad_{\DD+\wt A}= 2\Tad_{\DD+A}$.
\end{prop}

\begin{proof}
We already proved the following formula, for any $k\in \N$,
$$
\ncint |\DD_A|^{-(d-k)}=
\ncint |\DD|^{-(d-k)}+
\sum_{p=1}^k \sum_{r_1,\cdots, r_p =0}^{k-p}
\underset{s=d-k}{\Res} \, h(s,r,p) \, \Tr\big(\eps^{r_1}(Y)
\cdots\eps^{r_p}(Y) |\DD|^{-s}\big),
$$
with here $X:=\wt A \DD + \DD \wt A+\wt A^2, \wt A :=A+\epsilon JAJ^{-1}$.

As a consequence, for $k\neq n$, only the terms with $p=1$ contribute to the linear part:
$$
\Tad_{\DD + \wt A}(d-k)= \Lin_A(\ncint |\DD_A|^{-(d-k)}) =\sum_{r=0}^{k-1}
\underset{s=d-k}{\Res} \, h(s,r,1) \,
\Tr\big(\eps^{r}(\Lin_A(Y))|\DD|^{-s}\big)\, .
$$
We check that for any $N\in \N^*$, 
$$
\Lin_A(Y) \sim \sum_{l=0}^{N-1} \Ga_1^l(\wt A \DD + \DD \wt A)
\DD^{-2(l+1)} \mod OP^{-N-1}.
$$
Since $\Ga_1^l(\wt A \DD + \DD \wt A) = \tfrac{(-1)^{l}}{l+1}
\nabla^{l}
(\wt A \DD + \DD \wt A)= \tfrac{(-1)^{l}}{l+1} \{ \nabla^l(\wt A),\DD
\}$, we get, assuming the dimension spectrum to be simple
\begin{align*}
\Tad_{\DD+ \wt A}(d-k)&= \sum_{r=0}^{k-1}
\underset{s=d-k}{\Res} \, h(s,r,p) \,
\Tr\big(\eps^{r}(\Lin_A(Y))|\DD|^{-s}\big) \\
&= \sum_{r=0}^{k-1} h(n-k,r,1) \sum_{l=0}^{k-1-r}\tfrac{(-1)^{l}}{l+1}
\underset{s=d-k}{\Res} \,\Tr\big( \eps^{r}(\{ \nabla^l(\wt A),\DD\})|\DD|^{-s-2(l+1)}\big)\\
 & =  2\sum_{r=0}^{k-1} h(d-k,r,1) \sum_{l=0}^{k-1-r}\tfrac{(-1)^{l}}{l+1}
\ncint \nabla^{r+l}(\wt A) \DD |\DD|^{-(d-k + 2(r+l)) -2}  \\ 
& = -(n-k)
\ncint \wt A \DD |\DD|^{-(d-k) -2},
\end{align*}
because in the last sum it remains only the case $r+l=0$, so $r=l=0$.

Formula \eqref{tadpole0} is a direct application of Theorem \ref{difference}.

The link between  $\Tad_{\DD+\wt A}$ and $\Tad_{\DD+A}$ follows from $J\DD=\epsilon 
\DD J$ and Lemma \ref{adjoint}.
\end{proof}

\begin{corollary} 
\label{Atilde=0}
In a real spectral triple $(\A,\H,\DD)$, if $A=A^*\in \Omega_\DD^1(\A)$ 
is such that $\wt A=0$,  then $\Tad_{D+A}(k) =0$ for any $k\in \Z$, $k\leq d$. 
\end{corollary}

The vanishing tadpole of order 0 has the following equivalence (see \cite{CC1})
\begin{align}
  \label{equ}
\ncint A \DD^{-1}=0, \,\forall A \in \Omega_{\DD}^1(\A)
\,\Longleftrightarrow  \,\ncint a b=\ncint a \a(b), \,\forall a,b
\in \A, 
\end{align}
where $\a(b):=\DD b\DD^{-1}$.
\vskip0.2cm
The existence of tadpoles is important since, for instance, $A=0$ is not necessarily a stable solution of the classical 
 field equation deduced from spectral action expansion, 
\cite{GW}.

\subsection {Commutative geometry}
\label{Commutative geometry}

\begin{definition}
\label{rieman}
Consider a commutative spectral triple given by a compact Riemannian spin manifold  $M$ of 
dimension $d$ without boundary and its Dirac operator $\Ds$ associated to the Levi--Civita connection.
This means $\big( \A:=C^{\infty}(M),\, \H:=L^2(M,S),\,\Ds \big)$ where $S$ is the spinor bundle over $M$. 
This triple is real since, due to the existence of a spin structure, the charge conjugation operator
generates an anti-linear isometry $J$ on $\H$ such that
\begin{align}
\label{JaJ(-1)}
JaJ^{-1}=a^*,\quad \forall a \in \A,
\end{align}
and when $d$ is even, the grading is given by the chirality matrix
\begin{align}
\label{chi}
\chi:=(-i)^{d/2}\,\ga^1\ga^2\cdots\ga^d.
\end{align}
Such triple is said to be a commutative geometry.
\end{definition}

In the polynomial algebra $\DD(\A)$ of Definition \ref{defDA}, we added $\A^\circ$. In the commutative case, 
$\A^\circ \simeq J\A J^{-1} \simeq \A$ as indicated by \eqref{JaJ(-1)} which also gives
\begin{align}
\label{JAJ}
JAJ^{-1}=-\epsilon\,A^* , \quad \forall A\in \Omega^1_\DD(\A) \text{ or } \wt A=0 \text{ when } A=A^*.
\end{align}

As noticed by Wodzicki, $\ncint P$ is equal to $-2$ times the coefficient in log $t$ of the asymptotics of 
$\Tr(P\, e^{-t\,\Ds^2}$) as $t \rightarrow 0$. It is remarkable that this coefficient is independent of $\Ds$ as 
seen in Theorem \ref{Wresresult} and this gives a close relation between the $\zeta$ function and 
heat kernel expansion with {\it WRes}. Actually, by \cite[Theorem 2.7]{GS}

\begin{align}
\label{heat}
\Tr(P\, e^{-t\,\Ds^2}) \underset{t \downarrow 0^+}{\sim} \, \sum_{k=0}^{\infty} a_k\,t^{(k-ord(P)-d)/2} 
+ \sum_{k=0}^{\infty} (-a'_k\,\log t+b_k)\,t^k,
\end{align}
so 
$$
\ncint P=2 a'_0.
$$

Remark that $\ncint$, {\it WRes} are traces on $\Psi\big(C^\infty(M)\big)$, thus Corollary \ref{uniquetrace} implies 
\begin{align}
\label{int=cWres}
\ncint P=c \,\text{\it WRes }P.
\end{align}
Since, via Mellin transform, 
$\Tr(P\,\Ds^{-2s})=\tfrac{1}{\Gamma (s)} \int_0^\infty t^{s-1}\, \Tr(P\, e^{-t\,\Ds^2})\,dt$, the non-zero coefficient $a'_k$, 
$k\neq0$ creates a pole of $\Tr(P\,\Ds^{-2s}) $ of order $k+2$ because we get 
$\int_0^1 t^{s-1} \log(t)^k\, dt=\tfrac{(-1)^k k!}{s^{k+1}}$ and 
\begin{align}
\label{Gamma}
\Gamma(s)=\frac{1}{s} +\gamma+s\,g(s)
\end{align} 
where $\gamma$ is the Euler constant and the function $g$ is also holomorphic around zero.

We have $\ncint 1=0$ and more generally, $\text{\it WRes}(P)=0$ for all zero-order pseudodifferential projections 
\cite{Wodzicki1}.

As the following remark shows, being a commutative geometry is more than just having a commutative algebra:

\begin{remark}
\label{commutativenon}
Since $J\pi(a)J^{-1}=\pi(a^*)$ for all $a \in \A$ and $\tilde A=0$ for all $A=A^{*}\in \Omega_\DD^1$ 
when $\A$ is commutative by \eqref{JAJ}, one can only use $\DD_A=\DD+A$ to get fluctuation of $\DD$: 
It is amazing to see that in the context of noncommutative geometry, to get an abelian gauge field, we need to go 
outside of abelian algebras. In particular, as pointed out in \cite{Var}, a commutative manifold could support 
relativity but not electromagnetism. 

However,  we can have $\A$ commutative and 
$J\pi(a)J^{-1}\neq \pi(a^*)$ for some $a\in \A$ \cite{CGravity,Kraj}:\\
Let $\A_1=\C \oplus \C$ represented on $\H_1=\C^3$ with, for some complex number $m\neq0$, 
\begin{align*} 
\pi_1(a)&:=  \left( \begin{array}{ccc}
b_1 & 0 & 0\\
0 & b_1 & 0\\
0 & 0 & b_2
\end{array} \right), \,\,for \,\, a=(b_1,\,b_2) \in \A,
\end{align*}
and 
\begin{align*} 
\DD_1&:= \left( \begin{array}{ccc}
0 & m & m\\
\bar m & 0 & 0\\
\bar m & 0 & 0
\end{array} \right),  \,\,
\chi_1:=  \left( \begin{array}{ccc}
1 & 0 & 0\\
0 & -1 & 0\\
0 & 0 & -1
\end{array} \right),  
\,\,
J_1:=\left( \begin{array}{ccc}
1 & 0 & 0\\
0 & 0 & 1\\
0 & 1 & 0
\end{array} \right) \circ \, cc
\end{align*}
where $cc$ is the complex conjugation. Then ($\A_1,\,\H_1,\,\DD_1$) is a commutative real spectral triple of 
dimension $d=0$ with non zero one-forms and such that $J_1\pi_1(a)J_1^{-1}= \pi_1(a^*)$ only if $a=(b_1,b_1)$. 

Take now a commutative geometry \big($\A_2=C^{\infty}(M), \,\H=L^2(M,S),\,\DD_2,\, \chi_2,\, J_2$\big) defined in 
\ref{rieman} where $d=dim M$ is even, and then take the tensor product of the two spectral triples, see Section \ref{tens}, namely 
$\A=\A_1\otimes \A_2$, $\H=\H_1 \otimes \H_2$, $\pi=\pi_1\otimes \pi_2$, 
$\DD=\DD_1 \otimes \chi_2 + 1\otimes \DD_2$, $\chi=\chi_1 \otimes \chi_2$ and $J$ is either 
$\chi_1 J_1\otimes J_2$ when $d\in \set{2,6}$ mod 8 or $J_1 \otimes J_2$ in the other cases, see \cite{CGravity,Vanhecke}.

Then $(\A,\, \H,\, \DD)$ is a real commutative spectral triple of dimension $d$ such that $\tilde A \neq 0$ for 
some selfadjoint one-forms $A$, so is not exactly like in Definition \ref{rieman}.
\end{remark}

\begin{prop}
\label{spectrcomm}
Let $Sp(M)$ be the dimension spectrum of a commutative geometry of dimension $d$. Then $Sp(M)$ is simple and 
$Sp(M) = \set{d-k \, \vert  \, k \in \N}$.
\end{prop}

\begin{proof}
Let $a\in \A=C^{\infty}(M)$ such that its trace norm $\vert\vert a \vert\vert_{L^1}$ is non zero and for $k \in \N$, let 
$P_k:=a \vert D \vert^{-k}$. Then $P_k \in OP^{-k} \subset OP^{0}$ and its associated zeta-function has a pole at 
$d-k$:
\begin{align*}
\underset{s=d-k}{\Res} \,\zeta^P_\DD(s)&=\underset{s=0}{\Res}
\,\zeta_\DD^P(s+d-k)=\underset{s=0}{\Res} \,\Tr \big(a \vert
\DD\vert^{-k}\vert \DD\vert^{-(s+d-k)}\big)=\ncint a \vert \DD \vert ^{-d}\\
&=\int_M a(x) \int_{S_x^*M} \Tr\big(({\sigma_1^{\vert \DD\vert}})^{-d}(x,\xi) \big) \vert d\xi \vert \, \vert dx \vert
=c\int_M a(x) \int_{S_x^*M} \vert\vert \xi \vert \vert^{-d}  \vert d\xi \vert \, \vert dx \vert \\
&=c\int_M a(x) \, dvol_g(x) =c\,\vert\vert a \vert\vert_{L^1}\neq 0.
\end{align*}
Conversely, since $\Psi^0(\A)$ is contained in the algebra of all pseudodifferential operators of order less or equal 
to 0, it is known\cite{Guillemin,Wodzicki1,Wodzicki} that $Sp(M) \subset \set{d-k \, :\, k \in \N}$ as seen in Theorem 
\ref{WresD-p}.
\\
All poles are simple since $\DD$ being differential and $M$ being without boundary, $a'_k=0$, for all $k\in \N^*$ in 
\eqref{heat}.
\end{proof}

\begin{remark}
Due to our efforts to mimic the commutative case, we get as in Theorem \ref{Wresresult} that the noncommutative 
integral is a trace on $\Psi^*(\A)$. However, when the dimension spectrum is not simple, the analog of {\it WRes} 
is no longer a trace.

The equation \eqref{invDAd} can be obtained via \eqref{Wresint} and \eqref{int=cWres} since 
$\sigma^{\vert \DD_A\vert^{-d}}_d=\sigma^{\vert \DD\vert^{-d}}_d$.

In dimension $d=4$, the computation in \eqref{coefficientsa} of coefficient $a_4(1,\DD_A^2)$ gives
$$
\zeta_{\DD_A}(0)=c_1\int_M(5R^2-8R{\mu\nu}r^{\mu\nu}-7R_{\mu\nu\rho\sigma}R^{\mu\nu\rho\sigma} \,dvol 
+c_2\int_M \tr(F_{\mu\nu}F^{\mu\nu})\,dvol,
$$
see Corollary \ref{asymptidentic} to see precise correspondence between $a_k(1,\DD_A^2)$ and 
$\zeta_{\DD_A}(0)$. One recognizes the Yang--Mills action which will be generalized in Section \ref{YMA} to 
arbitrary spectral triples.

According to Corollary \ref{Atilde=0}, a commutative geometry has no tadpoles.
\end{remark}

\subsection{Scalar curvature}

What could be the scalar curvature of a spectral triple $(\A,\H,\DD)$? Of course, we need to consider first the case 
of a commutative geometry $(C^\infty(M), L^2(M,S), \Ds)$ of dimension $d=4$: We know that 
$\ncint f(x) \Ds^{-d+2}=\int_M f(x)\,s(x) \,dvol(x)$ where $s$ is the scalar curvature for any $f\in C^\infty(M)$. This 
suggests the following
\begin{definition}
Let $(\A,\H,\DD)$ be a spectral triple of dimension $d$. The scalar curvature is the map $\mathcal{R}:a\in \A \to \C$ 
defined by
$$
\mathcal{R}(a)\vc \ncint a\,\DD^{-d+2}.
$$
\end{definition}
In the commutative case, $\mathcal{R}$ is a trace on the algebra. More generally
\begin{prop}
If $\mathcal{R}$ is a trace on $\A$ and the tadpoles $\ncint A\,\DD^{-d+1}$ are zero for all $A\in \Omega^1_\DD$, 
$\mathcal{R}$ is invariant by inner fluctuations $\DD\to\DD+A$.
\end{prop}
\noindent See \cite[Proposition 1.153]{ConnesMarcolli} for a proof.

\subsection{Tensor product of spectral triples}
\label{tens}

There is a natural notion of tensor for spectral triples which corresponds to direct product of manifolds in the 
commutative case. 
Let $(\A_i,\DD_i,\H_i)$, $i=1,2$, be two spectral triples of dimension $d_i$ with simple dimension spectrum. 
Assume the first to be of even dimension, with grading $\chi_1$.
\\
The spectral triple $(\A,\DD,\H)$ associated to the tensor product is defined by
\begin{align*}
\A:=\A_1 \otimes \A_2, \quad \DD:=\DD_1\otimes 1+\chi_1 \otimes \DD_2, \quad \H:=\H_1 \otimes \H_2.
\end{align*}
The interest of $\chi_1$ is to guarantee additivity: $\DD^2=\DD_1^2\otimes 1 + 1\otimes \DD_2^2$.
\\
We assume that
\begin{align}
\label{dimensionsomme}
\Tr(e^{-t\DD_1^2})\sim_{t\to 0} a_1\,t^{-d_1/2},\quad \Tr(e^{-t\DD_2^2})\sim_{t\to 0} a_2\,t^{-d_2/2}.
\end{align}

\begin{lemma}
\label{tensorproduct}
The triple $(\A,\DD,\H)$ has dimension $d=d_1+d_2$. \\
Moreover, the function $\zeta_\DD(s)=\Tr(\vert D\vert^{-s})$ has a simple pole 
at $s=d_1+d_2$ with
\begin{align*}
\Res_{s=d_1+d_2}\big(\zeta_{\DD}(s)\big)=\tfrac{1}{2}\, \tfrac{\Gamma(d_1/2) \Gamma(d_2/2)}{\Gamma(d/2)} 
\Res_{s=d_1}\big(\zeta_{\DD_1}(s)\big)\,\Res_{s=d_2}\big(\zeta_{\DD_2}(s)\big).
\end{align*}
\end{lemma}

\begin{proof}
If $(\mu_{n}(A))$ are the singular values of $A$, 
\begin{align*}
\zeta_\DD(2s)=\sum_{n=0}^\infty \mu_n(D_1^2\otimes 1+1\otimes \DD_2^2)^{-s}
=\sum_{n,m=0}^\infty \big( \mu_n(\DD_1^2)+\mu_m(\DD_2^2)\big)^{-s}.
\end{align*}
Since 
$\big(\mu_n(\DD_1)^2+\mu_m(\DD_2)^2 \big)^{-(c_1+c_2)/2} \leq \mu_n(\DD_1)^{-c_1}\mu_m(\DD_2)^{-c_2}$, 
this shows in particular that $\zeta_\DD(c_1+c_2) \leq \zeta_{\DD_1}(c_1)\zeta_{\DD_2}(c_2)$ if $c_i>d_i$, 
and in particular that 
$$
d:=\inf\set{c\in\R^+ \,:\, \zeta_\DD(c)<\infty}\leq d_1+d_2.
$$
We claim that $d=d_1+d_2$: recall first that in \eqref{dimensionsomme}
\begin{align}
\label{coefa}
a_i:=\underset{s=d_i/2}{\Res}\big(\Gamma(s)\zeta_{\DD_i}(2s)\big)=\Gamma(d_i/2)
\underset{s=d_i/2}{\Res}\big(\zeta_{\DD_i}(2s)\big) 
=\tfrac{1}{2}\Gamma(d_i/2)\,\underset{s=d_i}{\Res}\big(\zeta_{\DD_i}(s)\big).
\end{align}
If $f(s):=\Gamma(s)\,\zeta_D(2s)$,
\begin{align*}
f(s)&=\Gamma(s)\,\Tr\big(\DD^{-2s}\big)=\Tr(\int_0^\infty e^{-t\DD^2}t^{s-1}dt\big)
=\int_0^1\Tr\big(e^{-t\DD^2}\big)\,t^{s-1}\,dt +g(s) \\
&\,\,=\int_0^1 \Tr\big(e^{-t\DD_1^2}\big)\Tr\big( e^{-t\DD_2^2}\big)\,t^{s-1}\,dt +g(s)
\end{align*}
where $g$ is a holomorphic function since the map $x\in \R\to \int_1^\infty e^{-tx^2}t^{x-1} \,dt$ is in Schwartz space.\\
Since $\Tr\big(e^{-t\DD_1^2}\big)\Tr\big( e^{-t\DD_2^2}\big) \sim_{t\to 0} a_1a_2\,t^{-(d_1+d_2)/2}$, we get that 
the function $f(s)$ has a simple pole at $s=(d_1+d_2)/2$. We conclude that $\zeta_\DD(s)$ has a simple pole at 
$s=d_1+d_2$. 
\\
Moreover, thanks to \eqref{coefa},
\begin{align*}
\tfrac{1}{2}\Gamma((d_1+d_2)/2) \Res_{s=d}\big(\zeta_{\DD}(s)\big)=\tfrac{1}{2}\Gamma(d_1/2) 
\Res_{s=d_1}\big(\zeta_{\DD_1}(s)\big)\, \tfrac{1}{2}\Gamma(d_2/2) \Res_{s=d_2}\big(\zeta_{\DD_2}(s)\big).
\end{align*}
\end{proof}

Remark that we can apply the last lemma to our Remark \ref{commutativenon} in the commutative case, but we could also use $\A_2=M_n(\C)$ acting on $\C^n$ and $\DD_2$ any selfadjoint $n\times n$-matrix. This is the way the internal degrees of freedom for fermions are implemented in particle physics, see for instance references in \cite{ConnesMarcolli}.

\newpage
\section{Spectral action}
\label{Spectral action}

\subsection{On the search for a good action functional}

We would like to obtain a good action for any spectral triple and for this it is useful to look at some examples in 
physics. 

In any physical theory based on geometry, the interest of an action functional is, by a minimization 
process, to exhibit a particular geometry, for instance, trying to distinguish between different metrics. 
This is the case in general relativity with the Einstein--Hilbert action (with its Riemannian signature).

\subsubsection{Einstein--Hilbert action}
\label{Einstein--Hilbert action}
This action is 
\begin{align}
\label{EH action}
S_{EH}(g)\vc-\int_M s_g(x) \, dvol_g(x)
\end{align}
where $s$ is the scalar curvature (chosen positive for the 
sphere). This is nothing else (up to a constant, in dimension 4) than $\ncint \Ds^{-2}$ as quoted after \eqref{a2(D2)}.
\\
This action is interesting for the following reason: Let $\mathcal{M}_1$ be the set of Riemannian metrics $g$ on 
$M$ such that $\int_M dvol_g=1$. By a theorem of Hilbert \cite{Besse}, $g\in \mathcal{M}_1$ is a critical point of 
$S_{EH}(g)$ restricted to $\mathcal{M}_1$ if and only if $(M,g)$ is an Einstein manifold (the Ricci curvature $R$ 
of $g$ is proportional by a constant to $g$: $R=c\,g$). Taking the trace, this means that $s_g=c\,\dim(M)$ and such 
manifold have a constant scalar curvature.

But in the search for invariants under diffeomorphisms, they are more quantities than the 
Einstein--Hilbert action, a trivial example being $\int_M f\big(s_g(x)\big) \,dvol_g(x)$ and they are others 
\cite{GMH}. In this desire to implement gravity in noncommutative geometry, the eigenvalues of the Dirac 
operator look as natural variables \cite{LR}.
However we are looking for observables which add up under disjoint unions of different geometries.

\subsubsection{Quantum approach and spectral action}
In a way, a spectral triple fits quantum field theory since $\DD^{-1}$ can be seen as the propagator (or line element 
$ds$) for (Euclidean) fermions and we can compute Feynman graphs with fermionic internal lines. As glimpsed in 
section \ref{oneform}, the gauge bosons are only derived objects obtained from internal fluctuations via Morita 
equivalence given by a choice of a connection which is associated to a one-form in $\Omega_\DD^1(\A)$. 
Thus, the guiding principle followed by Connes and Chamseddine is to use a theory which is pure gravity 
with a functional action based on the spectral triple, namely which depends on the spectrum of $\DD$ \cite{CC}. 
They proposed the following

\begin{definition}
The spectral action of a spectral triple $(\A,\H,\DD)$ is defined by
\begin{align*}
\SS(\DD,f,\Lambda)&:=\Tr \big(f(\DD^2/\Lambda^2) \big)
\end{align*}
where $\Lambda\in \R^+$ plays the role of a cut-off and $f$ is any positive function (such that $f(\DD^2/\Lambda^2)$ 
is a trace-class operator). 
\end{definition}
\begin{remark}
We can also define $\SS(\DD,f,\Lambda)=\Tr \big(f(\DD/\Lambda) \big)$ when $f$ is positive and even. 
With this second definition, $S(\DD,g,\Lambda) = \Tr \big(f(\DD^2/\Lambda^2)\big)$ with $g(x)\vc f(x^2)$.
\end{remark}

For $f$, one can think of the characteristic function of $[-1,1]$, thus $f(\DD/\Lambda)$ is 
nothing else but the number of eigenvalues of $\DD$ within $[-\Lambda,\Lambda]$. 

When this action has an asymptotic series in $\Lambda \to \infty$, we deal with an effective theory.
Naturally, $\DD$ has to be replaced by $\DD_A$ which is a just a decoration. To this bosonic part of the action, 
one adds a fermionic term $\tfrac{1}{2}\langle J \psi,\DD \psi\rangle$ for $\psi\in \H$ to get a full action. 
In the standard model of particle physics, this latter corresponds  to the 
integration of the Lagrangian part for the coupling between gauge bosons and Higgs bosons with fermions. 
Actually, the finite dimension part of the noncommutative standard model is of $KO$-dimension 6, thus 
$\langle \psi,\DD \psi\rangle$ has to be replaced by $\tfrac 12 \langle J \psi,\DD \psi\rangle$ for 
$\psi =\chi \psi \in \H$, see \cite{ConnesMarcolli}.

\subsubsection{Yang--Mills action}
\label{YMA}

This action plays an important role in physics so it is natural to consider it in the 
noncommutative framework. Recall first the classical situation: let G be a compact Lie group with its Lie algebra 
$\mathfrak{g}$ and let $A\in \Omega^1(M,\mathfrak{g})$ be a connection. If 
$F \vc da+ \tfrac{1}{2}[A,A] \in \Omega^2(M,\mathfrak{g})$ is the curvature (or field strength) of $A$, then the 
Yang-Mills action is $S_{YM}(A)=\int_M \tr(F \wedge \star F)\,dvol_g$.  In the abelian 
case $G=U(1)$, it is the Maxwell action and its quantum version is the quantum electrodynamics (QED) since the 
un-gauged $U(1)$ of electric charge conservation can be gauged and its gauging produces electromagnetism 
\cite{Schucker}. It is conformally invariant when the dimension of $M$ is $d=4$. 

The study of its minima and its critical values can also been made for a spectral triple $(\A,\H,\DD)$ of dimension 
$d$ \cite{Connes88,Book}: let $A \in \Omega_\DD^1(\A)$ and curvature $\theta=dA+A^2$; then it is natural 
to consider 
$$
I(A)\vc \Tr_{Dix}(\theta^2 \DDD^{-d})
$$
since it coincides (up to a constant) with the previous Yang-Mills 
action in the commutative case: if $P=\theta^2 \vert \DD \vert^{-d}$, then Theorems \ref{Wresresult} and \ref{TrDix} 
give the claim since for the principal symbol, $\tr\big(\sigma^P(x,\xi)\big)=c\,\tr(F \wedge \star F)(x)$.

There is nevertheless a problem with the definition of $dA$: if $A=\sum_j \pi(a_j)[\DD,\pi(b_j)]$,
then $dA=\sum_j[\DD,\pi(a_j)][\DD,\pi(b_j)]$ can be non-zero while $A=0$. This ambiguity means that, to get a 
graded differential algebra $\Omega_\DD^*(\A)$, one must divide by a junk, for instance 
$\Omega_\DD^2\simeq \pi(\Omega^2/\pi \big( \delta (\text{Ker}(\pi) \cap\Omega^1) \big)$ where  
$\Omega^{k}(\A)$ is the set of universal $k$-forms over $\A$ given by the set of $a_0\delta a_1\cdots \delta a_k$ 
(before representation on $\H$: $\pi(a_0\delta a_1\cdots \delta a_k)\vc a_0[\DD,a_1]\cdots[\DD,a_k]$). 
\\
Let $\H_k$ be the Hilbert space completion of $\pi(\Omega^k(\A))$ with the scalar product defined by
$\langle A_1,A_2 \rangle_k \vc \Tr_{Dix}(A_2^*A_1 \DDD^{-d})$ for $A_j \in \pi(\Omega^k(\A))$.
\\The Yang--Mills action on $\Omega^1(\A)$ is 
\begin{align}
\label{YMaction}
S_{YM}(V) \vc \langle \delta V+V^2, \delta V+V^2 \rangle.
\end{align}
It is positive, quartic and gauge invariant under $V \to \pi(u)V\pi(u^*)+\pi(u)[\DD,\pi(u^*)]$ when $u\in \U(\A)$. 
Moreover, 
$$
S_{YM}(V)=\inf \set{I(\omega) \, \vert \, \omega \in \Omega^1(\A), \,\pi(\omega)=V}
$$
since the above ambiguity disappears when taking the infimum.

This Yang--Mills action can be extended to the equivalent of Hermitean vector bundles on $M$, namely finitely 
projective modules over $\A$. 

The spectral action is more conceptual than the Yang--Mills action since it gives no fundamental role to the 
distinction between gravity and matter in the artificial decomposition $\DD_A=\DD+A$. For instance, for the minimally 
coupled standard model, the Yang--Mills action for the vector potential is part of the spectral action, as far as the 
Einstein--Hilbert action for the Riemannian metric \cite{CC2}. 

As quoted in \cite{CC4}, the spectral action has conceptual advantages:

- Simplicity: when $f$ is a cutoff function, the spectral action is just the counting function.

- Positivity: when $f$ is positive (which is the case for a cutoff function), the action 
$\Tr \big( f(\DD/\Lambda)\big) \geq 0$ has the correct sign for a Euclidean action: the positivity of the function 
$f$ will insure that the actions for gravity, Yang-Mills, Higgs couplings are all positive and the Higgs mass term 
is negative.

- Invariance:  the spectral action has a much stronger invariance group than the usual diffeomorphism group as for 
the gravitational action; this is the unitary group of the Hilbert space $\H$.

However, this action is not local. It only becomes so when it is replaced by the asymptotic expansion:

\subsection{\texorpdfstring{Asymptotic expansion for $\Lambda \to \infty$}{Asymptotic expansion for Lambda 
going to infinity}}
\label{asymptoticspectralaction}

The heat kernel method already used in previous sections will give a control of spectral action 
$S(\DD,f,\Lambda)$ when $\Lambda$ goes to infinity.

\begin{theorem}
Let $(\A,\H,\DD)$ be a spectral triple with a simple dimension spectrum $Sd$. 
\\
We assume that 
\begin{align}
\label{hyp}
\Tr\big(e^{-t\DD^2}\big)  \, \underset{t \,\downarrow \,0}{\sim}  \,\sum_{\a\in Sd} a_\a\,t^\a \quad \text{with }a_\a\neq 0.
\end{align}
Then, for the zeta function $\zeta_\DD$ defined in \eqref{zetaPD}
\begin{align}
\label{calculdea}
a_\a=\tfrac{1}{2}\Res_{s=-2\a}\big(\Gamma(s/2) \zeta_\DD(s)\big).
\end{align}

(i) If $\a<0$,  $\zeta_\DD$ has a pole at $-2\a$ with 
$a_\a=\tfrac{1}{2}\Gamma(-\a) \underset{s=-2\a}{\Res} \, \zeta_\DD(s)$.

(ii) For $\a=0$, we get $a_0=\zeta_\DD(0)+\dim \,\Ker \,\DD$.

(iii) If $\a>0$, $a_\a=\zeta(-2\a)\underset{s=-\a}{\Res} \, \Gamma(s)$.

(iv) The spectral action has the asymptotic expansion over the positive part $Sd^+$of $Sd$:
\begin{align}
\label{asympspectral}
\Tr\big(f(\DD/\Lambda)\big)  \, \underset{\Lambda \,\to+\infty}{\sim}  \, 
\sum_{\beta \in Sd^+} f_\beta\,\Lambda^\beta \ncint \vert \DD \vert^{-\beta}+ f(0)\,\zeta_\DD(0)+ \cdots
\end{align}
where the dependence of the even function $f$ is $f_\beta:=\int_0^\infty f(x)\,x^{\beta-1}\,dx$ and $\cdots$ involves 
the full Taylor expansion of $f\,$at 0. 
\end{theorem}

\begin{proof}
$(i)$ Since 
$\Gamma(s/2)\,\vert \DD \vert^{-s}=\int_0^\infty e^{-t\DD^2}\,t^{s/2-1}\,dt=\int_0^1 e^{-t\DD^2}\,t^{s/2-1}\,dt+f(s)$, 
where the function $f$ is holomorphic (since the map $s\to \int_1^\infty e^{-tx^2}t^{s/2-1}\,dt$ is in the Schwartz 
space), the swap of $\Tr\big(e^{-t\DD^2}\big)$ with a sum of $a_\a\,t^\a$ and 
$a_\a\int_0^1t^{\a+s/2-1}\,dt=\tfrac{2\a_a}{s+2\a}$ yields \eqref{calculdea}.

$(ii)$ The regularity of $\Gamma(s/2)^{-1}  \, \underset{s \to 0}{\sim}  \, s/2$ around zero implies that only 
the pole part at $s=0$ of $\int_0^\infty \Tr\big(e^{-t\DD^2}\big)\,t^{s/2-1}\,dt$ contributes to $\zeta_\DD(0)$. 
This contribution is $a_0\int_0^1 t^{s/2-1}\,dt=\tfrac{2a_0}{s}$.

$(iii)$ follows from  \eqref{calculdea}.

$(iv)$ Assume $f(x)=g(x^2)$ where $g$ is a Laplace transform: $g(x):=\int_0^\infty e^{-sx}\,\phi(s)\,ds$. 
We will see in Section \ref{useofLaplace} how to relax this hypothesis. \\
Since 
$g(t\DD^2)=\int_0^\infty e^{-st\DD^2} \,\phi(s)\,ds$, 
$\Tr\big(g(t\DD^2)\big) \, \underset{t \,\downarrow \,0}{\sim} \, \sum_{\a\in \Sp} \, 
a_\a\,t^\a\int_0^\infty s^\a\, \phi(s)\,ds$. When $\a<0$, $s^\a=\Gamma(-\a)^{-1}\int_0^\infty e^{-sy}\,y^{-\a-1}\,dy$ and 
$\int_0^\infty s^\a\,\phi(s)\,ds=\Gamma(-\a)^{-1}\int_0^\infty g(y)\,y^{-\a-1}\,dy$. Thus
$$
\Tr\big(g(t\DD^2)\big)  \, \underset{t \,\downarrow \,0}{\sim} \, \sum_{\a\in \Sp^-} \, 
\big[\tfrac{1}{2}\underset{s=-2\a}{\Res} \, \zeta_\DD(s) \int_0^\infty g(y)\,y^{-\a-1} \,dy\big]\,t^\a.
$$
Finally \eqref{asympspectral} follows from $(i),\,(ii)$ and $\tfrac{1}{2}\int_0^\infty g(y)\,y^{\beta/2-1}\,dy=\int_0^\infty f(x)\,x^{\beta-1}\,dx$.
\end{proof}
It can be useful to make a connection with \eqref{traceasympt} of Section \ref{Wodzicki residue and heat expansion}:

\begin{corollary}
\label{asymptidentic}
Assume that the spectral triple $(\A,\H,\DD)$ has dimension $d$. If 
\begin{align}
\label{Traceasympto}
{\rm Tr} \big( e^{-t \, \DD^2} \big) \, \underset{t \,\downarrow \,0}{\sim}  \,
\sum_{k\in \set{0,\cdots,d}} t^{(k-d)/2} \,{a}_k (\DD^{2})+\cdots, 
\end{align}
then
\begin{align*}
\SS(\DD,f,\Lambda)\underset{\Lambda \to \infty}{\sim}  \, 
 \sum_{k \in \set{1,\cdots ,d}} \,f_k \, \Lambda^k\,a_{d-k}(\DD^2) +f(0)\,a_d(\DD^2)+\cdots
\end{align*}
with $f_k:=\tfrac{1}{\Ga(k/2)}\int_0^\infty f(s)s^{k/2-1}ds$.\\
Moreover,
\begin{align}
& a_k(\DD^2)=\tfrac 12\,\Gamma(\tfrac{d-k}{2})\,\ncint \vert\DD \vert^{-d+k} \text{ for }k=0,\cdots,d-1,\label{nlccoeff}\\
& a_d(\DD^2)= \dim\, \Ker \DD+\zeta_{\DD^2}(0). \nonumber
\end{align}
\end{corollary}

\begin{proof}
We rewrite the hypothesis on $\Tr \big( e^{-t \, \DD^2} \big)$ as
\begin{align*}
{\rm Tr} \big( e^{-t \, \DD^2} \big) \, \underset{t \,\downarrow \,0}{\sim}  \,
\sum_{\a\in\set{-d/2,\cdots ,-1/2}} A_\a \, t^\a+A_0=\sum_{k\in \set{1,\cdots,d}} t^{(k-d)/2} \,{a}_k (\DD^{2}) +a_d(\DD^2)
\end{align*}
with $a_k(\DD^2):=A_{(k-d)/2}$.

For $\a<0$, we repeat the above proof:
\begin{align*}
\SS(\DD,f,\Lambda)&  \underset{\Lambda\to\infty}{\sim}  \, 
\sum_{\a \in \set{-d/2,\cdots ,-1/2}} A_\a \, \Lambda^{-2\a} \tfrac{1}{\Ga(-\a)} \int_0^\infty f(s)s^{-\a-1}ds +A_0\\
&=  \sum_{l \in \set{1,\cdots ,d}} A_{(l-d)/2} \, \Lambda^{d-l} \tfrac{1}{\Ga((d-l)/2)} \int_0^\infty f(s)s^{(d-l)/2-1}ds 
+ A_0f(0)\\
&=  \sum_{l \in \set{1,\cdots ,d}} a_l\, \Lambda^{d-l} \tfrac{1}{\Ga((d-l)/2)} \int_0^\infty f(s)s^{(d-l)/2-1}ds
+a_df(0)\\
&=  \sum_{k \in \set{1,\cdots ,d}} a_{d-k}\, \Lambda^{k} \tfrac{1}{\Ga(k/2)} \int_0^\infty f(s)s^{k/2-1}ds 
+a_df(0).
\end{align*}
Again, for $\a<0$,
\begin{align*}
A_\a&=\tfrac 12 \Gamma(-\a)\,\underset{s=-2\a}{\Res} \, \Tr\big( \vert \DD \vert^{-s}\big)
=\tfrac 12 \Gamma(-\a)\,\underset{s=0}{\Res} \, \Tr\big( \vert \DD \vert^{-(s-2\a)} \big)\\
&=\tfrac 12 \Gamma(-\a)\,\underset{s=0}{\Res} \, \Tr\big(\vert \DD \vert^{2\a}\vert \DD \vert^{-s}\big)
=\tfrac 12 \Gamma(-\a) \ncint \vert \DD \vert^{2\a}.
\end{align*}
Thus, for $\a=\tfrac{k-d}{2}<0$, $k=0,\cdots,d-1$,
\begin{align*}
a_k(\DD^2)=A_{(k-d)/2}=\tfrac 12  \Gamma(\tfrac{d-k}{2})\ncint\vert \DD \vert^{-d+k}.
\end{align*}
\end{proof}

The asymptotics \eqref{asympspectral} uses the value of $\zeta_\DD(0)$ in the constant term $\Lambda^0$, so it is 
fundamental to look at its variation under a gauge fluctuation $\DD \to \DD+A$ as we saw in Theorem \ref{difference}.

\subsection{Remark on the use of Laplace transform}
\label{useofLaplace}

The spectral action asymptotic behavior 
\begin{align}
\label{asympto}
S(\DD,f, \Lambda) \, \underset{\Lambda \,\to+\infty}{\sim}  \,  \sum_{n=0}^\infty c_n \,\Lambda^{d-n}\,a_n(\DD^2)
\end{align}
has been proved for a smooth function $f$ which is a Laplace transform for an arbitrary 
spectral triple (with simple dimension spectrum) satisfying \eqref{hyp}.
However, this hypothesis is too restrictive since it does not cover the heat kernel case where $f(x)=e^{-x}$.
\\
When the triple is commutative and $\DD^2$ is a generalized Laplacian on sections of a vector bundle over a 
manifold of dimension 4, Estrada--Gracia-Bond\'{\i}a--V\'arilly proved in \cite{EGBV} that previous asymptotics is
\begin{align*}
\Tr \big( f(\DD^2/\Lambda^2) \big) &\sim \tfrac{1}{(4\pi)^2} \Big[ \text{rk(E)}\int_0^\infty  xf(x)\, dx 
\, \Lambda^4 + b_2(\DD^2) \int_0^\infty f(x)\,dx\, \Lambda^2  \\
&\hspace{2cm} +\sum_{m=0}^\infty (-1)^m\, f^{(m)}(0)\,b_{2m+4} (\DD^2) \, \Lambda^{-2m} \Big],\quad 
\Lambda \to \infty
\end{align*}
where $(-1)^m b_{2m+4}(\DD^2) = \tfrac{(4\pi)^2}{m!}  \,\mu_m(\DD^2)$ are suitably normalized, integrated 
moment terms of the spectral density of $\DD^2$.

The main point is that {\it this asymptotics makes sense in the Ces\`aro sense }(see \cite{EGBV} for definition) for 
$f$ in $\mathcal{K'}(\R)$, which is the dual of $\mathcal{K}(\R)$. This latter is the space of smooth functions 
$\phi$ such that for some $a\in\R$, $\phi^{(k)}(x)=\mathcal{O}( \vert x\vert^{a-k} )$ as $\vert x \vert \to \infty$, 
for each $k\in \N$. In particular, the Schwartz functions are in $\mathcal{K}(\R)$ (and even dense). 

Of course, the counting function is not smooth but is in $\mathcal{K'}(\R)$, so such behavior \eqref{asympto} 
is wrong beyond the first term, but is correct in the Ces\`aro sense. Actually there are more derivatives of $f$ at $0$ 
as explained on examples in \cite[p. 243]{EGBV}. See also Section \ref{The commutative case}.

\subsection{About convergence and divergence, local and global aspects of the asymptotic expansion}
\label{aboutconvergence}

The asymptotic expansion series \eqref{Traceasympto} of the spectral action may or may not converge. It is known 
that each function $g(\Lambda^{-1})$ defines at most a unique expansion series when $\Lambda \to \infty$ 
but the converse is not true since several functions have the same asymptotic series. 
We give here examples of convergent and divergent series of this kind.

When $M$ is the torus $\T^d$ as in Example \ref{Dixtrace} with $\Delta=\delta^{\mu\nu} \partial_\mu\partial_\nu$, 
\begin{align*}
\Tr(e^{t\Delta})=\tfrac{(4\pi)^{-d/2} \,\Vol(\T^d)}{t^{d/2}} +\mathcal{O}(t^{-d/2}\,e^{-1/4t}),
\end{align*}
thus the asymptotic series $\Tr(e^{t\Delta}) \simeq \frac{(4\pi)^{-d/2} \,\Vol(T^d)}{t^{d/2}}$, $t\to 0$, has only one term. 

In the opposite direction, let now $M$ be the unit four-sphere $\mathbb{S}^4$ and $\Dslash$ be the usual Dirac 
operator. By Propostion \ref{spectrcomm}, equation \eqref{hyp} yields (see \cite{CC5}):
\begin{align*}
& \Tr(e^{-t\Dslash^2})=  \tfrac{1}{t^2} \big(  \tfrac{2}{3} + \tfrac{2}{3} \,t 
+ \sum_{k=0}^n a_k \,t^{k+2} +\mathcal{O}(t^{n+3}) \big) , \\
& a_k:=\tfrac{(-1)^k \,4}{3\, k!} \big( \tfrac{B_{2k+2}}{2k+2}-\tfrac{B_{2k+4}}{2k+4} \big)
\end{align*}
with Bernoulli numbers $B_{2k}$. 
Thus $ t^2\Tr(e^{-t\Dslash^2})\simeq \frac{2}{3} + \tfrac{2}{3} \,t + \sum_{k=0}^\infty a_k \,t^{k+2}$ when 
$t \to 0$ and this series is a not convergent but only asymptotic: \\
$a_k > \tfrac{4}{3\,k!}\tfrac{\vert B_{2k+4}\vert }{2k+4}  > 0$ and $\vert B_{2k+4} \vert 
=2\, \tfrac{(2k+4)!}{(2\pi)^{2k+4}} \, \zeta(2k+4) \simeq 4 \sqrt{\pi (k+2)} 
\left(\tfrac{k+2}{ \pi e} \right)^{2k+4}  \to \infty \text{ if }k \to \infty$.

More generally, in the commutative case considered above and when $\DD$ is a differential operator---like a 
Dirac operator, the coefficients of the asymptotic series of $\Tr(e^{-t\DD^2})$ are locally defined by the symbol of 
$\DD^2$ at point $x\in M$ but this is not true in general: in \cite{GilkeyGrubb} is given a positive elliptic 
pseudodifferential such that non-locally computable coefficients especially appear in \eqref{Traceasympto} 
when $2k> d$. Nevertheless, all coefficients are local for $2k \leq d$.

Recall that a locally computable quantity is the integral on the manifold of a local frame-independent smooth 
function of one variable, depending only on a finite number of derivatives of a finite number of terms in the 
asymptotic expansion of the total symbol of $\DD^2$. 
For instance, some nonlocal information contained in the ultraviolet asymptotics can be recovered if one looks at 
the (integral) kernel of $e^{-t\sqrt{-\Delta}}$: in $\T^1$, with $\Vol(\T^1)=2\pi$, we get \cite{Fulling}
$$
\Tr(e^{-t\sqrt{-\Delta}})=2\sum_{n=1}^\infty e^{-tn} +1=\tfrac{\sinh (t)}{\cosh(t)-1}=\coth(\tfrac{t}{2})=
\tfrac{2}{t}\,\sum_{k=0}^\infty \tfrac{B_{2k}}{(2k)!} \,t^{2k}=\tfrac{2}{t}[1+\tfrac{t^2}{12}-\tfrac{t^4}{720} +\mathcal{O}(t^6)]
$$
and the series converges when $t<2\pi$, since $\frac{B_{2k}}{(2k)!}=(-1)^{k+1} \,\tfrac{2\,\zeta(2k)}{(2\pi)^{2k}}$, thus 
$\tfrac{\vert B_{2k} \vert}{(2k)!} \simeq \tfrac{2}{(2\pi)^{2k}}$ when $k\to \infty$.

Thus we have an example where $t \to \infty$ cannot be used with the asymptotic series.

Thus the spectral action of Corollary \ref{asymptidentic} precisely encodes these local and 
nonlocal behavior which appear or not
in its asymptotics for different $f$. The coefficient of the action for the positive part (at least) of the dimension 
spectrum correspond to renormalized traces, namely the noncommutative integrals of \eqref{nlccoeff}. 
In conclusion, the asymptotics \eqref{asympto} of spectral action may or may not have nonlocal coefficients.

For the flat torus $\T^d$, the difference between $\Tr(e^{t\Delta})$ and its asymptotic series is an term 
which is related to periodic orbits of the geodesic flow on $\T^d$. Similarly, the counting function $N(\lambda)$ 
(number of eigenvalues including multiplicities of $\Delta$ less than $\lambda$) obeys Weyl's law: 
$N(\lambda)= \frac{(4\pi)^{-d/2} \,\Vol(\T^d)}{\Gamma(d/2 +1)}\, \lambda^{d/2} + o(\lambda^{d/2})$ --- see \cite{Arendt} 
for a nice historical review on these fundamental points.
The relationship between the asymptotic expansion of the heat kernel and the formal expansion of the spectral 
measure is clear: the small-$t$ asymptotics of heat kernel is determined by the large-$\lambda$ asymptotics of the 
density of eigenvalues (and eigenvectors). However, the latter is defined modulo some average: 
Ces\`aro sense as reminded in Section \ref{useofLaplace}, or Riesz mean of the measure which washes out 
ultraviolet oscillations, but also gives informations on intermediate values of $\lambda$ \cite{Fulling}.

In \cite{CC4,MPT1} are given examples of spectral actions on (compact) commutative geometries of dimension 4 
whose asymptotics have only two terms. In the quantum group $SU_q(2)$, the spectral action 
itself has only 4 terms, independently of the choice of function $f$.

See \cite{ILV1} for more examples.

\subsection{About the physical meaning of the spectral action via its asymptotics}

As explained before, the spectral action is non-local. Its localization does not cover all situations: consider for instance 
the commutative geometry of a spin manifold $M$ (see Section \ref{Commutative geometry}) of dimension 4. One 
adds a gauge connection $A\in \Gamma^\infty \big(M,End(S)\big)$ to the Dirac operator $\Ds$ such that 
$\DD=i\ga^\mu(\partial_\mu+A_\mu)$, thus with a field strength 
$F_{\mu \nu}=\partial_\mu A _\nu-\partial_\nu A_\mu+ [A_\mu,A_\nu]$. We can apply \eqref{Lapl} with $P=\DD^2$ 
and find the coefficients $a_i(1,P)$ of \eqref{coefficientsa} with $i=0,2,4$. The expansion \eqref{asympto} 
corresponds to a weak field expansion. 

Moreover a commutative geometry times a finite one where the finite one is algebra is a sum of matrices 
(like in Remark \ref{commutativenon}) has been deeply and intensively investigated for the noncommutative 
approach to  standard model of particle physics, see \cite{CCM,ConnesMarcolli}. This approach offers a 
lot of interesting perspectives, for instance, the possibility to compute the Higgs representations and mass 
(for each noncommutative model) is 
particularly instructive \cite{CC,CC6,CC9,IKS,LMM,JS,JS1,MGV}. Of course, since the first term in
 \eqref{asympspectral} is a cosmological term, one may be worried by its 
large value (for instance in the noncommutative standard model where the cutoff is, roughly speaking 
the Planck scale). At the classical level, one can work with unimodular gravity where the metric (so the 
Dirac operator) $\DD$ varies within the set $\M_1$ of metrics which preserve the volume as in 
Section \ref{Einstein--Hilbert action}. Thus it remains only (!) to control the inflaton: see \cite{CC3}.

The spectral action has been computed in \cite{ILS} for the quantum group $SU_q(2)$ which is not a 
deformation of $SU(2)$ of the type considered in Section \ref{The Moyal plane} on the Moyal plane. It is quite 
peculiar since \eqref{asympspectral} has only a finite number of terms.

Due to the difficulties to deal with non-compact manifolds (see nevertheless Section \ref{The non-compact case}),
the case of spheres $\mathbb{S}^4$ or $\mathbb{S}^3\times \mathbb{S}^1$ has been investigated in 
\cite{CC4,CC5} for instance in the case of  Robertson--Walker metrics.

All the machinery of spectral geometry as been recently applied to cosmology, computing the spectral action in 
few cosmological models related to inflation, see \cite{KM, MPT1,MPT2, MP,NS,Sa}. 

Spectral triples associated to manifolds with boundary have been considered in \cite{CC7,CC8,IL1,IL2,CC8,ILV}. The 
main difficulty is precisely to put nice boundary conditions to the operator $\DD$ to still get a selfadjoint 
operator and then, to define a compatible algebra $\A$. This is probably a must to obtain a result in a  
noncommutative Hamiltonian theory in dimension 1+3.

The case of manifolds with torsion has also been studied in \cite{HPS,PS,PS1}, and even with boundary 
in \cite{ILV}. These works show that the Holst action appears in spectral actions and that torsion could 
be detected in a noncommutative world.


\newpage
\section{Residues of series and integral, holomorphic continuation, etc}
\label{Residues of series}

The aim of this section is to control the holomorphy of series of holomorphic functions. 
\\The necessity of a Diophantine condition appears quite naturally. This section has its own interest, but will 
be fully applied in the next one devoted to the noncommutative torus. The main idea is to get a condition which 
guarantee the commutation of a residue and a series. 

This section is quite technical, but with only non-difficult notions. Nevertheless, the devil is hidden into the 
details and I recommend to the reader to have a look at the proofs despite their lengths. 
\\
Reference: \cite{EILS}.
\\
Notations:
\\
In the following, the prime in ${\sum}'$ means that we omit terms with division by zero in the summand.
$B^{n}$ (resp. $S^{n-1}$) is the closed ball (resp. the sphere) of $\R^n$ with center $0$ and radius 1 and the
Lebesgue measure on $S^{n-1}$ will be noted $dS$.

For any $x=(x_1,\dots,x_n) \in \R^n$ we denote by $|x|=\sqrt{x_1^2+\dots+x_n^2}$ the Euclidean norm and 
$|x|_1 :=|x_1|+\dots +|x_n|$.

By  $f(x,y) \ll_{y} g(x)$ uniformly in $x$, we mean that 
$\vert f(x,y)\vert \leq a(y) \, \vert g(x) \vert$ for all $x$ and $y$ for some $a(y)>0$.

\subsection{Residues of series and integral}
In order to be able to compute later the residues of certain series, we prove here the following

\begin{theorem}
\label{res-int} 
Let $P(X)=\sum_{j=0}^{d} P_j(X) \in \C[X_1,\cdots,X_n]$ be a polynomial function where $P_j$
is the homogeneous part of $P$ of degree $j$. The function
$$
\zeta^P(s):={\sum}'_{k\in\Z^n} \tfrac{P(k)}{|k|^s}, \,\,\, s \in \C
$$
has a meromorphic continuation to the whole complex plane
$\C$. \par Moreover $\zeta^P(s)$ is not entire if and only if
$\mathcal{P}_P:= \{j \,\,\vert\,\,\int_{u\in S^{n-1}} P_j(u)\, dS(u)\neq 0 \}\neq \varnothing$.
In that case, $\zeta^P$ has only simple poles at the points $j+n$,
$j\in \mathcal{P}_{P}$, with
$$
\underset{s=j+n}{\Res} \, \zeta^P(s) = \int_{u\in S^{n-1}} P_j(u)\, dS(u).
$$
\end{theorem}
\medskip

The proof of this theorem is based on the following lemmas.

\begin{lemma}\label{majPQs}
For any polynomial $P\in \C[X_1,\dots,X_n]$ of total
degree $\delta (P):=\sum_{i=1}^n deg_{X_i}P$ and any $\alpha \in
\N_0^n$, we have
$$\partial^{\alpha} \left(P(x) |x|^{-s}\right)\ll_{P, \alpha, n}
  (1+|s|)^{|\alpha|_1} |x|^{-\sigma -|\alpha|_1 +\delta (P)}$$
uniformly  in  $x \in \R^n$, $|x|\geq 1$, where
$\sigma=\Re(s)$.
\end{lemma}
\begin{proof}
By linearity, we may assume without loss of generality
that $P(X)=X^{\gamma}$ is a monomial. It is easy to
prove (for example by induction on $|\alpha|_1$) that for all $\alpha
\in \N_0^n$ and
$x \in \R^n \setminus \{0\}$:
$$
\partial^{\alpha} \left(|x|^{-s}\right)
=\alpha! \sum_{\genfrac{}{}{0pt}{2}{\beta, \mu \in \N_0^n}{\beta + 2
\mu =
\alpha }}
\genfrac(){0pt}{1}{-s/2}{|\beta|_1 +|\mu|_1}
\tfrac{(|\beta|_1+|\mu|_1)!}{\beta! ~ \mu!}
\tfrac{x^{\beta}}{|x|^{\sigma+2(|\beta|_1+|\mu|_1)}}.
$$
It follows that for all $\alpha \in \N_0^n$, we have uniformly in  $x
\in \R^n$,  $|x|\geq 1$:
\begin{equation}\label{majQs}
\partial^{\alpha} \left(|x|^{-s}\right) \ll_{\alpha,  n}
  (1+|s|)^{|\alpha|_1} |x|^{-\sigma -|\alpha|_1}\,.
\end{equation}

By Leibniz formula and (\ref{majQs}), we have uniformly in $x\in \R^n$,  $|x|\geq 1$:
\begin{align*}
\partial^{\alpha} \left(x^\gamma |x|^{-s}\right) & = \, \sum_{\beta
\leq \alpha}
\genfrac(){0pt}{1}{\alpha}{\beta} \,  \partial^{\beta} (x^\gamma)~
\partial^{\alpha
  -\beta} \left(|x|^{-s}\right) \\
& \ll_{\gamma, \alpha, n}  \sum_{\beta \leq \alpha; \beta \leq
\gamma} x^{\gamma
-\beta}~ (1+|s|)^{|\alpha|_1-|\beta|_1}~
|x|^{-\sigma-|\alpha|_1+|\beta|_1} \\
& \ll_{\gamma, \alpha, n}  (1+|s|)^{|\alpha|_1}~
|x|^{-\sigma-|\alpha|_1+|\gamma|_1}.
\end{align*}
\end{proof}

\begin{lemma}
\label{deltaf} 
Let $P\in \C[X_1,\dots,X_n]$ be a polynomial of degree $d$.
Then, the difference
$$
\Delta_P(s):={\sum}'_{k\in\Z^n}
\tfrac{P(k)}{|k|^s}-\int_{\R^n\setminus B^{n}}
\tfrac{P(x)}{|x|^s} \, dx
$$
which is defined for $\Re(s)>d+n$, extends holomorphically on the whole complex plane $\C$.
\end{lemma}

\begin{proof}
We fix in the sequel a function $\psi\in C^\infty(\R^n ,\R)$ such that for all $x\in\R^n$
$$
0\leq \psi(x) \leq 1, \quad \psi(x)=1 \text{ if }|x|\geq 1
\quad \text{and} \quad \psi(x)=0 \text{ if } |x|\leq 1/2.
$$
The function $f(x,s) :=  \psi (x)~P(x)~ |x|^{-s}$, $x\in \R^n$ and
$s\in \C$, is in ${\cal C}^\infty (\R^n \times \C)$  and depends holomorphically on $s$.

Lemma \ref{majPQs} above shows that $f$ is
a ``gauged symbol'' in the terminology of \cite[p. 4]{GSW}.
Thus \cite[Theorem 2.1]{GSW} implies that $\Delta_P(s)$ extends
holomorphically on the whole complex plane $\C$. However, to be
complete, we will give here a short proof of Lemma \ref{deltaf}:
\par It follows from the classical Euler--Maclaurin formula that for
any function $h: \R \rightarrow \C$ of class ${\cal C}^{N+1}$
satisfying $ \lim_{|t|\rightarrow +\infty} h^{(k)}(t)=0$ and $\int_{\R}
|h^{(k)} (t)| ~dt <+\infty$ for any
$k=0 \dots,N+1$, that we have
$$
\sum_{k\in \Z} h(k) = \int_{\R} h(t) + \tfrac{(-1)^N}{(N+1)!}
\int_{\R} B_{N+1}(t)~h^{(N+1)}(t) ~dt
$$
where $B_{N+1}$ is the Bernoulli function of
order $N+1$ (it is a bounded periodic function.) \par Fix $m' \in \Z^{n-1}$ and $s\in
\C$. Applying to the function $h(t):= \psi (m',t)~P(m',t)
~|(m',t)|^{-s}$ (we use Lemma \ref{majPQs} to verify hypothesis),
 we obtain that for any $N\in \N_0$:
\begin{equation}\label{*1}
\sum_{m_n \in \Z} \psi (m',m_n) ~P(m',m_n) ~|(m',m_n)|^{-s}
= \int_{\R} \psi(m',t)
~P(m',t) ~|(m',t)|^{-s} ~dt +{\cal R}_N(m';s)
\end{equation}
where $ {\cal R}_N(m';s):=\tfrac{(-1)^N}{(N+1)!} \int_{\R} B_{N+1}(t)~
\tfrac{\partial^{N+1}}{{\partial x_n}^{N+1}} \left(\psi(m',t)~P(m',t)
~|(m',t)|^{-s}\right)~dt$.\\
By Lemma \ref{majPQs},
$$
\int_{\R} {\Big |} B_{N+1}(t)~
\tfrac{\partial^{N+1}}{{\partial x_n}^{N+1}} \left(\psi (m',t)
~P(m',t)~|(m',t)|^{-s}\right){\Big |}~dt \ll_{P,n, N}
(1+|s|)^{N+1}~ (|m'|+1)^{-\sigma
-N+ \delta(P)}.
$$
Thus $ \sum_{m' \in \Z^{n-1}} {\cal R}_N(m';s)$
converges absolutely and define a holomorphic function in
the half plane $\{\sigma
=\Re (s) > \delta(P)+n-N\}$. \par Since $N$ is an arbitrary integer,
by letting
$N\rightarrow \infty$ and using $(\ref{*1})$ above, we conclude that:
$$
s\mapsto \sum_{(m',m_n) \in \Z^{n-1}\times \Z} \psi (m',m_n) P(m',m_n) |(m',m_n)|^{-s}-\sum_{m'
  \in \Z^{n-1}} \int_{\R} \psi (m',t) ~P(m',t)~|(m',t)|^{-s} dt
$$
has
a holomorphic continuation to the whole complex plane $\C$.\par
After $n$ iterations, we obtain that
$$s\mapsto {\sum}_{m\in \Z^{n}} \psi(m)~P(m)
~|m|^{-s}-\int_{\R^n} \psi(x)~P(x) ~|x|^{-s}~dx$$ has a
holomorphic continuation to the whole $\C$.\\
To finish the proof of Lemma \ref{deltaf}, it is enough to notice
that:

\hspace{1 cm} $\bullet$ $\psi(0)=0$ and  $\psi (m)=1$, $\forall m\in
\Z^n\setminus \{0\}$;

\hspace{1 cm} $\bullet$ $s\mapsto \int_{B^n} \psi(x)~P(x)~|x|^{-s}~dx
= \int_{\{x\in \R^n : 1/2\leq |x|\leq 1\}} \psi(x)~P(x)~|x|^{-s}~dx$
is a holomorphic
function on $\C$.
\end{proof}

\begin{proof}[Proof of Theorem \ref{res-int}]

Using the polar decomposition of the volume form
$dx=\rho^{n-1}\,d\rho\, dS$ in $\R^n$, we get for $\Re(s)>d+n$,
$$
\int_{\R^n \setminus B^{n}} \tfrac{P_j(x)}{|x|^s}dx =
\int_{1}^{\infty}
\tfrac{\rho^{j+n-1}}{\rho^s}\int_{S^{n-1}} P_j(u)\, dS(u) =
\tfrac{1}{j+n-s}
\int_{S^{n-1}} P_j(u)\, dS(u).
$$
Lemma \ref{deltaf} now gives the result.
\end{proof}

\subsection{Holomorphy of certain series}
Before stating the main result of this section, we give first in the following some
preliminaries from Diophantine approximation theory:

\begin{definition}\label{ba}
(i) Let $\delta >0$. A vector $a \in \R^n$ is said to be
$\delta$-badly approximable
if there exists $c >0$ such that $|q . a -m| \geq c \,|q|^{-\delta}$,
$\forall q \in \Z^n \setminus \set{0}$ and $\forall m \in \Z$. \\
We note ${\cal BV}(\delta )$ the set of $\delta$-badly approximable
vectors and ${\cal
BV} :=\cup_{\delta >0} {\cal BV}(\delta)$ the set of badly
approximable vectors.\par
(ii) A matrix $\Th \in {\cal M}_{n}(\R)$ (real $n \times n$ matrices)
will be
said to be badly approximable if there
exists $u \in \Z^n$ such that  ${}^t\Th (u)$ is a badly approximable vector of $\R^n$.
\end{definition}
{\bf Remark.} A classical result from Diophantine approximation asserts
that for $\delta >n$, the Lebesgue measure of
$\R^n \setminus {\cal BV}(\delta)$ is zero (i.e almost any element of
$\R^n$ is $\delta-$badly approximable.)
\par Let $\Th \in {\cal M}_n(\R)$. If its row
of index $i$ is a badly approximable vector of $\R^n$ (i.e. if $L_i
\in {\cal BV}$)
then ${}^t \Th (e_i) \in {\cal BV}$ and thus $\Th$ is a badly approximable matrix. It
follows that almost any matrix of ${\cal M}_n(\R)\approx \R^{n^2}$ is badly approximable.\par

\medskip

The goal of this section is to show the following

\begin{theorem}\label{analytic}
Let $P\in \C[X_1,\cdots,X_n]$ be a homogeneous polynomial of degree
$d$ and let $b$ be in
$\mathcal{S}(\Z^{n} \times \dots \times \Z^{n})$ ($q$ times,
$q\in\N$). Then,

(i)
Let $a \in \R^n$. We define $f_a(s):={\sum}'_{k\in \Z^n}
\frac{P(k)}{|k|^s}\,
e^{2\pi i k.a}$.

\quad 1.
If $a\in \Z^n$, then $f_a$ has a meromorphic
continuation to the whole complex plane
$\C$.\\ Moreover if  $S$ is the unit sphere and $dS$
its Lebesgue measure, then
$f_a$ is not entire if and only if
$\int_{u\in S^{n-1}} P(u)\, dS(u)\neq 0$. In that case, $f_a$
has only a simple pole at the point $d+n$, with
$\underset{s=d+n}{\Res} \, f_a(s) = \int_{u\in  S^{n-1}} P(u)\, dS(u)$.

\quad 2.
If $a\in \R^n\setminus \Z^n$, then $f_a(s)$ extends holomorphically
to the whole
complex plane $\C$.

(ii)
 Suppose that  $\Th \in {\cal M}_{n}(\R)$ is badly approximable. For any $(\eps_i)_i\in \{-1,0,1\}^{q}$, the function
$$
g(s):={\sum}_{l\in (\Z^n)^{q}} \, b(l) \,f_{\Th\,
\sum_i \eps_i l_i}(s)
$$
extends meromorphically to the whole complex plane  $\C$ with only one possible pole
on $s= d+n$.\par Moreover, if we set ${\cal Z}:=\{l\in(\Z^n)^{q} \,\, \vert \,\, \sum_{i=1}^q \eps_i
l_i= 0\}$ and $V:=\sum_{l\in {\cal Z}} \, b(l)$, then

1. If $V\int_{S^{n-1}} P(u)\, dS(u)\neq 0$, then $s=d+n$ is a simple pole of $g(s)$ and
$$
\underset{s=d+n}{\Res} \, g(s) = V\,\int_{u\in S^{n-1}} P(u)\, dS(u).
$$

2. If $V\int_{S^{n-1}} P(u)\, dS(u)=0$, then $g(s)$ extends holomorphically to the whole complex plane $\C$.

(iii)
Suppose that $\Th \in \mathcal{M}_n(\R)$ is badly approximable. For any 
$(\eps_i)_i\in \{-1,0,1\}^{q}$, the function
$$
g_0(s):={\sum}_{l\in (\Z^n)^{q}\setminus
  {\cal Z}} \,
b(l)\,f_{\Th\, \sum_{i=1}^q \eps_i l_i}(s)
$$
where ${\cal Z}:=\{l\in(\Z^n)^{q} \,\,\vert\,\, \sum_{i=1}^q \eps_i l_i= 0\}$ extends holomorphically to the whole
complex plane $\C$.
\end{theorem}

{\it Proof of Theorem \ref{analytic}:}

First we remark that 

$\hspace{1cm}$ If $a \in \Z^n$ then $f_a(s)={\sum}'_{k\in \Z^n} \frac{P(k)}{|k|^s}$. So,
the point $(i.1)$ follows from Theorem \ref{res-int};

$\hspace{1cm}$
$ g(s):=\sum_{l\in (\Z^n)^{q}\setminus {\cal Z}} \, b(l) \,f_{\Th\, \sum_i \eps_i l_i}(s)
+\left(\sum_{l\in {\cal Z}} \, b(l)\right) {\sum}'_{k\in \Z^n}
\frac{P(k)}{|k|^s}$. Thus, the point $(ii)$ rises easily from $(iii)$ and Theorem \ref{res-int}.

So, to complete the proof, it remains to prove the items $(i.2)$ and $(iii)$.\par
The direct proof of $(i.2)$ is easy but is not sufficient to deduce $(iii)$ of which the proof is
more delicate and requires a more precise (i.e. more effective) version of
$(i.2)$. The next lemma gives such crucial version, but before, let us give some notations:
$$
{\cal F}:=\{\tfrac{P(X)}{(X_1^2+\dots +X_n^2+1)^{r/2}}  \,\,\vert\,\, P(X) \in \C[X_1,\dots, X_n] {\mbox { and }} r \in \N_0\}.
$$
\hspace{1cm}  We set $g=$deg$(G) =$deg$(P) -r \in \Z$, the degree of
$G=\frac{P(X)}{(X_1^2+\dots +X_n^2+1)^{r/2}}\in
{\cal F}$.

\hspace{0.4cm} By convention, we set deg$(0)=-\infty$.

\begin{lemma}
\label{ieffective}
Let $a \in \R^n$. We assume that $d\left(a . u, \Z\right):= \inf_{m\in \Z} |a . u  -m| >0$ for some $u \in \Z^n$.
For all $G\in {\cal F}$, we define formally,
\begin{align*}
    F_0(G;a;s):={\sum}'_{k\in\Z^n}
\tfrac{G(k)}{|k|^{s}}\, e^{2\pi i \,k . a}  \quad \text{and} \quad F_1(G;a;s):={\sum}_{k \in \Z^n}
\tfrac{G(k)}{(|k|^2+1)^{s/2}} \,e^{2\pi i \,k . a} .
\end{align*}
\\
Then for all $N\in \N$, $G\in {\cal F}$ and $i\in \{0,1\}$, there exist
positive constants $C_i:=C_i(G,N,u)$, $B_i:=B_i(G,N,u)$ and
$A_i:=A_i(G,N,u)$ such
that $s\mapsto F_i(G;\a;s)$ extends  holomorphically to the half-plane
$\{\Re(s)>-N\}$ and verifies in it:
$$F_i(G;a;s)\leq C_i (1+|s|)^{B_i} \,
\big(d\left(a . u, \Z\right)\big)^{-A_i}.$$
\end{lemma}

\begin{remark}
The important point here is that we obtain an explicit bound of $F_i(G;\a;s)$ in $\{\Re(s)>-N\}$ which depends on 
the vector $a$ only through $d(a.u,\Z)$, so depends on $u$ and
indirectly on $a$ (in the sequel, $a$ will vary.) In particular the
constants $C_i:=C_i(G,N,u)$, $B_i=B_i(G,N)$ and $A_i:=A_i(G,N)$ do not depend on the vector $a$
but only on $u$. {\it This is crucial for the proof of items $(ii)$ and $(iii)$ of Theorem \ref{analytic}!}
\end{remark}

\subsubsection{\texorpdfstring{Proof of Lemma \ref{ieffective} for $i=1$:}{Proof of Lemma for i=1}}

Let $N\in \N_0$ be a fixed integer, and set $g_0:= n+N+1$.\\
We will prove Lemma \ref{ieffective} by induction on $g=$deg$(G)\in
\Z$. More
precisely, in order to prove case $i=1$, it suffices to
prove that:

\hspace{1cm} Lemma \ref{ieffective} is true for all $G\in {\cal F}$ with  deg$(G)\leq -g_0$.

\hspace{1cm} Let $g\in \Z$ with $g\geq -g_0+1$. If Lemma \ref{ieffective} is true for all $G\in {\cal F}$ 
such that deg$(G)\leq g -1$,

\hspace{1cm} then it is also true for all $G\in {\cal F}$ satisfying deg$(G)= g$.

$\bullet$ {Step 1: Checking Lemma
\ref{ieffective} for
deg$(G) \leq -g_0:= -(n+N+1)$.}\\
Let $G(X)=\frac{P(X)}{(X_1^2+\dots +X_n^2+1)^{r/2}} \in {\cal F}$ with  deg$(G)\leq -g_0$.
It is easy to see that we have uniformly in  $s=\sigma +i\tau  \in \C$
and in $k \in
\Z^n$:
\begin{align*}
\tfrac{|G(k) \,e^{2\pi i \,k .
a}|}{(|k|^2+1)^{\sigma/2}}=&\tfrac{|P(k)|}{(|k|^2+1)^{(r+\sigma)/2}}\ll_G
\tfrac{1}{(|k|^2+1)^{(r+\sigma-deg(P))/2}}
\ll_G  \tfrac{1}{(|k|^2+1)^{(\sigma-deg(G))/2}}\ll_G
\tfrac{1}{(|k|^2+1)^{(\sigma+g_0)/2}}.
\end{align*}
It follows that $F_1(G;a;s)=\sum_{k \in \Z^n }
\frac{G(k)}{(|k|^2+1)^{s/2}} \,e^{2\pi i
\, k . a}$ converges absolutely and defines a holomorphic function in
the half plane
$\{\sigma >-N\}$. Therefore, we have for any $s\in \{\Re(s) >-N\}$:
$$|F_1(G;a;s)|\ll_G \sum_{k \in \Z^n }
\tfrac{1}{(|k|^2+1)^{(-N+g_0)/2}}\ll_G \sum_{k \in \Z^n }
\tfrac{1}{(|k|^2+1)^{(n+1)/2}}\ll_G 1.$$
Thus, Lemma \ref{ieffective} is true when deg$(G)\leq -g_0$.

$\bullet$ { Step 2: Induction. }\\
Now let $g\in \Z$ satisfying $g\geq -g_0+1$ and suppose that Lemma
\ref{ieffective} is
valid for all $G\in {\cal F}$ with deg$(G) \leq g-1$. Let $G\in
{\cal F}$ with deg$(G)=g$. We will prove that  $G$ also
verifies conclusions of Lemma \ref{ieffective}:\\
There exist $P\in \C[X_1,\dots,X_n]$ of degree $d\geq 0$ and
$r\in \N_0$ such that $G(X)=\frac{P(X)}{(X_1^2+\dots
+X_n^2+1)^{r/2}}$ and
$g=$deg$(G)=d-r$.\\
Since $G(k)\ll (|k|^2+1)^{g/2}$ uniformly in $k \in \Z^n$, we deduce
that
$F_1(G;a;s)$ converges absolutely in $\{\sigma=\Re(s)>n+g\}$.\\
Since $k\mapsto k+u$ is a bijection from  $\Z^n $ into $\Z^n$, it
follows
that we also have for $\Re(s)>n+g$
\begin{align*}
F_1(G;a;s)&=\sum_{k \in \Z^n }
\tfrac{P(k)}{(|k|^2+1)^{(s+r)/2}} \,e^{2\pi i \,k . a}
= \sum_{k \in \Z^n }
\tfrac{P(k+u)}{(|k+u|^2+1)^{(s+r)/2}} \,  e^{2\pi i \, (k+u) .  a}\\
&= e^{2\pi i \,u . a} \sum_{k \in \Z^n } \tfrac{P(k+u)}{(|k|^2+2 k . u
+|u|^2+1)^{(s+r)/2}} \,e^{2\pi i \,k . a}\\
&= e^{2\pi i \,u . a} \sum_{\alpha \in \N_0^n;
|\alpha|_1=\alpha_1+\dots+\alpha_n \leq
d} \tfrac{u^{\alpha}}{\alpha !} \sum_{k \in \Z^n }
\tfrac{\partial^{\alpha}
P(k)}{(|k|^2+2 k . u +|u|^2+1)^{(s+r)/2}} \, e^{2\pi i \,k . a}\\
&= e^{2\pi i \,u . a} \sum_{|\alpha|_1\leq d}
\tfrac{u^{\alpha}}{\alpha !} \sum_{k \in
\Z^n } \tfrac{\partial^{\alpha} P(k)}{(|k|^2+1)^{(s+r)/2}}
\big(1+\tfrac{2 k . u
+|u|^2}{(|k|^2+1)}\big)^{-(s+r)/2} \, e^{2\pi i \,k . a}.
\end{align*}
Let $M:= \sup(N+n+g, 0)\in \N_0$. We have uniformly in $k \in \Z^n$
$$
\big(1+\tfrac{2 k . u + |u|^2}{(|k|^2+1)}\big)^{-(s+r)/2}=
\sum_{j=0}^M \genfrac(){0pt}{1}{-(s+r)/2}{j} \tfrac{\left(2 k . u +
|u|^2\right)^j}{(|k|^2+1)^j}+
O_{M, u}\big(  \tfrac{(1+|s|)^{M+1}}{(|k|^2+1)^{(M+1)/2}}\big).
$$
Thus, for
$\sigma =\Re(s)>n+d$,
\begin{eqnarray}\label{f0dev}
F_1(G;a;s)&=& e^{2\pi i \,u . a} \sum_{|\alpha|_1 \leq d}
\tfrac{u^{\alpha}}{\alpha !}
\sum_{k \in \Z^n } \tfrac{\partial^{\alpha}
P(k)}{(|k|^2+1)^{(s+r)/2}} \big(1+\tfrac{2
k . u +|u|^2}{(|k|^2+1)}\big)^{-(s+r)/2}
e^{2\pi i \,k . a}\nonumber \\
&=& e^{2\pi i \,u . a} \sum_{|\alpha|_1 \leq d} \sum_{j=0}^M
\tfrac{u^{\alpha}}{\alpha
!}  \genfrac(){0pt}{1}{-(s+r)/2}{j} \sum_{k \in \Z^n }
\tfrac{\partial^{\alpha} P(k) \left(2 k . u
+|u|^2\right)^j } {(|k|^2+1)^{(s+r+2j)/2}} \, e^{2\pi i
  \,k . a}\nonumber \\
& & \hspace{1cm}+O_{G, M,u} \big((1+|s|)^{M+1}\sum_{k \in \Z^n }
\tfrac{1}{(|k|^2+1)^{(\sigma
+M+1-g)/2}}\big).
\end{eqnarray}
Set $I:=\left\{(\alpha ,j) \in \N_0^n \times \{0,\dots,M\} \mid
|\alpha|_1 \leq d\right\}$ and $I^*:=I\setminus \set{(0,0)}$.\\
Set also $ G_{(\alpha ,j);u}(X):= \tfrac{\partial^{\alpha} P(X)
\left(2 X . u
+|u|^2\right)^j }
{(|X|^2+1)^{(r+2j)/2}}\in {\cal F}$ for all $(\alpha ,j) \in I^*$.\\
Since $M\geq N+n+g$, it follows from (\ref{f0dev}) that
\begin{eqnarray}\label{crucial1}
(1-e^{2\pi i \,u . a})~F_1(G;a;s)= e^{2\pi i \,u . a} \sum_{(\alpha,
j)\in I^*}
\tfrac{u^{\alpha}}{\alpha !}  \genfrac(){0pt}{1}{-(s+r)/2}{j}
F_1\left(G_{(\alpha ,j);u};\a;s\right)  +R_N(G; a; u; s)
\end{eqnarray}
where $s\mapsto R_N(G; a; u; s)$ is a holomorphic function in the
half plane
$\{\sigma =\Re(s) >-N\}$, in which it satisfies the bound
$R_N(G; a; u; s)\ll_{G,N,u} 1 $.\\
Moreover it is easy to see that, for any $(\alpha, j)\in I^*$,
$$
\text{deg}\hspace{-.06cm}\left(G_{(\alpha ,j);u}\right)=\text{deg}\hspace{-.05cm}\left(\partial^{\alpha}
P \right)+j -(r+2j)\leq d-|\alpha|_1 +j -(r+2j)=g-|\alpha|_1 -j\leq g-1.
$$
Relation (\ref{crucial1}) and the induction hypothesis imply then that
\begin{equation}\label{crucial2}
(1-e^{2\pi i \,u . a })~F_1(G;a;s) {\mbox { verifies the conclusions
of Lemma \ref{ieffective}}}.
\end{equation}
Since $ |1-e^{2\pi i \,u . a}|=2|\sin(\pi u . a)|\geq d\left(u  . a, \Z\right)$,
then (\ref{crucial2}) implies that $F_1(G;a;s)$ satisfies conclusions of Lemma \ref{ieffective}. 
This completes the induction and the proof for $i=1$.

\subsubsection{\texorpdfstring{Proof of Lemma \ref{ieffective} for $i=0$:}{Proof of Lemma for i=0}}

Let  $N\in \N$ be a fixed integer.
Let $G(X)=\frac{P(X)}{(X_1^2+\dots +X_n^2+1)^{r/2}}
\in {\cal F}$ and $g=$ deg$(G)=d-r$ where $d\geq 0$ is the degree of
the polynomial $P$.
Set also $M:=\sup (N+g+n, 0)\in \N_0$.\par
Since $P(k)\ll |k|^d$ for $k \in \Z^n\setminus
\set{0}$, it follows that $F_0(G;a;s)$ and $F_1(G;a;s)$ converge
absolutely in the half plane
$\{\sigma =\Re(s)>n+g\}$.\\
Moreover, we have for $s=\sigma +i\tau  \in \C$ with $\sigma >n+g$:
\begin{align}\label{crucial3}
F_0(G;a;s)&= \sum_{k \in \Z^n \setminus \set{0}}
\tfrac{G(k)}{(|k|^2+1-1)^{s/2}} \, e^{2\pi i \,k . a}
={\sum_{k \in \Z^n }}' \tfrac{G(k)}{(|k|^2+1)^{s/2}}
\left(1-\tfrac{1}{|k|^2+1}\right)^{-s/2} \,
e^{2\pi i \,k . a} \nonumber \\
&= {\sum_{k \in \Z^n }}' \,  \sum_{j=0}^M
\genfrac(){0pt}{1}{-s/2}{j} (-1)^j
\tfrac{G(k)}{(|k|^2+1)^{(s+2j)/2}} \,
e^{2\pi i \,k . a} \nonumber\\
&\hspace{2cm} +O_M\big((1+|s|)^{M+1} {\sum_{k \in \Z^n }}'
\tfrac{|G(k)|}{(|k|^2+1)^{(\sigma +2M+2)/2}}\big)\nonumber \\
&=\sum_{j=0}^M  \genfrac(){0pt}{1}{-s/2}{j}
(-1)^j F_1(G;a; s+2j) \nonumber\\
&\hspace{2cm}+ O_M\big[(1+|s|)^{M+1} \big(1+{\sum_{k \in \Z^n }}'
\tfrac{|G(k)|}{(|k|^2+1)^{(\sigma +2M+2)/2}}\big)\big].
\end{align}
In addition we have  uniformly in $s=\sigma +i\tau \in \C$ with 
$\sigma >-N$,
$$
{\sum_{k \in \Z^n }}'
\tfrac{|G(k)|}{(|k|^2+1)^{(\sigma +2M+2)/2}}\ll {\sum_{k \in \Z^n }}'
\tfrac{|k|^g}{(|k|^2+1)^{(-N +2M+2)/2}}\ll {\sum_{k \in \Z^n }}'
\tfrac{1}{|k|^{n+1}}<+\infty.
$$
So (\ref{crucial3}) and Lemma \ref{ieffective} for $i=1$ imply that Lemma \ref{ieffective} is also true for $i=0$. This
completes the proof of Lemma \ref{ieffective}.

\subsubsection{\texorpdfstring{Proof of item $(i.2)$ of Theorem \ref{analytic}:}{Proof of item i.2 of Theorem}}

Since  $a\in \R^n\setminus \Z^n$, there exists $i_0\in \{1,\dots,n\}$ with $a_{i_0}\not \in \Z$. So 
 $d(a . e_{i_0}, \Z)=d(a_{i_0},\Z)>0$. Therefore,  $a$ satisfies the assumption of 
Lemma \ref{ieffective} with  $u=e_{i_0}$. Thus, for all $N\in \N$, $s\mapsto f_{a}(s)=F_0(P;a;s)$ has a
holomorphic continuation to the half-plane $\{\Re (s)>-N\}$. It follows, by letting
$N\rightarrow \infty$, that $s\mapsto f_{a}(s)$ has a holomorphic continuation
to the whole complex plane $\C$.

\subsubsection{\texorpdfstring{Proof of item $(iii)$ of Theorem \ref{analytic}:}{Proof of item (iii) of Theorem}}
Let $\Th\in {\cal M}_n(\R)$, $(\eps_i)_i\in \{-1,0,1\}^{q}$ and $b
\in {\cal
S}(\Z^n\times \Z^n)$. We assume that
$\Th$ is a badly approximable matrix. Set 
${\cal Z}:= \set{l=(l_1,\dots,l_q)\in (\Z^n)^q \,\,\vert\,  \, \sum_i \eps_i l_i =0}$ and
$P\in \C[X_1,\dots,X_n]$ of degree  $d\geq 0$.\\
It is easy to see that for $\sigma >n+d$:
\begin{align*}
\sum_{l\in (\Z^n)^{q}\setminus {\cal Z}} \, |b(l)| \,  {\sum_{k \in
\Z^n }}'
\tfrac{|P(k)|}{|k|^{\sigma}} \,|e^{2\pi i \,k . \Th\, \sum_i \eps_i
l_i }| & \ll_P  \sum_{l\in (\Z^n)^q\setminus {\cal Z}} |b (l)| \,
{\sum_{k \in \Z^n }}'
\tfrac{1}{|k|^{\sigma -d}}
\ll_{P,\sigma} \sum_{l \in (\Z^n)^q \setminus {\cal Z}} |b (l)|\\
&<+\infty.
\end{align*}
So
$$
g_0(s):=\sum_{l \in (\Z^n)^q \setminus {\cal Z}}
b (l) \, f_{\Th\, \sum_i \eps_i l_i}(s)= \sum_{l \in (\Z^n)^q
\setminus
  {\cal Z}}  b(l) \, {\sum_{k \in \Z^n }}'
\tfrac{P(k)}{|k|^{s}} e^{2\pi i \,k . \Th \, \sum_i \eps_i l_i}
$$
converges absolutely in the half plane $\{\Re(s) >n+d\}$.\\
Moreover with the notations of Lemma \ref{ieffective}, we have for
all $s=\sigma +i\tau \in \C$ with $\sigma >n+d$:
\begin{equation}\label{crucial4}
g_0(s)=\sum_{l \in (\Z^n)^q\setminus {\cal Z}}
b(l) f_{\Th\, \sum_i \eps_i l_i}(s)=\sum_{l
\in (\Z^n)^q  \setminus {\cal Z}} b(l) F_0(P; \Th\, {\sum}_i \eps_i
l_i; s)
\end{equation}
But $\Th$ is badly approximable, so there exists $u \in \Z^n$ and
$\delta ,c>0$ such
$$
|q . \,{}^t\Th u -m| \geq c\, (1+|q|)^{-\delta},\,  \forall q \in
\Z^n \setminus \set{0},
\, \forall m\in \Z.
$$
We deduce that $\forall l \in (\Z^n)^q \setminus {\cal Z},$
$$
\quad |\big(\Th\, {\sum}_i \eps_i l_i \big) .  u -m|= |\big({\sum}_i
\eps_i l_i\big) . {}^t\Th u -m| \geq c \,
\big(1+|{\sum}_i \eps_i l_i| \big)^{-\delta} \geq c \,
(1+|l|)^{-\delta}.
$$
It follows that there exists $u \in \Z^n$, $\delta >0$ and $c>0$ such
that
\begin{equation}\label{hypothesisOK!}
\forall l \in (\Z^n)^q \setminus {\cal Z}, \quad d\big((\Th\,
{\sum}_i \eps_i l_i) .  u;
\Z\big) \geq c \, (1+|l|)^{-\delta}.
\end{equation}
{\it Therefore, for any $l \in (\Z^n)^q \setminus {\cal Z}$,
the vector $a=\Th\, \sum_i \eps_i l_i$
verifies the assumption of Lemma \ref{ieffective} with the same $u$.
Moreover
$\delta$ and $c$ in
(\ref{hypothesisOK!}) are also independent on $l$}.\\
We fix now $N\in \N$. Lemma \ref{ieffective} implies that there exist
positive
constants $C_0:=C_0(P,N,u)$, $B_0:=B_i(P,N,u)$ and $A_0:=A_0(P,N,u)$
such that for all
$l \in (\Z^n)^q \setminus {\cal Z}$, $s\mapsto F_0(P; \Th\, \sum_i
\eps_i l_i; s)$ extends
{\it holomorphically} to the half plane $\{\Re(s)>-N\}$ and verifies
in it the bound
$$
F_0(P;\Th\, \sum_i \eps_i l_i; s)\leq C_0 \left(1+|s|\right)^{B_0} \,
d\big((\Th\, {\sum}_i \eps_i l_i) .  u; \Z\big) ^{-A_0}.
$$
This and (\ref{hypothesisOK!}) imply that for any compact set $K$
included
in the half plane $\{\Re(s)>-N\}$, there exist two constants
$C:=C(P,N,c, \delta,u, K)$ and  $D
:=D(P,N,c,\delta, u)$ (independent on $l \in (\Z^n)^q \setminus {\cal
Z}$) such that
\begin{equation}\label{crucialbig}
\forall s\in K {\mbox { and }} \forall l \in (\Z^n)^q \setminus {\cal
Z}, \quad F_0(P; \Th\,
\sum_i \eps_i l_i; s)\leq C \left(1+|l|\right)^{D}.
\end{equation}
It follows that $s\mapsto \sum_{l \in (\Z^n)^q \setminus {\cal Z}}
b(l) F_0(P;
\Th\, {\sum}_i \eps_i l_i;s)$
has a holomorphic continuation to the half plane
$\{\Re(s)>-N\}$.\\
This and ( \ref{crucial4}) imply that $
s\mapsto g_0(s)=\sum_{l \in (\Z^n)^q \setminus {\cal Z}} b(l) f_{\Th\, \sum_i \eps_i l_i}(s)$
has a holomorphic continuation to
$\{\Re(s)>-N\}$. Since $N$ is an arbitrary integer, by letting
$N\rightarrow \infty$,
it follows that $s\mapsto g_0(s)$ has a holomorphic continuation to
the whole complex plane $\C$ which completes the proof of the theorem. 

\begin{remark}
By equation (\ref{crucial2}), we see that a Diophantine
condition is sufficient to get Lemma \ref{ieffective}. Our Diophantine
condition appears also (in equivalent form) in Connes
\cite[Prop. 49]{NCDG} (see Remark 4.2 below). The following heuristic
argument shows that our condition seems to be necessary in order
to get the result of Theorem \ref{analytic}:\\
For simplicity we assume $n=1$ (but the argument extends easily to
any $n$).\\
Let $\theta \in \R \setminus \Q$. We know that for any $l\in \Z\setminus\{0\}$,
$$
g_{\theta l}(s):={\sum_{k\in \Z}}' \, \tfrac{e^{2\pi i \theta l
k}}{|k|^s}=
\tfrac{\pi^{s-1/2}}{\Gamma (\tfrac{1-s}{2})}{\Gamma (\tfrac{s}{2})}
~h_{\theta l}(1-s) {\mbox { where }}
h_{\theta l}(s):={\sum_{k\in \Z}}'\,\tfrac{1}{|\theta l+ k|^s}.
$$
So, for any $(a_l) \in {\cal S}(\Z)$, the existence of meromorphic
continuation of
$g_0(s):=\sum_{l\in \Z} ' a_l \, g_{\theta l}(s)$ is equivalent to
the
existence of meromorphic continuation of
$$
h_0(s):={\sum_{l\in \Z}}' a_l \, h_{\theta l}(s)=
{\sum_{l\in \Z}} ' a_l \, {\sum_{k\in \Z}} ' \tfrac{1}{|\theta l+
k|^s}.
$$
So, for at least one $\sigma_0 \in \R$, we must have
$\tfrac{|a_l|}{|\theta l+ k|^{\sigma_0}} = O(1)
{\mbox { uniformly in }} k, l \in \Z^*.$

It follows that for any $(a_l) \in {\cal S}(\Z)$,
$|\theta l+ k| \gg |a_l|^{1/\sigma_0}$
uniformly in $k, l \in \Z^*$. Therefore, our Diophantine condition
seems
to be necessary.
\end{remark}

\subsubsection{Commutation between sum and residue}

Let $p\in \N$. Recall that $\mathcal{S}((\Z^n)^p)$ is the set of the
Schwartz
sequences on $(\Z^n)^p$. In other words, $b\in \mathcal{S}((\Z^n)^p)$
if and only if
for all $r\in \N_{0}$, $(1+|l_1|^2+\cdots
|l_p|^2)^r\,|b(l_1,\cdots,l_p)|^2$ is bounded on
$(\Z^n)^p$. We note that if $Q\in\R[X_1,\cdots,X_{np}]$ is a
polynomial, $(a_j)\in
\mathcal{S}(\Z^n)^{p}$, $b\in \mathcal{S}(\Z^n)$ and $\phi$ a
real-valued function, then
 $l:=(l_1,\cdots,l_p)\mapsto \wt a(l)\, b(-\wh l_p)\, Q(l)\,
 e^{i\phi(l)}$ is a Schwartz
sequence on $(\Z^n)^p$, where
\begin{align*}
\wt a(l) &:= a_1(l_1)\cdots a_p(l_p),\\
\wh l_i &:= l_1+\ldots +l_i.
\end{align*}

In the following, we will use several times the fact that for any
$(k,l)\in (\Z^n)^2$
such that $k\neq0$ and $k\neq -l$, we have
\begin{equation}\label{trick-0}
\frac{1}{|k+l|^2} = \frac{1}{|k|^2} -
\frac{2k.l+|l|^2}{|k|^2|k+l|^2}\,.
\end{equation}

\begin{lemma}\label{R-poly} There exists a polynomial $P\in
\R[X_1,\cdots,X_p]$ of degree $4p$ and with positive coefficients
such that for any
$k\in \Z^n$, and $l:=(l_1,\cdots,l_p)\in (\Z^n)^p$ such that $k\neq
0$ and $k\neq -\wh
l_i$ for all $1\leq i \leq p$, the following holds:
$$
\frac{1}{|k+\wh l_1|^2\ldots |k+\wh l_p|^2} \leq \frac{1}{|k|^{2p}}\
P(|l_1|,\cdots,|l_p|).
$$
\end{lemma}
\begin{proof} Let's fix $i$ such that $1\leq i\leq p$.
Using two times (\ref{trick-0}), Cauchy--Schwarz inequality and the
fact that $|k+\wh
l_i|^2\geq 1$, we get
\begin{align*}
\tfrac{1}{|k+\wh l_i|^2} &\leq \tfrac{1}{|k|^2} + \tfrac{2|k||\wh
l_i|+|\wh
l_i|^2}{|k|^4} + \tfrac{(2|k||\wh l_i|+|\wh l_i|^2)^2}{|k|^4|k+\wh
l_i|^2}\\
&\leq \tfrac{1}{|k|^2}+ \tfrac{2}{|k|^3}|\wh l_i| +
\big(\tfrac{1}{|k|^4}+\tfrac{4}{|k|^2}\big)|\wh l_i|^2 +
\tfrac{4}{|k|^3}|\wh l_i|^3 +
\tfrac{1}{|k|^4}|\wh l_i|^4.
\end{align*}
Since $|k|\geq 1$, and $|\wh l_i|^j \leq |\wh l_i|^4$ if $1\leq j\leq
4$, we find
\begin{align*}
&\tfrac{1}{|k+\wh l_i|^2} \leq \tfrac{5}{|k|^2} {\sum}_{j=0}^4 \,|\wh
l_i|^j \leq
\tfrac{5}{|k|^2} \big(1 + 4 |\wh l_i|^4\big)
\leq \tfrac{5}{|k|^2} \big(1 + 4 ({\sum}_{j=1}^p \, |l_j|)^4\big), \\
&\tfrac{1}{|k+\wh l_1|^2\ldots |k+\wh l_p|^2}  \leq
\tfrac{5^p}{|k|^{2p}} \big(1 + 4
({\sum}_{j=1}^p \,|l_j|)^{4}\big)^p.
\end{align*}
Taking $P(X_1,\cdots,X_p):= 5^p \big(1 + 4 ({\sum}_{j=1}^p
X_j)^{4}\big)^p$ now gives
the result.
\end{proof}

\begin{lemma}
\label{abs-som}
Let $b\in \mathcal{S}((\Z^n)^p)$, $p\in\N$, $P_j\in
\R[X_1,\cdots,X_n]$ be a
homogeneous polynomial function of degree $j$, $k\in \Z^n$,
$l:=(l_1,\cdots,l_p)\in
(\Z^n)^p$, $r\in \N_{0}$, $\phi$ be a real-valued function on
$\Z^n\times (\Z^n)^{p}$ and
$$
h(s,k,l):=\frac{b(l)\, P_j(k)\, e^{i\phi(k,l)}}{|k|^{s+r}|k+\wh
l_1|^2\cdots |k+\wh
l_p|^2} \, ,
$$
with $h(s,k,l):=0$ if, for $k\neq 0$, one of the
denominators is zero.

For all $s\in \C$ such that $\Re(s)>n+j-r-2p$, the series
$$
H(s):={{\sum}'}_{(k,l)\in (\Z^n)^{p+1}} h(s,k,l)
$$
is absolutely summable. In particular,
$$
{\sum_{k\in\Z^n}}' \sum_{l\in (\Z^n)^p} h(s,k,l) = \sum_{l\in
  (\Z^n)^p}
{\sum_{k\in\Z^n}} ' h(s,k,l)\,.
$$
\end{lemma}

\begin{proof}
Let $s=\sigma+i\tau \in \C$ such that
$\sigma>n+j-r-2p$. By Lemma \ref{R-poly} we get, for $k\neq 0$,
$$
|h(s,k,l)|\leq |b(l)\, P_j(k)|\, |k|^{-r-\sigma-2p}\, P(l),
$$
where $P(l):=P(|l_1|,\cdots,|l_p|)$ and $P$ is a polynomial of degree
$4p$ with
positive coefficients. Thus,
$|h(s,k,l)|\leq  F(l)\, G(k)$
where $F(l):=|b(l)|\, P(l)$ and $G(k):= |P_j(k)| |k|^{-r-\sigma-2p}$.
The summability
of $\sum_{l\in (\Z^n)^p} F(l)$ is implied by the fact that $b\in
\mathcal{S}((\Z^n)^p)$. The summability of ${\sum}'_{k\in \Z^n} G(k)$
is a consequence
of the fact that $\sigma>n+j-r-2p$. Finally, as a product of two
summable series,
${\sum_{k,l}} F(l) G(k)$ is a summable series, which proves that
${\sum_{k,l}}h(s,k,l)$ is also absolutely summable.
\end{proof}

\begin{definition}
Let $f$ be a function on $D\times (\Z^n)^p$ where $D$ is an open
neighborhood of $0$ in $\C$.

We say that $f$ satisfies (H1) if and only if there exists $\rho>0$ such that

\hspace{1cm} (i) for any $l$, $s\mapsto f(s,l)$ extends as a
holomorphic function on $U_\rho$,
where $U_\rho$ is the
open disk of center 0 and radius $\rho$,

\hspace{1cm}(ii) if $\norm{H(\cdot,l)}_{\infty,\rho}:=\sup_{s\in U_{\rho}}|H(s,l)|$, the series 
$\sum_{l\in (\Z^n)^p} \norm{H(\cdot,l)}_{\infty,\rho}$ is summable.\\
 We say that
$f$ satisfies (H2) if and only if there exists $\rho>0$ such that

\hspace{1cm}
(i) for any $l$, $s\mapsto f(s,l)$ extends as a holomorphic function
on
$U_\rho-\{0\}$,

\hspace{1cm}
(ii) for any $\delta$ such that $0<\delta<\rho$, the series
$\sum_{l\in (\Z^n)^p}
\norm{H(\cdot,l)}_{\infty,\delta,\rho}$ is summable, where
$\norm{H(\cdot,l)}_{\infty,\delta,\rho}:=\sup_{\delta<|s|<\rho}|H(s,l)|$.
\end{definition}

\begin{remark}
Note that (H1) implies (H2). Moreover, if $f$ satisfies (H1)
(resp. (H2) for $\rho>0$, then it is straightforward to check that
$f:s\mapsto \sum_{l\in (\Z^n)^p} f(s,l)$
extends as an holomorphic function on $U_\rho$ (resp. on $U_\rho
\setminus \set{0}$).
\end{remark}

\begin{corollary}
\label{res-somH}
With the same notations of Lemma \ref{abs-som}, suppose that $r+2p-j>n$, then, the function
$H(s,l):={\sum}'_{k\in \Z^n} h(s,k,l)$ satisfies (H1).
\end{corollary}

\begin{proof}
$(i)$ Let's fix $\rho>0$ such that $\rho < r+2p-j-n$. Since $r+2p-j>n$, $U_\rho$ is inside the half-plane of absolute
convergence of the series defined by $H(s,l)$. Thus, $s\mapsto H(s,l)$ is
holomorphic on $U_\rho$.\\ $(ii)$ Since $\big||k|^{-s}\big|\leq |k|^{\rho}$ for all $s\in U_\rho$ and
$k\in\Z^n \setminus \set{0}$, we get as in the above proof
$$
|h(s,k,l)|\leq |b(l)\, P_j(k)| \, |k|^{-r+\rho-2p} \,
P(|l_1|,\cdots,|l_p|).
$$
Since $\rho < r+2p-j-n$, the series ${\sum}'_{k\in \Z^n} |P_j(k)|
|k|^{-r+\rho-2p}$ is
summable.

Thus, $\norm{H(\cdot,l)}_{\infty,\rho} \leq K \, F(l)$ where $K :=
{{\sum_k}}'|P_j(k)| |k|^{-r+\rho-2p}<\infty$. We have already seen
that the series
$\sum_l F(l)$ is summable, so we get the result.
\end{proof}

We note that if $f$ and $g$ both satisfy (H1) (or (H2)),
then so does $f+g$. In the
following, we will use the equivalence relation
$$
f\sim g \Longleftrightarrow f-g \text{ satisfies (H1)}.
$$

\begin{lemma}
\label{res-som}
Let $f$ and $g$ be two functions on $D\times (\Z^n)^p$ where $D$ is an open
neighborhood of $0$ in $\C$, such that $f\sim g$ and such that $g$ satisfies (H2). Then
$$
\underset{s=0}{\Res}\sum_{l \in (\Z^n)^p} f(s,l)=\sum_{l \in (\Z^n)^p} \underset{s=0}{\Res}\ g(s,l)\, .
$$
\end{lemma}

\begin{proof}
Since $f\sim g$, $f$ satisfies (H2) for a certain
$\rho>0$.
Let's fix $\eta$ such that $0<\eta<\rho$ and define $C_\eta$ as the
circle of center 0
and radius $\eta$. We have
$$
\underset{s=0}{\Res}\ g(s,l) = \underset{s=0}{\Res}\ f(s,l) = \tfrac
{1}{2\pi
i}\oint_{C_\eta} f(s,l)\, ds = \int_I u(t,l) dt\, .
$$
where $I=[0,2\pi]$ and $u(t,l):=\tfrac {1}{2\pi} \eta e^{it}
f(\eta\,e^{i t},l) $. The
fact that $f$ satisfies (H2) entails that the series $\sum_{l\in
(\Z^n)^p}
\norm{f(\cdot,l)}_{\infty,C_\eta}$ is summable. Thus, since
$\norm{u(\cdot,l)}_{\infty} = \tfrac {1}{2\pi} \eta
\norm{f(\cdot,l)}_{\infty,C_\eta}$, the series $\sum_{l\in (\Z^n)^p}
\norm{u(\cdot,l)}_{\infty}$ is summable, so, 
$\int_I \sum_{l\in
(\Z^n)^p}   u(t,l) dt = \sum_{l\in (\Z^n)^p}  \int_I u(t,l)dt$ which
gives the result.
\end{proof}

\subsection{Computation of residues of zeta functions}

Since, we will have to compute residues of series, let us introduce the following

\begin{definition}
\begin{align*}
\zeta(s)&:=\sum_{n=1}^{\infty} n^{-s},\\
Z_{n}(s)&:={\sum_{k\in\Z^{n}}}'\,\,\vert k\vert^{-s},\\
\zeta_{p_{1},\dots,p_{n}}(s)&:={\sum_{k\in\Z^{n}}}'\,\,
\frac{k_1^{p_{1}}\cdots
k_n^{p_{n}}}{\vert k\vert^{s}}\,\text{ , for } p_{i}\in \N,
\end{align*}
\end{definition}
\noindent where $\zeta(s)$ is the Riemann zeta function (see \cite{HW} or \cite{Edery}).

By the symmetry $k\rightarrow -k$, it is clear that these functions
$\zeta_{p_{1},\dots,p_{n}}$ all vanish for odd values of $p_{i}$.

Let us now compute
$\zeta_{0,\cdots,0,1_{i},0\cdots,0,1_{j},0\cdots,0}(s)$ in terms of
$Z_{n}(s)$:\\
Since $\zeta_{0,\cdots,0,1_{i},0\cdots,0,1_{j},0\cdots,0}(s)
=A_{i}(s)\,\delta_{ij}$,
exchanging the components $k_{i}$ and $k_{j}$, we get
\begin{align*}
\zeta_{0,\cdots,0,1_{i},0\cdots,0,1_{j},0\cdots,0}(s)
=\tfrac{\delta_{ij}}{n}\,Z_{n}(s-2).
\end{align*}
Similarly,
\begin{align*}
{\sum}'_{\Z^{n}}\,\tfrac{k_{1}^{2}k_{2}^{2}}{\vert k\vert^{s+8}}
=\tfrac{1}{n(n-1)}
Z_{n}(s+4)- \tfrac{1}{n-1}{\sum}'_{\Z^{n}}\, \tfrac{k_{1}^4}{\vert
k\vert^{s+8}}
\end{align*}
but it is difficult to write explicitly
$\zeta_{p_{1},\dots,p_{n}}(s)$ in terms of
$Z_{n}(s-4)$ and other $Z_{n}(s-m)$ when at least four indices
$p_{i}$ are non zero.

When all $p_{i}$ are even, $\zeta_{p_{1},\dots,p_{n}}(s)$ is a
nonzero series of
fractions $\tfrac{P(k)}{\vert k\vert ^s}$ where $P$ is a homogeneous
polynomial of
degree $p_{1}+\cdots +p_{n}$. Theorem \ref{res-int} now gives us the
following

\begin{prop}
    \label{calculres}
$\zeta_{p_{1},\dots,p_{n}}$ has a meromorphic extension to the whole
plane with a
unique pole at $n+p_{1} +\cdots +p_{n}$. This pole is simple and the
residue at this
pole is
\begin{align}
    \label{formule1}
\underset{s=n+p_{1} +\cdots
+p_{n}}{\Res} \,\zeta_{p_{1},\dots,p_{n}}(s)= 2 \,
\tfrac{\Gamma(\tfrac{p_{1}+1}{2}) \cdots
\Gamma(\tfrac{p_{n}+1}{2})}
{\Gamma
(\tfrac{n+p_{1}+ \cdots + p_{n}}{2})}
\end{align}
when all $p_{i}$ are even or this residue is zero otherwise.\\
In particular, for $n=2$,
\begin{align}
\underset{s=0}{\Res} \,\,{\sum_{k\in\Z^2}}'\,\tfrac{k_{i}k_{j}} {\vert
k\vert^{s+4}}=\delta_{ij}\,\pi\, , \label{formulen=2}
\end{align}
and for $n=4$,
\begin{align}
&\underset{s=0}{\Res} \,\,{\sum_{k\in\Z^{4}}}'\,\tfrac{k_{i}k_{j}}
{\vert k\vert^{s+6}}=\delta_{ij}\tfrac{\pi^2}{2}\, ,\nonumber\\
&\label{formule2} \underset{s=0}{\Res}
\,{\sum_{k\in\Z^{4}}}'\,\tfrac{k_{i}k_{j}k_{l}k_{m}}
{\vert k\vert^{s+8}}=(\delta_{ij}\delta_{lm}+\delta_{il}\delta_{jm}
+\delta_{im}\delta_{jl})\,\tfrac{\pi^2}{12}\, .
\end{align}
\end{prop}
\begin{proof}
Equation (\ref{formule1}) follows from Theorem (\ref{res-int})
$$
\underset{s=n+p_{1} +\cdots
+p_{n}}{\Res}\, \, \zeta_{p_{1},\dots,p_{n}}(s)=
\int_{k \in
S^{n-1}}k_{1}^{p_{1}} \cdots k_{n}^{p_{n}}\, dS(k)
$$
and standard formulae (see for instance \cite[VIII,1;22]{Schwartz}).
Equation
(\ref{formulen=2}) is a straightforward consequence of Equation
(\ref{formule1}).
Equation (\ref{formule2}) can be checked for the cases $i=j\neq l=m$
and
$i=j=l=m$.
\end{proof}
Remark that $Z_n(s)$ is an Epstein zeta-function which is associated to the quadratic form $q(x):=x_1^2+...+x_n^2$, 
so $Z_n$ satisfies the following functional equation
$$
Z_n(s)= \pi^{s-n/2} \Gamma (n/2 -s/2)\Gamma (s/2)^{-1}\,
Z_n(n-s).
$$
Since $\pi^{s-n/2} \Gamma (n/2 -s/2) \,\Gamma (s/2)^{-1}=0$
for any negative even integer $n$ and $Z_n(s)$ is meromorphic on $\C$
with only one pole at $s=n$ with residue $2 \pi^{n/2} \Gamma
(n/2)^{-1}$ according to previous proposition, so we get $Z_n(0)=
-1$. We have proved that
\begin{align}
\label{formule}
    \underset{s=0}{\Res} \,\, Z_{n}(s+n)&=2\pi^{n/2} \, \Gamma
(n/2)^{-1},\\
Z_n(0)&= -1.
\label{Zn0}
\end{align}
There are many applications of Proposition \ref{calculres} for instance in $\zeta$-regularization, multiplicative 
anomalies or Casimir effect, see for instance \cite{Edery}.

\subsection{Meromorphic continuation of a class of zeta functions}

Let $n,q\in \N$, $q\geq2$, and $p=(p_1,\dots,p_{q-1}) \in
\N_0^{q-1}$.\\
Set $I:=\{ i \mid p_i\neq 0\}$ and assume that $I\neq \emptyset$ and
$${\cal I}:=\{\alpha =(\alpha_i)_{i\in I} \mid \forall i\in I ~
\alpha_i=(\alpha_{i,1},\dots, \alpha_{i,p_i})\in
\N_0^{p_i}\}=\prod_{i\in I} \N_0^{p_i}.$$

We will use in the sequel also the following notations:

\hspace{1cm} - for $x=(x_1,\dots,x_t) \in \R^t$ recall that
$|x|_1=|x_1|+\dots+|x_t|$ and $|x|=\sqrt{x_1^2+\dots+x_t^2}$;

\hspace{1cm} - for all $\alpha =(\alpha_i)_{i\in I}
  \in {\cal I} =\prod_{i\in I} \N_0^{p_i}$,
$$|\alpha|_1=\sum_{i\in I} |\alpha_i|_1 =\sum_{i\in I}
\sum_{j=1}^{p_i} |\alpha_{i,j}| {\mbox { and }}
\genfrac(){0pt}{1}{1/2}{\alpha} =\prod_{i\in I}
\genfrac(){0pt}{1}{1/2}{\alpha_{i}}=
\prod_{i\in I} \prod_{j=1}^{p_i}
\genfrac(){0pt}{1}{1/2}{\alpha_{i,j}}.$$

\subsubsection{A family of polynomials}
In this paragraph we define a family of polynomials which plays an
important role later.

Consider first the variables:

- for $X_1,\dots, X_n$ we set $X=(X_1,\dots,X_n)$;

- for any $i=1,\dots,2q$, we consider the variables $Y_{i,1},\dots,
  Y_{i,n}$ and set
$Y_i:=(Y_{i,1},\dots, Y_{i,n})$ and $Y:=(Y_1,\dots,Y_{2q})$;

- for $Y=(Y_1,\dots,Y_{2q})$, we set for any $1\leq j\leq q$,
$\wt Y_j:= Y_1+\cdots+ Y_j+  Y_{q+1}+\cdots + Y_{q+j}$.

We define for all $\alpha =(\alpha_i)_{i\in I}\in {\cal I}
=\prod_{i\in I} \N_0^{p_i}$ the polynomial
\begin{equation}
\label{palphaxy}
P_\alpha(X,Y):= \prod_{i\in I} \prod_{j=1}^{p_i}
(2\langle X, \wt Y_i\rangle + |\wt
    Y_i|^2)^{\alpha_{i,j}}.
\end{equation}

It is clear that $P_\alpha(X,Y) \in \Z[X,Y]$, deg$_X P_\alpha \leq
|\alpha|_1$ and deg$_Y P_\alpha \leq 2 |\alpha|_1$.

Let us fix a polynomial $Q\in \R[X_1,\cdots,X_n]$ and
note $d:= \deg Q$.
For $\a\in  {\cal I}$, we want to expand $P_\alpha(X,Y) \, Q(X)$ in
homogeneous polynomials in $X$ and $Y$ so defining
$$
L(\a):=\set{\beta\in\N_0^{(2q+1)n}\, \,\vert\,\,
|\beta|_1-d_\beta \leq 2|\a|_1 \text{ and } d_\beta\leq |\a|_1+d}
$$
where
$d_\beta := \sum_1^n \beta_i$, we set
$$
\genfrac(){0pt}{1}{1/2}{\alpha} P_\alpha(X,Y) \, Q(X) =:
\sum_{\beta\in L(\a)} c_{\a,\beta} \,  X^\beta Y^\beta
$$
where $c_{\a,\beta}\in \R$, $X^\beta:= X_1^{\beta_{1}} \cdots
X_n^{\beta_{n}}$ and $Y^{\beta}:=
Y_{1,1}^{\beta_{n+1}}\cdots Y_{2q,n}^{\beta_{2(q+1)n}}$. By
definition, $X^\beta$ is a
homogeneous polynomial of degree in $X$ equals to $d_\beta$.
We note $$M_{\a,\beta}(Y):=c_{\a,\beta} \, Y^\beta.$$

\subsubsection{Residues of a class of zeta functions}

In this section we will prove the following result, used in
Proposition \ref{zeta(0)} for
the computation of the spectrum dimension of the noncommutative
torus:

\begin{theorem}
\label{zetageneral}
(i) Let $\tfrac{1}{2\pi}\Th$ be a
badly approximable matrix, and $\wt a \in \mathcal{S}
\big((\Z^{n})^{2q}\big)$. Then
$$
s\mapsto f(s):= \sum_{l\in [(\Z^n)^{q}]^2} \wt a_{l}\ {\sum_{k\in
\Z^n}}'\, \prod_{i=1}^{q-1}|k+\wt l_i|^{p_i} |k|^{-s}\, Q(k)\,
e^{ik.\Th \sum_1^{q} l_j}
$$
has a meromorphic continuation to the whole complex plane $\C$ with at most simple
possible poles at the points $s=n+d+|p|_1-m$ where $m\in \N_0$.

(ii) Let $m\in \N_0$ and set
$I(m):= \set{(\a,\beta)\in \mathcal{I}\times \N_0^{(2q+1)n} \, \vert \, \beta\in L(\a)$ where we have taken 
$m=2|\a|_1 -d_\beta +d } $.
Then $I(m)$ is a finite set and $s=n+d+|p|_1-m$ is a pole of $f$
if and only if
$$
C(f,m):= \sum_{l\in Z} \wt a_l
\sum_{(\a,\beta)\in I(m)} M_{\a,\beta}(l)  \int_{u\in S^{n-1}}
u^\beta\, dS(u) \neq 0,
$$
with $Z:=\{l \,\,\vert\,\, \sum_1^{q} l_j=0 \}$ and the convention
$\sum_{\emptyset} =0$.
In that case $s=n+d+|p|_1-m$ is a simple pole of residue
$\underset{s= n+d+|p|_1 -m}{\Res} \, f(s) = C (f,m)$.
\end{theorem}

In order to prove the theorem above we need the following

\begin{lemma}
\label{zetageneral-lem}
For all $N\in \N$ we have
$$
 \prod_{i=1}^{q-1} |k+\wt l_i|^{p_i}=
\sum_{\alpha =(\alpha_i)_{i\in I}
  \in \prod_{i\in I}\{0,\dots,N\}^{p_i}}
  \genfrac(){0pt}{1}{1/2}{\alpha}\,
\tfrac{P_\alpha (k,l)}{|k|^{2|\alpha|_1-|p|_1}}
+\mathcal{O}_N(|k|^{|p|_1- (N+1)/2})
$$
uniformly in $k\in \Z^n$ and $l\in (\Z^n)^{2q}$ such that $|k| > U(l):=36\,
(\sum_{i=1,\, i\neq q}^{2q-1}|l_{i}|)^4$.
\end{lemma}
\begin{proof}
For $i=1,\dots,q-1$, we have uniformly in $k\in \Z^n$ and $l\in
(\Z^n)^{2q}$ with $|k| > U(l)$,
\begin{equation}
\label{devjustification}
\tfrac{\big|2\langle k, \wt l_i \rangle+|\wt l_i|^2\big|}{|k|^2}
\leq\tfrac{\sqrt{U(l)}}{2|k|} < \tfrac{1}{2\sqrt{|k|}}.
\end{equation}
In that case,
\begin{eqnarray*}
|k+\wt l_i|&=& \big(|k|^2+2\langle k, \wt l_i\rangle
  + |\wt l_i|^2\big)^{1/2} =
|k| \big(1+ \tfrac{2\langle k, \wt l_i\rangle
  + |\wt l_i|^2}{|k|^2}\big)^{1/2} =
 \sum_{u=0}^{\infty}  \genfrac(){0pt}{1}{1/2}{u} \,
\tfrac{1}{|k|^{2u-1}}P^i_u(k,l)
\end{eqnarray*}
where for all $i=1,\dots, q-1$ and for all $u\in \N_0$,
\begin{equation*}
P^i_u(k,l):=\big(2\langle k, \wt l_i\rangle + |\wt l_i|^2\big)^u,
\end{equation*}
with the convention $P^i_0(k,l):=1$.

In particular $P^i_u(k,l)\in \Z[k,l]$,
$\deg_{k} P^i_u\leq u$ and $\deg_{l} P^i_u\leq 2u$.
Inequality (\ref{devjustification}) implies that for all
$i=1,\dots,q-1$
and for all $u\in \N$,
$$
\tfrac{1}{|k|^{2u}}\,|P^i_u(k,l)|\leq \big(2\sqrt{|k|}\big)^{-u}
$$
uniformly in $k\in \Z^n$ and $l\in (\Z^n)^{2q}$ such that $|k| > U(l)$.

Let $N\in \N$. We deduce from the previous that for any
$k\in \Z^n$ and $l\in (\Z^n)^{2q}$ with $|k| > U(l)$ and
for all $i=1,\dots,q-1$, we have
\begin{eqnarray*}
|k+\wt l_i|&=& \sum_{u=0}^{N}
\genfrac(){0pt}{1}{1/2}{u} \, \tfrac{1}{|k|^{2u-1}}P^i_u(k,l)+
\mathcal{O}\big(\sum_{u>N}|k|\,|\genfrac(){0pt}{1}{1/2}{u}|\,
(2\sqrt{|k|})^{-u}\big)\\
&=& \sum_{u=0}^{N} \genfrac(){0pt}{1}{1/2}{u} \,
\tfrac{1}{|k|^{2u-1}}P^i_u(k,l)+\mathcal{O}_N
\big(\tfrac{1}{|k|^{(N-1)/2}}\big).
\end{eqnarray*}
It follows that for any $N\in \N$, we have uniformly in
$k\in \Z^n$ and $l\in (\Z^n)^{2q}$ with $|k| > U(l)$ and
for all $i\in I$,
$$
|k+\wt l_i|^{p_i}=\sum_{\alpha_i \in
  \{0,\dots,N\}^{p_i}} \genfrac(){0pt}{1}{1/2}{\alpha_i} \,
\tfrac{1}{|k|^{2|\alpha_i|_1-p_i}}P^i_{\alpha_i} (k,l)
+\mathcal{O}_N \left(\tfrac{1}{|k|^{(N+1)/2-p_i}}\right)
$$
where $P^i_{\alpha_i} (k,l)=\prod_{j=1}^{p_i} P^i_{\alpha_{i,j}}(k,l)$
for all $\alpha_i =(\alpha_{i,1},\dots,\alpha_{i,p_i})\in
  \{0,\dots,N\}^{p_i}$ and
$$
\prod_{i\in I} |k+\wt l_i|^{p_i}=\sum_{\alpha=(\alpha_i) \in
  \prod_{i\in I} \{0,\dots,N\}^{p_i}} \genfrac(){0pt}{1}{1/2}{\alpha}
\,
\tfrac{1}{|k|^{2|\alpha|_1-|p|_1}}P_{\alpha} (k,l)+\mathcal{O}_N
\big(\tfrac{1}{|k|^{(N+1)/2 -|p|_1}}\big)
$$
where $P_{\alpha } (k,l)=\prod_{i\in I} P^i_{\alpha_{i}}(k,l)=
\prod_{i\in I} \prod_{j=1}^{p_i} P^i_{\alpha_{i,j}}(k,l)$.
\end{proof}

\medskip

\begin{proof}[Proof of Theorem \ref{zetageneral}]

$(i)$ All $n$, $q$, $p=(p_1,\dots,p_{q-1})$ and $\wt a
\in {\cal S}\left((\Z^n)^{2q}\right)$ are fixed as above and we
define formally for any $l \in (\Z^n)^{2q}$
\begin{equation}
\label{fls}
F(l,s):= {\sum_{k\in \Z^n}}' \,
\prod_{i=1}^{q-1} |k+\wt l_i|^{p_i}\, Q(k)\,
e^{ik.\Th \sum_1^{q}l_j}\,|k|^{-s}.
\end{equation}
Thus, still formally,
\begin{equation}\label{fsexpfls}
f(s):=\sum_{l\in (\Z^n)^{2q}} \wt a_l\ F(l,s).
\end{equation}
It is clear that $F(l,s)$ converges absolutely in the half
plane $\{\sigma=\Re(s) >n+d+|p|_1\}$ where $d=\deg Q$.

Let $N\in \N$. Lemma \ref{zetageneral-lem} implies
that for any $l\in (\Z^n)^{2q}$ and for $s\in \C$ such
that $\sigma >n+|p|_1+d$,
\begin{align*}
F(l,s)&= {\sum_{|k|\leq U(l)}}' \,
\prod_{i=1}^{q-1} |k+\wt l_i|^{p_i}\, Q(k)\,
e^{ik.\Th \sum_1^{q}l_j}\,|k|^{-s} \\
& \quad \quad  +\sum_{\alpha =(\alpha_i)_{i\in I}
  \in \prod_{i\in
I}\{0,\dots,N\}^{p_i}}\genfrac(){0pt}{1}{1/2}{\alpha}
\sum_{|k| > U(l)} \tfrac{1}{|k|^{s+2|\alpha|_1-|p|_1}}P_\alpha (k,l)
Q(k)\, e^{ik.\Th \sum_1^{q}l_j}
+ G_N(l,s).
\end{align*}
where $s\mapsto G_N(l,s)$ is a holomorphic function in the half-plane
$D_N:=\{\sigma > n+d+|p|_1-\tfrac{N+1}{2}\}$ and verifies in it the
bound
$G_N(l,s) \ll_{N,\sigma} 1$ uniformly in $l$.

It follows that
\begin{equation}
\label{flsexpzeta}
F(l,s)= \sum_{\alpha =(\alpha_i)_{i\in I}
  \in \prod_{i\in I}\{0,\dots,N\}^{p_i}} H_{\a}(l,s)+ R_N(l,s),
\end{equation}
where
\begin{eqnarray*}
H_{\a}(l,s)&:=&{\sum_{k\in \Z^n}}'\,
\genfrac(){0pt}{1}{1/2}{\alpha} \,
\tfrac {1}{|k|^{s+2|\alpha|_1-|p|_1}}P_\alpha (k,l)\, Q(k)\,
e^{ik.\Th \sum_1^{q}l_j},\\
R_N(l,s)&:=& {\sum_{|k|\leq U(l)}}' \,
\prod_{i=1}^{q-1} |k+\wt l_i|^{p_i}\, Q(k)\,
e^{ik.\Th \sum_1^{q}l_j}\,|k|^{-s}\\
& & \quad -{\sum_{|k|\leq U(l)}}' \quad\sum_{\alpha =
(\alpha_i)_{i\in I} \in \prod_{i\in I}\{0,\dots,N\}^{p_i}}
  \genfrac(){0pt}{1}{1/2}{\alpha} \tfrac{P_\alpha (k,l)}
{|k|^{s+2|\alpha|_1-|p|_1}}Q(k)\,
  e^{ik.\Th \sum_1^{q}l_j}
+ G_N(l,s).
\end{eqnarray*}
In particular there exists $A(N)>0$ such that
$s\mapsto R_N(l,s)$ extends holomorphically to the half-plane
$D_N$ and verifies in it the bound
$R_N(l,s) \ll_{N,\sigma} 1 +|l|^{A(N)}$ uniformly in $l$.

Let us note formally
$$
h_\a(s):= \sum_l \wt a_l\, H_\a(l,s).
$$
Equation (\ref{flsexpzeta}) and $R_N(l,s) \ll_{N,\sigma} 1
+|l|^{A(N)}$ imply that
\begin{equation}
\label{fssimN}
f(s) \sim_N \sum_{\alpha =(\alpha_i)_{i\in I}
  \in \prod_{i\in I}\{0,\dots,N\}^{p_i}} h_\a(s),
\end{equation}
where $\sim_N$ means modulo a holomorphic function in $D_N$.

Recall the decomposition
$\genfrac(){0pt}{1}{1/2}{\alpha} \,P_\a(k,l) \, Q(k)=\sum_{\beta\in L(\a)} M_{\a,\beta}(l) \, k^\beta$ and we
decompose similarly $h_{\a}(s) =\sum_{\beta\in L(\a)} h_{\a,\beta}(s).$\\
Theorem \ref{analytic} now implies that for all 
$\alpha =(\alpha_i)_{i\in I}  \in \prod_{i\in I}\{0,\dots,N\}^{p_i}$ and $\beta\in L(\a)$,

\quad - the map $s\mapsto h_{\a,\beta}(s)$ has a meromorphic
continuation to the whole complex plane $\C$ with only one
simple possible pole at $s=n+ |p|_1 - 2|\a|_1 +d_\beta$,

\quad - the residue at this point is equal to
\begin{equation}
\label{res-halphaj}
\underset {s=n+ |p|_1 - 2|\a|_1 +d_\beta}{\Res}\,
h_{\a,\beta}(s) =
\sum_{l\in \mathcal{Z}} \wt a_l\, M_{\a,\beta}(l) \int_{u\in S^{n-1}}
u^\beta dS(u)
\end{equation}
where $\mathcal{Z}:=\{l\in (\Z)^{n})^{2q} \, : \, \sum_1^{q} l_j =0
\}$.
If the right hand side is zero, $h_{\a,\beta}(s)$ is holomorphic on
$\C$.

By (\ref{fssimN}), we deduce therefore that
$f(s)$ has a meromorphic continuation on the halfplane $D_N$,
with only simple possible poles in the set
$
\set{n+|p|_1 + k  \,: \,-2N|p|_1\leq k \leq d}.
$
Taking now $N\to \infty$ yields the result.

$(ii)$ Let $m\in \N_0$ and set
$I(m):= \set{(\a,\beta)\in \mathcal{I}\times \N_0^{(2q+1)n} \,\,\vert\,\, \beta\in L(\a)$ and  
$m=2|\a|_1 -d_\beta +d } $.
If $(\a,\beta)\in I(m)$, then $|\a|_1 \leq m$ and
$|\beta|_1\leq 3m+d$, so $I(m)$ is finite.

With a chosen $N$ such that $2N|p|_1+d>m$, we get by (\ref{fssimN})
and (\ref{res-halphaj})
$$
\underset{s= n+d+|p|_1 -m}{\Res} \,f(s) =
\sum_{l\in \mathcal{Z}} \wt a_l
\sum_{(\a,\beta)\in I(m)} M_{\a,\beta}(l)  \int_{u\in S^{n-1}}
u^\beta\, dS(u)=C(f,m)
$$
with the convention $\sum_{\emptyset} =0$. Thus, $n+d+|p|_1 - m$ is a
pole of $f$ if and only if
$C(f,m)\neq 0$.
\end{proof}

\newpage
\section{The noncommutative torus}
\label{The noncommutative torus}

The aim of this section is to compute the spectral action of the noncommutative torus. After the basic definitions, 
the result is presented in Theorem \ref{main}. Due to a fundamental appearance of small divisors, the number theory 
is involved via a Diophantine condition. As a consequence, the result which essentially says that the spectral 
action of the noncommutative torus coincide with the action of the ordinary torus (up few constants) is awfully 
technical and use the machinery of Section \ref{Residues of series}. A bunch of proofs are not given, but the 
essential lemmas are here: they show to the reader how life can be hard in noncommutative geometry!

Reference: \cite{EILS}.

\subsection{Definition of the nc-torus}
\label{Definition of the nc-torus}

Let $\Coo(\T^n_\Th)$ be the smooth noncommutative $n$-torus associated to a non-zero skew-symmetric 
deformation matrix $\Th \in M_n(\R)$. It was introduced by Rieffel \cite{RieffelRot} and Connes  \cite{ConnesTorus} 
to generalize the $n$-torus $\T^n$. \\
This means that $\Coo(\T^n_\Th)$ is the algebra generated by $n$
unitaries $u_i$, $i=1,\dots,n$ subject to the relations
\begin{equation}
\label{rel}
u_l\,u_j=e^{i\,\Th_{lj}}\,u_j\,u_l,
\end{equation}
and with Schwartz coefficients: an element $a\in\Coo(\T_\Th^n)$ can be written as
$a=\sum_{k\in\Z^n}a_k\,U_k$, where $\{a_k\}\in\SS(\Z^n)$ with the Weyl elements defined by 
$$
U_k\vc e^{-\frac i2 k.\chi k}\,u_1^{k_1}\cdots u_n^{k_n},
$$ 
$k\in\Z^n$, relation \eqref{rel} reads
\begin{equation}
\label{rel1}
U_{k}U_{q}=e^{-\frac i2 k.\Theta q} \,U_{k+q}, \text{ and } U_{k}U_{q}=e^{-i k.\Theta q} \,U_{q}U_{k}
\end{equation}
where $\chi$ is the matrix restriction of $\Theta$ to its upper triangular part.
Thus unitary operators $U_{k}$ satisfy 
\begin{align*}
U_{k}^*=U_{-k} \text{ and }[U_{k},U_{l}]=-2i\,\sin(\tfrac 12 k.\Th l)\,U_{k+l}.
\end{align*}

Let $\tau$ be the trace on $\Coo(\T^n_\Th)$ defined by 
$$
\tau\big( \sum_{k\in\Z^n}a_k\,U_k \big)\vc a_0
$$
and $\H_{\tau}$ be the GNS Hilbert space obtained by completion of $ \Coo(\T_\Th^n)$
with respect of the norm induced by the scalar product 
$$
\langle a,b\rangle\vc \tau(a^*b).
$$

On $\H_{\tau}=\set{\sum_{k\in\Z^n}a_k\,U_k \,\, \vert \, \, \{a_{k}\}_{k} \in l^2(\Z^n) }$, we consider the left and 
right regular representations of $\Coo(\T_\Th^n)$ by bounded operators, that we denote respectively
by $L(.)$ and $R(.)$.

Let also $\delta_\mu$, $\mu\in \set{1,\dots,n}$, be the $n$ (pairwise commuting) canonical derivations, defined by
\begin{equation}
\delta_\mu(U_k)\vc ik_\mu U_k. \label{dUk}
\end{equation}

We need to fix notations: let 
$$
\A_{\Th}\vc C^{\infty}(\T_{\Th}^n) \text{ acting on } \H\vc \H_{\tau}\otimes \C^{2^m}
$$ 
with $n=2m$ or $n=2m+1$ (i.e., $m=\lfloor \tfrac n2 \rfloor$ is the integer part of 
$\tfrac n2$), the square integrable sections of the trivial spin bundle over $\T^n$.

Each element of $\A_{\Th}$ is represented on $\H$ as $L(a)\otimes1_{2^m}$. The Tomita conjugation 
$$
J_{0}(a)\vc a^*
$$
satisfies $[J_{0},\delta_{\mu}]=0$ and we define 
$$
J\vc J_{0}\otimes C_{0}
$$
where $C_{0}$ is an operator on $\C^{2^m}$. The Dirac-like operator is given by
\begin{align}
\label{defDirac}
\DD\vc -i\,\delta_{\mu}\otimes \gamma^{\mu},
\end{align}
where we use hermitian Dirac matrices $\gamma$. It is defined and symmetric on the dense subset of $\H$ 
given by $C^{\infty}(\T_{\Th}^n) \otimes \C^{2^{m}}$. We still note $\DD$ its selfadjoint extension. This
implies
\begin{align}
    \label{CGamma}
C_{0}\ga^{\alpha}=-\eps \ga^\alpha C_{0},
\end{align}
and
$$
\DD\ U_k \otimes e_i = k_\mu U_k \otimes \gamma^\mu e_i ,
$$
where $(e_i)$ is the canonical basis of $\C^{2^m}$. Moreover, $C_{0}^2=\pm 1_{2^m}$
depending on the parity of $m$. Finally, one introduces the chirality, which in the even case is 
$$
\chi\vc id \otimes (-i)^{m} \gamma^1 \cdots \gamma^{n}.
$$
This yields a spectral triple:
\begin{theorem}
The 5-tuple $(\A_{\Th},\H,\DD,J,\chi)$ is a real regular spectral triple of dimension $n$. It satisfies the finiteness 
and orientability conditions of Definition \ref{defspectraltriple}. It is $n$-summable and its $KO$-dimension is 
also $n$.
\end{theorem}
We do not give a proof since most of its arguments will be emphasized in this section; see however 
\cite{Book,Polaris} for a specific proof.

For instance, we prove in Proposition \ref{zeta(0)} that this triple has simple dimension spectrum when 
$\Theta$ is badly approximable (see Definition \ref{ba}).

The perturbed Dirac operator $V_{u}\,\DD\, V_{u}^*$ by the unitary
$$
V_{u}\vc \big(L(u)\otimes 1_{2^m}\big)J\big(L(u)\otimes 1_{2^m}\big)J^{-1},
$$
defined for every unitary $u \in \A$, $uu^{*}=u^{*}u=U_{0}$, must satisfy condition $J\DD=\epsilon \DD J$ 
(which is equivalent to $\H$ being endowed with a structure of $\A_{\Th}$-bimodule). This
yields the necessity of a symmetrized covariant Dirac operator
$$
\DD_{A}\vc \DD + A + \epsilon J\,A\,J^{-1}$$
since
$V_{u}\,\DD\, V_{u}^{*}=\DD_{L(u)\otimes 1_{2^m}[\DD,L(u^{*}) \otimes 1_{2^m}]}$:
in fact, for $a \in \A_{\Th}$, using $J_{0}L(a){J_{0}}^{-1}=R(a^*)$, we get 
$$
\epsilon J\big(L(a)\otimes \gamma^{\alpha}\big)J^{-1}=-R(a^*)\otimes \gamma^{\alpha}
$$
and that the representation $L$ and the anti-representation $R$ are $\C$-linear, commute and satisfy
$$
[\delta_{\alpha},L(a)]=L(\delta_{\alpha}a),\quad [\delta_{\alpha},R(a)]=R(\delta_{\alpha}a).
$$
This induces some covariance property for the Dirac operator: one checks that for all $k \in \Z^{n}$,
\begin{align}
\label{puregauge1}
L(U_{k})\otimes 1_{2^m}[\DD,L(U_{k}^{*})\otimes 1_{2^m}]&=1\otimes
(-k_{\mu}\ga^{\mu}),
\end{align}
so with (\ref{CGamma}), we get $U_{k}[\DD,U_{k}^{*}]+\epsilon
JU_{k}[\DD,U_{k}^{*}]J^{-1}=0$ and
\begin{align}
\label{covariance}
V_{U_{k}} \,\DD \, V_{U_{k}}^{*}=
\DD=\DD_{L(U_{k})\otimes 1_{2^m}[\DD,L(U_{k}^{*})\otimes 1_{2^m}]}.
\end{align}
Moreover, we get the gauge transformation (see Lemma \ref{fluctuationoffluctuation}):
\begin{align}
\label{gaugeDirac}
V_{u} \DD_{A} V_{u}^{*}= \DD_{\gamma_{u}(A)}
\end{align}
where the gauged transform one-form of $A$ is
\begin{align}
\label{gaugetransform}
\gamma_{u}(A)\vc u[\DD,u^{*}]+uAu^{*},
\end{align}
with the shorthand $L(u)\otimes 1_{2^m} \longrightarrow u$. As a consequence, the spectral action is gauge invariant:
$$
\SS(\DD_{A},f,\Lambda)=\SS(\DD_{\gamma_{u}(A)},f,\Lambda).
$$

An arbitrary selfadjoint one-form $A\in \Omega_\DD^1(\A)$, can be written as
\begin{equation}
\label{connection}
A = L(-iA_{\alpha})\otimes\gamma^{\alpha},\,\, A_{\alpha} =-A_{\alpha}^* \in \A_{\Th},
\end{equation}
thus
\begin{equation}
\label{dirac}
\DD_{A}=-i\,\big(\delta_{\alpha}+L(A_{\alpha})-R(A_{\alpha})\big)
 \otimes \gamma^{\alpha}.
\end{equation}
Defining 
$$
\tilde A_{\alpha}\vc L(A_{\alpha})-R(A_{\alpha}),
$$
we get
$\DD_{A}^2=-g^{{\alpha}_{1} {\alpha}_{2}}(\delta_{{\alpha}_{1}}+\tilde A_{{\alpha}_{1}})(\delta_{{\alpha}_{2}}
+\tilde A_{{\alpha}_{2}})\otimes 1_{2^m} - \tfrac 12 \Omega_{{\alpha}_{1} {\alpha}_{2}} 
\otimes \gamma^{{\alpha}_{1}{\alpha}_{2}}
$
where
\begin{align*}
\gamma^{{\alpha}_{1} {\alpha}_{2}}
&\vc \tfrac 12(\gamma^{{\alpha}_{1}}\gamma^{{\alpha}_{2}} -\gamma^{{\alpha}_{2}}\gamma^{{\alpha}_{1}}) ,\\
\Omega_{{\alpha}_{1} {\alpha}_{2}}
&\vc [\delta_{{\alpha}_{1}}+ \tilde A_{{\alpha}_{1}},\delta_{{\alpha}_{2}} +\tilde A_{{\alpha}_{2}}]\,
=L(F_{{\alpha}_{1} {\alpha}_{2}}) - R(F_{{\alpha}_{1} {\alpha}_{2}})
\end{align*}
with
\begin{align}
\label{Fmunu}
F_{{\alpha}_{1} {\alpha}_{2}}\vc \delta_{{\alpha}_{1}}(A_{{\alpha}_{2}})
-\delta_{{\alpha}_{2}}(A_{{\alpha}_{1}})+[A_{{\alpha}_{1}},A_{{\alpha}_{2}}].
\end{align}
In summary,
\begin{align}
\label{D2}
\DD_{A}^2=-\delta^{{\alpha}_{1} {\alpha}_{2}} \Big( \delta_{{\alpha}_{1}}+L(A_{{\alpha}_{1}})-R(A_{{\alpha}_{1}})\Big)
\Big(\delta_{{\alpha}_{2}}+L(A_{{\alpha}_{2}})-R(A_{{\alpha}_{2}})\Big)
\otimes 1_{2^m} \nonumber\\
-\tfrac 12\,\big(L(F_{{\alpha}_{1} {\alpha}_{2}}) - R(F_{{\alpha}_{1} {\alpha}_{2}})\big)
\otimes \gamma^{{\alpha}_{1} {\alpha}_{2}}.
\end{align}

\subsection{Kernels and dimension spectrum}

We now compute the kernel of the perturbed Dirac operator:

\begin{prop}
\label{noyaux}
(i) $\Ker \DD=U_0\otimes \C^{2^m}$, so $\dim \Ker \DD =2^m$.

(ii) For any selfadjoint one-form $A$, $\Ker \DD \subseteq \Ker \DD_A$.

(iii) For any unitary $ u\in \A$, $\Ker \DD_{\gamma_{u}(A)}=V_{u}\, \Ker \DD_{A}$.
\end{prop}

\begin{proof}
$(i)$ Let $\psi =\sum_{k,j} c_{k,j} \, U_k \otimes e_j \in \Ker \DD$. Thus, 
$0=\DD^2 \psi = \sum_{k,i} c_{k,j} |k|^2\, U_k \otimes e_j$ which entails that $c_{k,j}|k|^2=0$
for any $k \in \Z^n$ and $1\leq j\leq 2^m$. The result follows.

$(ii)$ Let $\psi \in \Ker \DD$. So, $\psi = U_0 \otimes v$ with $v\in \C^{2^m}$ and from (\ref{dirac}), we get
\begin{align*}
\DD_A \psi &= \DD \psi + (A+\epsilon J AJ^{-1})\psi = (A+\epsilon J AJ^{-1})\psi=
-i[A_\a,U_0]\otimes \ga^\a v = 0
\end{align*}
since $U_{0}$ is the unit of the algebra, which proves that $\psi \in \Ker \DD_A$.

$(iii)$ This is a direct consequence of (\ref{gaugeDirac}).
\end{proof}

\begin{corollary}
Let $A$ be a selfadjoint one-form. Then $\Ker \DD_A=\Ker \DD$ in the following cases:

 (i) $A=A_{u}\vc L(u)\otimes 1_{2^m}[\DD,L(u^*)\otimes 1_{2^m}]$ when $u$ is a unitary in $\A$.

 (ii) $\vert \vert A \vert \vert <\tfrac12$.

 (iii) The matrix $\tfrac{1}{2\pi}\Th$ has only integral coefficients.
 \end{corollary}

\begin{proof}
$(i)$ This follows from previous result because $V_{u} (U_{0}\otimes v)= U_{0} \otimes v$ for any $v\in \C^{2^m}$.

$(ii)$ Let $\psi=\sum_{k,j}c_{k,j}\, U_{k}\otimes e_{j} $ be in $ \Ker \DD_{A}$ (so $\sum_{k,j} \vert c_{k,j}\vert^2<
\infty$) and $\phi\vc \sum_{j}c_{0,j}\, U_{0}\otimes e_{j}$. Thus
$\psi'\vc \psi-\phi \in \text{Ker }\DD_{A}$ since $\phi \in \Ker \DD \subseteq \Ker \DD_A$ and
$$
\vert \vert \sum_{0\neq k \in \Z^n,\,j} c_{k,j}\,k_{\alpha}\,U_{k}\otimes
\gamma^{\alpha}e_{j}\vert \vert^2=\vert \vert \DD \psi'\vert\vert^2=\vert \vert -(A + \epsilon
JAJ^{-1})\psi'\vert \vert^2 \leq 4\vert\vert A \vert \vert^{2}\vert \vert \psi' \vert \vert^{2} <\vert \vert \psi' \vert \vert^{2}.
$$
Defining $X_{k}\vc \sum_{\alpha}k_{\alpha}\gamma_{\alpha}$,
$X_{k}^{2}=\sum_{\alpha}\vert k_{\alpha}\vert^{2}\, 1_{2^{m}}$ is invertible and the vectors 
$\set{U_{k}\otimes X_{k}e_{j}}_{0\neq k\in \Z^{n},\,j}$ are orthogonal in $\H$, so
$$
\sum_{0\neq k\in \Z^{n},\,j}\big( \sum_{\alpha} \vert k_{\alpha}\vert^{2} \big)\, \vert
c_{k,j}\vert^{2} < \sum_{0\neq k\in \Z^{n},\,j}\vert c_{k,j}\vert^{2}
$$
which is possible only if $c_{k,j}=0, \, \forall k,\,j$ that is $\psi'=0$ and $\psi=\phi \in \text{Ker }\DD$.

$(iii)$ This is a consequence of the fact that the algebra is commutative, thus the arguments of \eqref{JAJ} 
apply and $\wt A=0$.
\end{proof}

Note that if $\wt A_{u}\vc A_{u}+\epsilon JA_{u}J^{-1}$, then by (\ref{puregauge1}), $\wt A_{U_{k}}=0$ for all
$k \in \Z^n$ and $\norm{A_{U_{k}}}=\vert k\vert$, but for an arbitrary unitary $u\in \A$, $\wt A_{u}\ne 0$ so
$\DD_{A_{u}}\ne \DD$.

Naturally the above result is also a direct consequence of the fact that the eigenspace of an isolated eigenvalue 
of an operator is not modified by small perturbations. However, it is interesting to compute the last result directly 
to emphasize the difficulty of the general case:

Let $\psi=\sum_{l\in \Z^n, 1\leq j \leq 2^m}c_{l,j}\, U_{l}\otimes
e_{j}\in \Ker
\DD_A$, so $\sum_{l\in \Z^n, 1\leq j \leq 2^m} \vert c_{l,j}\vert^2< \infty$. We
have to show that $\psi\in$ Ker $\DD$ that is $c_{l,j}=0$ when $l\ne
0$.

Taking the scalar product of $\langle U_{k} \otimes e_{i}\vert$ with
$$
0=\DD_{A}\psi=\sum_{l,\,\a,\,j} c_{l,\,j}(l^{\a}U_{l}-i[A_{\a},U_{l}] )\otimes \gamma^{\a}e_{j},
$$
we obtain
$$
0=\sum_{l,\,\a,\,j} c_{l,\,j} \big(l^{\a}\delta_{k,l}-i\langle U_{k},[A_{\a},U_{l}]\rangle \big)\langle e_{i},\gamma^{\a}e_{j}
\rangle.
$$
If $A_{\a}=\sum_{\a,l}a_{\a,l}\, U_{l} \otimes \gamma^{\a}$ with 
$\set{a_{\a,l}}_{l} \in \SS(\Z^n)$, note that $[U_{l},U_{m}]=-2i \sin(\tfrac 12 l.\Th m) \, U_{l+m}$ and
$$
\langle U_{k},[A_{\a},U_{l}]\rangle = \sum_{l'\in \Z^{n}}a_{\a,l'}(-2i \sin (\tfrac 12 l'.\Th l) \langle U_{k}, U_{l'+l}
\rangle=-2i\, a_{\a,k-l} \,\sin(\tfrac 12 k.\Th l).
$$
Thus
\begin{align}
\label{contraintenoyau}
0=\sum_{l\in \Z^{n}}\sum_{\a=1}^{n}\sum_{j=1}^{2^{m}} c_{l,\,j}
\big(l^{\a}\delta_{k,l} -2a_{\a,k-l} \, \sin(\tfrac 12 k.\Th l) \big)
\, \langle
e_{i},\gamma^{\a}e_{j} \rangle, \quad \forall k\in \Z^n, \, \forall
i, 1\leq i \leq
2^{m}.
\end{align}

\medskip

{\it We conjecture that $\Ker \DD=\Ker \DD_A$ at least for generic $\Th$'s}: the constraints (\ref{contraintenoyau}) 
should imply $c_{l,j} = 0$
for all $j$ and all $l \neq 0$ meaning $\psi \in \Ker \DD$. When
$\tfrac{1}{2\pi}\Th$ has only integer coefficients, the sin part of
these constraints disappears giving the result.
\medskip

We will use freely the notation \eqref{DAdroit} about the difference between $\DD$ and $D$.
\begin{lemma}
\label{spectrumset}
If $\tfrac{1}{2\pi} \Th$ is badly approximable (see Definition \ref{ba}),
$Sp\big(\Coo(\T^n_\Th),\H,\DD\big)=\Z$ and all these poles are simple.
\end{lemma}

\begin{proof}
Let $B\in \DD(\A)$ and $p\in \N_0$. Suppose that $B$ is of the form
$$
B= a_r b_r
\DD^{q_{r-1}}|\DD|^{p_{r-1}} a_{r-1}b_{r-1}\cdots
\DD^{q_1}|\DD|^{p_1} a_1 b_1
$$
where $r\in \N$, $a_i \in \A$, $b_i\in J\A J^{-1}$, $q_i, p_i \in
\N_0$.
We note $a_i=:\sum_l a_{i,l}\,U_l$ and
$b_i=:\sum_i b_{i,l} \,U_l$. With the shorthand
$k_{\mu_1,\mu_{q_i}}\vc k_{\mu_1}\cdots
k_{\mu_{q_i}}$ and $\ga^{\mu_1,\mu_{q_i}}=\ga^{\mu_1}\cdots
\ga^{\mu_{q_i}}$, we get
$$
\DD^{q_1}|\DD|^{p_1}  a_1 b_1 \, U_k \otimes e_j =  \sum_{l_1,\,l'_1}
a_{1,l_1} b_{1,l'_1}
U_{l_1}U_k U_{l'_1}
\,|k+l_1+l'_1|^{p_1}\,(k+l_1+l'_1)_{\mu_1,\mu_{q_1}} \otimes
\ga^{\mu_1,\mu_{q_1}} e_j
$$
which gives, after $r$ iterations,
$$
B U_k \otimes e_j = \sum_{l,l'} \wt a_{l} \wt b_{l} U_{l_r}\cdots
U_{l_1} U_k
U_{l'_1}\cdots U_{l'_r} \prod_{i=1}^{r-1} |k+\wh l_i+\wh
l'_i|^{p_i}(k+\wh l_{i} +\wh
l'_{i})_{\mu^{i}_1,\mu^i_{q_i}} \otimes
\ga^{\mu^{r-1}_1,\mu^{r-1}_{q_{r-1}}}\cdots
\ga^{\mu^1_1,\mu^1_{q_1}} e_j
$$
where $\wt a_l : = a_{1,l_1}\cdots a_{r,l_r}$ and $\wt b_{l'} : =
b_{1,l'_1}\cdots
b_{r,l'_r}$.

Let us note $F_\mu(k,l,l')\vc \prod_{i=1}^{r-1}|k+\wh l_i+\wh
l'_i|^{p_i}
(k+\wh l_{i} +\wh l'_{i})_{\mu^{i}_1,\mu^i_{q_i}}$ and $\ga^\mu
\vc \ga^{\mu^{r-1}_1,\mu^{r-1}_{q_{r-1}}}\cdots
\ga^{\mu^1_1,\mu^1_{q_1}}$. Thus, with
the shortcut 
$$
\sim_c \text{ meaning modulo a constant function towards the variable } s,
$$
$$
\Tr \big(B|D|^{-2p-s}\big) \sim_c {\sum_k}' \, \sum_{l,l'} \wt a_l
\wt b_{l'} \,
\tau\big(U_{-k}U_{l_r}\cdots U_{l_1} U_k U_{l'_1}\cdots U_{l'_r}\big)
\tfrac{F_\mu(k,l,l')}{|k|^{s+2p}} \Tr (\ga^\mu)\, .
$$
Since $U_{l_r}\cdots U_{l_1} U_k = U_k U_{l_r}\cdots U_{l_1}
e^{-i\sum_1^r l_i .\Th
k}$ we get
$$\tau\big(U_{-k}U_{l_r}\cdots U_{l_1} U_k U_{l'_1}\cdots
U_{l'_r}\big)=
\delta_{\sum_1^r l_i+l'_i,0} \, e^{i\phi(l,l')} \, e^{-i\sum_1^r
l_i.\Th k}$$ where $\phi$
is a real valued function. Thus,
\begin{align*}
\Tr \big(B |D|^{-2p-s} \big)&\sim_c {\sum_k}' \, \sum_{l,l'}
e^{i\phi(l,l')}\,\delta_{\sum_1^r l_i+l'_i,0}\, \wt a_l \wt b_{l'}\,
\tfrac{F_\mu(k,l,l')\,e^{-i\sum_1^r l_i.\Th k}}{|k|^{s+2p}} \Tr
(\ga^\mu) \\
&\sim_c f_\mu(s)\Tr (\ga^\mu).
\end{align*}
The function $f_\mu(s)$ can be decomposed as a linear combination of zeta function
of type described in Theorem \ref{zetageneral} (or, if $r=1$ or all the $p_i$ are zero,
in Theorem \ref{analytic}). Thus, $s\mapsto \Tr \big(B |D|^{-2p-s}\big)$ has only poles in $\Z$ and each 
pole is simple. Finally, by linearity, we get the result.
\end{proof}
The dimension spectrum of the noncommutative torus is simple:

\begin{prop}
\label{zeta(0)}

(i) If $\tfrac{1}{2\pi} \Th$ is badly approximable, the spectrum dimension of the spectral triple 
$\big(\Coo(\T^n_\Th),\H,\DD\big)$ is equal to the set $\set{n-k \, :\,  k\in \N_0}$ and all these poles
are simple.

(ii) $\zeta_D(0)=0.$
\end{prop}

\begin{proof}
$(i)$ Lemma \ref{spectrumset} and Remark \ref{remark-spectrum}.

$(ii)$  $\zeta_D(s)={\sum}_{k\in \Z^n}
\sum_{1\leq j\leq 2^m} \<
U_k\otimes e_j, |D|^{-s}U_k\otimes e_{j}>=2^m( {\sum}'_{k\in\Z^n}
\frac{1}{|k|^{s}} + 1) =2^m(\,Z_n(s)+1).$ By (\ref{Zn0}), we get the
result.
\end{proof}

We have computed $\zeta_D(0)$ relatively easily but the main difficulty of the present Section is precisely to 
calculate $\zeta_{D_A}(0)$.

\subsection{Noncommutative integral computations}

We fix a self-adjoint one-form $A$ on the noncommutative torus of dimension $n$.

\begin{prop}
\label{invariance}
If $\tfrac{1}{2\pi}\Th$ is badly approximable, then the first elements of the spectral action 
expansion (\ref{asympspectral}) are given by
\begin{align*}
&\ncint {|D_{A}|}^{-n}\,=\ncint |D|^{-n}= 2^{m+1}\pi^{n/2}\,\Gamma(\tfrac{n}{2})^{-1}.\\
&\ncint \vert D_{A}\vert^{-n+k}=0 \text{ for k odd}.\\
&\ncint \vert D_{A}\vert^{-n+2}=0.
\end{align*}
\end{prop}
We need a few technical lemmas:

\begin{lemma}
\label{traceAD}
On the noncommutative torus, for any $t\in \R$,
$$
\ncint \wt A \DD |D|^{-t}= \ncint \DD \wt A |D|^{-t} =0.
$$
\end{lemma}

\begin{proof}
Using notations of (\ref{connection}), we have
\begin{align*}
\Tr (\wt A \DD |D|^{-s})&\sim_c {\sum}_j {\sum}'_k \langle U_k\otimes e_j,-i k_\mu|k|^{-s}
[A_\a,U_k] \otimes \ga^\a \ga^\mu e_j \rangle \\
&\sim_c -i\Tr(\ga^\a\ga^\mu) \, {\sum}'_k k_\mu |k|^{-s} \langle U_k,[A_\a,U_k] \rangle=0
\end{align*}
since $\langle U_k,[A_\a,U_k] \rangle = 0$. Similarly
\begin{align*}
\Tr ( \DD \wt A  |D|^{-s})&\sim_c {\sum}_j {\sum}'_k \langle U_k\otimes
e_j,|k|^{-s}{\sum}_l a_{\a,l}\,2 \sin \tfrac{k. \Th l}{2} (l+k)_\mu U_{l+k} \otimes \ga^\mu \ga^\a e_j \rangle\\
&\sim_c 2\Tr(\ga^\mu \ga^\a){\sum}'_k {\sum}_l a_{\a,l}\sin \tfrac{k. \Th l}{2}\,(l+k)_\mu \,|k|^{-s}\langle U_k,U_{l+k} 
\rangle =0.
\end{align*}
\end{proof}

Any element $h$ in the algebra generated by $\A$ and $[\DD,\A]$ can be written as a
linear combination of terms of the form ${a_1}^{p_1}\cdots {a_n}^{p_r}$ where
$a_i$ are elements of $\A$ or $[\DD,\A]$. Such a term can be written
as a series $b\vc \sum a_{1,\a_1,l_1}\cdots a_{q,\a_q,l_q} U_{l_1}\cdots
U_{l_q} \otimes \ga^{\a_1}\cdots \ga^{\a_q}$ where $a_{i,\a_i}$ are Schwartz sequences and when 
$a_i=:\sum_l a_l U_l \in \A$, we set $a_{i,\a,l}=a_{i,l}$ with $\ga^\a =1$. We define
$$
L(b)\vc  \tau \big({\sum}_l a_{1,\a_1,l_1}\cdots a_{q,\a_q,l_q} U_{l_1}
\cdots U_{l_q}\big) \Tr (\ga^{\a_1}\cdots \ga^{\a_q}).
$$
By linearity, $L$ is defined as a linear form on the whole algebra generated by $\A$ and $[\DD,\A]$.

\begin{lemma}
\label{tracehD}
If $h$ is an element of the algebra generated by $\A$ and $[\DD,\A]$,
$$
\Tr \big(h |D|^{-s}\big) \sim_c L(h)\,  Z_n(s).
$$
In particular, $\Tr \big(h |D|^{-s}\big)$ has at most one pole at $s=n$.
\end{lemma}

\begin{proof} We get with $b$ of the form $\sum a_{1,\a_1,l_1}\cdots a_{q,\a_q,l_q} U_{l_1}\cdots U_{l_q} \otimes
\ga^{\a_1}\cdots \ga^{\a_q}$,
\begin{align*}
\Tr\big(b|D|^{-s}\big)&\sim_c {\sum_{k\in\Z^n}}' \langle  U_k, \sum_l a_{1,\a_1,l_1}\cdots a_{q,\a_q,l_q} U_{l_1}
\cdots U_{l_q}U_k \rangle \, \Tr (\ga^{\a_1}\cdots \ga^{\a_q})|k|^{-s} \\
&\sim_c \tau(\sum_l a_{1,\a_1,l_1}\cdots a_{q,\a_q,l_q} U_{l_1}\cdots U_{l_q})\Tr(\ga^{\a_1}\cdots \ga^{\a_q})\, 
Z_n(s) =L(b) \,Z_n(s).
\end{align*}
The results follows now from linearity of the trace.
\end{proof}

\begin{lemma}
\label{traceJAJA}
If $\tfrac{1}{2\pi}\Th$ is badly approximable, the function $s\mapsto\Tr \big( \eps JAJ^{-1} A |D|^{-s} \big)$ extends 
meromorphically on the whole plane with only one possible pole at $s=n$. Moreover, this pole is simple and
$$
\underset{s=n}{\Res}\, \Tr \big(\eps JAJ^{-1} A |D|^{-s}\big) = a_{\a,0}\,a^\a_{0}\ 2^{m+1}\pi^{n/2}\,\Ga(n/2)^{-1}.
$$
\end{lemma}

\begin{proof}
With $A=L(-i A_\a)\otimes \ga^\a$, we get $\epsilon J A J^{-1}=R(i A_\a)\otimes \ga^\a$, and by multiplication
$\eps JAJ^{-1} A=R(A_\beta)
L(A_\a)\otimes \ga^{\beta}\ga^\a$. Thus,
\begin{align*}
\Tr\big(\eps JAJ^{-1} A |D|^{-s}\big)&\sim_c {\sum_{k\in\Z^n}}' \langle U_k,A_\a U_k A_\beta \rangle \,|k|^{-s}
\Tr (\ga^{\beta}\ga^{\a}) \\
&\sim_c{\sum_{k\in\Z^n}}' \,\sum_{l} a_{\a,l}\,a_{\beta,-l}\,e^{ik.\Th l}\, |k|^{-s}\Tr (\ga^{\beta}\ga^{\a})\\
&\sim_c 2^m {\sum_{k\in\Z^n}}' \, \sum_{l} a_{\a,l}\,a^\a_{-l}\,e^{ik.\Th l}\, |k|^{-s}.
\end{align*}
Theorem \ref{analytic} $(ii)$ entails that ${\sum}'_{k\in\Z^n} \, \sum_{l} a_{\a,l} \,a^{\a}_{-l}\,e^{ik.\Th l}\, |k|^{-s}$ 
extends meromorphically on the whole plane $\C$ with only one possible pole at $s=n$.
Moreover, this pole is simple and we have
$$
\underset{s=n}{\Res}\, {\sum_{k\in\Z^n}}' \,\sum_{l} a_{\a,l}
\,a^\a_{-l}\,e^{ik.\Th l}\,
|k|^{-s} =  a_{\a,0}\,a^\a_{0} \, \underset{s=n}{\Res}\, Z_n(s).
$$
Equation (\ref{formule}) now gives the result.
\end{proof}

\begin{lemma}
\label{traceXD}
If $\tfrac{1}{2\pi}\Th$ is badly approximable, then for any $t\in \R$,
$$
\ncint X|D|^{-t} = \delta_{t,n}\, 2^{m+1}\big(-\sum_l a_{\a,l}\,a^\a_{-l}+ \,a_{\a,0}\,a^\a_{0}\big)\ 2\pi^{n/2}\,
\Ga(n/2)^{-1}  .
$$
where $X=\wt A\DD + \DD \wt A + {\wt A}^2$ and $A=:-i\sum_{l} a_{\a,l}\,U_l\otimes \ga^\a$.
\end{lemma}

\begin{proof} 
By Lemma \ref{traceAD}, we get $\ncint X|D|^{-t}=\Res_{s=0} \Tr({\wt A}^2 |D|^{-s-t})$. Since
$A$ and $\eps JAJ^{-1}$ commute, we have $\wt A ^2 = A^2 + JA^2J^{-1} + 2\eps JAJ^{-1}A$.
Thus,
$$
\Tr({\wt A}^2 |D|^{-s-t})=\Tr( A^2 |D|^{-s-t})+\Tr( JA^2J^{-1} |D|^{-s-t})+2\Tr ( \eps JAJ^{-1}A |D|^{-s-t}).
$$
Since $|D|$ and $J$ commute, we have with Lemma \ref{tracehD},
$$
 \Tr \big({\wt A}^2 |D|^{-s-t}\big)\sim_c 2L (A^2) \, Z_n(s+t) + 2 \Tr \big(\eps JAJ^{-1}A |D|^{-s-t}\big).
$$
Thus Lemma \ref{traceJAJA} entails that $\Tr({\wt A}^2 |D|^{-s-t})$ is holomorphic at 0 if $t\neq n$. When $t=n$,
\begin{equation}
\label{TrA^2}
\underset{s=0}{\Res}\, \Tr\big({\wt A}^2 |D|^{-s-t}\big) = 2^{m+1}\big(-\sum_l a_{\a,l} \, a^\a_{-l}+
\,a_{\a,0}\,a^\a_{0}\ \big)\, 2\pi^{n/2}\,\Ga(n/2)^{-1},
\end{equation}
which gives the result.
\end{proof}

\begin{lemma}
\label{traceAA}
If $\tfrac{1}{2\pi}\Th$ is badly approximable, then
$$
\ncint \wt A \DD \wt A \DD |D|^{-2-n}=-\tfrac{n-2}{n} \ncint \wt A^2 |D|^{-n}.
$$
\end{lemma}

\begin{proof}
With $\DD J = \eps J \DD$, we get
$$
\ncint \wt A \DD \wt A \DD |D|^{-2-n} = 2 \ncint  A \DD  A \DD |D|^{-2-n} + 2 \ncint \eps JAJ^{-1} \DD A \DD |D|^{-2-n}.
$$
Let us first compute $\ncint  A \DD  A \DD |D|^{-2-n}$. We have, with 
$A=:-i L(A_\a)\otimes \ga^\a=: -i\sum_{l} a_{\a,l} U_l \otimes \ga^\a$,
$$
\Tr \big(A\DD A \DD |D|^{-s-2-n}\big) \sim_c -{{\sum_{k}}}' \sum_{l_1,l_2} a_{\a_2,l_2}\,a_{\a_1,l_1}
\,\tau(U_{-k} U_{l_2} U_{l_1} U_k) \, \tfrac{k_{\mu_1}(k+l_1)_{\mu_2}}{|k|^{s+2+n}} \Tr(\ga^{\a,\mu})
$$
where $\ga^{\a,\mu}\vc  \ga^{\a_2}\ga^{\mu_2}  \ga^{\a_1}\ga^{\mu_1}$.
Thus,
$$
\ncint  A \DD  A \DD |D|^{-2-n} = -\sum_{l} a_{\a_2,-l}\,a_{\a_1,l} \, \underset{s=0}{\Res}\,\big({{\sum_{k}}}'
\tfrac{k_{\mu_1}k_{\mu_2}}{|k|^{s+2+n}} \big) \Tr(\ga^{\a,\mu}).
$$
We have also, with $\eps JAJ^{-1} = iR(A_\a)\otimes \ga^{a}$,
$$
\Tr \big(\eps JAJ^{-1}\DD A \DD |D|^{-s-2-n}\big) \sim_c {{\sum_{k}}}' \sum_{l_1,l_2}
a_{\a_2,l_2}a_{\a_1,l_1} \tau(U_{-k} U_{l_1} U_k U_{l_2}) \tfrac{k_{\mu_1}(k+l_1)_{\mu_2}}{|k|^{s+2+n}} 
\Tr(\ga^{\a,\mu}).
$$
which gives
$$
\ncint  \eps JAJ^{-1} \DD  A \DD |D|^{-2-n} =a_{\a_2,0}a_{\a_1,0} \, \underset{s=0}{\Res}\,\big({{\sum_{k}}}'
\tfrac{k_{\mu_1}k_{\mu_2}}{|k|^{s+2+n}} \big) \Tr(\ga^{\a,\mu}).
$$
Thus,
$$
\half \ncint \wt A \DD \wt A \DD |D|^{-2-n} = \big(a_{\a_2,0}a_{\a_1,0}-\sum_{l}
a_{\a_2,-l}a_{\a_1,l}\big) \Res_{s=0}\big({{\sum_{k}}}' \tfrac{k_{\mu_1}k_{\mu_2}}{|k|^{s+2+n}} \big) \Tr(\ga^{\a,\mu}).
$$
With ${\sum}'_k \tfrac{k_{\mu_1}k_{\mu_2}}{|k|^{s+2+n}} = \tfrac{\delta_{\mu_1\mu_2}}{n}Z_n(s+n)$ and
$C_n\vc \Res_{s=0} Z_n(s+n) = 2\pi^{n/2} \Ga(n/2)^{-1}$ we obtain
$$
\half \ncint \wt A \DD \wt A \DD |D|^{-2-n} = \big(a_{\a_2,0}a_{\a_1,0}-\sum_{l} a_{\a_2,-l}a_{\a_1,l}\big)\tfrac{C_n}{n}
\Tr(\ga^{\a_2}\ga^{\mu}\ga^{\a_1}\ga_\mu).
$$
Since  $\Tr(\ga^{\a_2}\ga^{\mu}\ga^{\a_1}\ga_\mu)= 2^m(2-n)\delta^{\a_2,\a_1}$,
we get
$$
\half \ncint \wt A \DD \wt A \DD |D|^{-2-n} = 2^m\big(-a_{\a,0}\,a^\a_0+\sum_{l} a_{\a,-l}\,a^\a_l\big)\tfrac{C_n (n-2)}{n}.
$$
Equation (\ref{TrA^2}) now proves the lemma.
\end{proof}

\begin{lemma}
\label{ncint-odd-pdo}
If $\tfrac{1}{2\pi}\Th$ is badly approximable, then for any $P\in \Psi_{1}(\A)$ and $q \in \N$, $q$ odd,
$$
\ncint P |D|^{-(n-q)} = 0.
$$
\end{lemma}

\begin{proof}
There exist $B\in \DD_{1}(\A)$ and $p\in \N_0$ such that $P= BD^{-2p}+R$ where $R$ is in $OP^{-q-1}$.
As a consequence, $\ncint P |D|^{-(n-q)} = \ncint B|D|^{-n-2p+q}$. Assume 
$B= a_r b_r\DD^{q_{r-1}}a_{r-1}b_{r-1}\cdots \DD^{q_1} a_1 b_1 $ where
$r\in \N$, $a_i \in \A$, $b_i\in J\A J^{-1}$, $q_i\in \N$. If we prove that $\ncint B|D|^{-n-2p+q} =0$, then
the general case will follow by linearity. We note $a_i=:\sum_l a_{i,l}\,U_l$ and
$b_i=:\sum_l b_{i,l} \,U_l$. With the shorthand $k_{\mu_1,\mu_{q_i}}\vc k_{\mu_1}\cdots
k_{\mu_{q_i}}$ and $\ga^{\mu_1,\mu_{q_i}}=\ga^{\mu_1}\cdots \ga^{\mu_{q_i}}$, we get
$$
\DD^{q_1}  a_1 b_1 U_k \otimes e_j = \sum_{l_1,l'_1} \,a_{1,l_1}\, b_{1,l'_1}\, U_{l_1}U_k
U_{l'_1} \,(k+l_1+l'_1)_{\mu_1,\mu_{q_1}} \otimes \ga^{\mu_1,\mu_{q_1}} e_j
$$
which gives, after iteration,
$$
B\, U_k \otimes e_j = \sum_{l,l'} \wt a_{l} \wt b_{l} U_{l_r} \cdots U_{l_1} U_k
U_{l'_1}\cdots U_{l'_r} \prod_{i=1}^{r-1} (k+\wh l_{i} +\wh l'_{i})_{\mu^{i}_1,\mu^i_{q_i}} \otimes \ga^{\mu^{r-1}_1,
\mu^{r-1}_{q_{r-1}}}\cdots \ga^{\mu^1_1,\mu^1_{q_1}} e_j
$$
where $\wt a_l : = a_{1,l_1}\cdots a_{r,l_r}$ and$\wt b_{l'} : = b_{1,l'_1}\cdots b_{r,l'_r}$. Let's note 
$Q_\mu(k,l,l')\vc \prod_{i=1}^{r-1} (k+\wh l_{i} +\wh l'_{i})_{\mu^{i}_1,\mu^i_{q_i}}$ and
$\ga^\mu \vc \ga^{\mu^{r-1}_1,\mu^{r-1}_{q_{r-1}}}\cdots \ga^{\mu^1_1,\mu^1_{q_1}}$. Thus,
$$
\ncint B\,|D|^{-n-2p+q} =\underset{s=0}{\Res}\,  {\sum_k}' \, \sum_{l,l'} \wt a_l
\,\wt b_{l'} \, \tau\big(U_{-k}U_{l_r}\cdots U_{l_1} U_k U_{l'_1}\cdots U_{l'_r}\big)
\,\tfrac{Q_\mu(k,l,l')}{|k|^{s+2p+n-q}} \,\Tr (\ga^\mu)\, .
$$
Since $U_{l_r}\cdots U_{l_1} U_k = U_k U_{l_r} \cdots U_{l_1} e^{-i\sum_1^r l_i .\Th k}$, we get
$$
\tau\big(U_{-k}U_{l_r}\cdots U_{l_1} U_k U_{l'_1}\cdots U_{l'_r}\big)= \delta_{\sum_1^r l_i+l'_i,0} \,e^{i\phi(l,l')} \,
e^{-i\sum_1^r l_i.\Th k}$$ where $\phi$ is a real valued function. Thus,
\begin{align*}
\ncint B\,|D|^{-n-2p+q} &=\underset{s=0}{\Res}\, {\sum_k}'\,\sum_{l,l'}
e^{i\phi(l,l')}\,\delta_{\sum_1^r l_i+l'_i,0}\, \wt a_l \,\wt b_{l'} \,\tfrac{Q_\mu(k,l,l')e^{-i\sum_1^r l_i.\Th k}}{|k|^{s+2p+n-q}}
\Tr (\ga^\mu) \\
&=:\underset{s=0}{\Res}\, f_\mu(s)\Tr (\ga^\mu).
\end{align*}

We decompose $Q_{\mu}(k,l,l')$ as a sum $\sum_{h=0}^r M_{h,\mu}(l,l') \, Q_{h,\mu}(k)$
where $Q_{h,\mu}$ is a homogeneous polynomial in $(k_1,\cdots,k_n)$
and $M_{h,\mu}(l,l')$ is a polynomial in $\big((l_1)_1,\cdots,(l_{r})_n,(l'_1)_1,\cdots,(l'_{r})_n \big)$.

Similarly, we decompose $f_{\mu}(s)$ as $\sum_{h=0}^{r} f_{h,\mu}(s)$. Theorem \ref{analytic} $(ii)$ entails that 
$ f_{h,\mu}(s)$ extends meromorphically to the whole complex plane $\C$ with only one
possible pole for $s+2p+n-q=n+d$ where $d\vc \text{deg } Q_{h,\mu}$. In other words, if $d+q-2p\neq 0$,
$f_{h,\mu}(s)$ is holomorphic at $s=0$. Suppose now $d+q-2p =0$
(note that this implies that $d$ is odd, since $q$ is odd by hypothesis), then, by Theorem \ref{analytic} $(ii)$
$$
\underset{s=0}{\Res}\ f_{h,\mu}(s) = V \int_{u\in S^{n-1}} Q_{h,\mu}(u)\,dS(u)
$$
where $V\vc \sum_{l,l'\in Z} M_{h,\mu}(l,l')\,e^{i \phi(l,l')}\, \delta_{\sum_1^r l_i+l'_i,0}\, \wt a_{l} \,\wt b_{l'}$ and
$Z\vc \set{l,l' \, :\, \sum_{i=1}^{r} l_i=0}$. Since $d$ is odd, $Q_{h,\mu}(-u)=-Q_{h,\mu}(u)$ and 
$\int_{u\in S^{n-1}}Q_{h,\mu}(u)\,dS(u)=0$. Thus, $\underset{s=0}{\Res} \ f_{h,\mu}(s)=0$ in any case, which 
gives the result.
\end{proof}

As we have seen, the crucial point of the preceding lemma is the decomposition of the numerator of the series 
$f_\mu(s)$ as polynomials in $k$. This has been possible because we restricted our pseudodifferential
operators to $\Psi_1(\A)$.

\bigskip

\begin{proof}[Proof of Proposition \ref{invariance}]
The top element follows from Proposition \ref{invariance1} and according to (\ref{formule}),
\begin{align*}
\ncint |D| ^{-n}= \underset{s=0}{\Res}\ \Tr\big(|D|^{-s-n}\big)=2^m\,\underset{s=0}{\Res}\,
Z_n(s+n)=\tfrac{2^{m+1}\pi^{n/2}}{\Gamma(n/2)}\, .
\end{align*}

For the second equality, we get from Lemmas \ref{tracehD} and \ref{Res-zeta-n-k}
$$
\underset{s=n-k}{\Res}\,  \zeta_{D_A}(s)= \sum_{p=1}^k \sum_{r_1,\cdots, r_p =0}^{k-p} h(n-k,r,p)
\ncint \eps^{r_1}(Y)\cdots\eps^{r_p}(Y) |D|^{-(n-k)}.
$$
Corollary  \ref{eps-pdo} and Lemma \ref{ncint-odd-pdo} imply that
$\ncint \eps^{r_1}(Y)\cdots\eps^{r_p}(Y) |D|^{-(n-k)} =0$, which gives the result.

Last equality follows from Lemma \ref{traceAA} and Corollary \ref{res-n-2-A}.
\end{proof}

\subsection{The spectral action}

Here is the main result of this section.

\begin{theorem}
\label{main}
Consider the noncommutative torus $\big(\Coo(\T^n_\Th),\H,\DD\big)$ of dimension $n\in \N$ 
where $\tfrac{1}{2\pi}\Th$ is a real $n\times n$ real skew-symmetric badly approximable matrix, and a 
selfadjoint one-form
$A=L(-iA_{\alpha})\otimes \ga^{\alpha}$. Then, the full spectral action of $\DD_{A}=\DD +A +\epsilon JAJ^{-1}$ is
\\
(i) for $n=2$,
$$
\SS(\DD_{A},f,\Lambda)=4\pi\,f_{2} \, \Lambda^{2} + \mathcal{O}(\Lambda^{-2}),
$$
(ii) for $n=4$,
$$
\SS(\DD_{A},f,\Lambda)= 8\pi^2\,f_{4} \, \Lambda^{4} -\tfrac{4\pi^{2}}{3}\
f(0) \,\tau(F_{\mu\nu}F^{\mu\nu})+  \mathcal{O}(\Lambda^{-2}),
$$
(iii) More generally, in
$$
\SS(\DD_{A},f,\Lambda) \, = \,\sum_{k=0}^n f_{n-k}\, c_{n-k}(A) \,\Lambda^{n-k} +\mathcal{O}(\Lambda^{-1}),
$$
$c_{n-2}(A)=0$, $c_{n-k}(A)=0$ for $k$ odd. In particular, $c_0(A)=0$ when $n$ is odd.
\end{theorem}

\quad

This result (for $n=4$) has also been obtained in \cite{GIVas} using the heat kernel method. It is however 
interesting to get the result via direct computations of (\ref{asympspectral}) since it shows how this formula is efficient.
As we will see, the computation of all the noncommutative integrals require a lot of technical steps.
One of the main points, namely to isolate where the Diophantine condition on $\Th$ is assumed, is outlined here.

\begin{remark}
Note that all terms must be gauge invariants, namely, according
to (\ref{gaugetransform}), invariant by
$A_{\alpha}\longrightarrow \gamma_{u}(A_{\alpha})=
uA_{\alpha}u^{*}+u\delta_{\alpha}(u^{*})$. A particular case is
$u=U_{k}$ where
$U_{k}\delta_{\alpha}(U_{k}^{*})=-ik_{\alpha} U_{0}$.

In the same way, note that there is no contradiction with
the commutative case where, for any selfadjoint one-form
$A$, $\DD_{A}=\DD$ (so $A$ is equivalent to 0!), since
we assume in Theorem \ref{main} that
$\Th$ is badly approximable, so $\A$ cannot be commutative.
\end{remark}

\begin{conjecture}
The constant term of the
spectral action of $\DD_{A}$ on the noncommutative n-torus is
proportional to the constant term of the spectral action
of $\DD+A$ on the commutative n-torus.
\end{conjecture}

\begin{remark}
The appearance of a Diophantine condition for $\Th$ has been
characterized in dimension 2 by Connes \cite[Prop. 49]{NCDG} where in
this case, $\Th=\th\genfrac{(}{)}{0pt}{1}{\,0 \quad1}{-1\,\,\,\, 0 }$
with $\th \in \R$. In fact, the Hochschild cohomology
$H(\A_{\Th},{\A_{\Th}}^*)$ satisfies dim
$H^j(\A_{\Th},{\A_{\Th}}^*)=2$ (or $1$) for $j=1$ (or $j=2$) if and
only if the irrational number $\th$ satisfies a Diophantine condition
like $\vert 1-e^{i2\pi n \th} \vert^{-1} =\mathcal{O}(n^k)$ for some
$k$.

Recall that when the matrix $\Th$ is quite irrational (the lattice generated by its columns is dense after 
translation by $\Z^n$, see \cite[Def. 12.8]{Polaris}), then the C$^*$-algebra generated by $\A_{\Th}$ is simple.
\end{remark}

\begin{remark}
It is possible to generalize above theorem to the case $\DD=-i\,{g^{\mu}}_{\nu} \, \delta_\mu \otimes \ga^\nu$ 
instead of \eqref{defDirac} when $g$ is a positive definite constant matrix. The formulae in Theorem \ref{main} 
 are still valid, up to  obvious modifications due to volume variation.
\end{remark}

\subsection{\texorpdfstring{Computations of $\ncint$}{Computations of noncommutative integral}}

In order to get this theorem, let us prove a few technical lemmas.

We suppose
from now on that $\Th$ is a skew-symmetric matrix in
$\mathcal{M}_n(\R)$. No other
hypothesis is assumed for $\Th$, except when it is explicitly stated.

When $A$ is a selfadjoint one-form, we define for $n\in N$, $q\in \N$,
$2\leq q \leq n$ and
$\sigma\in \{-,+\}^q$
\begin{align*}
\mathbb{A}^{+}&:=A\DD D^{-2},\\
\mathbb{A}^{-}&:= \epsilon JAJ^{-1} \DD D^{-2},\\
\mathbb{A}^{\sigma}&:=\mathbb{A}^{\sigma_q}\cdots
\mathbb{A}^{\sigma_1} .
\end{align*}

\begin{lemma}
\label{ncadmoins1}
We have for any $q\in \N$,
$$
\ncint (\wt A D^{-1})^q = \ncint (\wt A \DD D^{-2})^q =
\sum_{\sigma\in \set{+,-}^{q}}
\ncint \mathbb{A}^{\sigma}.
$$
\end{lemma}

\begin{proof} Since $P_0 \in OP^{-\infty}$, $D^{-1} = \DD D^{-2} \mod
OP^{-\infty}$
and $\ncint (\wt A D^{-1})^q = \ncint (\wt A \DD D^{-2})^q$.
\end{proof}

\begin{lemma}
\label{symetrie}
Let $A$ be a selfadjoint one-form, $n\in \N$ and $q\in \N$ with
$2\leq q \leq n$ and
$\sigma\in \{-,+\}^q$. Then
$$
\ncint \mathbb{A}^{\sigma} =\ncint \mathbb{A}^{-\sigma}.
$$
\end{lemma}

\begin{definition}

In \cite{CC1} has been introduced the vanishing tadpole hypothesis:
\begin{align}
    \label{vanishtad}
    \ncint A D^{-1}=0, \text{ for all } A\in \Omega_{\DD}^{1}(\A).
\end{align}
\end{definition}
By the following lemma, this condition is satisfied for the noncommutative torus.

\begin{lemma}
\label{tadpole}
Let $n\in \N$, $A= L(-iA_{\a})\otimes \gamma^{\a}=-i\sum_{l\in \Z^n}a_{\alpha,l} \, U_{l}\otimes \ga^{\alpha}$, 
$A_{\a}\in \A_{\Th}$, where $\set{a_{\alpha,l}}_{l}\in \SS (\Z^n)$, be a hermitian one-form.
Then, \\
    (i) $\ncint A^p D^{-q}  =  \ncint (\epsilon JAJ^{-1})^p
D^{-q}=0$ for $p\geq0$ and $1\leq q < n$ (case $p=q=1$ is tadpole
hypothesis.)\\
    (ii) If $\tfrac{1}{2\pi} \Th$ is badly-approximable, then $\ncint
BD^{-q}=0$ for $1\leq q < n$ and any $B$ in the algebra generated by
$\A$, $[\DD,\A]$, $J\A J^{-1}$ and $J[\DD,\A]J^{-1}$. 
\end{lemma}

\begin{proof}
$(i)$ Let us compute $$\ncint A^p(\epsilon
JAJ^{-1})^{p'}D^{-q}.$$ 
With $A=L(-i A_\a)\otimes \ga^\a$ and
$\epsilon J A
J^{-1}=R(i A_\a)\otimes \ga^\a$, we get
$$
A^p=L(-i A_{\a_1})\cdots L(-i A_{\a_p}) \otimes \ga^{\a_1}\cdots
\ga^{\a_p}
$$
and
$$
(\epsilon J A J^{-1})^{p'}=R(i A_{\a'_1})\cdots R(i A_{\a'_{p'}})
\otimes
\ga^{\a'_1}\cdots \ga^{\a'_{p'}}.
$$
We note $\wt a_{\a,l}:= a_{\a_1,l_1}\cdots a_{\a_p,l_p}$. Since
$$
L(-i A_{\a_1})\cdots L(-i A_{\a_p})R(i A_{\a'_1})\cdots R(i
A_{\a'_{p'}}) U_k= (-i)^p
\, i^{p'} \sum_{l,l'} \wt a_{\a,l} \, \wt a_{\a',l'} \, U_{l_1}\cdots
U_{l_p} U_k
U_{l'_{p'}}\cdots U_{l'_1},
$$
and
$$
U_{l_1}\cdots U_{l_p} U_k= U_k U_{l_1}\cdots U_{l_p} \, e^{-i(\sum_i
l_i).\Th k},
$$
we get, with
\begin{align*}
&U_{l,l'}:=U_{l_1}\cdots U_{l_p}U_{l'_{p'}}\cdots U_{l'_1},\\
&g_{\mu,\a,\a'}(s,k,l,l'):= e^{ik. \Th \sum_j l_j} \,
\tfrac{k_{\mu_{1}}\ldots
k_{\mu_{q}}}{\vert k \vert^{s+2q}} \, \wt a_{\a,l} \,  \wt a_{\a',l'},
\\
&\ga^{\a,\a',\mu}:=\ga^{\a_1}\cdots\ga^{\a_{p}}\ga^{\a'_1}\cdots
\ga^{\a'_{p'}} \ga^{\mu_1}\cdots \ga^{\mu_q},
\end{align*}
$$
A^p(\epsilon JAJ^{-1})^{p'}D^{-q}|D|^{-s} U_k\otimes e_i\sim_c (-i)^p
\, i^{p'}
\sum_{l,l'} g_{\mu,\a,\a'}(s,k,l,l') \, U_k U_{l,l'} \otimes
\ga^{\a,\a',\mu} e_i.
$$
Thus, $\ncint A^p(\epsilon JAJ^{-1})^{p'}D^{-q} =
\underset{s=0}{\Res}\ f(s)$
where
\begin{align*}
f(s):&=\Tr\big(A^p(\epsilon JAJ^{-1})^{p'}D^{-q}|D|^{-s}\big)\\
&\sim_c (-i)^p \, i^{p'} {\sum_{k\in\Z^{n}}}' \langle U_{k}\otimes
e_{i},\sum_{l,l'}
g_{\mu,\a,\a'}(s,k,l,l')  U_k
U_{l,l'} \otimes \ga^{\a,\a',\mu} e_i \rangle\\
&\sim_c (-i)^p \, i^{p'} {\sum_{k\in\Z^{n}}}' \, \,\tau\big(
\sum_{l,l'}
g_{\mu,\a,\a'}(s,k,l,l')
U_{l,l'} \big) \Tr(\ga^{\mu,\a,\a'})\\
&\sim_c (-i)^p \, i^{p'} {\sum_{k\in\Z^{n}}}' \sum_{l,l'}
g_{\mu,\a,\a'}(s,k,l,l') \, \tau
\big(U_{l,l'} \big) \Tr(\ga^{\mu,\a,\a'}).
\end{align*}
It is straightforward to check that the series
${\sum}'_{k,l,l'} g_{\mu,\a,\a'}(s,k,l,l') \, \tau\big( U_{l,l'}
\big)$
is absolutely sum-mable if $\Re(s)>R$ for a $R>0$. Thus, we can
exchange the summation
on $k$ and $l,l'$, which gives
$$
f(s)\sim_c (-i)^p \, i^{p'} \sum_{l,l'}  {\sum_{k\in\Z^{n}}}'
g_{\mu,\a,\a'}(s,k,l,l') \,
\tau \big( U_{l,l'} \big) \Tr(\ga^{\mu,\a,\a'}).
$$
If we suppose now that $p'=0$, we see that,
$$
f(s)\sim_c  (-i)^p \sum_{l}  {\sum_{k\in\Z^{n}}}'   \,
\tfrac{k_{\mu_{1}}\ldots
k_{\mu_{q}}}{\vert k \vert^{s+2q}} \, \wt a_{\a,l} \,
\delta_{\sum_{i=1}^p l_i,0}
\Tr(\ga^{\mu,\a,\a'})
$$
which is, by Proposition \ref{calculres}, analytic at 0. In
particular, for
$p=q=1$, we see that $\ncint A D^{-1} =0$, i.e. the
vanishing tadpole
hypothesis is satisfied. Similarly, if we suppose $p=0$,
we get
$$
f(s)\sim_c  (-i)^{p'} \sum_{l'}  {\sum_{k\in\Z^{n}}}'   \,
\tfrac{k_{\mu_{1}}\ldots
k_{\mu_{q}}}{\vert k \vert^{s+2q}} \, \wt a_{\a,l'} \,
\delta_{\sum_{i=1}^{p'} {l'}_i,0}
\Tr(\ga^{\mu,\a,\a'})
$$
which is holomorphic at 0.

$(ii)$ Adapting the proof of Lemma \ref{ncint-odd-pdo} to our setting
(taking $q_i=0$, and 
adding gamma matrices components), we see that
$$
\ncint B\,D^{-q} =\underset{s=0}{\Res}\, {\sum_k}'\,\sum_{l,l'}
e^{i\phi(l,l')}\,\delta_{\sum_1^r l_i+l'_i,0}\, \wt a_{\a,l} \,\wt
b_{\beta,l'}
\,\tfrac{k_{\mu_1}\cdots k_{\mu_q}\,e^{-i\sum_1^r l_i.\Th
k}}{|k|^{s+2q}}
\Tr (\ga^{(\mu,\a,\beta)})
$$
where $\ga^{(\mu,\a,\beta)}$ is a complicated product of gamma matrices.
By Theorem \ref{analytic} $(ii)$, since we suppose here that
$\tfrac{1}{2\pi} \Th$ is
badly approximable, this residue is 0.
\end{proof}

\subsubsection{Even dimensional case}
\begin{corollary}

Same hypothesis as in Lemma \ref{tadpole}.

(i) Case $n=2$:
    \begin{align*}
\ncint A^q D^{-q}= -\delta_{q,2}\,4\pi  \,\tau\big(A_{\a} A^{\a}
\big) \,.
\end{align*}
    (ii) Case $n=4$: with the shorthand
$\delta_{\mu_1,\ldots,\mu_4}:=
\delta_{\mu_1\mu_2}\delta_{\mu_3\mu_4}+\delta_{\mu_1\mu_3}\delta_{\mu_2\mu_4}
+\delta_{\mu_1\mu_4}\delta_{\mu_2\mu_3}$,
\begin{align*}
\ncint A^q D^{-q}= \delta_{q,4}\,\tfrac{\pi^2}{12} \tau\big(A_{\a_1}\cdots A_{\a_4}
\big)\Tr(\ga^{\a_1}\cdots\ga^{\a_{4}}\ga^{\mu_1}\cdots\ga^{\mu_4}) \delta_{\mu_1,\ldots,\mu_4} \,.
\end{align*}
\end{corollary}

\begin{proof}
$(i,ii)$ The same computation as in Lemma \ref{tadpole} $(i)$ (with
$p'=0$, $p=q=n$) gives
\begin{align*}
\ncint A^n D^{-n}=\underset{s=0}{\Res}(-i)^{n}
\big({\sum_{k\in\Z^{n}}}'\tfrac{k_{\mu_{1}}\ldots k_{\mu_{n}}}{\vert k
    \vert^{s+2n}}\big) \, \tau\big(\sum_{l\in (\Z^n)^n}
\wt a_{\a,l} U_{l_1}\cdots U_{l_{n}}
    \big) \,
\Tr(\ga^{\a_1}\cdots\ga^{\a_{n}}\ga^{\mu_1}\cdots\ga^{\mu_n})
\end{align*}
and the result follows from Proposition \ref{calculres}.
\end{proof}

We will use few notations:

For $n\in \N$, $q\geq 2$, $l:=(l_1,\cdots,l_{q-1})\in
(\Z^n)^{q-1}$, $\a:=(\a_1,\cdots,\a_q)\in \{1,\cdots,n\}^q$, $k\in \Z^n \backslash \{0\}$, $\sigma\in \{-,+\}^q$, 
$(a_i)_{1\leq i\leq n}\in (\mathcal{S}(\Z^n))^n$,
\begin{align*}
&l_q:=-\sum_{1\leq j\leq q-1} l_j \, , \quad
\lambda_\sigma:=(-i)^q\prod_{j=1\ldots
q}\sigma_j \, , \quad
\wt a_{\a,l}:= a_{\a_1,l_1}\ldots a_{\a_q,l_q}\,,\\
& \phi_\sigma(k,l):=\sum_{1\leq j\leq q-1} (\sigma_j-\sigma_q)\,
k.\Th l_j +
\sum_{2\leq
j\leq q-1} \sigma_j \, (l_1+\ldots +l_{j-1}).\Th l_j \, ,\\
& g_\mu(s,k,l):=\tfrac{k_{\mu_1}(k+l_1)_{\mu_2}\ldots (k+l_1+\ldots
+l_{q-1})_{\mu_q}}{|k|^{s+2}|k+l_1|^2\ldots|k+l_1+\ldots+l_{q-1}|^2}
\, ,
\end{align*}
with the convention $\sum_{2\leq j\leq q-1} = 0$ when $q=2$, and
$g_\mu(s,k,l)=0$
whenever $\wh l_i=-k$ for a $1\leq i\leq q-1$.

\begin{lemma}\label{formegenerale}
Let $A= L(-iA_{\a})\otimes \gamma^{\a}=-i\sum_{l\in \Z^n}a_{\alpha,l}
\, U_{l}\otimes
\ga^{\alpha}$ where $A_{\a}=-A_{\a}^*\in \A_{\Th}$ and
$\set{a_{\alpha,l}}_{l}\in \SS
(\Z^n)$, with $n\in \N$, be a hermitian one-form, and let $2\leq q
\leq n$, $\sigma\in
\{-,+\}^q$.

Then, $\ncint \mathbb{A}^{\sigma}= \underset{s=0}{\Res}\ f(s)$ where
$$
f(s):=\sum_{l\in (\Z^{n})^{q-1}} {\sum_{k\in\Z^n}}' \,
\lambda_\sigma\  e^{\tfrac i2
\phi_\sigma(k,l)}\ g_\mu(s,k,l)\ \wt a_{\a,l}\
\Tr(\ga^{\a_q}\ga^{\mu_q}\cdots\ga^{\a_1}\ga^{\mu_1}).
$$
\end{lemma}

In the following, we will use the shorthand
$$
c:=\tfrac{4\pi^{2}}{3}.
$$

\begin{lemma}
\label{Termaterm}
Suppose $n=4$. Then, with the same hypothesis of Lemma \ref{formegenerale},
\begin{align*}
\hspace{-2cm}\text{(i)} \quad \quad &
\tfrac 12 \ncint (\mathbb A^+)^2= \tfrac 12 \ncint (\mathbb A^-)^2= c
\,
\sum_{l\in\Z^4} \, a_{\alpha_{1},l} \, a_{\alpha_{2},-l} \,
\big(l^{\alpha_{1}}l^{\alpha_{2}}
- \delta^{\alpha_{1}\alpha_{2}} \vert l \vert^2\big).\\
\hspace{-2cm}\text{(ii)} \quad  \quad & \hspace{-0.45cm}
-\tfrac 13\ \ncint (\mathbb A^+)^3=-\tfrac 13 \ncint (\mathbb
A^-)^3=4c
\,\sum_{l_i \in \Z^4}
a_{\a_3,-l_1-l_2}\,a^{\a_1}_{l_2}\,a_{\a_1,l_1}\ \sin \tfrac{l_1.\Th
l_2}{2}\,l_1^{\a_3}.\\
\hspace{-1cm}\text{(iii)} \quad \quad &
\tfrac 14 \ncint (\mathbb A^+)^4=\tfrac 14 \ncint (\mathbb A^-)^4=
2c\,
\sum_{l_i \in \Z^4} a_{\alpha_{1},-l_1-l_2-l_3}\,
a_{{\alpha_{2}},l_3} \, a^{\alpha_{1}}_{l_2}
\, a^{\alpha_{2}}_{l_1} \sin \tfrac{l_1 .\Th (l_2+l_3)}{2}\, \sin
\tfrac{ l_2 .\Th
l_3}{2}.
\end{align*}
\vspace{-0.3cm}

(iv) Suppose $\tfrac {1}{2\pi}\Th$ badly approximable. Then the
crossed terms in
$\ncint (\mathbb A^+ + \mathbb A^-)^q$ vanish:  if $C$ is the set of
all $\sigma\in \{-,+\}^q$ with $2\leq q\leq 4$, such that there exist
$i,j$ satisfying
$\sigma_i \neq \sigma_j$, we have $ \sum_{\sigma \in C} \, \ncint \mathbb{A}^{\sigma} =0.$
\end{lemma}

\begin{lemma}
\label{double}
Suppose $n=4$ and $\tfrac{1}{2\pi}\Th$ badly approximable. For any self-adjoint one-form $A$,
$$
\zeta_{D_A}(0)-\zeta_{D}(0)=-c\, \tau(F_{\a_1,\a_2}F^{\a_1\a_2}).
$$
\end{lemma}

\begin{proof}
By (\ref{termconstanttilde}) and Lemma \ref{ncadmoins1} we get
$$
\zeta_{D_A}(0)-\zeta_{D}(0) =\sum_{q=1}^n \tfrac {(-1)^q}{q} \sum_{\sigma\in \set{+,-}^q}\ncint
\mathbb{A}^\sigma.
$$
By Lemma \ref{Termaterm} $(iv)$, we see that the crossed terms all
vanish. Thus,
with Lemma \ref{symetrie}, we get
\begin{equation}
\label{n4eq1}
\zeta_{D_A}(0)-\zeta_{D}(0) = 2 \sum_{q=1}^n
\tfrac {(-1)^q}{q} \ncint (\mathbb{A}^+)^q.
\end{equation}

By definition,
\begin{align*}
F_{\alpha_{1} \alpha_{2}}&=i\sum_{k}\big(a_{\alpha_{2},k}\,
k_{\alpha_{1}}-a_{\alpha_{1},k}\, k_{\alpha_{2}}\big)U_{k}
+\sum_{k,\,l}a_{\alpha_{1},k}\,a_{\alpha_{2},l}\,[U_{k},U_{l}]\\
&=i\sum_{k}\big[(a_{\alpha_{2},k} \, k_{\alpha_{1}} -a_{\alpha_{1},k}
\,
k_{\alpha_{2}}) -2\sum_{l}a_{\alpha_{1},k-l}\,a_{\alpha_{2},l}\,
\sin(\tfrac{k.\Th
l}{2})\big] \, U_{k}.
\end{align*}
Thus
\begin{align*}
\tau(F_{\alpha_{1}\alpha_{2}} F^{\alpha_{1}\alpha_{2}})
&=\sum_{\alpha_{1},\,\alpha_{2}=1}^{2^m}\, \sum_{k\in \Z^4}\big[
(a_{\alpha_{2},k} \, k_{\alpha_{1}}-a_{\alpha_{1},k} \,
k_{\alpha_{2}})
-2\sum_{l'\in\Z^4}a_{\alpha_{1},k-l'}\,a_{\alpha_{2},l'}\,
\sin(\tfrac{k.\Th l'}{2})\big]\\
&\hspace{2.2cm} \big[(a_{\alpha_{2},-k} \, k_{\alpha_{1}}
-a_{\alpha_{1},-k}\, k_{\alpha_{2}})
-2\sum_{l"\in
\Z^4}a_{\alpha_{1},-k-l"}\,a_{\alpha_{2},l"}\,\sin(\tfrac{k.\Th
l"}{2})\big].
\end{align*}
One checks that the term in $a^q$ of $\tau(F_{\alpha_{1}\alpha_{2}}
F^{\alpha_{1}\alpha_{2}})$ corresponds to the term $\ncint
(\mathbb{A}^+)^q$ given by
Lemma \ref{Termaterm}. For $q=2$, this is
$$
-2\sum_{l\in\Z^4,\,\alpha_{1},\,\alpha_{2}} \, a_{\alpha_{1},l} \,
a_{\alpha_{2},-l}
\, \big(l_{\alpha_{1}}l_{\alpha_{2}} - \delta_{\alpha_{1}\alpha_{2}}
\vert l \vert
^2\big).
$$
\noindent For $q=3$, we compute the crossed terms:
$$
i\sum_{k,k',l} (a_{{\a_{2}},k} \, k_{\a_{1}} - a_{{\a_{1}},k} \,
k_{\a_{2}})\
a_{k'}^{\a_1}\ a_{l}^{\a_2} \big(U_k[U_{k'},l]+[U_{k'},U_l]U_k \big),
$$
which gives the following $a^3$-term in
$\tau(F_{\alpha_{1}\alpha_{2}} F^{\alpha_{1}\alpha_{2}})$
$$
-8\sum_{l_i} a_{\a_3,-l_1-l_2}\,a^{\a_1}_{l_2}\,a_{\a_1,l_1}\ \sin
\tfrac{l_1.\Th
l_2}{2}\,l_1^{\a_3} .
$$
\noindent For $q=4$, this is
$$-4\sum_{l_i}a_{\alpha_{1},-l_1-l_2-l_3}\,
a_{{\alpha_{2}},l_3} \, a^{\alpha_{1}}_{l_2} \, a^{\alpha_{2}}_{l_1}
\sin \tfrac{l_1
.\Th (l_2+l_3)}{2}\, \sin \tfrac{ l_2 .\Th l_3}{2}
$$
which corresponds to the term $\ncint (\mathbb{A}^+)^4$.
We get finally,
\begin{equation}\label{n4eq2}
\sum_{q=1}^n
\tfrac {(-1)^q}{q} \ncint (\mathbb{A}^+)^q
=- \tfrac c2 \tau(F_{\a_1,\a_2}F^{\a_1\a_2}).
\end{equation}
Equations (\ref{n4eq1}) and (\ref{n4eq2}) yield the result.
\end{proof}

\begin{lemma}
    \label{term-n=2}
Suppose $n=2$. Then, with the same hypothesis as in Lemma
\ref{formegenerale},
\begin{align*}
\hspace{-5.9cm}\text{(i)} \quad \quad &  \ncint (\mathbb
A^+)^2= \ncint (\mathbb A^-)^2=0.
\end{align*}
(ii) Suppose $\tfrac{1}{2\pi}\Th$ badly approximable. Then
$$
\ncint \mathbb A^+ \mathbb A^{-}
=  \ncint \mathbb A^- \mathbb A^+=0.
$$
\end{lemma}

\begin{lemma}
\label{termeconstantn=2}
Suppose $n=2$ and $\tfrac{1}{2\pi}\Th$ badly approximable. Then, for any self-adjoint one-form $A$,
\begin{align*}
\zeta_{D_A}(0)-\zeta_{D}(0)= 0.
\end{align*}
\end{lemma}

\begin{proof}
As in Lemma \ref{double}, we use (\ref{termconstanttilde}) and Lemma \ref{ncadmoins1}
so the result follows from Lemma \ref{term-n=2}.
\end{proof}

\subsubsection{Odd dimensional case}

\begin{lemma}
\label{impair}
Suppose $n$ odd and $\tfrac{1}{2\pi}\Th$ badly approximable.
Then for any self-adjoint one-form $A$ and $\sigma\in \{-,+\}^q$ with
$2\leq q\leq n$,
$$
\ncint \mathbb{A}^\sigma = 0\, .
$$
\end{lemma}

\begin{proof}
Since $\mathbb{A}^\sigma \in \Psi_1(\A)$, Lemma \ref{ncint-odd-pdo} with $k=n$ gives the result.
\end{proof}

\begin{corollary}
\label{zetaimpair}
With the same hypothesis of Lemma \ref{impair}, $\zeta_{D_A}(0)-\zeta_{D}(0)=0.$
\end{corollary}

\begin{proof}
As in Lemma \ref{double}, we use (\ref{termconstanttilde}) and Lemma \ref{ncadmoins1}
so the result follows from Lemma \ref{impair}.
\end{proof}

\subsection{Proof of the main result}

\begin{proof}[Proof of Theorem \ref{main}.]
$(i)$ By (\ref{asympspectral}) and Proposition \ref{invariance}, we
get
$$
\SS(\DD_{A},f,\Lambda) \, = \, 4f_{2}\, \Lambda^{2}  + f(0) \, \zeta_{D_{A}}(0) + \mathcal{O}
(\Lambda^{-2}),
$$
where $f_{2}= \half\int_{0}^{\infty} f(t) \, dt$. By Lemma \ref{termeconstantn=2},
$\zeta_{D_A}(0) - \zeta_{D}(0) = 0$ and from Proposition \ref{zeta(0)}, $\zeta_{D}(0)=0$, so we get the result.

$(ii)$ Similarly,
$\SS(\DD_{A},f,\Lambda) \, =  8 \pi^2\, f_{4}\, \Lambda^{4} + f(0) \,\zeta_{D_{A}}(0) + \mathcal{O}(\Lambda^{-2})$
with $f_{4}= \half\int_{0}^{\infty} f(t) \, t \, dt$.
Lemma \ref{double} implies that $\zeta_{D_A}(0) - \zeta_D(0)=-c\,\tau(F_{\mu\nu}F^{\mu\nu})$
and Proposition \ref{zeta(0)} yields  the equality $\zeta_{D_A}(0)=-c\,\tau(F_{\mu\nu}F^{\mu\nu})$ and the result.

$(iii)$ is a direct consequence of (\ref{asympspectral}), Propositions \ref{zeta(0)}, \ref{invariance}, and Corollary
\ref{zetaimpair}.
\end{proof}

\subsection{Beyond Diophantine equation}

This section is an attempt to understand what happens if $\Theta$ is `in between' rational numbers and 
``Diophantine numbers''. Consider the simplest case: $\T^2$ with 
\begin{equation*}
\Theta=\theta \left( \begin{array}{cc} 0 & -1 \\ 1 & 0 \end{array} \right),
\end{equation*}

To proceed, we need some results from number theory \cite{Bugeaud}: 

\begin{definition}
Let $f\,:\, \R_{\geq 1} \rightarrow \R_{>0}$ be a continuous
function such that $x \rightarrow x^2\,f(x)$ is
non-increasing. Consider the set
$$
\mathcal{F}(f):=\set{\theta \in \R \, : \, \vert\theta q -p
\vert < q\,f(q) \text{ for infinitely many rational numbers
$\tfrac pq$} }.
$$
The elements of $\mathcal{F}(f)$ are termed $f$-approximable. 
\end{definition}
Note that we cannot expect the above estimate to be valid for all rational 
numbers $\tfrac pq$ since for all irrational numbers $\th$, the set of
fractional values of $(\th q)_{q\geq 1}$ is dense in $[0,1]$.

\begin{theorem}
There exists an uncountable family of real numbers $\theta/(2\pi)$ which are $f$-approximable but not
$cf$-approximable for any $0<c<1$.
\end{theorem}
\noindent See \cite[Exercise 1.5]{Bugeaud} for a proof.

Let us choose 
\begin{align*}
f(x)= (2\pi x)^{-1}e^{-2x},\label{fx}
\end{align*}
and fix a constant $c < 1$. Let us pick a $\theta$ which is $f$-approximable, but not $cf$-approximable. 
Consider now $g(t)\vc\Tr\big(aJbJ^{-1}\,e^{-t\,\DD^2}\big)$. It is shown in \cite{GIVas} that, by tuning 
$a,b\in \A_\Th$, it is possible make the difference $g(t)-g(t)_{Dioph}$ (of $g(t)$ and its value if we 
suppose that $\th$ is a Diophantine number) of arbitrary order in $t$.

This shows how subtle can be the computation of spectral action!

\newpage
\section{The non-compact case}
\label{The non-compact case}

\subsection{The matter is not only technical}

When a Riemannian spin manifold $M$ is non-compact, the Dirac operator, which exists as a selfadjoint 
extension when $M$ is (geodesically) complete has no more a compact resolvent: its spectrum is not discrete but is 
$\R$ (\cite[Theorem 7.2.1]{Ginoux} and similar results for hyperbolic spaces \cite[p. 106]{Ginoux}.)
\\
To see what happens, let us consider for instance the flat space $M=\R^d$ and the Hilbert space $\H=L^2(\R^d)$. 
Then the operator $f(x)\,g(-i\nabla)$ is formally given on $\psi,\phi$ in appropriate domains by 
$$
\langle \psi, f(x)\,g(-i\nabla)\,\phi \rangle= \int_{\R^d} \bar f(x)\, \phi(x)\,  (g \,\widehat \psi\,)\,\,\widecheck{ }\,\,(x) \,dx.
$$
For $k\in\Z^d$, let $\chi_k$ be the characteristic function of the unit cube in $\R^d$ with center at $k$ and 
define for $p,q>0$ 
$$
\ell^q\big(L^p(\R^d)\big) \vc \set{f \,\, \vert \,\, \norm{f}_{p,q}\vc\big(\sum_k \norm{f\,\chi_k}_p^q\big)^{1/q} < \infty }
$$
where $\norm{g}_p \vc  \big(\int_{\R^d} \vert g(x)\vert^p\,dx\big)^{1/p}$ is the usual norm of $L^p(\R^d)$.

\begin{theorem}
Birman--Solomjak. 

(i) If $f,g \in \ell^p\big(L^2(\R^d)\big)$ for $1\leq p \leq2$, then $f(x)\,g(-i\nabla)$ is in the Schatten class $\L^p$ and 
$\norm{f(x)\,g(-i\nabla)}_p \leq c_p \norm{f}_{2,p} \,\norm{g}_{2,p}$.

(i) If $f,g$ are non zero, then $f(x)\,g(-i\nabla)\in \L^1(\H)$ if and only $f$ and $g$ are in $\ell^1\big(\H)\big)$.
\end{theorem}
For a proof, see \cite[Chapter 4]{SimonTrace}.
\\
This shows that even if $g(x)=e^{-tx^2}$, the heat kernel $e^{-t\Delta}$ is never trace-class since $f=1$ is not in 
$\ell^1\big(L^2(\R^d) \big)$.

Thus, to cover at least the non-compact manifold case, Definition \ref{deftriplet} has to be improved: 

\begin{definition}
\label{deftripletnoncompact}
A non-compact spectral triple $(\A,\H,\DD)$ is the data of an involutive algebra $\A$ with a faithful representation 
$\pi$ on a Hilbert space $\H$, a preferred unitization $\wt \A$ of $\A$ and a selfadjoint operator $\DD$ such that

- $a(\DD-\lambda)^{-1}$ is compact for all $a\in \A$ and $\lambda \notin \Sp \,\DD$. 

- $[\DD,\pi(a)]$ is bounded for any $a \in \wt A$.
\end{definition}
All definitions of regularity, finiteness and orientation have to be modified with $\wt \A$ instead of $\A$, see also 
\cite{CGRS}.

In the first constraint of this definition we recover a certain discreteness which, with $a=1$, is the compact case 
(the algebra can have a unit). This matter is not only technical since now there is a deeper intertwining of the 
choice of the algebra $\A$ and the operator $\DD$ to get a spectral triple. Moreover, a tentative of 
modification of $\DD$ is quite often forbidden by the second constraint.
 
The case of non-compact spin manifold has been considered by Rennie 
\cite{RennieSpin,RennieProj, RennieLocal}. This has been improved in \cite{GGBISV} which studied the Moyal 
plane. Actually, a compactification of this plane is the noncommutative torus!

\subsection{The Moyal product}

Reference: \cite{GGBISV}.

For any finite dimension $k$, let $\Theta$ be a real skewsymmetric $k \x k$ matrix, let $s\.t$ denote the usual scalar
product on Euclidean $\R^k$ and let $\SS(\R^k)$ be the space of
complex Schwartz functions on $\R^k$. One
defines, for $f,h \in \SS(\R^k)$, the corresponding Moyal or twisted product:
\begin{equation}
f \star_\Theta h(x) := (2\pi)^{-k} \iint f(x - \half\Theta u) \, h(x + t) \, e^{-iu\.t} \,d^ku \,d^kt.
\label{eq:moyal-prod-slick}
\end{equation}
In Euclidean field theory, the entries of~$\Theta$ have the dimensions of
an area. Because $\Theta$ is skewsymmetric, complex conjugation
reverses the product: $(f \star_\Theta h)^* = h^* \star_\Theta f^*$.

Assume $\Theta$ to be nondegenerate, that is to say,
$\sigma(s,t) := s\.\Theta t$ to be symplectic. This implies even
dimension, $k = 2N$. We note that $\Theta^{-1}$ is also skewsymmetric;
let $\th > 0$ be defined by $\th^{2N} := \det\Theta$. Then
formula~\eqref{eq:moyal-prod-slick} may be rewritten as
\begin{equation}
f \star_\Theta h(x) = (\pi\th)^{-2N} \iint f(x + s)\, h(x + t)\, e^{-2is\.\Theta^{-1}t} \,d^{2N}s \,d^{2N}t.
\label{eq:moyal-prod-gen}
\end{equation}

This form is very familiar from phase-space quantum mechanics, where $\R^{2N}$ is parametri-zed by $N$
conjugate pairs of position and momentum variables, and the entries
of $\Theta$ have the dimensions of an action; one then selects
$
\Theta = \hbar S
:= \hbar \left(\begin{smallmatrix} 0 & 1_N \\  -1_N & 0 \end{smallmatrix}\right)
$
Indeed, the product $\star$ (or rather, its commutator) was introduced
in that context by Moyal~\cite{Moyal}, using a series development in
powers of~$\hbar$ whose first nontrivial term gives the Poisson
bracket; later, it was rewritten in the above integral form. These are
actually oscillatory integrals, of which Moyal's series development,
\begin{equation}
f\star_\hbar g(x)
= \sum_{\a\in\N^{2N}} \bigl(\tfrac{i\hbar}{2}\bigr)^{|\a|}
\tfrac{1}{\a!}\, \pd{f}{x^\a}(x) \, \pd{g}{(Sx)^\a}(x),
\label{eq:moyal-asymp}
\end{equation}
is an asymptotic expansion. The first integral form~\eqref{eq:moyal-prod-slick}
of the Moyal product was exploited by Rieffel in a remarkable
monograph~\cite{RieffelDefQ}, who made it the starting point for a
more general deformation theory of $C^*$-algebras.

With the choice $\Theta = \th S$ made, the Moyal product can also
be written
\begin{equation}
f \mop g(x) := (\pi\th)^{-2N} \iint f(y) g(z)\,
e^{\tfrac{2i}{\th}(x-y)\,\.\,S(x-z)} \,d^{2N}y \,d^{2N}z.
\label{eq:moyal-prod}
\end{equation}

Of course, our definitions make sense only under certain hypotheses on $f$ and $g$ \cite{Phobos,Deimos}.

\begin{lem} {\rm\cite{Phobos}}
\label{lm:propriete}
Let $f,g\in \SS(\R^{2N})$. Then
\begin{enumerate}
\item[(i)]
$f \mop g \in \SS(\R^{2N})$.
\item[(ii)]
$\mop$ is a bilinear associative product on $\SS(\R^{2N})$. Moreover,
complex conjugation of functions $f \mapsto f^*$ is an
involution for~$\mop$.
\item[(iii)]
Let $j = 1,2,\dots,2N$. The Leibniz rule is satisfied:
\begin{equation}
\pd{}{x_j}(f\mop g) = \pd{f}{x_j} \mop g + f \mop \pd{g}{x_j}.
\label{eq:Leibniz}
\end{equation}
\item[(iv)]
Pointwise multiplication by any coordinate $x_j$ obeys
\begin{equation}
x_j (f \mop g)
= f \mop (x_j g) + \tfrac{i\,\th}{2} \pd{f}{(Sx)_j} \mop g
= (x_j f) \mop g - \tfrac{i\,\th}{2} f \mop \pd{g}{(Sx)_j}.
\label{eq:mult-rule}
\end{equation}
\item[(v)]
The product has the tracial property:
$$
\<f,g> := \tfrac{1}{(\pi\th)^N} \int f \mop g(x) \,d^{2N}x
= \tfrac{1}{(\pi\th)^N} \int g \mop f(x) \,d^{2N}x
= \tfrac{1}{(\pi\th)^N} \int f(x) \,g(x) \,d^{2N}x.
$$
\item[(vi)]
Let $L^\th_f \equiv L^\th(f)$ be the left multiplication
$g \mapsto f \mop g$. Then
$\lim_{\th\downarrow0}{L^\th_f\,g}(x) = f(x)\,g(x)$, for
$x \in \R^{2N}$.
\end{enumerate}
\end{lem}

Property~(vi) is a consequence of the distributional identity
$$\lim_{\eps\downarrow0}\eps^{-k} e^{ia\.b/\eps} =
(2\pi)^k \delta(a) \delta(b),
$$
for $a,b\in\R^k$; convergence takes
place in the standard topology~\cite{Schwartz} of $\SS(\R^{2N})$. To
simplify notation, we put $\SS := \SS(\R^{2N})$ and let
$\SS' := \SS'(\R^{2N})$ be the dual space of tempered distributions.
In view of~(vi), we may denote by $L^0_f$ the pointwise product
by~$f$.

\begin{thm} {\rm\cite{Phobos}}
$\A_\th := (\SS,\mop)$ is a nonunital associative, involutive
Fr\'echet algebra with a jointly continuous product and a
distinguished faithful trace.
\end{thm}

\begin{defn}
\label{df:basis-fns}
The algebra $\A_\th$ has a natural basis of eigenvectors $f_{mn}$
of the harmonic oscillator, indexed by $m,n \in \N^N$. If
$$
H_l := \half(x_l^2 + x_{l+N}^2) \sepword{for $l=1,\dots,N$\quad and}
H := H_1 + H_2 +\cdots+ H_N,
$$
then the~$f_{mn}$ diagonalize these harmonic oscillator Hamiltonians:
\begin{align}
H_l \mop f_{mn} &= \th(m_l+\half) f_{mn},
\nonumber \\
f_{mn} \mop H_l &= \th(n_l+\half) f_{mn}.
\label{eq:moyal-haml}
\end{align}
They may be defined by
\begin{equation}
f_{mn}
:= \tfrac{1}{\sqrt{\th^{|m|+|n|}\,m!n!}}\,(a^*)^m \mop f_{00} \mop a^n,
\label{eq:basis}
\end{equation}
where $f_{00}$ is the Gaussian function $f_{00}(x) := 2^N e^{-2H/\th}$,
and the annihilation and creation functions respectively are
\begin{equation}
a_l := \tfrac{1}{\sqrt{2}} (x_l + ix_{l+N})  \sepword{and}
a_l^* := \tfrac{1}{\sqrt{2}} (x_l - ix_{l+N}).
\label{eq:crea-annl}
\end{equation}
One finds that $a^n := a_1^{n_1} \dots a_N^{n_N} =
a_1^{\mop n_1} \mop\cdots\mop a_N^{\mop n_N}$.
\end{defn}

\begin{prop} {\rm \cite[p.~877]{Phobos}}
\label{pr:factorization}
The algebra $(\SS,\mop)$ has the (nonunique) factorization
property: for all $h \in \SS$ there exist $f,g \in \SS$ such that
$h = f \mop g$.
\end{prop}

\begin{lem} {\rm\cite{Phobos,Deimos}}
\label{lm:propriete-bis}
Let $f,g\in L^2(\R^{2N})$. Then
\begin{enumerate}
\item[(i)]
For $\th \neq 0$, $f \mop g$ lies in $L^2(\R^{2N})$. Moreover,
$f \mop g$ is uniformly continuous.
\item[(ii)]
$\mop$ is a bilinear associative product on $L^2(\R^{2N})$. The
complex conjugation of functions $f \mapsto f^*$ is an involution
for~$\mop$.
\item[(iii)]
The linear functional $f \mapsto \int f(x)\,dx$ on $\SS$ extends to
$\I_{00}(\R^{2N}) := L^2(\R^{2N}) \mop L^2(\R^{2N})$, and the product
has the tracial property:
$$
\<f,g> := (\pi\th)^{-N} \! \int f \mop g(x) \,d^{2N}x
= (\pi\th)^{-N} \! \int g \mop f(x) \,d^{2N}x
= (\pi\th)^{-N} \! \int f(x) \,g(x) \,d^{2N}x.
$$
\item[(iv)]
$\lim_{\th\downarrow0}{L^\th_f\,g}(x) = f(x)\,g(x)$ almost
everywhere on $\R^{2N}$.
\end{enumerate}
\end{lem}

\begin{defn}
Let $A_\th := \set{T \in \SS' : T \mop g \in L^2(\R^{2N})
\text{ for all } g \in L^2(\R^{2N})}$, provided with the operator norm
$\|L^\th(T)\|_{\mathrm{op}} :=
\sup\set{\|T \mop g\|_2/\|g\|_2 : 0 \neq g \in L^2(\R^{2N})}$.

Obviously $\A_\th = \SS \hookto A_\th$. But $\A_\th$ is not dense in
$A_\th$.
\end{defn}

Note that $\G_{00} \subset A_\th$. This is clear from the following
estimate.

\begin{lem} {\rm\cite{Phobos}}
\label{lm:norm-HS}
If $f,g \in L^2(\R^{2N})$, then $f \mop g \in L^2(\R^{2N})$ and
$\|L^\th_f\|_{\mathrm{op}} \leq (2\pi\th)^{-N/2} \|f\|_2$.
\end{lem}

\begin{proof}
Expand $f = \sum_{m,n} c_{mn} \a_{mn}$ and
$g = \sum_{m,n} d_{mn} \a_{mn}$ with respect to the orthonormal basis
$\{\a_{nm}\} := (2\pi\th)^{-N/2} \{f_{nm}\}$ of $L^2(\R^{2N})$. Then
\begin{align*}
\|f\mop g\|_2^2
&= (2\pi\th)^{-2N} \biggl\| \sum_{m,l}
\Bigl( \sum_n c_{mn}\,d_{nl} \Bigr) f_{ml} \biggr\|_2^2
= (2\pi\th)^{-N} \sum_{m,l} \Bigl|\sum_n c_{mn}\,d_{nl}\Bigr|^2
\\
&\leq (2\pi\th)^{-N} \sum_{m,j} |c_{mj}|^2 \sum_{k,l} |d_{kl}|^2
= (2\pi\th)^{-N} \|f\|_2^2 \, \|g\|_2^2,
\end{align*}
on applying the Cauchy--Schwarz inequality.
\end{proof}

\begin{prop} {\rm\cite{Deimos}}
\label{pr:algebra}
$(A_\th,\|.\|_{\mathrm{op}})$ is a unital $C^*$-algebra of
operators on $L^2(\R^{2N})$, isomorphic to $\L(L^2(\R^N))$ and
including $L^2(\R^{2N})$. Moreover, there is a continuous injection of $*$-algebras
$\A_\th \hookto A_\th$, but $\A_\th$ is not dense in~$A_\th$.
\end{prop}

\begin{prop}
\label{pr:pre-S}
$\A_\th$ is a (nonunital) Fr\'echet pre-$C^*$-algebra.
\end{prop}

\begin{proof}
We adapt the argument for the commutative case
in~\cite[p.~135]{Polaris}. To show that $\A_\th$ is stable under
the holomorphic functional calculus, we need only check that if
$f \in \A_\th$ and $1 + f$ is invertible in $A^0_\th$
with inverse $1 + g$, then the quasi-inverse $g$ of $f$ must lie in
$\A_\th$. From  $f + g + f \mop g = 0$, we obtain
$f \mop f + g \mop f + f \mop g \mop f = 0$, and it is enough to show
that $f \mop g \mop f \in \A_\th$, since the previous relation then
implies $g \mop f \in \A_\th$, and then
$g = -f - g \mop f \in \A_\th$ also.

Now, $A_\th \subset \G_{-r,0}$ for any $r > N$ \cite[p.~886]{Deimos}.
Since $f \in \G_{s,p+r} \cap \G_{qt}$, for $s,t$ arbitrary and $p,q$
positive, we conclude that $f \mop g \mop f \in
\G_{s,p+r} \mop \G_{-r,0} \mop \G_{qt} \subset \G_{st}$; as
$\SS = \bigcap_{s,t\in\R} \G_{st}$, the proof is complete.
\end{proof}

\begin{lem}
\label{lm:cojoreg}
If $f \in \SS$, then $L^\th_f$ is a regularizing $\PsiDO$.
\end{lem}

\begin{proof}
{}From~\eqref{eq:moyal-prod-slick}, one at once sees that left Moyal
multiplication by~$f$ is the pseudodifferential operator on~$\R^{2N}$
with symbol $f(x - \frac{\th}{2}S\xi)$. Clearly $L^\th_f$ extends to a
continuous linear map from
$\Coo(\R^{2N})' \hookto \SS'$ to $\Coo(\R^{2N})$.
The lemma also follows from the inequality
$$
|\del_x^\a \del_\xi^\b f(x - \tfrac{\th}{2} S\xi)|
\leq C_{K\a\b} (1 + |\xi|^2)^{(d-|\b|)/2},
$$
valid for all $\a,\b \in \N^{2N}$, any compact $K \subset \R^{2N}$,
and any $d \in \R$, since $f\in \SS$.
\end{proof}

\begin{rem}
Unlike for the case of a compact manifold, regularizing $\PsiDO$s are not necessarily compact operators!
\end{rem}

\subsection{The preferred unitization of the Schwartz Moyal algebra}

\begin{defn}
\label{df:distr-zoo}
Following Schwartz, we denote $\B := \Oh_0$, the space of smooth functions bounded together with all derivatives.
\end{defn}
A unitization of~$\A_\th$ is given by the algebra $\Aun_\th := (\B,\mop)$. The
inclusion of~$\A_\th$ in~$\B$ is not dense, but this is not needed.
$\Aun_\th$ contains the constant functions and the plane waves, but no
nonconstant polynomials and no imaginary-quadratic exponentials, such
as $e^{iax_1x_2}$ in the case $N = 1$ (we will see later the
pertinence of this).

Since $\B$ is a unital $*$-algebra with the Moyal
product, 

\begin{prop}
\label{pr:pre-Aun}
$\Aun_\th$ is a unital Fr\'echet pre-$C^*$-algebra.
\end{prop}

An advantage of~$\Aun_\th$ is that the covering relation of the
noncommutative plane to the NC torus is made transparent. To wit, the
smooth noncommutative torus algebra $\Coo(\T_\Theta^{2N})$ seen in Section \ref{Definition of the nc-torus} can be
embedded in~$\B$ as periodic functions (with a fixed period
parallelogram). This is in fact a Hopf algebra homomorphism: recall
that $\Coo(\T_\Theta^{2N})$ is a cotriangular Hopf algebra by exploiting the integral
form~\eqref{eq:moyal-prod-slick} of (a periodic version of) the Moyal
product.

We finally note the main reason for suitability of $\Aun_\th$, namely, that each
$[\Dslash, L^\th(f) \ox 1_{2^N}]$ lies in $A_\th \ox M_{2^N}(\C)$, for
$f \in \Aun_\th$ and $\Dslash$ the Dirac operator on~$\R^{2N}$.

\subsection{The commutative case}
\label{The commutative case}

When $\Theta=0$ the Moyal product is the ordinary product. 

Let $\A$ be some appropriate subalgebra of $C^\infty(M)$ and $\Dslash$
be the Dirac operator, with $k$ equal to the ordinary dimension of the
spin manifold $M=\R^k$. Let $\H$ be the space of square-integrable spinors.
Then $[\Dslash, f] = \Dslash(f)$, just as in the unital case, and so
the boundedness of $[D,\A]$ is unproblematic. In order to check
whether $(\A,\H,\Dslash,\chi)$ is a spectral triple, one
first needs to determine whether products of the form
$f(|\Dslash| + \eps)^{-k}$ are compact operators of Dixmier trace
class, whose Dixmier trace is (a standard multiple of)
$\int f(x) \,d^kx$. This compactness condition is guaranteed in the
flat space case (taking $\A = \SS(\R^k)$, say) by celebrated estimates
in scattering theory~\cite{SimonTrace}. 

The summability condition is a bit tougher. The Ces\`aro summability
theory of~\cite{EGBV} establishes that, for a positive
pseudodifferential operator $H$ of order~$d$, acting on spinors, the
spectral density asymptotically behaves as
$$
d_H(x,x;\la') \sim \tfrac{2^{\piso{k/2}}}{d\,(2\pi)^k} \bigl(\text{\it WRes } H^{-k/d}\,(\la')^{(k-d)/d} +\cdots \bigr),
$$
in the Ces\`aro sense. (If the operator is not positive, one uses the
``four parts'' argument.) In our case, $H = a(|\Dslash| + \eps)^{-k}$ is pseudodifferential of order~$-k$, so
$$
d_H(x,x;\la')
\sim -\tfrac{2^{\piso{k/2}}\,\Omega_k\,a(x)}{k\,(2\pi)^k} \, ({\la'}^{-2} + \cdots),
$$
as $\la' \to \infty$ in the Ces\`aro sense; here $\Omega_k$ is the
hyper-area of the unit sphere in $\R^k$. We independently know that $H$
is compact, so on integrating the spectral density over~$x$ and over
$0 \leq \la' \leq \la$, we get that the number of eigenvalues of $H$ less than $\lambda$ is
$$
N_H(\la)
\sim \tfrac{2^{\piso{k/2}}\,\Omega_k\int a(x)\,d^kx}{k\,(2\pi)^k}
\,\la^{-1}\, \as \la \to \infty.
$$
This holds in the ordinary asymptotic sense, and not merely the
Ces\`aro sense, by the ``sandwich'' argument used in the proof of
\cite[Cor.~4.1]{EGBV}. So finally,
\begin{equation}
\la_m(H)
\sim \tfrac{2^{\piso{k/2}}\,\Omega_k\int a(x)\,d^kx}{k\,(2\pi)^k}\,
m^{-1} \as m \to \infty,
\label{eq:eastindian-discovery}
\end{equation}
and the Dixmier traceability of $a(|\Dslash| + \eps)^{-k}$, plus the
value of its trace, follow at once.

The rest is a long but almost trivial verification. For instance, $J$
is the charge conjugation operator on spinors; the algebra
$(\B,\star_0)$ is a suitable compactification; the domain $\H^\infty$
consists of the smooth spinors; and so on. 
Thus, we get the following

\begin{thm}
\label{th:commuters}
The triple $(\SS(\R^k), L^2(\R^k)\ox\C^{2^{\piso{k/2}}}, \Dslash)$ on
$\R^k$ defines a noncompact commutative geometry of spectral
dimension~$k$.
\end{thm}

What about the non-flat case (of a spin manifold such that $\Dslash$
is selfadjoint)? Mainly because the previous Ces\`aro summability
argument is purely local, everything carries over, if we choose
for~$\A$ the algebra of smooth and compactly supported functions. Of
course, in some contexts it may be useful to demand that $M$ also has
conic exits.

\subsection{The Moyal plane}
\label{The Moyal plane}

Let $\A = (\SS(\R^{2N}),\mop)$, with preferred unitization
$\Aun := (\B(\R^{2N}),\mop)$. The Hilbert space will be
$\H := L^2(\R^{2N}) \ox \C^{2^N}$ of ordinary square-integrable
spinors. The representation of $\A$ is given by
$\pi^\th \: \A \to \L(\H) : f \mapsto L^\th_f \ox 1_{2^N}$, where
$L^\th_f$ acts on the ``reduced'' Hilbert space
$\H_r := L^2(\R^{2N})$. In other words, if $a \in \A$ and
$\Psi \in \H$, to obtain $\pi^\th(a)\Psi$ we just left Moyal multiply
$\Psi$ by~$a$ componentwise.

This operator $\pi^\th(f)$ is bounded, since it acts diagonally
on $\H$ and $\|L^\th_f\| \leq (2\pi\theta)^{-N/2} \|f\|_2$ was proved
in Lemma~\ref{lm:norm-HS}. Under this action, the elements of $\H$ get
the lofty name of \textit{Moyal spinors}.

The selfadjoint Dirac operator is not ``deformed'': it will be the
ordinary Euclidean Dirac operator $\Dslash := -i\,\ga^\mu \del_\mu$,
where the hermitian matrices $\ga^1,\dots,\ga^{2N}$ satisfying 
$\{\ga^\mu, \ga^\nu\} = +2\,\delta^{\mu\nu}$ irreducibly represent the
Clifford algebra $\CC\ell \,\R^{2N}$ associated to $(\R^{2N},\eta)$, with
$\eta$ the standard Euclidean metric.

As a grading operator $\chi$ we take the usual chirality associated to
the Clifford algebra:
$$
\chi := \ga_{2N+1} := 1_{\H_r} \ox (-i)^N \ga^1 \ga^2 \dots \ga^{2N}.
$$
The notation $\ga_{2N+1}$ is a nod to physicists' $\ga_5$. Thus
$\chi^2 = (-1)^N (\ga^1\dots\ga^{2N})^2 = (-1)^{2N} = 1$ and
$\chi \ga^\mu = -\ga^\mu \chi$.

The real structure $J$ is chosen to be the usual charge conjugation
operator for spinors on $\R^{2N}$ endowed with an Euclidean metric.
Here, we only assume that $J^2 = \pm 1$ according to the ``sign table''
\eqref{commu} and that
$$
J (1_{\H_r} \ox \ga^\mu) J^{-1} = - 1_{\H_r} \ox \ga^\mu
$$
which guarantees the other requirements of the table. In
general, in a given representation, it can be written as
\begin{equation}
J := CK,
\label{eq:charge-conjn}
\end{equation}
where $C$ denotes a suitable $2^N \x 2^N$ unitary matrix and $K$ means
complex conjugation. 
An important property of~$J$ is
\begin{equation}
J( L^\th(f^*) \ox 1_{2^N}) J^{-1} = R^\th(f) \ox 1_{2^N},
\label{eq:right-Moyal}
\end{equation}
where $R^\th(f) \equiv R^\th_f$ is the right Moyal multiplication by
$f$; this follows from the antilinearity of~$J$ and the reversal of the
twisted product under complex conjugation.

\medskip

Lemma~\ref{lm:propriete}(iii) implies that $[\Dslash,\pi^\th(f)] = -i L^\th(\del_\mu f)\ox \gamma^\mu 
=: \pi^\th(\Dslash(f))$ and this is bounded for $f \in \Aun_\th = \B(\R^{2N})$ just as in the commutative case.

\subsubsection{The compactness condition}
\label{sec:axiom-comp}

In this subsection and the next, the main tools are techniques
developed some time ago for scattering theory problems, as summarized
in Simon's booklet~\cite[Chap.~4]{SimonTrace}. We adopt the convention
that $\L^\infty(\H) := \K(\H)$, with
$\|A\|_\infty := \|A\|_{\mathrm{op}}$.

Let $g\in L^\infty(\R^{2N})$. We define the operator $g(-i\nb)$ on
$\H_r$ as
$$
g(-i\nb)\psi := \F^{-1}(g\,\F\psi),
$$
where $\F$ is the ordinary Fourier transform. More in detail, for
$\psi$ in the correct domain,
$$
g(-i\nb)\psi(x)
= (2\pi)^{-2N}\iint e^{i\xi\.(x-y)}\,g(\xi)\psi(y)\,d^{2N}\xi\,d^{2N}y.
$$
The inequality
$\|g(-i\nb)\psi\|_2 = \|\F^{-1}g\F\psi\|_2\leq \|g\|_\infty\|\psi\|_2$
entails that $\|g(-i\nb)\|_\infty \leq \|g\|_\infty$.

\begin{thm}
\label{th:compacite}
Let $f\in \A$ and $\la \notin \spec \Dslash$. Then, if
$R_\Dslash(\la)$ is the resolvent operator of $\Dslash$, then
$\pi^\th(f)\, R_\Dslash(\la)$ is compact.
\end{thm}

Thanks to the first resolvent equation, $R_\Dslash(\la) =
R_\Dslash(\la') + (\la' - \la) R_\Dslash(\la) R_\Dslash(\la')$, we may
assume that $\la = i\mu$ with $\mu \in \R^*$. The theorem will follow
from a series of lemmas interesting in themselves.

\begin{lem}
\label{lm:comp}
If $f \in \SS$ and $0 \neq \mu \in \R$, then
$$
\pi^\th(f)  R_\Dslash(i\mu)    \in \K(\H)
\iff
\pi^\th(f) |R_\Dslash(i\mu)|^2 \in \K(\H).
$$
\end{lem}

\begin{proof}
We know that $L^\th(f)^* = L^\th(f^*)$. The ``only if'' part is
obvious since $R_\Dslash(i\mu)$ is a bounded normal operator.
Conversely, if $\pi^\th(f)|R_\Dslash(i\mu)|^2$ is compact, then the operator 
$\pi^\th(f)|R_\Dslash(i\mu)|^2 \pi^\th(f^*)$ is compact. Since
an operator $T$ is compact if and only if $TT^*$ is compact, the proof
is complete.
\end{proof}

The usefulness of this lemma stems from the diagonal nature of the
action of the operator $\pi^\th(f)|R_\Dslash(i\mu)|^2$ on $\H = \H_r \ox \C^{2^N}$;
so in our arguments it is feasible to replace $\H$ by $\H_r$,
$\pi^\th(f)$ by $L^\th_f$, and to use the scalar Laplacian
$-\Delta := -\sum_{\mu=1}^{2N} \del_\mu^2$ instead of the square of
the Dirac operator $\Dslash^2$.

\begin{lem}
\label{lm:HSO}
When $f,g \in \H_r$, $L^\th_f \,g(-i\nb)$ is a Hilbert--Schmidt
operator such that, for all real~$\th$,
$$
\|L^\th_f \,g(-i\nb)\|_2 = \|L_f^0 \,g(-i\nb)\|_2
= (2\pi)^{-N} \|f\|_2\,\|g\|_2.
$$
\end{lem}

\begin{proof}
To prove that an operator $A$ with integral kernel $K_A$ is
Hilbert--Schmidt, it suffices to check that
$\int |K_A(x,y)|^2 \,dx\,dy$ is finite, and this will be equal to
$\|A\|_2^2$~\cite[Thm.~2.11]{SimonTrace}. So we compute
$K_{L^\th(f)\,g(-i\nb)}$. In view of Lemma~\ref{lm:cojoreg},
$$
[L^\th(f)\,g(-i\nb)\psi](x) = \tfrac{1}{(2\pi)^{2N}} \iint
f(x - \tfrac{\th}2 S\xi) \,g(\xi) \psi(y)\, e^{i\xi\.(x-y)}
\,d^{2N}\xi \,d^{2N}y.
$$
Thus
$$
K_{L^\th(f)\,g(-i\nb)}(x,y) = \tfrac{1}{(2\pi)^{2N}} \int
f(x - \tfrac{\th}2 S\xi)\, g(\xi) \,e^{i\xi\.(x-y)} \,d^{2N}\xi,
$$
and $\int |K_{L^\th(f)\,g(-i\nb)}(x,y)|^2 \,dx\,dy$ is given by
\begin{align*}
\tfrac{1}{(2\pi)^{4N}}
&\int\cdots\int \bar{f}(x - \tfrac{\th}2 S\xi)\, \bar{g}(\xi)\,
f(x - \tfrac{\th}2 S\zeta)\, g(\zeta)\, e^{i(x-y)\.(\zeta-\xi)}
\,d^{2N}x \,d^{2N}y \,d^{2N}\zeta \,d^{2N}\xi
\\
&\quad = \tfrac{1}{(2\pi)^{2N}} \iint |f(x - \tfrac{\th}2 S\xi)|^2\, |g(\xi)|^2 \,d^{2N}x \,d^{2N}\xi
= (2\pi)^{-2N} \|f\|_2^2 \,\|g\|_2^2 < \infty.
\end{align*}
\end{proof}

\begin{rem}
As a consequence, we get
$$
\|.\|_2\mbox{-}\lim_{\th\to 0} L^\th_f \,g(-i\nb) = L_f^0 \,g(-i\nb).
$$
\end{rem}

\begin{lem}
\label{lm:schatten}
If $f \in \H_r$ and $g \in L^p(\R^{2N})$ with $2 \leq p < \infty$,
then $L^\th_f\,g(-i\nb) \in \L^p(\H_r)$ and
$$
\|L^\th_f\,g(-i\nb)\|_p
\leq (2\pi)^{-N(1/2+1/p)} \th^{-N(1/2-1/p)} \,\|f\|_2\,\|g\|_p.
$$
\end{lem}

\begin{proof}
The case $p = 2$ (with equality) is just the previous lemma. For
$p = \infty$, we estimate $\|L^\th_f\,g(-i\nb)\|_\infty \leq
(2\pi\th)^{-N/2} \|f\|_2 \,\|g\|_\infty$: since
$\|L^\th_f\,g(-i\nb)\|_\infty \leq
\|L^\th_f\|_\infty \,\|g(-i\nb)\|_\infty$, this follows from
Lemma~\ref{lm:norm-HS} and a previous remark.

Now use complex interpolation for $2 < p < \infty$. For that, we first
note that we may suppose $g \geq 0$: defining the function $a$ with
$|a| = 1$ and $g = a|g|$, we see that
\begin{align*}
\|L^\th_f\,g(-i\nb)\|_2^2
&= \Tr(|L^\th_f\,g(-i\nb)|^2)
= \Tr(\bar{g}(-i\nb)\,L^\th_{f^*}\,L^\th_f\,g(-i\nb))
\\
&= \Tr(|g|(-i\nb)\,\bar{a}(-i\nb)\,L^\th_{f^*}
\,L^\th_f\,a(-i\nb)\,|g|(-i\nb))
\\
&= \Tr(\bar{a}(-i\nb)\,|g|(-i\nb)\,L^\th_{f^*}
\,L^\th_f\,|g|(-i\nb)\,a(-i\nb))
\\
&= \Tr(|L^\th_f\,|g|(-i\nb)|^2) = \|L^\th_f\,|g|(-i\nb)\|_2^2,
\end{align*}
and
\begin{align*}
\|L^\th_f\,g(-i\nb)\|_\infty
&= \|L^\th_f\,a(-i\nb)\,|g|(-i\nb)\|_\infty
=\|L^\th_f\,|g|(-i\nb)\,a(-i\nb)\|_\infty
\\
&\leq \|L^\th_f\,|g|(-i\nb)\|_\infty \,\|a(-i\nb)\|_\infty
= \|L^\th_f\,|g|(-i\nb)\|_\infty.
\end{align*}
Secondly, for any positive, bounded function $g$ with compact support,
we define the maps:
$$
F_p : z \mapsto  L^\th_f\, g^{zp}(-i\nb) : S = \set{z\in\C \,\, \vert \,\, 0 \leq \Re z \leq \half} \to \L(\H_r).
$$
For all $y\in \R$,
$F_p(iy) = L^\th_f \, g^{iyp}(-i\nb) \in \L^\infty(\H_r)$ by
Lemma~\ref{lm:HSO} since $g$, being compactly supported, lies in
$\H_r$. Moreover, $\|F_p(iy)\|_\infty \leq(2\pi\th)^{-N/2} \|f\|_2$.

Also, by Lemma~\ref{lm:HSO}, $F_p(\half+iy) \in \L^2(\H_r)$ and
$\|F_p(\half+iy)\|_2 = (2\pi)^{-N} \|f\|_2 \,\|g^{p/2}\|_2$. Then
complex interpolation (see \cite[Chap.~9]{ReedSII}
and~\cite{SimonTrace}) yields $F(z) \in \L^{1/\Re z}(\H_r)$, for all
$z$ in the strip $S$. Moreover,
$$
\|F_p(z)\|_{1/\Re z}
\leq \|F(0)\|_\infty^{1-2\Re z}\,\|F(\half)\|_2^{2\Re z}
= \|f\|_2 (2\pi\th)^{-\frac{N}{2}(1-2\Re z)}(2\pi)^{-2N\Re z}\,
\|g^{p/2}\|_2^{2\Re z},
$$
and applying this result at $z = 1/p$, we get for such $g$:
$$
\|L^\th_f\,g(-i\nb)\|_p =  \|F(1/p)\|_p
\leq(2\pi)^{-N(1/2+1/p)}\th^{-N(1/2-1/p)}\|f\|_2\,\|g\|_p .
$$
We finish by using the density of compactly supported bounded
functions in $L^p(\R^{2N})$.
\end{proof}

\begin{lem}
\label{lm:neo-schatten}
If $f \in \SS$ and $0 \neq \mu \in \R$, then
$\pi^\th(f)\,|R_\Dslash(i\mu)|^2 \in \L^p$ for $p > N$.
\end{lem}

\begin{proof} We see that
$$
\pi^\th(f)\,|R_\Dslash(i\mu)|^2
= (L^\th_f \ox 1_{2^N})\,(\Dslash-i\mu)^{-1} (\Dslash+i\mu)^{-1}
= L^\th_f \,(-\del^\nu\del_\nu +\mu^2)^{-1} \ox 1_{2^N}.
$$
So this operator acts diagonally on $\H_r \ox \C^{2^N}$ and
Lemma~\ref{lm:schatten} implies that
$$
\bigl\| L^\th_f \,(-\del^\nu\del_\nu +\mu^2)^{-1} \bigr\|_p
\leq (2\pi)^{-N(1/2+1/p)}\th^{-N(1/2-1/p)}\,\|f\|_2
\biggl(\int \tfrac{d^{2N}\xi}{(\xi^\nu\xi_\nu + \mu^2)^p}\biggr)^{1/p},
$$
which is finite for $p > N$.
\end{proof}

\begin{proof}[Proof of Theorem~\ref{th:compacite}]
By Lemma~\ref{lm:comp}, it was enough to prove that
$\pi^\th(f)\,|R_\Dslash(i\mu)|^2$ is compact for a nonzero
real~$\mu$.
\end{proof}

\subsubsection{Spectral dimension of the Moyal planes}

\begin{thm}
\label{th:dim-spectrale}
The spectral dimension of the Moyal $2N$-plane spectral triple is~$2N$.
\end{thm}

We shall first establish existence properties. \\
Thanks to
Lemma~\ref{lm:schatten} and because
$[\Dslash,\pi^\th(f)] = -iL^\th(\del_\mu f) \ox \ga^\mu$, we see that
$\pi^\th(f)(\Dslash^2 + \eps^2)^{-l}$ and
$[\Dslash,\pi^\th(f)]\,(\Dslash^2 + \eps^2)^{-l}$ lie in $\L^p(\H)$
whenever $p > N/l$ (we always assume $\eps > 0$). In the next lemma, we
show that $[|\Dslash|,\pi^\th(f)]\, (\Dslash^2 + \eps^2)^{-l}$ has the
same property of summability; this will become our main technical
instrument for the subsection.

\begin{lem}
\label{lm:commutateur}
If $f \in \SS$ and $\half \leq l \leq N$, then
$[|\Dslash|,\pi^\th(f)]\,(\Dslash^2 + \eps^2)^{-l} \in \L^p(\H)$
for $p > N/l$.
\end{lem}

\begin{proof}
We use the following spectral identity for a positive operator $A$:
$$
A = \tfrac{1}{\pi} \int_0^\infty \tfrac{A^2}{A^2 + \mu} \,\tfrac{d\mu}{\sqrt{\mu}},
$$
and another identity for any operators $A$, $B$ and
$\la \notin \spec A$:
\begin{equation}
[B, (A-\la)^{-1}] = (A - \la)^{-1} [A,B] (A - \la)^{-1}.
\label{eq:deriv-inv}
\end{equation}
\\
Hence, for any $\rho > 0$,
\begin{align}
\label{eq:spectid}
[|\Dslash|, \pi^\th(f)]
&= [|\Dslash|+\rho, \pi^\th(f)]
= \tfrac{1}{\pi} \int_0^\infty \biggl[
\tfrac{(|\Dslash|+\rho)^2}{(|\Dslash|+\rho)^2+\mu},
\pi^\th(f) \biggr] \, \tfrac{d\mu}{\sqrt{\mu}}
\nonumber\\
&= \tfrac{1}{\pi} \int_0^\infty
\biggl(1 - \tfrac{(|\Dslash|+\rho)^2}{(|\Dslash|+\rho)^2+\mu}\biggr)
\bigl[(|\Dslash|+\rho)^2, \pi^\th(f)\bigr]
\tfrac{1}{(|\Dslash|+\rho)^2+\mu} \, \tfrac{d\mu}{\sqrt{\mu}}
\nonumber\\
&= \tfrac{1}{\pi} \int_0^\infty
\tfrac{1}{(|\Dslash|+\rho)^2+\mu}
\bigl[(|\Dslash|+\rho)^2,\pi^\th(f)\bigr]
\tfrac{1}{(|\Dslash|+\rho)^2+\mu} \,\sqrt{\mu} \,d\mu
\\
&= \tfrac{1}{\pi} \int_0^\infty
\tfrac{1}{(|\Dslash|+\rho)^2+\mu}
\biggl(-\pi^\th(\del^\mu \del_\mu f)
- 2i(L^\th(\del_\mu f) \ox \ga^\mu) \Dslash
+ 2\rho \bigl[ |\Dslash|, \pi^\th(f) \bigr] \biggr)
\nonumber\\
&\hspace{13em} \x
\tfrac{1}{(|\Dslash|+\rho)^2+\mu} \,\sqrt{\mu} \,d\mu.
\nonumber
\end{align}
This implies that
\begin{align*}
\bigl\| [|\Dslash|,\pi^\th(f)]\,(\Dslash^2 + \eps^2)^{-l} \bigr\|_p
&\leq \tfrac{1}{\pi} \int_0^\infty
\bigl\| \tfrac{1}{(|\Dslash|+\rho)^2+\mu}
\Bigl( -\pi^\th(\del^\mu \del_\mu f)
- 2i(L^\th(\del_\mu f) \ox \ga^\mu) \Dslash
\\
&\hspace{4em} + 2\rho \bigl[ |\Dslash|, \pi^\th(f) \bigr] \Bigr)
\tfrac{1}{(|\Dslash|+\rho)^2+\mu}\,(\Dslash^2 + \eps^2)^{-l}\bigr\|_p
\,\sqrt{\mu} \,d\mu.
\end{align*}

Thus, the proof reduces to show that for any $f \in \SS$,
\begin{equation}
\label{eq:integrale}
\tfrac{1}{\pi} \int_0^\infty
\bigl\| \tfrac{1}{(|\Dslash|+\rho)^2+\mu} \,\pi^\th(f)\Dslash
\,\tfrac{1}{(|\Dslash|+\rho)^2+\mu} \,(\Dslash^2 + \eps^2)^{-l}
\bigr\|_p \,\sqrt{\mu} \,d\mu < \infty.
\end{equation}

Since the Schatten $p$-norm is a symmetric norm, and since, as in the
proof of Theorem~\ref{th:compacite}, only the reduced Hilbert space is
affected, expression~\eqref{eq:integrale} is bounded by by
\begin{align*}
\tfrac{1}{\pi} \int_0^\infty
&\bigl\| \tfrac{1}{(|\Dslash|+\rho)^2+\mu} \bigr\|^{3/2} \,
\bigl\| \tfrac{\Dslash}{(\Dslash^2 + \eps^2)^{1/2}} \bigr\| \,
\bigl\| \pi^\th(f) \tfrac{1}{(\Dslash^2 + \eps^2)^{l-1/2}}\,
\tfrac{1}{((|\Dslash|+\rho)^2+\mu)^{1/2}} \bigr\|_p
\,\sqrt{\mu} \,d\mu
\\
&\leq \tfrac{1}{\pi} \int_0^\infty
\bigl\| \pi^\th(f)\, (\Dslash^2 + \eps^2)^{-l+1/2}
((|\Dslash|+\rho)^2 + \mu)^{-1/2} \bigr\|_p
\, \tfrac{\sqrt{\mu} \,d\mu}{(\mu+\rho^2)^{3/2}}.
\end{align*}

Thanks to Lemma~\ref{lm:schatten}, we can estimate the
$\mu$-dependence of the last $p$-norm:
\begin{align*}
\bigl\| \pi^\th(f) &((|\Dslash|+\rho)^2+\mu)^{-1/2}
(\Dslash^2+\eps^2)^{-l+1/2} \bigr\|_p
\\
&\leq (2\pi)^{-N(1/2+1/p)}\th^{-N(1/2-1/p)} \|f\|_2\,
\bigl\| ((|\xi|+\rho)^2 +\mu)^{-1/2} (|\xi|^2 + \eps^2)^{-l+1/2}
\bigr\|_p
\\
&\leq C(p,\th) \bigl\| ((|\xi|+\rho)^2 +\mu)^{-1/2} \bigr\|_q \, \bigl\| (|\xi|^2+\eps^2)^{-l+1/2} \bigr\|_r;
\end{align*}
with $p^{-1} = q^{-1} + r^{-1}$ appropriately chosen, these integrals
are finite for all $q > 2N$ and $r > 2N/(2l-1)$; for $l = \half$, take
$r = \infty$ and $q = p$. For such values,
\begin{align*}
&\bigl\| \pi^\th(f) ((|\Dslash|+\rho)^2+\mu)^{-1/2}
(\Dslash^2+\eps^2)^{-l+1/2} \bigr\|_p
\\
&\leq C(p,\th,N;f) \|(|\xi|^2+\eps^2)^{-l+1/2}\|_r \, \Omega_{2N}^{1/q}
\biggl( \int_0^\infty \tfrac{R^{2N-1}}{((R+\rho)^2+\mu)^{q/2}}\,dR
\biggr)^{1/q}
\\
&= C(p,\th,N;f) \|(|\xi|^2+\eps^2)^{-l+1/2}\|_r \, \pi^{N/q} \,
\tfrac{\Ga^{1/q}(\tfrac{q}{2} - N)}{\Ga^{1/q}(\tfrac{q}{2})}\,
\mu^{-1/2 + N/q} =: C'(p,q,\th,N;f) \,\mu^{-1/2 + N/q}.
\end{align*}
Finally, the integral~\eqref{eq:integrale} is less than
$$
C'(p,q,\th,N;f) \int_0^\infty
\tfrac{\mu^{N/q}}{(\mu+\rho^2)^{3/2}} \,d\mu,
$$
which is finite for $q > 2N$ and $p > N/l$. This concludes the proof.
\end{proof}

\begin{lem}
\label{lm:Cwikel}
If  $f \in \SS$, then
$\pi^\th(f) \,(|\Dslash|+\eps)^{-1}\, \pi^\th(f^*) \in \L^{2N+}(\H)$.
\end{lem}

\begin{proof}
This is an extension to the Moyal context of the renowned inequality by
Cwikel~\cite{SimonTrace}. As remarked before, it is
possible to replace $\Dslash^2$ by $-\Delta$, $\pi^\th(f)$ by $L^\th_f$
and $\H$ by $\H_r$. Consider
$g(-i\nb) := (\sqrt{-\Delta} + \eps)^{-1}$. Since $g$ is positive, it
can be decomposed as $g = \sum_{n\in\Z} g_n$ where
$$
g_n(x) := \begin{cases} g(x) &\text{if $2^{n-1} < g(x) \leq 2^n$}, \\
0 &\text{otherwise}. \end{cases}
$$

For each $n \in \Z$, let $A_n$ and $B_n$ be the two operators
$$
A_n := \sum_{k\leq n} L^\th_f \,g_k(-i\nb) \,L^\th_{f^*},  \quad
B_n := \sum_{k>n}     L^\th_f \,g_k(-i\nb) \,L^\th_{f^*}.
$$
We estimate the uniform norm of the first part:
\begin{align*}
\|A_n\|_\infty
&\leq \|L^\th_f\|^2\, \bigl\|\sum_{k\leq n} g_k(-i\nb)\bigr\|_\infty
\leq (2\pi\th)^{-N} \|f\|_2^2
\bigl\| \sum_{k\leq n} g_k \bigr\|_\infty\\
&\leq (2\pi\th)^{-N} \|f\|_2^2 \,2^n =: 2^n\,c_1(\th,N;f).
\end{align*}
The trace norm of $B_n$ can be computed using Lemma~\ref{lm:HSO}:
\begin{align*}
\|B_n\|_1
&= \bigl\|\Bigl(\smash[b]{\sum_{k>n}g_k(-i\nb)}\Bigr)^{1/2}
L^\th_{f^*} \bigr\|_2^2
= \bigl\| L^\th_f \Bigl(\smash[b]{\sum_{k>n}g_k(-i\nb)}\Bigr)^{1/2}
\bigr\|_2^2
= (2\pi)^{-2N} \|f\|_2^2 \,
\bigl\| \Bigl(\smash[b]{\sum_{k>n} g_k}\Bigr)^{1/2} \bigr\|_2^2
\\
&= (2\pi)^{-2N} \|f\|_2^2\, \bigl\| \sum_{k>n} g_k \bigr\|_1
= (2\pi)^{-2N} \|f\|_2^2 \,\sum_{k>n} \|g_k\|_1
\\
&\leq (2\pi)^{-2N} \|f\|_2^2 \,
\sum_{k>n} \|g_k\|_\infty \,\nu\{\supp(g_k)\},
\end{align*}
where $\nu$ is the Lebesgue measure on $\R^{2N}$. By definition,
$\|g_k\|_\infty \leq 2^k$ and
\begin{align*}
\nu\{\supp(g_k)\}
&= \nu\set{x \in \R^{2N} : 2^{k-1} < g(x) \leq 2^k}
\leq \nu\set{x \in \R^{2N} \,\,\vert \,\, (|x|+\eps)^{-1} \geq 2^{k-1}}
\\
&\leq 2^{2N(1-k)} \,c_2.
\end{align*}
Therefore
\begin{align*}
\|B_n\|_1
&\leq (2\pi)^{-2N} \|f\|_2^2 \, 2^{2N} c_2\, \sum_{k>n} 2^{k(1-2N)}
\\
&< \pi^{-2N} \,c_2 \,\|f\|_2^2 \, 2^{n(1-2N)}
=: 2^{n(1-2N)} \,c_3(N;f),
\end{align*}
where the second inequality follows because $N > \half$.

We can now estimate the $m$th singular value $\mu_m$ of $B_n$
(arranged in decreasing order with multiplicity): $\|B_n\|_1 = \sum_{k=0}^\infty \mu_k(B_n)$. Note that, for
$m=1,2,3,\cdots$, we get that 
$\|B_n\|_1 \geq \sum_{k=0}^{m-1} \mu_k(B_n) \geq m\,\mu_m(B_n)$. Thus,
$\mu_m(B_n) \leq  \|B_n\|_1\, m^{-1}\leq 2^{n(1-2N)} \,c_3 \, 
m^{-1}$. Now Fan's
inequality~\cite[Thm.~1.7]{SimonTrace} yields
\begin{align*}
\mu_m(L^\th_f \,g(-i\nb) \,L^\th_{f^*})
&= \mu_m(A_n + B_n) \leq \mu_1(A_n) + \mu_m(B_n)
\\
&\leq \|A_n\| +  \|B_n\|_1 \,m^{-1}
\leq 2^n\,c_1 + 2^{n(1-2N)} \,c_3 \,m^{-1}.
\end{align*}
Given $m$, choose $n \in \Z$ so that $2^n \leq m^{-1/2N} < 2^{n+1}$.
Then
$$
\mu_m(L^\th_f \,g(-i\nb) \,L^\th_{f^*})
\leq  c_1 \, m^{-1/2N} + c_3 \, m^{-(1-2N)/2N} m^{-1}
=: c_4(\th,N;f) \, m^{-1/2N}.
$$
Therefore
$L^\th_f\,(\sqrt{-\Delta}+\eps)^{-1}\,L^\th_{f^*} \in \L^{2N+}(\H_r)$,
and the statement of the lemma follows.
\end{proof}

\begin{cly}
\label{cr:fragmentation}
If  $f,g \in \SS$, then
$\pi^\th(f)\, (|\Dslash|+\eps)^{-1} \,\pi^\th(g) \in \L^{2N+}(\H)$.
\end{cly}

\begin{proof}
Consider
$\pi^\th(f \pm g^*)\,(|\Dslash|+\eps)^{-1}\,\pi^\th(f^* \pm g)$ and
$\pi^\th(f \pm ig^*)\,(|\Dslash|+\eps)^{-1}\,\pi^\th(f^* \mp ig)$.
\end{proof}

\begin{cly}
\label{cr:spec-dim-one}
If $h \in \SS$, then
$\pi^\th(h)\,(|\Dslash|+\eps)^{-1} \in \L^{2N+}(\H)$.
\end{cly}

\begin{proof}
Let $h = f \mop g$. Then
$$
\pi^\th(h)\, (|\Dslash|+\eps)^{-1}
= \pi^\th(f)\, (|\Dslash|+\eps)^{-1}\,\pi^\th(g)
+ \pi^\th(f)\, [\pi^\th(g), (|\Dslash|+\eps)^{-1}],
$$
and we obtain from the identity \eqref{eq:deriv-inv} that
$$
\pi^\th(h) \,(|\Dslash|+\eps)^{-1}
= \pi^\th(f) \,(|\Dslash|+\eps)^{-1} \,\pi^\th(g)
+ \pi^\th(f) \,(|\Dslash|+\eps)^{-1} \,[|\Dslash|, \pi^\th(g)]\,
(|\Dslash|+\eps)^{-1}.
$$
By arguments similar to those of lemmata \ref{lm:schatten} and
\ref{lm:commutateur}, the last term belongs to $\L^p$ for $p > N$, and
thus to~$\L^{2N+}$.
\end{proof}

Boundedness of $(|\Dslash|+\eps)(\Dslash^2+\eps^2)^{-1/2}$ follows
from elementary Fourier analysis. And so the last corollary means that
the spectral triple is ``$2N^+$-summable''. We have taken care of the
first assertion of the theorem. The next lemma is the last property of
existence that we need.

\begin{lem}
\label{lm:existence}
If $f\in \SS$, then $\pi^\th(f) (|\Dslash|+\eps)^{-2N}$ and
$\pi^\th(f) (\Dslash^2+\eps^2)^{-N}$ are in $\L^{1+}(\H)$.
\end{lem}

\begin{proof}
It suffices to prove that
$\pi^\th(f) (|\Dslash|+\eps)^{-2N} \in \L^{1+}(\H)$. We factorize
$f \in \SS$ according to Proposition~\ref{pr:factorization}, with the
following notation:
\begin{align*}
f &= f_1 \mop f_2 = f_1 \mop f_{21} \mop f_{22}
= f_1 \mop f_{21} \mop f_{221} \mop f_{222}
\\
&= \cdots = f_1 \mop f_{21} \mop f_{221} \mop\cdots\mop
f_{22\cdots 21} \mop f_{22\cdots 22}.
\end{align*}
Therefore,
\begin{align}
\pi^\th(f) \,(|\Dslash|+\eps)^{-2N}
&= \pi^\th(f_1) \,(|\Dslash|+\eps)^{-1} \,\pi^\th(f_2)
\,(|\Dslash|+\eps)^{-2N+1}
\nonumber \\
&\qquad + \pi^\th(f_1) \,(|\Dslash|+\eps)^{-1}
\,[|\Dslash|,\pi^\th(f_2)] \,(|\Dslash|+\eps)^{-2N}.
\label{eq:modtrace}
\end{align}

By Lemma~\ref{lm:schatten},
$\pi^\th(f_1)(|\Dslash|+\eps)^{-1} \in \L^p(\H)$ whenever $p > 2N$;
and by Lemma~\ref{lm:commutateur}, the term
$[|\Dslash|,\pi^\th(f_2)] (|\Dslash|+\eps)^{-2N}$ lies in $\L^q(\H)$
for $q > 1$. Hence, the last term on the right hand side of
equation~\eqref{eq:modtrace} lies in $\L^1(\H)$. We may write the
following equivalence relation:
$$
\pi^\th(f) (|\Dslash|+\eps)^{-2N}
\sim \pi^\th(f_1) (|\Dslash|+\eps)^{-1} \pi^\th(f_2)
(|\Dslash|+\eps)^{-2N+1},
$$
where $A \sim B$ for $A,B \in \K(\H)$ means that $A - B$ is
trace-class. Thus,
\begin{align*}
\pi^\th(f) &(|\Dslash|+\eps)^{-2N}
\sim \pi^\th(f_1) (|\Dslash|+\eps)^{-1} \pi^\th(f_2)
(|\Dslash|+\eps)^{-2N+1}
\\
&= \pi^\th(f_1) (|\Dslash|+\eps)^{-1} \pi^\th(f_{21})
(|\Dslash|+\eps)^{-1} \pi^\th(f_{22}) (|\Dslash|+\eps)^{-2N+2}
\\
&\qquad + \pi^\th(f_1) (|\Dslash|+\eps)^{-1} \pi^\th(f_{21})
(|\Dslash|+\eps)^{-1} \,[|\Dslash|, \pi^\th(f_{22})]\,
(|\Dslash|+\eps)^{-2N+1}
\\
&\sim \pi^\th(f_1) (|\Dslash|+\eps)^{-1} \pi^\th(f_{21})
(|\Dslash|+\eps)^{-1} \pi^\th(f_{22}) (|\Dslash|+\eps)^{-2N+2}
\sim \cdots \\
&\sim \pi^\th(f_1) (|\Dslash|+\eps)^{-1} \pi^\th(f_{21})
(|\Dslash|+\eps)^{-1} \pi^\th(f_{221}) (|\Dslash|+\eps)^{-1}
\dots \pi^\th(f_{22\cdots 22}) (|\Dslash|+\eps)^{-1}.
\end{align*}

The second equivalence relation holds because
$\pi^\th(f_1)(|\Dslash|+\eps)^{-1}\pi^\th(f_{21})(|\Dslash|+\eps)^{-1}
\in \L^p(\H)$ for $p > N$ by Lemma~\ref{lm:schatten}, and
$[|\Dslash|,\pi^\th(f_{22})] (|\Dslash|+\eps)^{-2N+1} \in \L^q(\H)$
for $q > 2N/(2N - 1)$ by Lemma~\ref{lm:commutateur} again. The other
equivalences come from similar arguments.
Corollary~\ref{cr:fragmentation}, the H\"older inequality (see
\cite[Prop.~7.16]{Polaris}) and the inclusion
$\L^1(\H) \subset \L^{1+}(\H)$ finally yield the result.
\end{proof}

\smallskip

Now we go for the computation of the Dixmier trace. Using the
regularized trace for a $\Psi$DO:
$$
\Tr_\La(A) := (2\pi)^{-2N} \iint_{|\xi|\leq\La} \sigma[A](x,\xi)
\,d^{2N}\xi \,d^{2N}x,
$$
the result can be conjectured because
$\lim_{\La\to\infty}\, \Tr_\La(\.)/\log(\La^{2N})$ is heuristically
linked with the Dixmier trace, and the following computation:
\begin{align*}
\lim_{\La\to\infty} &\tfrac{1}{2N\log\La}
\Tr_\La \bigl(\pi^\th(f)(\Dslash^2 + \eps^2)^{-N}\bigr)
\\
&= \lim_{\La\to\infty}\,t\,\tfrac{2^N}{2N(2\pi)^{2N}\log\La}
\iint_{|\xi|\leq\La} f(x - \tfrac{\th}{2}S\xi)\,
(|\xi|^2 + \eps^2)^{-N} \,d^{2N}\xi\,d^{2N}x
\\
&= \tfrac{2^N\,\Omega_{2N}}{2N\,(2\pi)^{2N}} \int f(x) \,d^{2N}x.
\end{align*}
This is precisely the same result of~\eqref{eq:eastindian-discovery},
in the commutative case, for $k = 2N$. However, to establish it
rigorously in the Moyal context requires a subtler strategy. We shall
compute the Dixmier trace of $\pi^\th(f)\,(\Dslash^2 + \eps^2)^{-N}$
as the residue of the ordinary trace of a related meromorphic family
of operators. 
In turn we are allowed to introduce the explicit symbol formula that
will establish measurability~\cite{Book,Polaris}, too.

We seek first to verify that $\A_\th$ has \textit{analytical dimension} equal to $2N$; that is, for
$f \in \A_\th$ the operator $\pi^\th(f)\,(\Dslash^2 + \eps^2)^{-z/2}$
is trace-class if $\Re z>2N$.

\begin{lem}
\label{lm:in-extremis}
If $f \in \SS$, then $L^\th_f\,(\Dslash^2 + \eps^2)^{-z/2}$ is
trace-class for $\Re z>2N$, and
$$
\Tr[L^\th_f\,(\Dslash^2+\eps^2)^{-z/2}] = (2\pi)^{-2N} \iint
f(x)\, (|\xi|^2 + \eps^2)^{-z/2} \,d^{2N}\xi \,d^{2N}x.
$$
\end{lem}

\begin{proof}
If $a(x,\xi) \in \K_p(\R^{2k})$, for $p < -k$, is the symbol of a
pseudodifferential operator $A$, then the operator is trace-class and
moreover
$$
\Tr A = (2\pi)^{-k} \iint a(x,\xi) \,d^kx\,d^k\xi.
$$
This is easily proved by taking $a \in \SS(\R^{2k})$ first and
extending the resulting formula by continuity.

In our case, the symbol formula for a product of $\PsiDO$s yields, for
$p > N$,
\begin{align*}
\sigma \bigl[L^\th_f(-\Delta + \eps^2)^{-p}\bigr](x,\xi)
&= \sum_{\a\in\N^N} \frac{(-i)^{|\a|}}{\a!} \;
\del_\xi^\a \sigma[L^\th_f](x,\xi) \,
\del_x^\a \sigma\bigl[(-\Delta + \eps^2)^{-p}\bigr](x,\xi)
\\
&= \sigma[L^\th_f](x,\xi) \,
\sigma\bigl[(-\Delta + \eps^2)^{-p}\bigr](x,\xi)
\\
&= f(x - \tfrac{\th}{2}S\xi)\,(|\xi|^2 + \eps^2)^{-p}.
\end{align*}
Therefore, for $p > N$,
\begin{align*}
\Tr\bigl(L^\th_f(-\Delta + \eps^2)^{-p}\bigr)
&= (2\pi)^{-2N}
\iint f(x - \tfrac{\th}{2} S\xi)\,(|\xi|^2 + \eps^2)^{-p}
\,d^{2N}\xi \,d^{2N}x
\\
&= (2\pi)^{-2N} \iint f(x)\, (|\xi|^2 + \eps^2)^{-p} \,d^{2N}\xi \,d^{2N}x.
\end{align*}
\end{proof}

We continue with a technical lemma, in the spirit of~\cite{RennieLocal}. Consider the approximate unit
$\{\ee_K\}_{K\in\N} \subset \A_c$ where
$\ee_K := \sum_{0\leq|n|\leq K} f_{nn}$. These $\ee_K$ are projectors
with a natural ordering: $\ee_K \mop \ee_L = \ee_L \mop \ee_K = \ee_K$
for $K \leq L$, and they are local units for~$\A_c$.

\begin{lem}
\label{lm:traceclass}
Let $f \in \A_{c,K}$. Then
$$
\pi^\th(f) \,(\Dslash^2+\eps^2)^{-N} - \pi^\th(f) \,
\bigl( \pi^\th(\ee_K)(\Dslash^2+\eps^2)^{-1}\pi^\th(\ee_K) \bigr)^N
\in \L^1(\H).
$$
\end{lem}

\begin{proof}
For simplicity we use the notation $e := \ee_K$ and
$e_n := \ee_{K+n}$. By the boundedness of $\pi^\th(f)$, we may assume
that $f = e \in \A_{c,K}$.

Because $e_n \mop e = e \mop e_n = e$, it is clear that
\begin{equation}
\pi^\th(e) (\Dslash+\la)^{-1} \bigl( 1 - \pi^\th(e_n) \bigr)
= \pi^\th(e) \,(\Dslash+\la)^{-1} \,[\Dslash,\pi^\th(e_n)]
\,(\Dslash+\la)^{-1}.
\label{eq:difference}
\end{equation}
Also, $\pi^\th(e)\,[\Dslash, \pi^\th(e_n)] =
[\Dslash,\pi^\th(e\mop e_n)] - [\Dslash,\pi^\th(e)]\,\pi^\th(e_n) = 0$
because we have the relation $[\Dslash,\pi^\th(e)]\,\pi^\th(e_n) = [\Dslash,\pi^\th(e)]$ for
$n \geq1$. We
obtain
\begin{align*}
& A_n := \pi^\th(e) (\Dslash+\la)^{-1} [\Dslash,\pi^\th(e_n)]
(\Dslash+\la)^{-1}
\nonumber\\
&= \pi^\th(e) (\Dslash+\la)^{-1} [\Dslash,\pi^\th(e_1)]
(\Dslash+\la)^{-1} [\Dslash,\pi^\th(e_n)] (\Dslash+\la)^{-1}
\nonumber\\
&= \pi^\th(e) (\Dslash+\la)^{-1} [\Dslash,\pi^\th(e_1)] \pi^\th(e_2)
(\Dslash+\la)^{-1} [\Dslash,\pi^\th(e_n)] (\Dslash+\la)^{-1}
= \cdots
\nonumber\\
&= \bigl( \pi^\th(e) (\Dslash+\la)^{-1} \bigr)
\bigl( [\Dslash,\pi^\th(e_1)] (\Dslash+\la)^{-1} \bigr)
\bigl( [\Dslash,\pi^\th(e_2)] (\Dslash+\la)^{-1} \bigr)
\cdots \\
&\hspace{4cm} \cdots \bigl( [\Dslash,\pi^\th(e_n)] (\Dslash+\la)^{-1} \bigr).
\end{align*}

Taking $n = 2N$ here, $A_{2N}$ appears as a product of $2N+1$ terms in
parentheses, each in $\L^{2N+1}(\H)$ by Lemma~\ref{lm:schatten}. Hence,
by H\"older's inequality, $A_{2N}$ is trace-class and therefore
$\pi^\th(e) (\Dslash+\la)^{-1} (1 - \pi^\th(\ee_{2N})) \in \L^1(\H)$.
Thus,
\begin{align}
\pi^\th(e)\, &(\Dslash^2+\eps^2)^{-1} \bigl(1-\pi^\th(\ee_{4N})\bigr)
\nonumber\\
&= \pi^\th(e) (\Dslash-i\eps)^{-1}
\bigl(1-\pi^\th(\ee_{2N}) + \pi^\th(\ee_{2N})\bigr)
(\Dslash+i\eps)^{-1} \bigl(1-\pi^\th(\ee_{4N})\bigr)
\nonumber\\
&= \pi^\th(e)(\Dslash-i\eps)^{-1} \bigl(1-\pi^\th(\ee_{2N})\bigr)
(\Dslash+i\eps)^{-1} \bigl(1-\pi^\th(\ee_{4N})\bigr)
\nonumber\\
&\qquad + \pi^\th(e) (\Dslash-i\eps)^{-1} \pi^\th(\ee_{2N})
(\Dslash+i\eps)^{-1} \bigl(1-\pi^\th(\ee_{4N})\bigr) \in \L^1(\H).
\label{eq:difference2}
\end{align}
This is to say $\pi^\th(e) (\Dslash^2+\eps^2)^{-1} \sim
\pi^\th(e) (\Dslash^2+\eps^2)^{-1} \pi^\th(\ee_{4N})$.
Shifting this property, we get
\begin{align*}
\pi^\th(e) (\Dslash^2+\eps^2)^{-N}
&\sim \pi^\th(e) (\Dslash^2+\eps^2)^{-1} \pi^\th(\ee_{4N})
(\Dslash^2+\eps^2)^{-N+1}
\\
&\sim \pi^\th(e) (\Dslash^2+\eps^2)^{-1} \pi^\th(\ee_{4N})
(\Dslash^2+\eps^2)^{-1} \pi^\th(\ee_{8N}) (\Dslash^2+\eps^2)^{-N+2}
\sim \cdots
\\
&\sim \pi^\th(e) (\Dslash^2+\eps^2)^{-1} \pi^\th(\ee_{4N})
(\Dslash^2+\eps^2)^{-1} \pi^\th(\ee_{8N}) \cdots
(\Dslash^2+\eps^2)^{-1} \pi^\th(\ee_{4N^2}).
\end{align*}
By identity \eqref{eq:deriv-inv}, the last term on the right equals
\begin{align*}
&\pi^\th(e) (\Dslash+i\eps)^{-1} \pi^\th(e) (\Dslash-i\eps)^{-1}
\pi^\th(\ee_{4N}) (\Dslash^2+\eps^2)^{-1} \pi^\th(\ee_{8N}) \cdots
(\Dslash^2+\eps^2)^{-1} \pi^\th(\ee_{4N^2})
\\
&\quad + \pi^\th(e) (\Dslash+i\eps)^{-1} [\Dslash,\pi^\th(e)]
(\Dslash^2+\eps^2)^{-1} \pi^\th(\ee_{4N}) (\Dslash^2+\eps^2)^{-1}
\pi^\th(\ee_{8N}) \cdot\cdot (\Dslash^2+\eps^2)^{-1} \pi^\th(\ee_{4N^2}).
\end{align*}

The last term is trace-class because it is a product of $N$ terms in
$\L^p(\H)$ for $p > N$ and one term in $\L^q(\H)$ for $q > 2N$, by
Lemma~\ref{lm:schatten}. Removing the second $\pi^\th(e)$ once again,
by the ordering property of the local units $\ee_K$ yields
\begin{align*}
&\pi^\th(e) (\Dslash+i\eps)^{-1} \pi^\th(e) (\Dslash-i\eps)^{-1}
\pi^\th(\ee_{4N}) (\Dslash^2+\eps^2)^{-1} \pi^\th(\ee_{8N}) \cdots
(\Dslash^2+\eps^2)^{-1} \pi^\th(\ee_{4N^2})
\\
&= \pi^\th(e) (\Dslash^2+\eps^2)^{-1} \pi^\th(e)
(\Dslash^2+\eps^2)^{-1} \pi^\th(\ee_{8N}) \cdots
(\Dslash^2+\eps^2)^{-1} \pi^\th(\ee_{4N^2})
\\
&\quad + \pi^\th(e) (\Dslash^2+\eps^2)^{-1} [\Dslash,\pi^\th(e)]
(\Dslash-i\eps)^{-1} \pi^\th(\ee_{4N}) (\Dslash^2+\eps^2)^{-1}
\pi^\th(\ee_{8N}) \cdot\cdot (\Dslash^2+\eps^2)^{-1} \pi^\th(\ee_{4N^2}).
\end{align*}
The last term is still trace-class, hence
$$
\pi^\th(e) (\Dslash^2+\eps^2)^{-N} \sim \pi^\th(e) (\Dslash^2+\eps^2)^{-1} \pi^\th(e)
(\Dslash^2+\eps^2)^{-1} \pi^\th(\ee_{8N}) \cdots (\Dslash^2+\eps^2)^{-1} \pi^\th(\ee_{4N^2}).
$$
This algorithm, applied another $(N-1)$ times, yields the result:
\begin{align*}
\pi^\th(e) (\Dslash^2+\eps^2)^{-N} \sim \bigl( \pi^\th(e) (\Dslash^2+\eps^2)^{-1} \pi^\th(e) \bigr)^N.
\end{align*}
\end{proof}

We retain the following consequence.

\begin{cly}
\label{cr:killproj}
$\Tr^+\bigl(\pi^\th(g)\,[\pi^\th(f),(\Dslash^2+\eps^2)^{-N}]\bigr) = 0$ for any $g \in \SS$ and any projector $f \in \A_c$.
\end{cly}

\begin{proof}
This follows from Lemma~\ref{lm:traceclass} applied to $\pi^\th(f)\,(\Dslash^2+\eps^2)^{-N}$ and its adjoint.
\end{proof}

Now we are finally ready to evaluate the Dixmier traces.

\begin{prop}
\label{pr:calcul}
For $f \in \SS$, any Dixmier trace $\Tr^+$ of
$\pi^\th(f)\,(\Dslash^2 + \eps^2)^{-N}$ is independent of~$\eps$, and
$$
\Tr^+ \bigl( \pi^\th(f)\,(\Dslash^2 + \eps^2)^{-N} \bigr)
= \tfrac{2^N\,\Omega_{2N}}{2N\,(2\pi)^{2N}} \int f(x) \,d^{2N}x
= \tfrac{1}{N!\,(2\pi)^N} \int f(x) \,d^{2N}x.
$$
\end{prop}

\begin{proof}
We will first prove it for $f \in \A_c$. Choose $e$ a unit for $f$,
that is, $e \mop f = f \mop e = f$. By Lemmata~\ref{lm:existence}
and~\ref{lm:traceclass}, and because $\L^1(\H)$ lies inside the kernel
of the Dixmier trace, we obtain
$$
\Tr^+( \pi^\th(f) \,(\Dslash^2+\eps^2)^{-N}) = \Tr^+ \bigl(
\pi^\th(f) \,(\pi^\th(e) (\Dslash^2+\eps^2)^{-1} \pi^\th(e))^N \bigr).
$$
Lemma~\ref{lm:traceclass} applied to $f = e$ implies that
$\bigl( \pi^\th(e) (\Dslash^2+\eps^2)^{-1} \pi^\th(e) \bigr)^N$ is a
positive operator in $\L^{1+}(\H)$, since it is equal to
$\pi^\th(e) (\Dslash^2+\eps^2)^{-N}$ plus a term in $\L^1(\H)$. Thus,
\cite[Thm.~5.6]{CareyPS} yields (since the limit converges, any
Dixmier trace will give the same result):
\begin{align}
\Tr^+ \bigl( \pi^\th(f) \,(\Dslash^2+\eps^2)^{-N} \bigr)
&= \lim_{s\downarrow 1} (s-1) \Tr\bigl[ \pi^\th(f)\,
(\pi^\th(e) (\Dslash^2+\eps^2)^{-1} \pi^\th(e))^{Ns} \bigr]
\nonumber \\
&= \lim_{s\downarrow 1} (s-1) \Tr\bigl( \pi^\th(f) \pi^\th(e)
(\Dslash^2+\eps^2)^{-Ns} \pi^\th(e) + E_{Ns} \bigr),
\label{eq:decompo}
\end{align}
where
$$
E_{Ns} := \pi^\th(f) \,
\bigl( \pi^\th(e) (\Dslash^2+\eps^2)^{-1} \pi^\th(e) \bigr)^{Ns}
- \pi^\th(f)\pi^\th(e) (\Dslash^2+\eps^2)^{-Ns} \pi^\th(e).
$$
Lemma~\ref{lm:traceclass} again shows that $E_N \in \L^1(\H)$.

Now for $s > 1$,  the first term $\pi^\th(f) \,\bigl(\pi^\th(e)
(\Dslash^2+\eps^2)^{-1} \pi^\th(e) \bigr)^{Ns}$ of $E_{Ns}$ is in
$\L^1(\H)$. In effect, using Lemma~\ref{lm:schatten} and since
$\pi^\th(e) (\Dslash^2+\eps^2)^{-1} \in \L^p(\H)$ for $p > N$, we get that 
$\pi^\th(e) (\Dslash^2+\eps^2)^{-1} \pi^\th(e) \in \L^{Ns}(\H)$. This
operator being positive, one concludes
$$
\bigl( \pi^\th(e) (\Dslash^2+\eps^2)^{-1} \pi^\th(e) \bigr)^{Ns}
\in \L^1(\H).
$$
The second term
$\pi^\th(f) \pi^\th(e) (\Dslash^2+\eps^2)^{-Ns} \pi^\th(e)$ lies in
$\L^1(\H)$ too, because
$$
\|\pi^\th(e) (\Dslash^2+\eps^2)^{-Ns} \pi^\th(e)\|_1
= \|(\Dslash^2+\eps^2)^{-Ns/2} \pi^\th(e)\|_2^2
= \|\pi^\th(e) (\Dslash^2+\eps^2)^{-Ns/2}\|_2^2
$$
is finite by Lemma~\ref{lm:HSO}. So $E_{Ns} \in \L^1(\H)$ for
$s \geq 1$, and \eqref{eq:decompo} implies
\begin{align*}
\Tr^+ \bigl( \pi^\th(f) \,(\Dslash^2+\eps^2)^{-N} \bigr)
&= \lim_{s\downarrow 1} (s-1) \Tr \bigl(
\pi^\th(f) \pi^\th(e) (\Dslash^2+\eps^2)^{-Ns} \pi^\th(e) \bigr)
\\
&= \lim_{s\downarrow 1}
(s-1) \Tr \bigl(\pi^\th(f) (\Dslash^2+\eps^2)^{-Ns} \bigr).
\end{align*}

Applying now Lemma~\ref{lm:in-extremis}, we obtain
\begin{align*}
\Tr^+ \bigl( \pi^\th(f) \,(\Dslash^2+\eps^2)^{-N} \bigr)
&= \lim_{s\downarrow 1} (s-1) \Tr(1_{2^N}) \Tr\bigl(L^\th_f (-\Delta +\eps^2)^{-Ns}\bigr)\\
&= 2^N(2\pi)^{-2N} \lim_{s\downarrow 1} (s-1)
\iint f(x)\, (|\xi|^2+\eps^2)^{-Ns} \,d^{2N}\xi \,d^{2N}x\\
&= \tfrac{1}{N!\,(2\pi)^N} \int f(x) \,d^{2N}x,
\end{align*}
where the identity
$$
\int (|\xi|^2+\eps^2)^{-Ns} \,d^{2N}\xi
= \pi^N \tfrac{\Ga(N(s-1))}{\Ga(Ns)\,\eps^{2N(s-1)}},
$$
and $\Ga(N\a) \sim 1/N\a \as \a \downarrow 0$ have been used. The
proposition is proved for $f \in \A_c$.

Finally, take $f$ arbitrary in $\SS$, and recall that $\{e_K\}$ is an
approximate unit for $\A_\th$. Since $f = g \mop h$ for some
$g,h \in \SS$, Corollary~\ref{cr:killproj} implies
\begin{align*}
\bigl| \Tr^+ &\bigl( (\pi^\th(f) - \pi^\th(\ee_K \mop f \mop \ee_K))
(\Dslash^2+\eps^2)^{-N} \bigr) \bigr|
\\
&= \bigl| \Tr^+ \bigl( (\pi^\th(f) - \pi^\th(\ee_K \mop f))\,
(\Dslash^2+\eps^2)^{-N} \bigr) \bigr|
\\
&= \bigl| \Tr^+ \bigl( (\pi^\th(g) - \pi^\th(\ee_K \mop g))\,
\pi^\th(h) (\Dslash^2+\eps^2)^{-N} \bigr) \bigr|
\\
&\leq \| \pi^\th(g) -\pi^\th(\ee_K \mop g) \|_\infty \,
\Tr^+ \bigl| \pi^\th(h)\,(\Dslash^2+\eps^2)^{-N} \bigr|.
\end{align*}
Since $\|\pi^\th(g) -\pi^\th(\ee_K \mop g)\|_\infty \leq (2\pi\th)^{-N/2}\|g - \ee_K \mop g\|_2$ tends to zero when $K$
increases, the proof is complete because $\ee_K\mop f\mop \ee_K$ lies in $\A_c$ and
\begin{align*}
\int[\ee_K\mop f\mop \ee_K](x)\,d^{2N}x \to \int f(x)\,d^{2N}x\as K\uparrow\infty.
\end{align*}
\end{proof}

\begin{rem}
Similar arguments to those of this section (or a simple comparison 
argument) show that for
$f\in\SS$,
$$
\Tr^+ \bigl( \pi^\th(f)\,(|\Dslash| + \eps)^{-2N} \bigr)
= \Tr^+ \bigl( \pi^\th(f)\,(\Dslash^2 + \eps^2)^{-N} \bigr).
$$
\end{rem}

In conclusion: the analytical and spectral dimension of Moyal planes coincide. And 
Lemma~\ref{lm:existence}, Proposition~\ref{pr:calcul} and the previous remark have concluded the proof of
Theorem~\ref{th:dim-spectrale}. 

\vspace{1cm}

The conclusion is that $(\A,\Aun,\H,\Dslash,\chi,J)$ defines a non-compact spectral triple; recall that we already 
know that both $\A$ and its preferred compactification $\Aun$ are pre-$C^*$-algebras.

\begin{thm}
\label{th:main}
The Moyal planes $(\A,\Aun,\H,\Dslash,J,\chi)$ are connected real non-compact spectral triples of spectral 
dimension~$2N$.
\end{thm}

One can compute the Yang--Mills action \eqref{YMaction} of this triple:

\begin{thm}
Let $\omega = -\omega^* \in \Omega^1\A_\th$. Then the Yang--Mills action
$YM(V)$ of the universal connection $\delta + \omega$, with
$V= \tilde\pi^\th(\omega)$, is equal to
$$
S_{YM}(V) = -\tfrac{1}{4g^2} \int F^{\mu\nu} \mop F_{\mu\nu}(x) \,d^{2N}x
= -\tfrac{1}{4g^2} \int F^{\mu\nu}(x) \,F_{\mu\nu}(x) \,d^{2N}x,
$$
where $F_{\mu\nu} :=
\half(\del_\mu A_\nu - \del_\nu A_\mu + [A_\mu,A_\nu]_\mop)$ and
$A_\mu$ is defined by $V = L^\th(A_\mu) \ox \ga^\mu$.
\end{thm}

The spectral action has been computed in \cite{GI,GIVas}. As one can expected, it is the same, up to few universal 
coefficients, to the one of Theorem \ref{main}.

\section*{Acknowledgments}

I would like to thank Driss Essouabri, Victor Gayral, Jos\'e Gracia-Bond\'{\i}a, Cyril Levy, Pierre Martinetti, 
Thierry Masson, Thomas Sch\"ucker, Andrzej Sitarz, Jo V\'arilly and Dmitri Vassilevich, for our discussions during 
our collaborations along years. Some of the results presented here are directly extracted from these collaborations. 
I also took benefits from the questions of participants (professors and students) of the school during the lectures, 
to clarify few points.

\end{document}